\PassOptionsToPackage{hypertexnames=false}{hyperref}
\PassOptionsToPackage{inline}{enumitem}
\PassOptionsToPackage{clock}{ifsym}
\documentclass{theoretics}

\title{DeepSec: Deciding Equivalence Properties for Security Protocols --- Improved theory and practice}

\ThCSauthor[1]{Vincent Cheval}{vincent.cheval@cs.ox.ac.uk}[0000-0002-3622-2129]
\ThCSauthor[2]{Steve Kremer}{steve.kremer@inria.fr}[0009-0004-6946-0678]
\ThCSauthor[3]{Itsaka Rakotonirina}{itsaka.rakotonirina@mpi-sp.org}[0000-0002-6587-971X]
\ThCSaffil[1]{University of Oxford, United Kingdom}
\ThCSaffil[2]{Inria Centre at Université de Lorraine, France}
\ThCSaffil[3]{MPI-SP, Germany}

\ThCSshortnames{V. Cheval, S. Kremer, I. Rakotonirina}
\ThCSshorttitle{DeepSec: Deciding Equivalence Properties for Security Protocols}

\ThCSyear{2024}
\ThCSarticlenum{4}
\ThCSreceived{Nov 14, 2022}
\ThCSdoi{10.46298/theoretics.24.4}
\ThCSrevised{Dec 12, 2023}
\ThCSaccepted{Jan 28, 2024}
\ThCSpublished{Mar 12, 2024}

\ThCSkeywords{Verification, cryptographic protocols, process equivalences}

\ThCSthanks{A preliminary version of this work appeared at IEEE Symposium on Security and Privacy (S\&P) 2018~\cite{CKR18}.
This work has been partly supported by the ANR Research and teaching chair in AI ASAP (ANR-20- CHIA-0024) with support from the region Grand Est and ANR France 2030 project SVP (ANR-22-PECY-0006).}

\usepackage[english]{babel}
\usepackage{csquotes}
\usepackage[T1]{fontenc}
\usepackage{stmaryrd} 
\usepackage{xspace}
\usepackage{multirow}
\usepackage{xifthen}
\usepackage{tikz}
\usepackage{booktabs}
\usepackage{rotating}
\usetikzlibrary{shapes.misc}
\usetikzlibrary{shapes.geometric}
\usetikzlibrary{shapes}
\usetikzlibrary{arrows,automata}
\usetikzlibrary{circuits.logic.US}
\usetikzlibrary{arrows}
\usetikzlibrary{decorations.pathmorphing}

\usepackage{colortbl}
\definecolor{light}{RGB}{230,230,230}
\newcommand{\light}[1]{\cellcolor{light}#1}

\RequirePackage{cleveref}

\LinesNotNumbered

 \usepackage{mdframed}
\mdfdefinestyle{leftbar}{
  leftmargin = 3pt,
  innerleftmargin = 8pt,
  innertopmargin = 0pt,
  innerbottommargin = 0pt,
  innerrightmargin = 0pt,
  rightmargin = 0pt,
  linewidth = 0.6pt,
  topline = false,
  rightline = false,
  bottomline = false,
  linecolor=gray,
  skipabove=1.2\smallskipamount,
}

\mdfdefinestyle{margindefstyle}{
	leftmargin = 7pt,
  skipabove = \baselineskip,
  skipbelow = \baselineskip,
  innerleftmargin = 0.5em,
  innertopmargin = 0em,
  innerbottommargin = 2pt,
  innerrightmargin = 0em,
  rightmargin = 0pt,
  linewidth = 1.5pt,
  topline = false,
  rightline = false,
  bottomline = false,
}

\newcommand\problemdescr[3][]{
  \ifthenelse{\isempty{#1}}
  {
    \begin{itemize}[itemsep=0pt]
      \item[\(\triangleright\)] \textsc{Input:} {#2}
      \item[\(\triangleright\)] \textsc{Question:} {#3}
    \end{itemize}
  }
  {

    \noindent #1:
    \begin{itemize}[topsep=0pt,itemsep=0pt]
      \item[\(\triangleright\)] \textsc{Input:} {#2}
      \item[\(\triangleright\)] \textsc{Question:} {#3}
    \end{itemize}
  }
}

\usepackage{mathpartir}
\usepackage{pifont} 

\newcommand{\mvs}{\fontfamily{mvs}\fontencoding{U}%
\fontseries{m}\fontshape{n}\selectfont}
\newcommand{\mvchr}[1]{{\mvs\char#1}}
\newcommand\Lightning{\mvchr{69}}

\newcommand{\downsquigarrow}{\mathbin{\text{\rotatebox[origin=c]{270}{\(\rightsquigarrow\)}}}}
\makeatletter
  \newcommand{\xRightarrow}[2][]{\ext@arrow 0359\Rightarrowfill@{#1}{#2}}
  \newcommand{\xleftrightarrow}[2][]{\ext@arrow 0359\leftrightarrowfill@{#1}{#2}}
\makeatother
\newcommand{\eqdef}{\triangleq}

\newcommand{\eint}[2]{\llbracket #1,#2\rrbracket}
\newcommand{\N}{\mathbb N}
\newcommand{\R}{\mathcal R}

\let\emptyset\varnothing

\let\epsilon\varepsilon

\let\rho\varrho

\newcommand\complexityfont{\small \sffamily}
\newcommand\nexp{{\complexityfont NEXP}\xspace}
\newcommand\np{{\complexityfont NP}\xspace}
\newcommand\logspace{{\complexityfont LOGSPACE}\xspace}

\newcommand\pspace{{\complexityfont PSPACE}\xspace}
\newcommand\ptime{{\complexityfont PTIME}\xspace}
\newcommand\ctime{{\complexityfont TIME}\xspace}
\newcommand\cspace{{\complexityfont SPACE}\xspace}
\newcommand\polyh[1]{\(\Pi_{#1}\)\xspace}

\newcommand\problemfont{\small \scshape}

\newcommand\sat{{\problemfont SAT}\xspace}
\newcommand\sucsat{{\problemfont SuccinctSAT}\xspace}
\newcommand\qbf[1][]{{\problemfont QBF}\(_{#1}\)\xspace}

\newcommand\StaticEquiv{{\problemfont StatEq}\xspace}
\newcommand\TraceEquiv{{\problemfont TraceEq}\xspace}
\newcommand\TraceInclus{{\problemfont TraceIncl}\xspace}
\newcommand\Bisimilarity{{\problemfont Bisimilarity}\xspace}
\newcommand\Similarity{{\problemfont Similarity}\xspace}
\newcommand\Simulation{{\problemfont Simulation}\xspace}

\newcommand\toolfont{\scshape}
\newcommand\proverif{{\toolfont ProVerif}\xspace}
\newcommand\tamarin{{\toolfont Tamarin}\xspace}
\newcommand\deepsec{{\toolfont DeepSec}\xspace}
\newcommand\akiss{{\toolfont Akiss}\xspace}
\newcommand\spec{{\toolfont Spec}\xspace}
\newcommand\satequiv{{\toolfont SatEquiv}\xspace}
\newcommand\adecs{{\toolfont ADECS}\xspace}
\newcommand\apte{{\toolfont Apte}\xspace}
\newcommand\avispa{{\toolfont AVISPA}\xspace}
\newcommand\scyther{{\toolfont Scyther}\xspace}

\newcommand\ukano{{\toolfont Ukano}\xspace}

\newcommand\senc{\mathsf{senc}}
\newcommand\sdec{\mathsf{sdec}}
\newcommand\aenc{\mathsf{aenc}}
\newcommand\adec{\mathsf{adec}}
\newcommand\testaenc{\mathsf{test\_aenc}}
\newcommand\getkey{\mathsf{get\_key}}
\newcommand\rsenc{\senc}
\newcommand\rsdec{\sdec}
\newcommand\raenc{\aenc}
\newcommand\radec{\adec}
\newcommand\pk{\mathsf{pk}}
\newcommand\pks{\mathit{pk}}
\newcommand\sk{\mathit{sk}}
\newcommand\sks{\mathit{sk}}
\newcommand\vpk{\mathsf{vpk}}
\newcommand\sign{\mathsf{sign}}
\newcommand\checksign{\mathsf{verify}}

\newcommand\fst{\mathsf {fst}}
\newcommand\snd{\mathsf {snd}}
\newcommand\okfun{\mathsf{ok}}

\newcommand\0{\mathsf 0}
\newcommand\1{\mathsf 1}
\newcommand\ffun{\mathsf{f}}
\newcommand\gfun{\mathsf{g}}
\newcommand\hfun{\mathsf{h}}

\newcommand\otherfun{\mathsf{getOther}}

\newcommand\args[2]{\overrightarrow{#1}^{#2}}
\newcommand\call[1]{{\tt Goto}\left\langle#1\right\rangle}
\newcommand\getEnv[1]{{\tt GetEnv}\left\langle#1\right\rangle}

\newcommand{\msg}{\mathsf {msg}}
\newcommand\norm{\!\!\downarrow}
\newcommand\ax{\mathsf{ax}}
\newcommand\InP[2]{{#1}(#2)}
\newcommand\OutP[2]{\overline {#1}\langle #2 \rangle}
\newcommand\LetP{\mathsf{let}}

\newcommand\IfP{\mathsf{if}}
\newcommand\ThenP{\mathsf{then}}
\newcommand\ElseP{\mathsf{else}}

\newcommand\BangP[1][]{!^{\scriptscriptstyle #1}}

\newcommand\multi[1]{\{\!\!\{ #1\}\!\!\}}
\newcommand\pair[1]{\langle #1\rangle}
\newcommand\sig{\mathcal {F}}
\newcommand\sigc{\sig_{\mathsf {c}}}
\newcommand\sigd{\sig_{\mathsf {d}}}
\newcommand\termset{\mathcal T}
\newcommand\recipeset{\termset^2}
\newcommand\Nall{\mathcal N}

\renewcommand\P{\mathcal P}
\newcommand\Q{\mathcal Q}
\newcommand\X[1][]{\mathcal X^{#1}}
\newcommand\Xfst{\X[1]}

\newcommand\Xsndi[2][\leqslant]{\X[2]_{\raisebox{1pt}{\(\scriptscriptstyle #1\)} #2}}
\newcommand\AX{\mathcal {AX}}
\newcommand\dom{\mathit{dom}}
\newcommand\im{\mathit{img}}
\newcommand\id{\mathit{id}}
\newcommand\rootf{\mathit{root}}
\newcommand\ptree{\mathsf{PTree}}
\newcommand\poslab{\mathsf{pos}}
\newcommand\neglab{\mathsf{neg}}
\newcommand\witness{\mathsf{w}}
\newcommand\fsol{f_{\mathsf{sol}}}
\newcommand\inp{\mathsf{in}}
\newcommand\outp{\mathsf{out}}
\newcommand\process[2]{(#1,#2)}
\newcommand\frameh[1]{\Phi_{#1}}
\newcommand\frameb[1]{\Phi_{#1}^\B}
\newcommand\framen[1]{\Phi_{#1}^{\mathsf N}}

\newcommand\tsize[1]{\left|#1\right|}
\newcommand\dagsize[1]{\left|#1\right|_{\mathsf{dag}}}
\newcommand\subterms[1][]{\mathit{st}^{#1}}
\newcommand\strsubterms[1][]{\mathit{sst}^{#1}}
\newcommand\vars[1][]{\mathit{vars}^{#1}}
\newcommand\varsfst{\vars[1]}

\newcommand\axioms{\mathit{axioms}}
\newcommand\names{\mathit{names}}

\newcommand\simpStep[1]{\mathrel{\xrightarrow{#1}\!\!\simplnorm}}
\newcommand\SimpStep[1]{\mathrel{\xRightarrow{#1}\!\!\simplnorm}}

\newcommand\cstep[1]{\xrightarrow{#1}}
\newcommand\Cstep[1]{\xRightarrow{#1}}
\newcommand\silentstep\rightsquigarrow
\newcommand\silent{\overset\star\rightsquigarrow}
\newcommand\silentrev{\overset\star\leftsquigarrow}
\newcommand\silentpistep{\silentstep_{\mathsf{pi}}}
\newcommand\silentpi{\overset\star\silentstep_{\mathsf{pi}}}
\newcommand\procnorm[1]{{#1}_{\downarrow_{\mathsf{pi}}}}

\newcommand\tr{\mathsf {tr}}
\newcommand\StatEq{\sim}

\newcommand\LabBis{\approx_b}
\newcommand\Simu{\sqsubseteq_s}
\newcommand\Simuinv{\sqsupseteq_s}
\newcommand\Simi{\approx_s}
\newcommand\TraceEq{\approx_t}
\newcommand\TraceIncl{\sqsubseteq_t}
\newcommand\A{\mathcal A}
\newcommand\simpl\rightsquigarrow
\newcommand\simplnorm{\mbox{\scalebox{0.9}{$\downsquigarrow$}}}

\newcommand\val{\mathsf{val}}
\newcommand\TestNode{\mathsf{TestNode}}
\newcommand\TestBool{\mathsf{TestBool}}
\newcommand\Node{\mathsf{Node}}

\newcommand\hNode{\hfun_{\mathsf{N}}}
\newcommand\hBool{\hfun_{\mathbb{B}}}
\newcommand\invN{\mathsf{Test_N}}
\newcommand\invB{\mathsf{Test_B}}
\newcommand\recpos[2]{{#1}_{|#2}}
\newcommand\CheckTree[1]{\mbox{\texttt{CheckTree}}(#1)}
\newcommand\CheckSat[1]{\mbox{\texttt{CheckSat}}(#1)}
\newcommand\guessBinary[1]{{\tt Choose}(#1)}
\newcommand\evalFormula[2]{#1 \leftarrow #2}

\newcommand\sem[1]{\left\llbracket #1\right\rrbracket}
\newcommand\B{\mathbb B}

\newcommand\C{\mathcal C}
\newcommand\D{\mathcal D}

\newcommand\dedfact{\vdash^{\scriptscriptstyle ?}}
\newcommand\ndedfact{\nvdash^{\scriptscriptstyle ?}}

\newcommand\eqs{=^?} 
\newcommand\neqs{\neq^?}
\newcommand\eqf{=^?_f}

\newcommand\maxarity[1]{\#(#1)}

\newcommand\mgu{\mathit{mgu}}
\newcommand\mguR[1][\R]{\mgu_{#1}}
\newcommand\eqnset{\mathcal E}
\newcommand{\replacepos}[3]{#1[#3]_{#2}}

\newcommand\Df{\mathsf{D}}
\newcommand\Eq{\mathsf{E}}
\newcommand\Eqn[1][]{\Eq^{#1}}
\newcommand\Eqfst{\Eqn[1]}
\newcommand\Eqsnd{\Eqn[2]}
\newcommand\Solved{\mathsf{K}}
\newcommand\USolved{\mathsf{F}}
\newcommand\equality[1]{#1_{=}}

\newcommand\bisim{\mathrel{\mathcal{R}}}
\newcommand\disim{\mathrel{\mathcal{S}}}
\newcommand\cs{Constraint Solving}

\newcommand\mgs[1][]{\mathit{mgs}^{#1}}
\newcommand\Sol[1][]{\mathit{Sol}^{#1}}

\newcommand\freshlab{\mathit{fresh}}
\newcommand\sstep[1]{\xrightarrow{#1}_{\mathsf{s}}}
\newcommand\Sstep[1]{\xRightarrow{#1}_{\mathsf{s}}}
\newcommand\tstep[1]{\xrightarrow{#1}_{T}}
\newcommand\Tstep[1]{\xRightarrow{#1}_{T}}
\renewcommand\S{\mathbb {S}}
\newcommand\clause[3][]{
  \ifthenelse{\isempty{#1}}
    {#2 \Leftarrow #3}
    {\forall #1.\ #2 \Leftarrow #3}
}
\newcommand\conseq{\mathsf{Conseq}}
\newcommand\stc{\mathit{st}_{\mathsf{c}}}
\newcommand\getpos[2]{{#1}_{|#2}}
\newcommand\recipes{\mathsf{R}} 
\newcommand\quanti[2]{#1\text{:}#2}
\newcommand\FApply[3]{#2\text{:}(#1,#3)}
\newcommand\CApply[2]{#2\text{:}#1}
\newcommand\Fhyp{\mathsf{hyp}}

\newcommand\receq{\simeq_{\mathsf{r}}}
\newcommand\RewF[3]{\mathsf{RewF}(#1,#2,#3)}
\newcommand\Skel[2]{\mathsf{Skel}(#1,#2)}

\newcommand\measureNC{\mathit{unused}^1}
\newcommand\terms[1][]{T^{#1}}
\newcommand\compon[1][]{\mathit{M}_{#1}}
\newcommand\setSDF{\mathsf{set}_{\Solved}}
\newcommand\setRew{\mathsf{set}_{\textsc{Rew}}}
\newcommand\setEq{\mathsf{set}_{\textsc{Eq}}}
\newcommand\defcomp[2]{%
	\[\tag{Meas. #1}
	\compon[#1](\Gamma) = #2\]
}
\newcommand\CompatibleSubs{\mathsf{CompatSubs}}

\newcommand\predlab{\mathsf{Inv}}
\newcommand\PredAll{\predlab_{\mathit{all}}}
\newcommand\PredWellFormed{\predlab_{\mathit{wf}}}
\newcommand\PredCorrectFormula{\predlab_{\mathit{sound}}}
\newcommand\PredCompleteFormula{\predlab_{\mathit{comp}}}
\newcommand\PredConseq{\predlab_{\mathit{satur}}}
\newcommand\PredSymb{\predlab_{\mathit{sol}}}
\newcommand\PredStruct{\predlab_{\mathit{str}}}

\newcommand\simplifstep{\overset {\mathsf{simpl}} \simpl}
\newcommand\normstep{\overset {\mathsf{norm}} \simpl}
\newcommand\vectstep{\overset {\mathsf{vect}} \simpl}
\newcommand\satstep{\xrightarrow{\ref{rule:satisfiable}}}
\newcommand\eqstep{\xrightarrow{\ref{rule:equality}}}
\newcommand\rewstep{\xrightarrow{\ref{rule:rewrite}}}

\SetKwProg{Def}{\textbf{Procedure}}{ =}{\textbf{End}}
\SetKwProg{If}{\textbf{if}}{\ \textbf{then}}{}
\SetKwProg{Else}{\textbf{else}}{}{}
\SetKwProg{ElseIf}{\textbf{else if}}{\ \textbf{then}}{}
\SetKwProg{For}{\textbf{for}}{\ \textbf{do}}{\textbf{done}}
\SetKwProg{While}{\textbf{while}}{\ \textbf{do}}{\textbf{done}}
\newcommand\returnkw{\mathsf{return}}
\newcommand\algocomment[1]{{\color{gray}\textsf{/\!/\ #1}}}

\newcommand{\caseitem}[1]{\begin{itemize}
  \item[\(\triangleright\)] #1 
\end{itemize}}

\definecolor{greenpigment}{rgb}{0.0,0.65,0.31}
\definecolor{cadiumred}{rgb}{0.89,0.0,0.13}
\definecolor{cadiumorange}{rgb}{0.93,0.53,0.18}
\definecolor{camouflagegreen}{rgb}{0.47,0.53,0.42}
\definecolor{ceruleanblue}{rgb}{0.16, 0.32, 0.75}
\newcommand\warningsign{{\color{ceruleanblue} \raisebox{-1pt}{\scalebox{1.2}{\Lightning}}}}
\newcommand{\outoftime}{\raisebox{-2pt}{\scalebox{0.8}{\StopWatchEnd}}}
\newcommand{\outofmemory}{\raisebox{-1pt}{\scalebox{0.35}{%
	\tikz[baseline=-2ex]{
	\node[color=cadiumred,fill = white, shape=regular polygon, minimum size = 1cm, regular polygon sides=8, inner sep=0pt, draw, thick] at (0,0) {};
	\node[color=cadiumred,shape=regular polygon, text = cadiumred, regular polygon sides=8, inner sep=0pt, draw, thick] at (0,0) {\small\textbf{OM}};
	}}}}

\newcommand{\attacksimple}{\warningsign}
\newcommand{\unable}{{\color{cadiumred}\ding{55}}}
\newcommand{\verified}{{\color{greenpigment}\ding{51}}}

\newenvironment{bigproof}[1][\proofname]
  {\begin{proof}[#1]~

  \noindent}
  {\end{proof}}

\makeatletter
\providecommand{\leftsquigarrow}{%
\mathrel{\mathpalette\reflect@squig\relax}%
}
\newcommand{\reflect@squig}[2]{%
\reflectbox{$\m@th#1\rightsquigarrow$}%
}
\makeatother

\newcommand\step[2][]{{\noindent{#2} {#1}:\xspace}}
\newcommand\case[2][]{\step[#1]{\bf case #2}}
\newcommand{\sbt}{\,\begin{picture}(-1,1)(-1,-2.5)\circle*{2.5}\end{picture}\ \ }

\begin{document}

\maketitle

\begin{abstract}
  Automated verification has become an essential part in the security evaluation of cryptographic protocols.
  In this context privacy-type properties are often modelled by indistinguishability statements, expressed as behavioural equivalences in a process calculus.
  In this article we contribute both to the theory and practice of this verification problem.
  We establish new complexity results for static equivalence, trace equivalence and labelled bisimilarity and provide a decision procedure for these equivalences in the case of a bounded number of protocol sessions.
  Our procedure is the first to decide trace equivalence and labelled bisimilarity exactly for a large variety of cryptographic primitives---those that can be represented by a subterm convergent destructor rewrite system.
  We also implemented the procedure in a new tool, \deepsec.
  We showed through extensive experiments that it is significantly more efficient than other similar tools, while at the same time raising the scope of the protocols that can be analysed.
\end{abstract}


\section{Introduction}

The use of automated, formal methods has become indispensable for analysing complex security protocols, such as those for authentication, key exchange and secure channel establishment.
Nowadays there exist mature, fully automated such analysers; among others \avispa \cite{ABB05}, \proverif \cite{B16}, \scyther \cite{C08}, \tamarin \cite{SMC13} or Maude-NPA \cite{SEM14}.
These tools are able to automatically verify full fledged models of widely deployed protocols and standards, such as the TLS protocol for secure connexion \cite{BBK17,CHH17}, the Signal messaging protocol \cite{KBB17,CCG18}, authentication protocols of the 5G standard \cite{BDH18}, or deployed multi-factor authentication protocols \cite{JK18}.
Theory-wise, the tools operate in so-called \emph{symbolic} models, rooted in the seminal work by Dolev and Yao \cite{DY81}: 
the attacker has full control over the communication network, unbounded computational power, but cryptography is idealised.
This model is well suited for finding attacks in the protocol logic, and tools have indeed been extremely effective in discovering this kind of flaw or proving their absence.

While most works investigate \emph{reachability} properties, a later trend consists in adapting the tools---and the underlying theory---to the more complex \emph{indistinguishability} properties.
Such properties are generally modelled as a behavioural equivalence (bisimulation or trace equivalence) in a dedicated process calculus such as the spi calculus \cite{AG99} or the applied pi calculus \cite{ABF17}.
A typical example is real-or-random secrecy: after interacting with a protocol, an adversary is unable to distinguish the \emph{real} secret used in the protocol from a \emph{random} value.
Privacy-type properties can also be expressed as such:
anonymity may be modelled as the adversary's inability to distinguish two instances of a protocol executed by different agents;
vote privacy \cite{DKR09} has been expressed as indistinguishability of the situations where the votes of two agents have been swapped or not;
unlinkability \cite{ACR10} is seen as indistinguishability of two sessions, either both executed by the same agent \(A\), or by two different agents \(A\) and \(B\).

\subsection*{Contributions}
  We significantly improve the theoretical understanding
  and the practical verification of equivalences when the number of protocol sessions is bounded.
  We emphasise that even in this setting, the system under study has an infinite state space due to the term algebra modelling cryptographic primitives.
  Our work targets the wide class of cryptographic primitives that can be represented by a subterm convergent rewriting system.
  Concretely, we provide:

  \begin{enumerate}
    \item tight complexity results for several equivalence relations:
    static equivalence, trace equivalence and labelled bisimilarity.
    In addition to the conference paper~\cite{CKR18}, we showcase the generality of our approach by providing, with a negligible proof overhead, a tight analysis of other security relations, namely similarity, simulation, and trace inclusion;
    \item a novel procedure deciding all of the above mentioned security relations for a bounded number of sessions, for the class of cryptographic primitives modelled by a destructor subterm convergent rewrite system;
    \item an implementation of our procedure for trace equivalence in a tool called \deepsec (DEciding Equivalence Properties for SECurity protocols), improved compared to its initial presentation in the conference paper~\cite{CKR18}.
  \end{enumerate}

  \noindent
  We detail the three contributions below.


\paragraph{Complexity} 
  We provide the first complexity results for deciding trace equivalence and labelled bisimilarity in the applied pi calculus, without any syntactic or semantic restriction on the class of protocols (other than bounding the number of sessions), and for a large class of cryptographic primitives modelled as rewrite rules. 
  As mentioned above, our results extend to several other security relations such as simulation.
  Let us also highlight one small, yet substantial difference with existing work: we do not consider cryptographic primitives (rewrite systems) as constants of the problem.
  As most modern verification tools allow for user-specified primitives \cite{manual-proverif,SMC13,SEM14,CCC16}, our approach seems to better fit this reality.
  Typically, all existing procedures for static equivalence can only be claimed \ptime because of this difference and are actually exponential in the sizes of the signature or equational theory.
  Our complexity results are summarised in Table~\ref{fig:summary}.
  All our lower bounds hold for subterm convergent rewrite systems and even for the positive fragment (without \(\ElseP\) branches).
  \textit{En passant}, we present results for the pi calculus:
  although investigated in \cite{BT00}, complexity was unknown when restricted to a bounded number of sessions.
  Still, our main result is the co\nexp completeness (and in particular, the decidability) of trace equivalence and labelled bisimilarity for destructor subterm convergent rewrite systems.

  \begin{table}[ht]
    \centering
    \begin{tabular}{|c|c|c|}
      \hline
      & \multirow{2}{*}{Pure pi calculus} & Applied pi calculus\\
      & & \small with destr. subterm convergent theory\\
      \hline
      static equivalence & \logspace & co\np complete\\
      \hline
      trace equivalence & \polyh 2 complete & co\nexp complete \\
      \hline
      labelled (bi)similarity & \pspace complete & co\nexp complete \\
      \hline
    \end{tabular}
    \caption{Summary of complexity results}
    \label{fig:summary}
  \end{table}

\paragraph{Decision procedure} 
  We present a novel procedure based on a symbolic semantics and constraint solving.
  Unlike most other work, our procedure decides equivalences \emph{exactly}, i.e. without approximations.
  Moreover, it does not restrict the class of processes (except for replication), nor the use of \(\ElseP\) branches, and is correct for any cryptographic primitives that can be modelled by a subterm convergent destructor rewrite system (see Section~\ref{sec:model} for more details).
  The design of the procedure did greatly benefit from our complexity study, and was developed in order to obtain tight complexity upper bounds.
  The theory is also more mature compared to the initial conference paper \cite{CKR18} which allowed some significant optimisations of the constraint solving procedure.

\paragraph{Tool implementation}
  We implemented our procedure for trace equivalence in a tool, \deepsec.
  Its prototype has initially been presented in the conference paper \cite{CKR18} (and some implementation details in a tool paper \cite{CKR18b}), but has significantly matured since then.
  In addition of the improvements at the level of the theoretical procedure, the low level implementation has been more carefully engineered data-structure-wise.
  All in all, the \deepsec 2.0.0 release includes the following new features:
  \begin{itemize}
    \item A significantly \emph{reduced verification time} (several orders of magnitude on some examples).
    \item An \emph{optional procedure exploiting the symmetries} that often arise in practical verification.
    When used, this further reduces the verification time by orders of magnitude albeit for occasionally introducing false attacks.
    In this article we rather focus on the main procedure;
    details about this feature can be found in \cite{CKR19}.
    \item An \emph{improved user experience}.
    The html based pretty-print of the original prototype has been upgraded into a standalone graphical user interface.
    Verification queries and options can be managed directly from the interface and a simulator displays interactively equivalence proofs or attacks to better visualise the outcome of the analysis.
  \end{itemize}

  \noindent
  Naturally \deepsec still integrates already-present features such as multicore distribution and the partial order reductions presented in \cite{BDH15}.
  All in all this makes the tool more user friendly and scale well despite the high theoretical complexity of the problem (co\nexp).
  Installation guidelines can be found in the official website \cite{website} together with a manual and a tutorial.

  Through extensive benchmarks, we compare \deepsec to other tools limited to a bounded number of protocol sessions: \apte, \spec, \akiss, \satequiv and our previous prototype (as presented in \cite{CKR18}).
  This prior version was already more efficient---by several orders of magnitude---than \apte, \spec and \akiss, even though \deepsec covers a strictly larger class of protocols than \apte and \spec.
  Besides, its performances were comparable to \satequiv, which still outperforms \deepsec when the number of parallel processes significantly increase.
  This gap in performance seems unavoidable as \deepsec operates on a much larger class of protocols (more primitives, \(\ElseP\) branches, no limitation to simple processes, termination guaranteed).
  Part of the benchmarks consists of classical authentication protocols and focuses on demonstrating scalability of the tool when augmenting the number of parallel protocol sessions.
  The other examples include more complex protocols, such as Abadi and Fournet's anonymous authentication protocol \cite{AF04}, the protocols implemented in the European passport \cite{P04}, a model (without XOR) of the AKA protocol used in 3G mobile telephony, as well as the Pr\^et-\`a-Voter \cite{RS06} and the Helios \cite{A08} e-voting protocols.

\subsection*{Related Work}
  The problem of analysing security protocols is undecidable in general but several decidable subclasses have been identified.
  While many complexity results are known for trace properties \cite{DLM04,RT03}, the case of behavioural equivalences remains mostly open.
  When the attacker is an eavesdropper and cannot interact with the protocol, the indistinguishability problem---\emph{static equivalence}---has been shown \ptime for large classes of cryptographic primitives \cite{AC06,CDK12,CBC11}.
  For active attackers, bounding the number of protocol sessions is often sufficient to obtain decidability \cite{RT03} and is of practical interest: most real-life attacks indeed only require a small number of sessions.
  In this context Baudet \cite{B05}, and later Chevalier and Rusinowitch \cite{CR10}, showed that real-or-random secrecy was co\np for cryptographic primitives that can be modelled as subterm convergent rewrite systems, by checking whether two constraint systems admit the same set of solutions.
  These procedures do however not allow for \(\ElseP\) branches, nor do they verify trace equivalence in full generality.
  In \cite{CCD13}, Cheval et al. have used Baudet's procedure as a black box to verify trace equivalence of \emph{determinate} processes.
  This class of processes is however insufficient for most anonymity properties.
  Finally, decidability results for an unbounded number of sessions exist \cite{CCD15,CCD15b}, but with severe restrictions on processes and equational theories.

  Tool support also exists for verifying equivalence properties.
  We start discussing tools that are limited to a bounded number of sessions.
  The \spec tool \cite{TD10, TNH16} verifies a sound symbolic bisimulation, but is restricted to particular cryptographic primitives (pairing, encryption, signatures and hash functions) and does not allow for \(\ElseP\) branches. In a similar setting, restricting to particular primitives, Cheval et al.~\cite{CCD10} propose a procedure for deciding equivalence of constraint systems. This procedure can be used for deciding trace equivalence of determinate processes and has been implemented in the ADECS tool.
  The \apte tool \cite{C14} generalizes \adecs: it covers the same primitives but allows \(\ElseP\) branches and decides trace equivalence exactly.
  On the contrary, the \akiss tool \cite{CCC16} allows for user-defined cryptographic primitives.
  The procedure of this tool is correct for primitives modelled by an arbitrary convergent rewrite system that has the finite variant property \cite{CD05}, and termination is additionally guaranteed for subterm convergent rewrite systems.
  However, \akiss does only decide trace equivalence for a class of determinate processes; for other processes trace equivalence can be both over- and under-approximated.
  The recent \satequiv tool \cite{CDD17} uses a different approach: it relies on Graph Planning and SAT solving to verify trace equivalence, rather than a dedicated procedure.
  The tool is extremely efficient and several orders of magnitude faster than other tools.
  It does however not guarantee termination and is currently restricted to pairing and symmetric encryption and only considers a class of \emph{simple processes} (a subclass of determinate processes) that satisfy a type-compliance condition.
  These restrictions severely limit its scope.

  To mitigate the state explosion problem from which most of the above tools suffer, Baelde et al.~\cite{BDH15} developed \emph{partial order techniques} which avoid to explicitly consider all possible interleavings and which are compatible with a symbolic approach based on constrained solving. Substantial efficiency gains on practical examples have been illustrated through an implementation in the \apte tool. We also implemented these techniques in \deepsec. However, the techniques may only be applied on a class of \emph{action-determinate} processes. This limitation has been overcome in a follow-up work by Baelde et al.~\cite{BDH18}: while more general the new techniques also require additional, expensive computations resulting in less spectacular performance increase than the initial work. Baelde et al.~\cite{BDH18}. have implemented their technique in a standalone library, and plugged it into the \apte and \deepsec tools.

  Other tools support verification of equivalence properties, even for an unbounded number of sessions.
  This is the case of \proverif \cite{BAF08}, \tamarin \cite{BDS15} and Maude NPA \cite{SEM14} which all allow for user-defined cryptographic primitives.
  However, given that the underlying problem is undecidable, these tools may not terminate.
  Moreover, they only approximate trace equivalence by verifying the stronger \emph{diff-equivalence}.
  This equivalence is too strong on many examples.
  While some recent improvements on \proverif \cite{CB13,BS16} help covering more protocols, general verification of trace equivalence is still out of scope.
  For instance, the verification by Arapinis et al. \cite{AMR12} of unlinkability in the 3G mobile phone protocols required some ``tricks'' and approximations of the protocol to avoid false attacks.
  In \cite{CGL17}, Cortier et al. develop a type system and automated type checker for verifying equivalences.
  While extremely efficient, this tool only covers a fixed set of cryptographic primitives (the same as \spec and \apte) and verifies an approximated equivalence, similar to diff-equivalence.
  A different approach has been taken by Hirschi et al. \cite{HBD16}, identifying sufficient conditions provable by \proverif for verifying unlinkability properties, implemented in the tool \ukano, a front-end to the \proverif tool.
  \ukano does however not verify equivalence properties in general.

\subsection*{Article Outline}
We organize the article as follows. In \Cref{sec:model} we present our formal model of cryptographic protocols and the process equivalences used to express security properties. We also precisely define the decision problems that we address in this article. 

In \Cref{sec:symbolic} we provide an overview of our decision procedures. First, we define (sound and complete) symbolic semantics where we replace the infinite set of possible attacker inputs by a finite representation in the form of constraint systems. Second,  we define the notion of a \emph{partition tree}. The partition tree organizes all symbolic traces in a tree such that a node contains (a symbolic representation of) all statically equivalent processes that can be reached by a given trace. Third, we show how equivalences can be decided on such a partition tree. Next, we explain how to compute a partition tree \emph{assuming} we can compute solutions to constraint systems. Finally, we discuss how the procedure for deciding trace equivalence has been implemented in the \deepsec tool and provide a performance evaluation.

In~\Cref{sec:ptree} we present a rule-based procedure to effectively solve constraint systems. This requires the definition of \emph{extended} constraint systems that store additional information and the introduction of the new notion of \emph{most general solutions}. Reminiscent of the notion of most general unifiers, most general solutions are a set of solutions that guarantee that any solution can be obtained from a most general solution by substituting atomic names by more complex terms. After presenting all the rules of the procedure in detail we explain how to construct a partition tree.

In~\Cref{sec:complex} we give complexity results. To achieve upper bounds we prove termination of the constraint solving procedure and exponentially bound the number of rules and size. From these bounds we obtain that when two processes are not equivalent (for different notions of equivalence) there exists a witness of exponential size, yielding a co\nexp decision procedure for equivalence. Lower bounds are provided by reduction to the \sucsat problem.

Finally we conclude the article in \Cref{sec:conclusion} and sketch some directions for future work.


\section{Model} 
\label{sec:model}


We first present our model of cryptographic protocols and use it to
model the security of the Private Authentication Protocol as a running
example~\cite{AF04}.  Our framework is based on the applied pi
calculus~\cite{ABF17} and follows the tradition of symbolic models
rooted in the seminal work of Dolev and Yao \cite{DY81}.  In these
models, the low-level details of cryptography are abstracted by a term
algebra describing the ideal behaviour of cryptographic primitives,
whereas secret data such as cryptographic keys or nonces are
represented by symbolic values called \emph{names}.

\subsection{Messages and cryptography}

\paragraph{Protocol messages}
  Cryptographic operations are modelled by a set \(\sig\) of symbols of fixed arity denoted \(\sig = \{\ffun/n, \gfun/m, \ldots\}\), called a \emph{signature}.
  In this article, it is always partitioned into:
  \begin{itemize}
    \item The infinite set of \emph{constants} (\(\sig_0\)) that are the functions of arity 0 of \(\sig\), thus modelling the public values of the protocol such as identities, IP addresses or public communication channels.
    \item The finite set of \emph{constructors} (\(\sigc\)) modelling cryptographic operations used to build messages, typically encryption, signature, concatenation or hash.
    \item The finite set of \emph{destructors} (\(\sigd\)) modelling inversions or operations that may fail depending on the structure of their argument, typically decryption, signature verification or projection.
  \end{itemize}

  \begin{example}
    \label{ex:standard-signature}
    The following signature captures most of the cryptographic primitives that are used in our examples and benchmarks.
    We will use them throughout Section~\ref{sec:model} in examples.
    \begingroup
      \renewcommand*\box[1]{\parbox[top][1cm][c]{2.5cm}{\centering \footnotesize #1}}
      \renewcommand*\arraystretch{1.3}
      \[\begin{array}{c@{\qquad}ccccc}
        & \box{concatenation / pairs}
        & \box{symmetric encryption}
        & \box{asymmetric encryption}
        & \box{digital signature}
        & \box{one-way hash}
        \\\hline
        \sigc
          & \pair {\cdot, \cdot}/2
          & \senc/3
          & \pk/1, \aenc/3
          & \vpk/1, \sign/3
          & \hfun/1
        \\\hline
        \sigd
          & \fst/1, \snd/1
          & \sdec/2
          & \adec/2
          & \checksign/2
          & \emptyset
      \end{array}\]
    \endgroup
    
    For example \(\aenc(m,r,\pk(k))\) models a plaintext \(m\) encrypted with public key \(\pk(k)\) and a randomness \(r\).
    The corresponding decryption key would be \(k\).
    A similar description can be made for symmetric encryption, except that the encryption and decryption keys are identical.
    The model of hash functions contains no destructors on purpose, thus modelling an assumption that \(\hfun\) is a random oracle, i.e., no identities can be derived from \(\hfun\).
    Notation-wise, we also often use a tuple notation \(\pair{x_1, \ldots, x_n}\) instead of the nested \(n-1\) pairs \(\pair{x_1, \pair {x_2, \ldots \pair {x_{n-1}, x_n}}}\). 
    %
  \end{example}

  A protocol message \(m\) is then modelled by a \emph{term} over this signature, i.e. \(m\) is obtained by applying function symbols to other terms or \emph{names}.
  The infinite set of \emph{names} \(\Nall\) can be seen as a symbolic abstraction of private values such as encryption keys or nonces.
  The set of names occurring in a term \(t\) is written \(\names(t)\).
  In some models
  names are partitioned into public and private names, where the set of public names essentially plays the same role as~\(\sig_0\).
  Since constants and public names have a similar role (and are even treated identically in our tool implementation) we decided to merge them into the single set \(\sig_0\) similarly to other formalisations, e.g. \cite{CCD15b}.
  We write \(\termset(S)\), \(S \subseteq \sig \cup \Nall\), the set of terms built from functions, names, and constants of \(S\).

\paragraph{Specifying cryptographic assumptions}
  The behaviour of the primitives of the signature is modelled by a \emph{rewriting system}.
  For that we assume an infinite set of \emph{variables} \(\X = \{x, y, z, \ldots\}\) that may be used in terms, and write \(\vars(t)\) the set of variables occurring in a term \(t\).
  Mappings \(\sigma\) from variables to terms are called \emph{substitutions} and are homomorphically extended to mappings from terms to terms implicitly.
  We use the postfix notation \(t\sigma\) for \(\sigma(t)\), and \(\sigma\sigma'\) for the composition of subtitution \(\sigma'\circ\sigma\) (that is, \(t \sigma \sigma' = (t\sigma) \sigma'\)).
  We call the \emph{domain} of \(\sigma\) the set \(\dom(\sigma) = \{x \in \X \mid x \sigma \neq x\}\).
  For convenience we also use set notations, defining a substitution~\(\sigma\) such that \(\dom(\sigma) \subseteq \{x_1, \ldots, x_n\}\) with the notation
  \(\sigma = \{x_1 \mapsto \sigma(x_1), \ldots, x_n \mapsto \sigma(x_n)\}\).
  Going further we may refer to the substitution \(\sigma \cup \sigma'\) (provided \(\sigma\) and \(\sigma'\) coincide on \(\dom(\sigma) \cap \dom(\sigma')\)) or write \(\sigma \subseteq \sigma'\) to mean that \(\sigma'\) extends \(\sigma\).
  A \emph{rewriting system} \(\R\) is then a finite binary relation on terms.
  All pairs of \(\R\) are called \emph{rewrite rules} and are assumed to be of the form
  \begin{align*}
    f(\ell_1, \ldots, \ell_n) & \to r & \mbox{for some }\ f/n\, \in \sigd\ \mbox{ and }\ \ell_1, \ldots, \ell_n, r \in \termset(\sig_c \cup \sig_0 \cup \X)
  \end{align*}
  Such rewriting systems are usually qualified as \emph{constructor destructor} in the literature.
  By extension we also use notation \(t \to s\) (``\emph{\(t\) rewrites to \(s\)}'') when \(t\) and \(s\) are related by the closure of \(\R\) under application of substitution and term context.
  The reflexive transitive closure of this relation is written \(\to^*\).

  \begin{example} \label{ex:standard-theory}
    We give the rewrite rules for the primitives introduced in Example~\ref{ex:standard-signature}.
    \begin{align*}
      & \text{\em sym. encryption:} & &
        \sdec(\senc(x,y,z),z) \to x \\
      & \text{\em pairs:} & &
        \fst(\pair {x,y}) \to x \quad \text{\em and} \quad \snd(\pair{x,y}) \to y \\
      & \text{\em asym. encryption:} & &
        \adec(\aenc(x,y,\pk(z)),z) \to x \\
      & \text{\em signatures:} & &
        \checksign(\sign(x,y,z), \vpk(z)) \to x
    \end{align*}
    For example here one can decrypt (apply \(\adec\)) a ciphertext \(\aenc(x,y,\pk(z))\) with the corresponding key \(z\) to recover the plaintext \(x\).
    The rule for signature verification is the opposite, recovering the signed message \(x\) using the public verification key \(\vpk(z)\).
    The behaviour of these primitives is idealised by the absence of other rules, for example modelling an assumption that no information can be extracted from a ciphertext or a signature without access to the secret or verification keys.
    This idealisation can be partially lifted by adding more rewrite rules modelling specific imperfections of the cryptography.
    For example we can add the following new symbols and rewrite rules:
    \begin{align*}
      \testaenc(\aenc(x,y,\pk(z))) & \to \okfun &
      \getkey(\aenc(x,y,\pk(z))) & \to \pk(z)
    \end{align*}
    model two assumptions that
    \begin{enumerate*}
      \item it is possible to distinguish a correctly encrypted message from a random bitstring, and
      \item it is possible to retrieve the encryption key from the ciphertext itself (i.e. the scheme is not \emph{key concealing}).
    \end{enumerate*}
    Naturally even if a protocol is considered secure without these two rewrite rules, a security violation may arise upon adding them.
    It is therefore important to keep in mind the assumptions underlying the model when interpreting the result of an analysis.
  \end{example}

  We observe that the rewrite rules introduced in the example above verify a classical property, \emph{subterm convergence}, introduced in \cite{AC06} and benefiting from several decidability results in the context of protocol analysis \cite{AC06,CCD13}.
  It means that \(\R\) is convergent (i.e. confluent and strongly terminating) and that its rules \(\ell \to r\) verify that \(r\) is either a strict subterm of \(\ell\) or a \emph{ground term} (i.e. a term without variables) in \emph{normal form} (i.e. irreducible w.r.t. \(\to\)).
  The results of this article only apply to cryptographic primitives modelled by a constructor destructor subterm convergent rewriting systems.
  Imposing such restrictions is inevitable when aiming for decidability, since the problems we investigate are undecidable for arbitrary convergent rewriting systems \cite{AC06}.

  In particular, by convergence, all terms \(t \in \termset(\sig \cup \Nall)\) have a unique normal form w.r.t.~\(\R\) that we will write \(t \norm\).
  It is also common to identify messages whose destructors failed to be applied.
  For that we define a predicate \(\msg\) on terms:
  we say that \(t\) is a \emph{message}, written \(\msg(t)\), when for all subterms \(u\) of \(t\), \(u\norm\) does not contain any destructors.
  For example if \(m , r \in \sig_0\) and \(k \neq k'\), \(\adec(\aenc(m,r,\pk(k)), k)\) is a message but not \(\fst(\pair {m, t})\) with \(t = \sdec(\senc(m,r,k),k')\).

\subsection{Protocols} \label{sec:processes}

\paragraph{Processes}
  Security protocols are modelled by \emph{(plain) processes} in a concurrent process calculus defined by the following grammar:
  \[\begin{array}{rl@{\qquad}r}
    P,Q := & 0 & \text{null}\\
    & P \mid Q & \text{parallel}\\
    & \IfP\ u = v\ \ThenP\ P\ \ElseP\ Q& \text{conditional}\\
    & \OutP u v.P & \text{output}\\
    & \InP u x.P & \text{input}\\
    \end{array}
  \]
  where \(u,v\) are terms and \(x \in \X\).
  Intuitively the \(0\) models a terminated process (and is often omitted for succinctness), a conditional \(\IfP\ u = v\ \ThenP\ P\ \ElseP\ Q \) executes either \(P\) or \(Q\) depending on whether the terms \(u\) and \(v\) are messages and have the same normal form, and \(P \mid Q\) models two concurrent processes.
  Inter-process communications are performed with \(\InP {u} {x}.P\) and \(\OutP {u} {v}.P\) which are, respectively, inputs and outputs on a communication channel \(u\).
  When \(u\) is known to the attacker, for example when it belongs to \(\sig_0\), executing an output on \(u\) adds it to the adversary's knowledge, whereas an input on \(u\) is fetched from the adversary possibly forwarding a previously stored message, or computing a new message from previous outputs.
  Otherwise the communication is performed silently without adversarial interferences.
  The main difference with the calculus of \cite{ABF17} is the absence of replication, thus bounding the number of instructions of a process.
  This restriction does \emph{not} make protocol analysis trivially decidable:
  although the number of instructions are finite, the number of their possible executions is not, since the attacker can fetch arbitrary messages to public inputs.

  \begin{example} \label{ex:process}
    We define a process modelling the protocol for private authentication described in \cite{AF04} as a running example through the article.
    Denoting by \(\sk_X,\pks_X\) the secret and public keys of an agent
    \(X\), and by \(r_X\) fresh nonces, its control flow can be described as follows using an informal Alice-Bob notation:
    \begin{align*}
      X \to B:\ & \aenc(\pair {N_X, \pks_X}, r_X, \pks_B) \\
      B \to X:\ & \aenc(\pair {N_X , N_B, \pks_B}, r_A, \pks_A) && \mbox{if \(X = A\)}\\
      & \aenc(N_B,r_B,\pks_B) && \mbox{if the decryption fails or \(X \neq A\)}
    \end{align*}
    where \(N_X,N_B\) are two freshly generated nonces.
    Here the agent \(B\) accepts authentication requests from the agent \(A\) but not from other parties.
    Among the security goals stated in \cite{AF04} are
    \begin{enumerate}
      \item \emph{Secrecy:}
      At the end of a successful instance of the protocol between \(A\) and \(B\), \(N_A\) and \(N_B\) are secrets (i.e. the attacker cannot get information about them).
      \item \emph{Anonymity:}
      The attacker cannot tell whether the protocol is run by \(A\) and \(B\) or other agents.
      \item \emph{Private authentication:}
      The attacker cannot tell whether \(B\) accepts connections from \(A\) or not.
    \end{enumerate}
    The last two security goals explain in particular the decoy message \(\aenc(N_B,r_B,\pks_B)\) that \(B\) sends upon decryption failure or connection refusal:
    thus from an outside observer there is no observable difference between the situations where \(B\) answers or not.
    The roles of \(X\) and \(B\) can be specified as follows in the applied pi calculus;
    each process takes as an argument its secret key \(s\), the public key \(p\) of the agent it aims at communicating with, its fresh session nonces \(n,r\) and we write \(t = \adec(x,s)\):
    \[\begin{array}{l@{\ }l@{\qquad}l@{\ }l}
      X(s, p, n, r) = & \OutP c {\aenc(\pair {n,\pk(s)},r,p)}. &
      B(s, p, n, r) = & \InP c x.\\
      & \InP {c} {x} & & \IfP\, \snd(t) = p\, \ThenP\\
      & & & \phantom{\ElseP}\, \OutP c {\aenc(\pair {\fst(t), n, \pk(s)}, r, p)}\\
      & & & \ElseP\, \OutP c {\aenc(n,r,\pk(s))}
    \end{array}\]
    where \(c \in \sig_0\).
    The security goals are formalised in Section~\ref{sec:equivalence}.
  \end{example}

  \paragraph{Attacker's knowledge}
    In the next paragraphs we formalise how processes may be executed in an active adversarial environment.
    The first step is to model the capabilities of the underlying attacker that spies on the communication network and actively interferes with communications.
    For that we refine the set of variables to \(\X = \Xfst \uplus \AX\), thus introducing a new type of variables \(\AX = \{\ax_1, \ax_2, \ax_3, \ldots\}\) called \emph{axioms} that will serve as handles to make reference to attacker's observations.
    Concretely a term \(\xi \in \termset(\sig \cup \AX)\) is called a \emph{recipe} and is intuitively an algorithm for the attacker to construct a term from their prior observations.
    For example upon observing the messages \(\aenc(m,r,\pk(k))\) and \(k\) in this order, an attacker can use the recipe \(\xi = \adec(\ax_1,\ax_2)\) to retrieve \(m\) although it has not been observed directly.
    We observe in particular that by definition a recipe cannot contain names, modelling that they are assumed to be private and as such cannot be used directly by the adversary.
    
    On the other hand, the variables of \(\Xfst\), called \emph{first-order variables} for distinction, stick to the initial role of variables---namely, being used as binders for protocol inputs.
    For this reason, we call a term \(t \in \termset(\sig \cup \Nall \cup \Xfst)\) a \emph{protocol term}.
    However, we often more specifically consider \emph{constructor terms} \(\termset(\sigc \cup \sig_0 \cup \Nall \cup \Xfst)\) that are protocol terms whose destructors have all been successfully computed, that is, reduced by a rewrite rule.
    We also write \(\varsfst(t) = \vars(t) \cap \Xfst\) the set of \emph{first-order variables} of \(t\).
    Using all these notions we define \emph{extended processes},
    representing a set of processes executed in parallel together
    with the knowledge aggregated by the attacker interacting with the protocol:

    \begin{definition}
      \label{def:extended process}
      An extended process is a pair \(A = (\P,\Phi)\) with \(\P\) a multiset of ground processes and \(\Phi = \{ \ax_1 \mapsto u_1, \ldots, \ax_n \mapsto u_n\} = \Phi(A)\) is called a \emph{frame} that is a substitution from axioms to ground constructor terms.
    \end{definition}

    Formalising the example above, if the frame \(\Phi = \{\ax_1 \mapsto \aenc(m,r,\pk(k)), \ax_2 \mapsto k\}\) models the attacker's observations during the execution of a protocol, the fact that \(m\) can be retrieved with the recipe \(\xi = \adec(\ax_1,\ax_2)\) is expressed by the fact that \(\msg(\xi \Phi)\) and \(\xi \Phi \norm = m\).
    A typical security problem is to decide, given a frame \(\Phi\) and a term \(t\), whether \(t\) is \emph{deducible} by the attacker from \(\Phi\);
    that is, whether there exists a recipe \(\xi\) such that \(\msg(\xi \Phi)\) and \(\xi \Phi \norm = t \norm\).

  \paragraph{Operational semantics}
    We now formalise the semantics of processes.
    By manipulating extended processes this semantics carries the knowledge the attacker aggregates by spying on the communication outputs. 
    Besides, in our constructor destructor setting we assume that the agents only send and accept meaningful messages, namely terms that verify the \(\msg\) predicate.
    While this assumption is realistic for authenticated encryption for example, it may not hold for schemes with weaker security guarantees.
    In practice the semantics takes the form of a transition relation
    between extended processes labelled by so-called \emph{actions}:
    \begin{enumerate}[label=\emph{\arabic*.}]
      \item \emph{Input actions} \(\InP {\xi_c} {\xi_t}\), where \(\xi_c\) and \(\xi_t\) are recipes, model an input from the attacker of a message (crafted using recipe \(\xi_t\)) on some channel (known to the attacker using recipe \(\xi_c\))
      \item \emph{Output actions} \(\OutP {\xi_c} {\ax_n}\), where \(\xi_c\) is a recipe, model an output on a channel (known by the attacker using recipe \(\xi_c\)), recorded into the frame (at pointer \(\ax_n \in \AX\)).
      \item \emph{Silent actions} \(\tau\) that model actions that
        are unobservable by the attacker such as synchronous private
        communications or evaluation of a conditional.
    \end{enumerate}

    We call \(\A\) the alphabet of actions, and transitions are of the form \(A \cstep {a} B\), \(a \in \A\). The transition relation is defined by the rules given in Figure~\ref{fig:semantics}.
    More generally:
    
    \begin{definition}[trace]
      We write \(A \Cstep{w} B\) when $A \cstep{a_1} \ldots \cstep{a_n} B$ and \(w \in \A^*\) is the word obtained after removing the \(\tau\) actions from the word \(a_1 \cdots a_n\), and call such a sequence of transitions a \emph{trace}.
      We also write $\Cstep{\tau}$ for $\cstep{\tau}^*$, i.e. the reflexive, transitive closure of $\cstep{\tau}$.
    \end{definition}

    \begin{figure*}[ht]
      \begingroup
        \newcommand\skipspace{\hspace{3cm}}
        \centering
        \begin{align}
          \tag{\mbox{\textsc{In}}} \label{rule:in}
          & (\multi {\InP {u} {x}.P} \cup \P, \Phi)
            \cstep {\InP {\xi_c} {\xi_t}} (\multi {P\{x \mapsto \xi_t \Phi \norm\}} \cup \P, \Phi) &
          & \mbox{\small if \(\msg(\xi_c \Phi)\), \(\msg(\xi_t \Phi)\), \(\msg(u)\)} \\
          \nonumber
          & & & \mbox{\small and \(\xi_c \Phi \norm = u \norm\)}\\
          \tag{\mbox{\textsc{Out}}} \label{rule:out}
          & (\multi {\OutP {u} {v}.P} \cup \P, \Phi)
            \cstep {\OutP {\xi_c} {\ax_n}} (\multi {P} \cup \P, \Phi \cup \{\ax_n \mapsto v\norm\}) &
          & \mbox{\small if \(\msg(\xi_c \Phi)\), \(\msg(u)\), \(\msg(v)\)} \\
          \nonumber
          & & & \mbox{\small \(\xi_c \Phi \norm = u \norm\) and \(n = |\dom(\Phi)|+1\)}\\
          \tag{\mbox{\textsc{Comm}}} \label{rule:comm}
          & (\multi {\OutP {u} {v}.P, \InP {u'} {x}.Q} \cup \P, \Phi)
            \cstep {\tau} (\multi {P, Q\{x \mapsto v\}} \cup \P, \Phi) &
          & \mbox{\small if \(\msg(u)\), \(\msg(v)\), \(\msg(u')\)} \\
          \nonumber
          & & & \mbox{\small and \(u \norm = u' \norm\)}\\
          \tag{\mbox{\textsc{Then}}} \label{rule:then}
          & (\multi {\IfP\ u = v\ \ThenP\ P\ \ElseP\ Q} \cup \P, \Phi)
            \cstep {\tau} (\multi {P} \cup \P, \Phi) &
          & \mbox{\small if \(\msg(u)\), \(\msg(v)\) and \(u \norm = v \norm\)}\\
          \tag{\mbox{\textsc{Else}}} \label{rule:else}
          & (\multi {\IfP\ u = v\ \ThenP\ P\ \ElseP\ Q} \cup \P, \Phi)
            \cstep {\tau} (\multi {Q} \cup \P, \Phi) &
          & \mbox{\small if \(\neg \msg(u)\), \(\neg \msg(v)\) or \(u \norm \neq v \norm\)}\\
          \tag{\mbox{\textsc{Par}}} \label{rule:par}
          & (\multi {P \mid Q} \cup \P, \Phi)
            \cstep {\tau} (\multi {P, Q} \cup \P, \Phi)
        \end{align}
      \endgroup
      \caption{Semantics of the calculus}
      \label{fig:semantics}
    \end{figure*}

    Apart from the absence of replication, this semantics aims at being as close as possible to the original semantics of the applied pi calculus \cite{ABF17} although using a different formalism, as it is also the semantics used by tools such as \proverif. 

    \begin{example} \label{ex:semantics}
      We now illustrate how our running example can be executed in the operational semantics.
      We let two agents of respective secret keys \(\sk_A,\sk_B \in \Nall\).
      An instance of the protocol between \(A\) and \(B\) is thus modelled, using the notations of Example~\ref{ex:process}, by the process \(P = \bar {A} \mid \bar{B}\) where, given fresh names \(r_A,r_B\):
      \begin{align*}
        \bar{A} & = X(\sk_A, \pk(\sk_B), N_A, r_A) &
        \bar{B} & = B(\sk_B, \pk(\sk_A), N_B, r_B)
      \end{align*}
      In order to lighten the presentation we use the same notations as in Example~\ref{ex:process} and name the three messages of the protocol as follows:
      \begin{align*}
        m_A & = \aenc(\pair {N_A, \pk(\sk_A)}, r_A, \pk(\sk_B)) \\
        m_B & = \aenc(\pair {\fst(t), N_B, \pk(\sk_B)}, r_B, \pk(\sk_A)) \\
        m_B' & = \aenc(N_B, r_B, \pk(\sk_B))
      \end{align*}
      We assume that the public keys \(\pk(\sk_A)\) and \(\pk(\sk_B)\) are known to the attacker, which can be modelled by an initial frame \(\Phi_0 = \{\ax_1 \mapsto \pk(\sk_A), \ax_2 \mapsto \pk(\sk_B)\}\).
      Another possibility is to prefix the process \(P\) with two outputs of \(\pk(\sk_A)\) and \(\pk(\sk_B)\) respectively, which will produce the frame \(\Phi_0\) after two applications of rule \eqref{rule:out}.
      The normal execution of the process is the following sequence of reduction steps:
      \[\begin{array}{lll}
        (\multi {P}, \Phi_0)
          & \cstep {\tau} (\multi {\bar {A}, \bar{B}}, \Phi_0) \\
          & \cstep {\OutP {c} {\ax_3}} (\multi {\InP {c} {x}, \bar {B}}, \Phi_1)
            & \mbox{with } \Phi_1 = \Phi_0 \cup \{\ax_3 \mapsto m_A\} \\
          & \Cstep {\InP {c} {\ax_3}} (\multi {\InP {c} {x}, \OutP {c} {m_B}}, \Phi_1) \\
          & \cstep {\OutP {c} {\ax_4}} (\multi {\InP {c} {x},0}, \Phi_2)
            & \mbox{with } \Phi_2 = \Phi_1 \cup \{\ax_4 \mapsto m_B\} \\
          & \cstep {\InP {c} {\ax_4}} (\multi {0,0}, \Phi_2)
      \end{array}\]
      In this execution the attacker only forwards messages, that is, each input action uses the last axiom added to the frame as a recipe.
      However the adversary may actively engage in the protocol, for example for guessing whether \(B\) accepts communications from a third agent \(C\).
      For that they could generate fresh nonces \(N,R \in \sig_0\)
      (attacker-generated nonces are modelled by fresh constants)
      and send the message \(m_A' = \aenc(\pair {N, \pk(\sk_C)}, R, 
      \pk(\sk_B))\)  to check how~\(B\) responds.
      Note that the message \(m_A'\) can indeed be crafted by the attacker assuming \(\Phi_0' = \Phi_0 \cup \{\ax_3 \mapsto \pk(\sk_C)\}\) as an initial frame.
      This scenario corresponds to the following sequence of transitions:
      \[\begin{array}{lll}
        (\multi {P}, \Phi_0')
          & \cstep {\tau} (\multi {\bar {A}, \bar {B}}, \Phi_0') \\
          & \Cstep {\InP {c} {\aenc(\pair {N, \ax_3}, R, \ax_2)}} (\multi {\bar {A}, \OutP {c} {m_B'}}, \Phi_0') \\
          & \cstep {\OutP {c} {\ax_4}} (\multi {\bar {A}, 0}, \Phi_0' \cup \{\ax_4 \mapsto m_B'\})
      \end{array}\]
      This does not leak information to the attacker, assuming they cannot distinguish the messages \(m_B\) and \(m_B'\).
      All in all, the set of \emph{traces} of the process, i.e. the set of all possible sequences of reductions, characterises all possible executions of the protocol in an active adversarial environment.
    \end{example}

    As a final note, let us observe that the original pi calculus \cite{MPW92} (referred as the \emph{pure pi calculus} in this article) can be seen as a special case of our model.
    Indeed the fragment without replication is retrieved when \(\sigc\), \(\sigd\) and \(\R\) are empty.
    This restriction makes the transition relation finitely branching up to bijective renaming of attacker-generated constants.

\subsection{Security properties} \label{sec:equivalence}

  \paragraph{Against a passive attacker}
    We first define the notion of \emph{static equivalence} that is often used to model security against a passive attacker in that it is only an equivalence of frames, i.e. it does not involve the operational semantics.
    It expresses that the knowledge obtained by eavesdropping in two different situations does not permit the attacker to distinguish them.
    For example no differences can be observed between \(\{\ax_1 \mapsto k\}\) and \(\{\ax_1 \mapsto k'\}\) if \(k,k' \in \Nall\) because, intuitively, two fresh nonces look like random bitstrings from an external observer's point of view.
    However the situation is different with the frames
    \begin{align*}
      \Phi & = \{\ax_1 \mapsto k, \ax_2 \mapsto k\} &
      \Psi & = \{\ax_1 \mapsto k', \ax_2 \mapsto k\} &
      \mbox{with } & k \neq k'
    \end{align*}
    Indeed, even if no differences can be made between \(k\) and \(k'\) in isolation, the attacker observed two identical messages in the first situation but two different messages in the second situation.
    In particular we say that the equality test ``\(\ax_1 = \ax_2\)'' distinguishes the two frames (because it holds in \(\Phi\) but not in \(\Psi\)).
    Besides, in our constructor destructor algebra it is also possible to observe destructor failures.
    For example the following frames can be distinguished:
    \begin{align*}
      \Phi & = \{\ax_1 \mapsto k, \ax_2 \mapsto \aenc(m,r,\pk(k))\} &
      \Psi & = \{\ax_1 \mapsto k', \ax_2 \mapsto \aenc(m,r,\pk(k))\}
    \end{align*}
    Indeed crafting the recipe \(\adec(\ax_2, \ax_1)\) (i.e. decrypting the last observed message with the first one) succeeds in the first situation but triggers a decryption failure in the second.
    Static equivalence has been extensively studied in the literature (see e.g. \cite{AC06,CDK12,BCD13,CBC11}).
    Formally:

    \begin{definition} \label{def:static-equivalence}
      Two frames \(\Phi\) and \(\Psi\) of same domain are \emph{statically equivalent}, written \(\Phi \StatEq \Psi\), when for all recipes \(\xi,\zeta\):
      \begin{enumerate}
        \item \(\msg(\xi \Phi)\) if and only if \(\msg(\xi \Psi)\)
        \item assuming \(\msg(\xi \Phi)\) and \(\msg(\zeta \Phi)\), \(\xi \Phi \norm = \zeta \Phi \norm\) if and only if \(\xi \Psi \norm = \zeta \Psi \norm\).
      \end{enumerate}
      This definition is lifted to extended processes by writing \(A \StatEq B\) instead of \(\Phi(A) \StatEq \Phi(B)\).
    \end{definition}

    \begin{example}
      The fact that the two frames
      \begin{align*}
        \Phi & = \{\ax_1 \mapsto \aenc(m,r,\pk(k))\} &
        \Psi & = \{\ax_1 \mapsto k'\} &
        \mbox{with }\ m \in \sig_0\ \mbox{ and }\ k,k',r \in \Nall
      \end{align*}
      are statically equivalent intuitively models that encryption makes messages unintelligible (in that the attacker cannot distinguish a ciphertext from a fresh nonce).
      Naturally this does not hold any more once the decryption key is revealed.
      Formally:
      \(\Phi \cup \{\ax_2 \mapsto k\} \ \not \StatEq\ \Psi \cup \{\ax_2 \mapsto k\}\)
      as witnessed by the recipe \(\xi = \adec(\ax_1, \ax_2)\) whose computation succeeds in the first frame but triggers a decryption failure in the second.
      Without going to the extreme extent of revealing the key, the two situations are also distinguishable if we weaken the cryptographic assumptions on \(\aenc\).
      For example, recalling the considerations of
      Example~\ref{ex:standard-signature}, if we do not suppose the
      encryption scheme to be  \emph{key concealing} anymore by adding the rule
      \[\getkey(\aenc(x,y,\pk(z))) \to \pk(z)\]
      then \(\Phi\) and \(\Psi\) are distinguished by the recipe \(\getkey(\ax_1)\) whose destructor succeeds in \(\Phi\) but fails in \(\Psi\).
      The same fact would arise using the weaker rewrite rule \[\testaenc(\aenc(x,y,\pk(z))) \to \okfun\] that tests whether a bitstring is a ciphertext.
    \end{example}

  \paragraph{Against an active attacker}
    Dynamic extensions of static equivalence consider distinguishability for an attacker interacting actively with protocols.
    Consider for example a protocol modelled by a process \(P\) manipulating a nonce \(k\).
    A possible model of the secrecy of \(k\) can be formalised by a non-interference statement:
    there is no observable difference in the behaviour of the protocol when \(k\) is replaced by another term.
    In this article we study several relations modelling the underlying notion of indistinguishability.
    For completeness, we also present their associated pre-orders that can be useful modelling tools in situations where only inclusion relations are to be expressed.

    \begin{definition}[Trace equivalence] 
      If \(A\) and \(B\) are extended processes, we write \(A \TraceIncl B\) when for all traces \(A \Cstep{\tr} A'\), there exists a trace \(B \Cstep {\tr} B'\) such that \(A' \StatEq B'\).
      We say that \(A\) and \(B\) are \emph{trace equivalent}, written \(A \TraceEq B\), when \(A \TraceIncl B\) and \(B \TraceIncl A\).
    \end{definition}
    \begin{definition}[Simulation, (Bi)similarity]
      \label{def:(bi)simulation}
      A \emph{labelled simulation} (or simply \emph{simulation}) is a relation \(\R\) such that for all extended processes \(A,B\), \(A \mathrel{\R} B\) entails
      \begin{enumerate}
        \item \(A \StatEq B\)
        \item for all transitions \(A \cstep {\alpha} A'\), there exists a trace \(B \Cstep {\alpha} B'\) such that \(A' \mathrel{\R} B'\)
      \end{enumerate}
      We call \(\Simu\) (\emph{simulation preorder}) the largest simulation, and \(\Simi\) (\emph{labelled similarity}, or simply \emph{similarity}) the relation $\Simu \cap \Simuinv$.
      \emph{Bisimilarity} \(\LabBis\) is the largest symmetric simulation.
    \end{definition}

    Note in particular that
    \[\LabBis {\subset} \Simi {\subset} \TraceEq\]
    i.e. two bisimilar processes are always similar, and two similar processes are always trace equivalent.
    These equivalences are well established as means to express security properties \cite{AG99,ABF17}.
    Trace equivalence has been studied intensively for security protocols \cite{CCD11,ACK16,CCD13,CKR18} while, for example, labelled bisimilarity is used as a characterisation for \emph{observational equivalence}~\cite{ABF17}.

    Each equivalence implies slightly different adversaries. As shown in~\cite{CCD13}, \(\TraceEq\) characterizes may-testing, i.e., equivalence in the presence of an arbitrary adversarial process running in parallel. \(\LabBis\) characterizes observational equivalence~\cite{ABF17} and considers a more adaptive adversary; \(\LabBis\) was also introduced as a proof technique for may-testing in~\cite{AG99}. Finally, it was recently shown~\cite{CCK-csf22} that  \(\Simi\) characterizes 
    a may-testing equivalence in the presence of a probabilistic adversary, i.e. an adversarial process that is allowed to branch probabilistically.

    \begin{example}
      We refer again to the processes modelling the Private Authentication protocol as described in Example~\ref{ex:process}.
      We let for instance the processes \(P_a = B(\sk_B, \pk(\sk_A), N_B, r_B)\) and \(P_c = B(\sk_B, \pk(\sk_C), N_B, r_B)\) modelling the role of \(B\) accepting connections from \(A\) and \(C\), respectively.
      We want to verify whether an adversary would be able to distinguish the two situations.
      This could be modelled for example by
      \begin{align*}
        (\multi {P_a}, \Phi_0) & \TraceEq (\multi {P_c}, \Phi_0) &
        \mbox{with } \Phi_0 & = \{\ax_1 \mapsto \pk(\sk_A), \ax_2 \mapsto \pk(\sk_B), \ax_3 \mapsto \pk(\sk_C)\}
      \end{align*}
      The initial frame \(\Phi_0\) models that the attacker knows the public keys of all agents.
      It appears that this equivalence statement holds, the core argument being that for all messages \(u_1,u_2,r_1,r_2\) and \(\pks_1,\pks_2 \in \{\pk(\sk_A), \pk(\sk_B),\pk(\sk_C)\}\), the following frames are statically equivalent:
      \begin{align*}
        \Phi_0 \cup \{\ax_4 \mapsto \aenc(u_1, r_1, \pks_1)\} & &
        \Phi_0 \cup \{\ax_4 \mapsto \aenc(u_2, r_2, \pks_2)\}
      \end{align*}
      In particular this equivalence statement still holds if we weaken the cryptographic assumptions on \(\aenc\) by assuming that a ciphertext is distinguishable from an arbitrary term, which is modelled by adding the rewrite rule \(\testaenc(\aenc(x,y,\pk(z))) \to \okfun\).
      However trace equivalence is violated if we add the rule \(\getkey(\aenc(x,y, \pk(z))) \to \pk(z)\).
      A possible attack trace is, with \(N,R \in \sig_0\):
      \[\begin{array}{@{}ll@{\qquad}l@{}}
        (\multi {P_a}, \Phi_0)
          & \Cstep {\InP {c} {\aenc(\pair {N,\ax_1},R,\ax_2)}} (\multi {\OutP {c} {u}}, \Phi_0)
            & \mbox{with } u = \aenc(\pair {N, N_B, \pk(\sk_B)}, r_B, \pk(\sk_A))\\
          & \cstep {\OutP {c} {\ax_4}} (\multi {0}, \Phi)
            & \mbox{with } \Phi = \Phi_0 \cup \{\ax_4 \mapsto u\}
      \end{array}\]
      Indeed there is only one trace in the other process taking the same actions:
      \[\begin{array}{ll@{\qquad}l}
        (\multi {P_c}, \Phi_0)
          & \Cstep {\InP {c} {\aenc(\pair {N,\ax_1},R,\ax_2)}} (\multi {\OutP {c} {v}}, \Phi_0)
            & \mbox{with } v = \aenc(N_B, r_B, \pk(\sk_B))\\
          & \cstep {\OutP {c} {\ax_4}} (\multi {0}, \Psi)
            & \mbox{with } \Psi = \Phi_0 \cup \{\ax_4 \mapsto v\}
      \end{array}\]
      and \(\Phi \not\StatEq \Psi\) because the recipe \(\xi = \getkey(\ax_4)\) is evaluated to \(\pk(\sk_A)\) in \(\Phi\) and to \(\pk(\sk_B)\) in \(\Psi\).
      That is, the recipes \(\xi\) and \(\zeta = \ax_1\) are equal in \(\Phi\) but not in \(\Psi\).
    \end{example}

\paragraph{In practice: security goals for Private Authentication}

  We now demonstrate in more details how equivalence properties can be used to model security in practical scenarios through a complete case study.
  We model the three security goals of the Private Authentication Protocol described in Example~\ref{ex:process}.
  For simplicity we present the simplest scenario of a single session of the protocol in this section (i.e. only one instance of the roles of \(A\) and \(B\) communicating in parallel).
  Of course a more extensive analysis needs to consider more parallel sessions.
  In the following we write \(\pks_X\) and \(\sk_X\) the public and private keys of an identity \(X\) and
  \[P(A,C,N_A,r_A,B,D,N_B,r_B) =
    X(\sk_A, \pks_C, N_A, r_A) \mid B(\sk_B, \pks_D, N_B, r_B)\]
  the process that runs in parallel the roles of \(A\) attempting to initiate a communication with \(C\) and \(B\) accepting a connection from a unique identity \(D\).
  We assume an initial frame \(\Phi_0\) that contains the public keys of all identities involved in the process.

    The security goals state that the protocol should conceal the identities of the participants (including \(C\) the recipient of \(A\) and \(D\) the connection accepted by \(B\)) and the values of the exchanged nonces.
    A possible formalisation is that there should not be any observable difference in \(P(A,C,N_A,r_A,B,D,N_B,r_B)\) when replacing the identities by others and \(N_A,N_B\) by any other value.
    That is, for all identities \(A,B,C,D,A',B',C',D'\), all terms
    \(N_A,N_B,N_A',N_B'\), and fresh names \(r_A,r_B\),
    \[(\multi {P(A,C,N_A,r_A,B,D,N_B,r_B)}, \Phi_0) \approx (\multi {P(A',C',N_A',r_A,B',D',N_B',r_B)}, \Phi_0)\]
    where \(\approx\) is either \(\TraceEq\), \(\Simi\) or \(\LabBis\) and
    \(\Phi_0\) is a frame whose image contains the public keys of all indentities involved.
    This models a form of non-interference property and has been called \emph{strong secrecy} in \cite{B04}. 

\subsection{Complexity and decision problems} \label{sec:complexity}
So far we detailed how process equivalences can be used to model privacy preservation in security protocols.
Our goal in this article is to present decidability and complexity results for static equivalence, trace equivalence and labelled bisimilarity.

\paragraph{On sizes}
  Before going further we need to clarify the notion of \emph{size} of the inputs since it plays a central role in complexity analyses.
  This is particularly important for our purpose since there exist several conventions for representing terms.
  The \emph{tree size} of term \(t\) refers to its number of symbols and is written \(\tsize {t}\).
  It corresponds to a classical representation of a term as a tree.
  On the other hand some of our complexity results are stated w.r.t. a succinct representation of terms as Directed Acyclic Graphs (DAG) with maximal sharing (which may be exponentially more concise).
  If \(\subterms(t)\) is the set of subterms of \(t\), the \emph{DAG size} of \(t\) refers to the cardinality \(|\subterms(t)|\) and is written \(\dagsize {t}\).
  This definition is lifted to sets and sequences of terms with the sharing common to all elements of the structure.
  The size of a signature \(\sig\) is the sum of the arities of the symbols of \(\sig\) (which is finite since \(\sigc\) and \(\sigd\) are finite) and the size of a rewrite system \(\R\) is the sum of the sizes of the two hand sides of its rules.
  The size of a process is the sum of the number of operators of the process and of the sizes of all terms appearing in the process (in conditionals, channels, and output terms).
  We emphasise that
  \begin{itemize}
    \item A complexity upper bound stated w.r.t. the DAG size of the inputs is a stronger result than the same upper bound stated w.r.t. the tree size.
    \item On the contrary a complexity lower bound stated in DAG size is a weaker result than the corresponding result in tree size.
  \end{itemize}
  In this article we only address the strongest configurations: lower bounds in the tree representation of terms, upper bounds in DAG.

\paragraph{Complexity classes}
  We now shortly remind some background about complexity, mainly introducing our notations.
  Given \(f : \N \to \N\), we define \(\mbox{\ctime}(f)\) (resp. \(\mbox{\cspace}(f)\)) the class of problems decidable by a deterministic Turing machine running in time (resp. in space) at most \(f(n)\) where \(n\) is the size of the parameters of the problem.
  It is common to define the following classes:
  \[\begin{array}{rclcrcl}
    \mbox{\logspace} & = & \displaystyle \bigcup_{p \in \N} \mbox{\cspace}(\log(n^p)) & & \mbox{\ptime} & = & \displaystyle \bigcup_{p \in \N} \mbox{\ctime}(n^p)\\[5mm]
    \mbox{\pspace} & = & \displaystyle \bigcup_{p \in \N} \mbox{\cspace}(n^p) & & \mbox{\complexityfont EXPTIME} & = & \displaystyle \bigcup_{p \in \N} \mbox{\ctime}(2^{n^p})
  \end{array}\]
  One can define their non-deterministic counterparts {\complexityfont NLOGSPACE} ({\complexityfont NL} for short), {\complexityfont NPTIME}, {\complexityfont NPSPACE} and {\complexityfont NEXPTIME}.
  Given a (non-deterministic) class \(\C\), we call co-\(\C\) the class of problems whose negation is in \(\C\).
  From now on we often omit the suffix \ctime in the name of time complexity classes for the sake of succinctness.
  Then it is known that:
  \[
    \mbox{\logspace}
      \subseteq \mbox{\complexityfont NL}
      = \mbox{co{\complexityfont NL}}
      \subseteq \mbox{\complexityfont P}
      \subseteq \mbox{\np,co\np}
      \subseteq \mbox{\pspace}
      = \mbox{\complexityfont NPSPACE}
      \subseteq \mbox{\complexityfont EXP}
      \subseteq \mbox{\complexityfont NEXP,coNEXP}
  \]

  To define complete problems for complexity classes above \ptime we use  classical \emph{many-to-one polytime reductions}.
  We also mention the notion of \emph{oracle} reduction, deciding a problem with a constant-time black box for another problem:
  the class of problems decidable in \(\C\) with an oracle for a problem \(Q\) is noted \(\C^Q\).
  When \(Q\) is complete for a class \(\D\) w.r.t. a notion of reduction executable in \(\C\), we may write \(\C^{\D}\) instead;
  in particular \(\C^{\D} = \C^{\mbox{\scriptsize co}\D}\).
  This kind of reduction is needed to define the last complexity classes we will use in this article: the \emph{polynomial hierarchy}, which is a collection of complexity classes between \ptime and \pspace.
  Indeed the difference between \np and \pspace lies in their capacity to express quantifier alternation;
  the usual complete problems considered for these two complexity classes are, given a boolean formula \(\varphi\):
  \begin{itemize}
    \item \sat (\np complete): does \(\exists x_1, \ldots, x_n. \varphi(x_1, \ldots, x_n)\) hold?
    \item \qbf (\pspace complete): does \(\forall x_1, \exists y_1, \ldots, \forall x_n, \exists y_n. \varphi(x_1,y_1, \ldots, x_n,y_n)\) hold?
  \end{itemize}
  The polynomial hierarchy characterises all classes corresponding to intermediate alternations.

  \begin{definition}
    The polynomial hierarchy {\complexityfont PH} consists of the classes \(\Sigma_n\) defined by \(\Sigma_0 = \mbox{\ptime}\) and \(\Sigma_{i+1} = \mbox{\np}^{\Sigma_i}\).
    In particular, \(\Sigma_1= \mbox{\np}\).
    We also write \polyh {i} for co\(\Sigma_i\).
  \end{definition}

\paragraph{Problems studied in this article}
  We thus study the following decision problems:

  \problemdescr[\StaticEquiv]
    {A rewriting system \(\R\), two frames \(\Phi\) and \(\Psi\).}
    {\(\Phi \StatEq \Psi\) for \(\R\)?}

  \problemdescr[\TraceEquiv]
    {A rewriting system \(\R\), two processes \(P\) and \(Q\).}
    {\((\multi {P}, \emptyset) \TraceEq (\multi {Q}, \emptyset)\) for \(\R\)?}
  
  \medskip 

  We also consider \TraceInclus, \Simulation, \Similarity, \Bisimilarity to be the analogue problems of \TraceEquiv, replacing trace equivalence by the relations \(\TraceIncl\), \(\Simu\), \(\Simi\), and \(\LabBis\), respectively.
  As we explained previously these problems are undecidable in general and we need to put restrictions on the inputs, in addition to the restriction to a bounded number of sessions, which is inherent to our model.
  Typically our results all include the restriction (inherent to our model) to constructor destructor theories and bounded processes.
  When we say for example that ``\TraceEquiv is decidable for constructor destructor subterm convergent rewriting systems'' it means that we are studying the following decision problem:

  \problemdescr
    {A constructor destructor subterm convergent rewriting system \(\R\), two processes \(P\) and \(Q\).}
    {Are \(P\) and \(Q\) trace equivalent (for \(\R\))?}

  The way we state the problem implies that complexity analyses need to account for the size of all inputs, including the rewriting system.
  However the treatment of this question is not uniform in the literature.
  Complexity analyses in \cite{AC06,B07,CCD13} consider the rewriting system as a constant of the problem.
  For the example above this means considering, for each constructor destructor subterm convergent rewriting system \(\R\), the following decision problem:

  \problemdescr
    {Two bounded processes \(P\) and \(Q\).}
    {Are \(P\) and \(Q\) trace equivalent (for \(\R\))?}

  For this formulation of the problem, generic completeness results w.r.t. complexity classes are not possible in general because different complexities may arise for each rewriting system~\(\R\).
  This is for example the case in \cite{AC06}, where \StaticEquiv is proven \ptime for any fixed subterm convergent rewriting system:
  the problem is indeed \ptime-hard for some of them \cite{CKR20} but also \logspace for others as we prove it in this article.
  All existing procedures \cite{AC06,CDK12,CBC11} are actually exponential in the size of the rewriting system.
  This is why we refer to this problem as \emph{parametric equivalence} and say by opposition that \emph{general equivalence} is the initial variant with the rewriting system considered as part of the input.
  We argue that the latter is more relevant today as the rewriting system can now be specified by the user in many automated tools.
  This motivated for example to prove in \cite{CKR20} that the complexity results of \cite{B07,CCD13} (stated in the parametric setting) were also valid in the general setting.

\section{Structure of the decision procedure}
\label{sec:symbolic}


We detail in this section our overall decision procedure for equivalence properties, intuitively reducing them to solving some forms of symbolic constraints.
We express this through a novel notion of \emph{partition tree} that crisply characterises equivalence proofs.
We formalise in this section the main properties of this tree and describe how to derive an actual decision procedure from it;
the constraint solving procedure necessary to generate the tree itself is then later detailed in Section~\ref{sec:ptree}.

\subsection{The symbolic approach for decidability}

  Our decision procedures rely on a \emph{symbolic semantics}, by opposition to the usual semantics of the calculus (recall Figure~\ref{fig:semantics}) that we will call the \emph{concrete} semantics from now on.
  Specifically, rather than fetching concrete input terms from the active attacker, our symbolic semantics abstract these inputs and only record the constraints they should satisfy to execute the protocol. 
  This thus provides a finite representation of the infinite set of actions potentially available to the attacker.
  For example let \(c \in \sig_0\), \(\hfun/1 \in \sigc\), \(k \in \Nall\) and consider the process
  \[\begin{array}{l@{\ }l}
    P =
      & \OutP {c} {k}.\,
      \InP {c} {x}.\,
      \IfP\ \fst(x) = k\ \ThenP\ \OutP {c} {\hfun(x)}
  \end{array}\]
  The trace executing the output \(\hfun(x)\) will gather constraints that intuitively indicate that:
  \begin{enumerate*}
    \item \(x\) is a term deducible by the attacker from the frame \(\{\ax_1 \mapsto k\}\); and
    \item \(x = \pair {k,y}\) for some term \(y\).
  \end{enumerate*}
  A constraint solving algorithm, detailed in Section~\ref{sec:mgs-gen}, can then be used to show that these constraints have a \emph{solution}:
  the recipe \(\xi = \pair {\ax_1,a}\), \(a \in \sig_0\), can be used to compute the input term \(x\) and satisfy the constraints, which justifies that the output of \(\hfun(x)\) is reachable.
  Similar approaches are common to decide reachability or equivalence properties of bounded processes~\cite{B07,CCD13}; our approach is however more widely applicable due to our absence of syntactic restrictions on processes.

\paragraph{Formalising symbolic constraints}
  We first introduce a new type of variables, used in recipes:

  \begin{definition}[second-order terms]
    We consider a partition of the set of (non-axiom) variables \(\X \smallsetminus \AX = \X[1] \uplus \X[2]\). 
    The elements of \(\X[1]\) are called \emph{first-order variables} and correspond to those we used so far in terms (in processes, frames, rewrite rules).
    Those of \(\X[2]\) are called \emph{second-order variables} and are used to represent an undefined recipe.
    A \emph{first-order term} is an element of \(\termset(\sig \cup \sig_0 \cup \Nall \cup \X[1])\) and a \emph{second-order term} is an element of \(\termset(\sig \cup \sig_0 \cup \AX \cup \X[2])\).

    We now distinguish \(\vars[1](u) = \vars(u) \cap \X[1]\), \(\vars[2](u) = \vars(u) \cap \X[2]\), and \(\axioms(u) = \vars(u) \cap \AX\).
    Note that we say that a second-order term \(t\) is ground if \(\vars[2](t) = \emptyset\), i.e., \(t\) may contain axioms.
    By definition, a recipe is therefore a ground second-order term.
    We also adapt the other notations of the term algebra to reflect the separation:
    \(\subterms[1]\), \(\subterms[2]\), \ldots
  \end{definition}

  In practice, when executing an input instruction \(\InP{c}{x}\) in the symbolic semantics, \(x\) will be associated to a fresh second-order variable written \(\quanti{X}{i}\), where \(X \in \X[2]\) will serve as a placeholder for the recipe used to compute \(x\), and \(i \in \N\) indicates that only the first \(i\) axioms of the frame are available to compute the recipe in question.
  This is formalised by the following, natural extension of the notion of substitution:

  \begin{definition}[second-order substitutions]
    We suppose a partition \(\X[2] = \biguplus_{i \in \N} \Xsndi[=]{i}\) where each class \(\Xsndi[=]{i}\) is infinite.
    We also write \(\Xsndi{i} = \bigcup_{j=0}^i \Xsndi[=]{j}\).
    If \(X\) is a second-order variable we may write \(\quanti{X}{i}\) to emphasise that \(X \in \Xsndi[=]{i}\) and say in this case that \emph{\(X\) is of type \(i\)}.
    A \emph{second-order substitution} is then a substitution \(\Sigma\) of domain \(\dom(\Sigma) \subseteq \X[2]\) that \emph{respects types}:
    \begin{align*}
      \forall \quanti{X}{i} \in \dom(\Sigma),\ X \Sigma \in \recipeset_i & &
      \text{where } \recipeset_i = \termset(\sig \cup \sig_0 \cup \Xsndi{i} \cup \{\ax_1, \ldots, \ax_i\})
    \end{align*}
  \end{definition}

  Altogether, we can then define the constraints that we use to characterise the possible values that an input term \(x\) may take:

  
  \begin{definition}[atoms]\label{def:atoms}
    We consider the following three kinds of atoms:
    \begin{enumerate}
      \item \emph{deduction fact} \(\xi \dedfact u\) where \(u\) is a message in normal form and \(\xi\) is a second-order term such that \(\rootf(\xi) \notin \sigc\);
      \item \emph{second-order equations} \(\xi \eqs \zeta\) where \(\xi\) and \(\zeta\) are two second-order terms;
      \item \emph{(first-order) equations} \(u \eqs v\) where \(u\) and \(v\) are two first-order terms (not necessarily messages).
    \end{enumerate}
    The negation \(\neg(\alpha \eqs \beta)\) of an equation is written \(\alpha \neqs \beta\) and called a \emph{disequation}.
  \end{definition}


  \begin{definition}[constraint] \label{def:constraint}
    An \emph{atomic constraint} (or an \emph{atomic formula}) is an atom that is either a deduction fact, a second-order equation, or a first-order equation $u \eqs v$ where $u$ and $v$ are constructor terms.
    A \emph{constraint} is then a first-order formula over atomic constraints, that is, either an atomic constraint, \(\top\), \(\bot\), or of the form \(\varphi \wedge \psi\), \(\varphi \vee \psi\), \(\neg \varphi\), \(\forall x.\varphi\), or \(\forall \quanti{X}{n}.\varphi\) for \(\varphi,\psi\) constraints.
    Note that \(\vars(\varphi)\) then refers to the \emph{free} variables of the constraint \(\varphi\).
  \end{definition}

  A deduction fact \(\xi \dedfact u\) indicates that term \(u\) is deducible by the recipe \(\xi\) and second-order equations \(\xi \eqs \zeta\) are used to put restrictions on which recipes \(\xi\) may be used to do so.
  For example \(\quanti{X}{i} \dedfact x\) states that the variable \(x\) is to be replaced by a term deducible by the attacker using at most the \(i\) first outputs of the frame;
  a constraint solving procedure may then impose that \(\exists \quanti{Y}{i}.\,X \eqs \ffun(Y)\), i.e., that the underlying recipe should have a \(\ffun\) symbol at its root.
  Equations reflect the syntactic equalities that the first-order terms verify.
  Typically when executing \(\IfP\ \fst(x) = t\ \ThenP\ P\ \ElseP\ Q\), the positive branch will intuitively lead to the constraint \(\exists y.\, x \eqs \pair {t,y}\) and the negative branch to \(\forall y.\, x \neqs \pair {t,y}\).

\paragraph{Constraint systems}
  Finally we define and give some properties of \emph{constraint systems} that are used to collect the first-order constraints induced by a given execution of a process.
  
  \begin{definition}[constraint system]
    A \emph{constraint system} is a triple \(\C = (\Phi,\Df,\Eqfst)\) whose elements are of the following form:
    \begin{enumerate}
      \item \(\Phi = \{\ax_1 \mapsto t_1, \ldots, \ax_n \mapsto t_n\}\) is a frame (not necessarily ground)
      \item \(\Df\) is a set of constraints of the form \(X \dedfact x\), with \(X \in \Xsndi{n}\), \(x \in \X[1]\).
      We also require the \emph{origination property}:
      for all \(i \in \eint {1} {n}\), for all \(x \in \vars(t_i)\), there exists \(X \in \Xsndi{i-1}\) such that \((X \dedfact x) \in \Df\).
      \item \(\Eqfst\) is a set of constraints of the form \(u \eqs v\) or \(\forall z_1 \ldots \forall z_k. \bigvee_{j=1}^r u_j \neqs v_j\).
    \end{enumerate}
  \end{definition}

  The components of \(\C\) are also written \(\Phi(\C)\), \(\Df(\C)\) and \(\Eqfst(\C)\).
  The set \(\Df\) contains all input binders \(x\) that have been executed, each mapped to a second-order variable \(X\) that will serve as a placeholder for the corresponding recipe.
  Next the origination property expresses that when reference is made to an input \(x\) in an output \(t_i\), this input should be computed only from the previous outputs \(t_1, \ldots, t_{i-1}\).
  This is a natural invariant preventing cyclic input-output dependencies, always satisfied in practice.
  Finally \(\Eqfst\) is a set of (dis)equalities imposed on the protocol messages by conditionals, among others.
  We will formalise in Section \ref{sec:mgs-def} the semantics of these constraints through a notion of \emph{solution}.

  \begin{remark}[notational conventions]
    We use several convenient notations throughout the article to lighten the presentation of constraints.
    First of all we do not make a difference between sets and conjunctions of constraints:
    for instance we may write \(\Eqfst = \bigwedge_{i=1}^n \varphi_i\) instead of \(\Eqfst = \{\varphi_i\}_{i=1}^n\) and conversely.
    We also interpret a substitution \(\sigma\) as the set of equations \(\eqnset = \{x \eqs x \sigma \mid x \in \dom(\sigma)\}\).
  \end{remark}

\subsection{(Most general) unifiers} \label{sec:unification}
We now recall some basics on term unification, a key concept in symbolic models that has some specificities in our context, in particular regarding second-order terms.

\paragraph{Unification of first-order terms}
Two first-order terms \(u\) and \(v\) are \emph{unifiable} if there exists a substitution \(\sigma\), called a \emph{unifier}, such that \(u \sigma = v \sigma\).
For example the terms \(u = \pair{\sdec(x,y),z}\) and \(v = \pair{z_1,z_2}\) are unified by \(\sigma = \{z_1 \mapsto \sdec(x,y), z_2 \mapsto z\}\).
The terms \(u\) and \(z'\) are unifiable as well using \(\sigma = \{z' \mapsto u\}\), but the terms \(u\) and \(z\) are not.
More generally, a unifier \(\sigma\) of a set of equations \(\eqnset = \{u_i \eqs v_i\}_{i=1}^n\) is a unifier of \(u_i\) and \(v_i\) for all \(i\).
A classical characterisation of the set of unifiers of two terms is based on \emph{most general unifiers}:

\begin{definition}[most general unifier]
  A unifier \(\sigma\) of \(\eqnset\) is said to be a \emph{most general} one if, for any \(\theta\) unifier of \(\eqnset\), there exists \(\tau\) such that \(\theta = \sigma \tau\).
  In this case, we write \(\sigma = \mgu(\eqnset)\) (and it is unique up to variable renaming).
  When \(\eqnset\) is not unifiable, we write \(\mgu(\eqnset) = \bot\).
\end{definition}

A straightforward inductive procedure allows to decide whether \(\eqnset\) is unifiable and, if it is, to compute \(\mgu(\eqnset)\)
We assume w.l.o.g. that this computation does not introduce variables, that is, if \(\sigma = \mgu(\eqnset)\) then \(\dom(\sigma) \cup \vars(\im(\sigma)) \subseteq \vars(\eqnset)\).
We also require that \(\dom(\sigma) \cap \vars(\im(\sigma)) = \emptyset\), that is, applying a mgu twice has no more effect than applying it once.
Note as well that all unifiers are instances of the mgu but the converse is also true, that is,
all instances of a mgu are unifiers.
By convenience we also write \(\mgu(\eqnset)\) in the case where \(\eqnset\) contains disequations (typically when writing \(\mgu(\Eqfst(\C))\)):
in this case only equations are taken into account and nothing ensures that the mgu satisfies the disequations of \(\eqnset\).

However mgu's are only syntactic:
when taking the rewriting system \(\R\) into account we say that \(\sigma\) is a \emph{unifier modulo theory} of \(\eqnset\) when for all \((u \eqs v) \in \eqnset\), \(u\sigma\norm = v\sigma\norm\).
A standard procedure based on narrowing (not detailed here) allows to compute \emph{most general unifiers modulo \(\R\)} when~\(\R\) is subterm convergent among others~\cite{CD05}.
However unlike the syntactic case they are not unique in general:

\begin{definition}[most general unifier modulo theory]
  We let \(\eqnset\) be a set of equations and \(\R\) be a convergent rewriting system.
  A set of \emph{most general unifiers modulo \(\R\)} is a set of substitutions \(\mguR(\eqnset)\) that verifies the following properties:
  \begin{enumerate}
    \item for all \(\sigma \in \mguR(\eqnset)\), \(\sigma\) is a unifier of \(\eqnset\) modulo \(\R\)
    \item for all \(\theta\) unifier of \(\eqnset\) modulo \(\R\), there exists \(\sigma \in \mguR(\eqnset)\) and a substitution \(\tau\) such that for all \(x \in \vars(\eqnset)\), \(x \theta\norm = x \sigma \tau\norm\)
  \end{enumerate}
\end{definition}

Again we emphasise that equality modulo \(\R\) only operates on valid messages, that is, if \(\sigma \in \mguR(u \eqs v)\) then \(u \sigma\) and \(v \sigma\) verify the \(\msg\) predicate.
A typical use case we consider in the symbolic semantics is \(\mguR(u \eqs u)\), which is the most general substitution \(\sigma\) such that \(\msg(u \sigma)\) holds (if any).
For example if \(u = \adec(x,y)\) we have \(\mguR(u \eqs u) = \{\sigma\}\), where:
\begin{align*}
  \sigma & = \{x \mapsto \aenc(x',x_r,\pk(y')), y \mapsto y'\} & & x',x_r,y' \in \X\ \text{fresh}
\end{align*}
This example also highlights that, unlike the syntactic case, computing mgu's modulo theory may require to introduce new variables.
This also makes it possible to enforce that \(\dom(\sigma) \cap \vars(\im(\sigma)) = \emptyset\).

\paragraph{Unification of second-order terms}
Intuitively, the unification of two second-order terms \(\xi\) and \(\zeta\) modulo theory means that they deduce the same first-order term \(u\) w.r.t. a given frame~\(\Phi\).
This unusual kind of unification is performed as a part of our constraint solving algorithm using a dedicated kind of constraint written \(\xi \eqf \zeta\), detailed in Section~\ref{sec:formulas}.
However, even the computation of syntactic mgu's has some subtleties for second-order terms that we discuss below.

As in the first-order case, a syntactic unifier of \(\xi\) and \(\zeta\) is a second-order substitution \(\Sigma\) such that \(\xi \Sigma = \zeta \Sigma\).
However, computing \(\Sigma\) is not as simple as usual due to the variable types.
Indeed, we recall that by definition, a second-order substitution has to respect types, that is, a variable \(\quanti{X}{n}\) cannot be mapped to a term containing axioms \(\ax_i\) or variables \(\quanti{Y}{i}\) if \(i > n\).
Say for instance we want to unify the two second-order terms \(\quanti{X}{1}\) and \(\ffun(\quanti{Y}{2})\):
a regular computation of the mgu would yield the substitution \(\Sigma = \{X \mapsto \ffun(Y)\}\), which does not respect the type of \(X\).
In this case, one solution is to introduce a fresh variable \(\quanti{Z}{1}\) and to choose the following unifier:
\[\mgu(X \eqs \ffun(Y)) = \{X \mapsto \ffun(Z), Y \mapsto Z\} = \Sigma\{Y \mapsto Z\}\,.\]
%
Given a second-order term \(\xi\), let us write \(\maxarity{\xi}\) the maximal type of second-order variables and axioms appearing in \(\xi\), that is, the minimal type \(i\) such that \(\xi \in \recipeset_i\).
Formally,
\[\maxarity{\xi} = \min \{i \in \N \mid \xi \in \recipeset_i\}\,.\]
The mgu of a conjunction of equations \(\varphi\) is then computed inductively as follows:
\[\begin{array}{l}
  \mgu(\top) = \top \\[3mm]
  \mgu\left(\varphi \wedge \ffun(\xi_1, \ldots, \xi_n) \eqs \gfun(\zeta_1, \ldots, \zeta_n)\right) =
  \left\{\begin{array}{ll}
    \bot & \mbox{if } \ffun \neq \gfun \\
    \mgu\left(\varphi \wedge \bigwedge_{i=1}^n \xi_i \eqs \zeta_i\right) & \mbox{if } \ffun = \gfun
  \end{array}\right. \\[8mm]
  \mgu\left(\varphi \wedge \quanti{X}{i} \eqs \xi\right) =
  \left\{\begin{array}{l@{\qquad}l}
    \bot & \mbox{if \(X \in \vars[2](\xi)\) and \(\xi \neq X\)} \\[2mm]
    \bot & \mbox{else if \(\exists \ax_j \in \axioms(\xi), j > i\)} \\[2mm]
    \Sigma_0 \Sigma
      & \mbox{else if \(\xi \notin \X[2]\), \(\quanti{Y}{j} \in \vars[2](\xi)\), \(j > i\), \(\quanti{Z}{i}\) fresh and}\\
      & \text{with } \Sigma_0 = \{Y \mapsto Z\} \text{ and } \Sigma = \mgu\left(\varphi\Sigma_0 \wedge \quanti{X}{i} \eqs \xi\Sigma_0\right) \\[2mm]
    \Sigma_0 \Sigma
      & \mbox{else if \(\maxarity{\xi} \leqslant i\), with \(\Sigma_0 = \{X \mapsto \xi\}\) and \(\Sigma = \mgu\left(\varphi\Sigma_0\right)\)}
  \end{array}\right.
\end{array}\]

As before we extend this notation to arbitrary sets \(\eqnset\), that is, we may write \(\mgu(\eqnset)\) even if~\(\eqnset\) contains disequations (which are then ignored during the computation).
The correctness of this function is proved below.

\begin{proposition}[correctness of second-order mgu's]
  For all sets of second-order equations~\(\eqnset\), the computation of \(\mgu(\eqnset)\) terminates.
  Besides we have that \(\mgu(\eqnset) = \bot\) \textit{iff} there exist no unifiers of \(\eqnset\).
  When \(\mgu(\eqnset) \neq \bot\), we have that:
  \begin{enumerate}
    \item \(\mgu(\eqnset)\) is a second-order substitution, i.e., it respects types, and it is a unifier of \(\eqnset\);
    \item for all unifiers \(\Sigma\) of \(\eqnset\), there exists a second-order substitution \(\Sigma_0\) such that \(\Sigma = \mgu(\eqnset) \Sigma_0\).
  \end{enumerate}
\end{proposition}

\begin{proof}
  We only prove the termination since all other properties can be proved separately by straightforward inductions on the definition of \(\mgu\).
  We let the partial ordering on second order variables \(\preccurlyeq\) given by the types, i.e. \(\quanti{X}{i} \preccurlyeq \quanti{Y}{j}\) \textit{iff}  \(i \leqslant j\).
  Given a set of second-order equations \(\eqnset\) we then let
  \[\mu(\eqnset) = (\vars[2](\eqnset),M(\eqnset),F(\eqnset))\]
  where \(M(\eqnset)\) is the multiset of variables of \(\eqnset\), i.e. multiplicity included, and \(F(\eqnset)\) is the multiset of the sizes of the equations of \(\eqnset\) (where the size of \(\xi \eqs \zeta\) is the number of function symbols in \(\xi\) and \(\zeta\)).
  The first two components are ordered w.r.t. the multiset extension of \(\preccurlyeq\),
  and the third one w.r.t. the multiset extension of \(\leqslant\).
  The overall tuple is ordered w.r.t. the lexicographic composition of the three components.

  If we number from 1 to 7 the axioms defining \emph{mgu}, we can show that \(\mu\) decreases at each recursive call:
  (1), (2), (4) and (5) make no recursive calls;
  (3) preserves \(\vars[2]\) and \(M\), and makes~\(F\) decrease;
  (6) replaces all occurrences of \(Y\) with \(Z\) that has a lower type which makes \(\vars[2]\) decrease.
  Regarding (7) two cases can arise:
  either \(\xi = X\) or \(X \notin \vars[2](\xi)\).
  In the first case \(\vars[2]\) is non increasing and \(M\) is decreasing since two occurrences of \(X\) are removed and the rest of the formula \(\eqnset\) is left unchanged.
  In the second case \(\vars[2]\) is decreasing since all occurrences of \(X\) are removed and no variables are added.
\end{proof}


\subsection{(Most general) solutions} \label{sec:mgs-def}
\paragraph{Solutions}
Let us now formalise the semantics of constraints.
Given a constraint \(\varphi\), a frame \(\Phi\) and second- and first-order substitutions \(\Sigma\) and \(\sigma\)
we define the predicate \((\Phi,\Sigma,\sigma) \models \varphi\) by:
\[\begin{array}{l@{\quad\mathit{iff}\quad}l}
  (\Phi,\Sigma,\sigma) \models \xi \dedfact u & \xi \Sigma \Phi \sigma\norm = u \sigma\norm \\
  (\Phi,\Sigma,\sigma) \models \xi \eqs \zeta & \xi\Sigma = \zeta\Sigma \\
  (\Phi,\Sigma,\sigma) \models u \eqs v & u\sigma = v\sigma \\
  (\Phi,\Sigma,\sigma) \models \forall x.\,\varphi & \text{for all first-order ground terms } t, (\Phi,\Sigma,\sigma) \models \varphi\{x \mapsto t\} \\
  (\Phi,\Sigma,\sigma) \models \forall \quanti{X}{n}.\,\varphi & \text{for all } \xi \in \recipeset_n, (\Phi,\Sigma,\sigma) \models \varphi\{X \mapsto \xi\}
\end{array}\]
The definition is extended with logical connectives \(\neg, \wedge, \vee, \ldots\) in the natural way.
By convention, writing \((\Phi,\Sigma,\sigma) \models \varphi\) implicitly assumes that, for all \(x \in \vars[1](\varphi)\) and \(X \in \vars[2](\varphi)\), \(x \sigma\) and \(X \Sigma \Phi \sigma\) are ground.
Intuitively the second-order substitution \(\Sigma\) describes which recipes are used to deduce each input term appearing in \(\varphi\), while \(\sigma\) gives the actual values of these inputs.

\begin{definition}[solution of a constraint system]
  We say that \((\Sigma,\sigma)\) is a \emph{solution} of \(\C\) if \(\dom(\Sigma) = \vars[2](\C)\), \(\dom(\sigma) = \vars[1](\C)\) and \((\Phi(\C),\Sigma,\sigma) \models \Df(\C) \wedge \Eqfst(\C)\).
  We call \(\Sigma\) a \emph{second-order solution} of \(\C\) and \(\sigma\) its \emph{first-order solution}.
  The set of solutions of \(\C\) is written \(\Sol(\C)\).
\end{definition}

The solutions of a constraint system \(\C\) indicate how the inputs of \(\C\) (i.e., \(\vars[1](\Df(\C))\)) can be computed while satisfying the constraints imposed by \(\Eqfst(\C)\).
Due to the origination property, the values the first-order solution \(\sigma\) takes on \(\vars[1](\Df(\C))\) is uniquely determined by which recipes are used to deduce terms, i.e., by the second-order solution \(\Sigma\).

\begin{example} \label{ex:csys}
  Consider again the example \(P = \OutP {c} {k}.\,
  \InP {c} {x}.\,\IfP\ \fst(x) = k\ \ThenP\ \OutP {c} {\hfun(x)}\).
  The traces performing the final output \(\hfun(x)\) are characterised by the constraint system
  \begin{align*}
    \Phi(\C) & = \{\ax_1 \mapsto k, \ax_2 \mapsto \hfun(x)\} &
    \Df(\C) & = \{X \dedfact x\} &
    \Eqfst(\C) & = \{x \eqs \pair{k,y}\}
  \end{align*}
  where \(\quanti{X}{1}\) and \(y\) are fresh second- and first-order variables, respectively.
  Observe in particular that the informal constraint ``\textit{there exists a term \(y\) such that \(x = \pair {k,y}\)}'' is not formalised using an explicit \(\exists\) quantification but with a free variable \(y\).
  All second-order solutions of \(\C\) are instances of \(\Sigma_0 = \{X \mapsto \pair {\ax_1, Y}\}\) where \(\quanti{Y}{1}\) is fresh,
  for example, \(\Sigma = \{X \mapsto \pair {\ax_1, a}\}\) with \(a \in \sig_0\).
  The corresponding first-order solution is then \(\sigma = \{x \mapsto \pair {k,a}, y \mapsto a\}\).
\end{example}

\paragraph{Most general solutions}
Similarly to mgu's, we now introduce a novel characterisation of solutions as instances of so-called most general solutions (\emph{mgs}).
The definition is parametrised with a predicate \(\pi\) on second-order substitutions, writing \(\Sol[\pi](\C) = \{(\Sigma,\sigma) \in \Sol(\C) \mid \pi(\C) \text{ holds}\}\).
Filtering solutions this way will essentially permit, during the decision procedure, to perform case analyses on the form of the solutions.

\begin{definition}[most general solution] \label{def:mgs}
  A set of \emph{most general solutions of \(\C\) that satisfy \(\pi\)} is a set \(\mgs[\pi](\C)\) of second-order substitutions such that:
  \begin{enumerate}
    \item \label{it:mgs-sol}
    for all \(\Sigma_0 \in \mgs[\pi](\C)\), \(\dom(\Sigma_0) \subseteq \vars[2](\C)\), for all injections \(\Sigma_1\) to fresh constants and of domain \(\dom(\Sigma_1) = \vars[2](\im(\Sigma_0),\C) \smallsetminus \dom(\Sigma_0)\),
    \((\Sigma_0\Sigma_1,\sigma) \in \Sol[\pi](\C)\) for some \(\sigma\).
    \item \label{it:mgs-inst}
    for all \((\Sigma,\sigma) \in \Sol[\pi](\C)\), there exists \(\Sigma_0 \in \mgs[\pi](\C)\) and \(\Sigma_1\) such that \(\Sigma = \Sigma_0\Sigma_1\).
  \end{enumerate}
  We omit the predicate \(\pi\) in the case where \(\pi = \top\), i.e., \(\pi(\Sigma)\) holds for any substitution.
\end{definition}

The first condition of the definition states that a mgs \(\Sigma_0\) is ``almost'' a solution of \(\C\):
\(\Sigma_0\) is allowed to be given in a minimal form that does not instantiate all variables of \(\vars[2](\C)\), and that may not have a ground image;
but we obtain a solution by replacing all pending variables by fresh names using \(\Sigma_1\).
The second condition states that all solutions are instances of a mgs.

\begin{example} \label{ex:mgs}
  In Example \ref{ex:csys} we have \(\mgs(\C) = \{\Sigma_0\}\) and \(\Sol(\C)\) is the set of all ground instances of \(\Sigma_0\).
  However in general the situation may be less ideal.
  For example a constraint system may have several most general solutions;
  a simple example being, with \(\hfun/1\) and \(k \in \Nall\):
  \begin{align*}
    \Phi(\C) & = \{\ax_1 \mapsto \hfun(k), \ax_2 \mapsto k\} &
    \Df(\C) & = \{\quanti{X}{2} \dedfact x\} &
    \Eqfst(\C) & = \{x \eqs \hfun(k)\}
  \end{align*}
  The constraint system \(\C\) expresses that an input \(x\) should be instantiated by \(\hfun(k)\), potentially by using the two previous outputs \(\hfun(k)\) and \(k\).
  There are therefore two ways of computing \(x\):
  either using \(\ax_1\) or \(\hfun(\ax_2)\), which is reflected as the fact that \(\mgs(\C) = \{\Sigma_1,\Sigma_2\}\) with
  \begin{align*}
    \Sigma_1 & = \{X \mapsto \hfun(\ax_2)\} &
    \Sigma_2 & = \{X \mapsto \ax_1\}
  \end{align*}
  Still, it is possible to obtain unique mgs' by performing a case analysis and restricting the solutions accordingly;
  typically here we have \(\mgs[\pi_i](\C) = \{\Sigma_i\}\) with
  \begin{align*}
    \pi_1(\Sigma) & \eqdef \exists X'.\, X \eqs \hfun(X') &
    \pi_2(\Sigma) & \eqdef \forall X'.\, X \neqs \hfun(X')
  \end{align*}
  Another notable point is that some ground instances of a mgs may not be solutions themselves.
  Taking \(a \in \sig_0\) a simple example is given by \(\C = (\emptyset, X \dedfact x, x \neqs a)\)
  and \(\mgs(\C) = \{\id\}\):  the substitution \(\{X \mapsto a\}\) is a ground instance of the identity but not a solution (which does not contradict Item \ref{it:mgs-sol} of Definition \ref{def:mgs} since although \(a\) is a constant, it is not fresh).
\end{example}

We describe in Section~\ref{sec:mgs-gen} how to generate a finite set of most general solutions, at least in the context of our decision procedure.

\subsection{Symbolic semantics} \label{sec:symbolic-semantics}

\paragraph{Symbolic execution}
We now describe formally our symbolic semantics.
It shares some common ground with the concrete semantics of the calculus,
except that a constraint system collects the execution's constraints.
The semantics operates on so-called \emph{symbolic processes} \((\P,\C)\) where \(\P\) is a multiset of (non-necessarily ground) plain processes and \(\C\) is a constraint system.
All free variables of \(\P\) are bound by deductions facts, that is, for all \(x \in \vars(\P)\) there exists \((X \dedfact x) \in \Df(\C)\).
The semantics then takes the form of a labelled transition system \(\sstep{\alpha}\) between symbolic processes, defined in Figure~\ref{fig:semantics-symbolic}, where \(\alpha\) ranges over the following alphabet of \emph{symbolic actions}:
\begin{enumerate}
  \item \emph{symbolic input actions} \(\InP {X} {Y}\) where \(X\) and \(Y\) are second-order variables, modelling public inputs as in the concrete semantics except that the attacker recipes are replaced by the two placeholders \(X,Y\);
  \item \emph{symbolic output actions} \(\OutP {X} {\ax_i}\) that follow the same logic;
  \item the \emph{unobservable action} \(\tau\) which has the exact same role as in the concrete semantics.
\end{enumerate}

Before we define the semantics let us explain how we handle conditionals.
First of all we recall our convention to interpret substitutions as sets of equalities, that is, the positive branch of ``\(\IfP\ u = v\ \ThenP\ P\ \ElseP\ Q\)'' will add one mgu of \(u\) and \(v\) modulo theory to \(\Eqfst\).
Regarding the negative branch, we want to add a constraint that is satisfied \textit{iff} \(u\) and \(v\) are not equal modulo theory.
We write it \(\neg\mguR(u \eqs v)\) and define it as follows:
\begin{align*}
  \neg\mguR(u \eqs v) & = \bigwedge_{\sigma \in \mguR(u \eqs v)} \forall
  z_1, \ldots, z_n.\ 
  \bigvee_{x \in \vars(u,v)} x \neqs x \sigma
\end{align*}
where \(\{z_1, \ldots, z_n\} = \vars(u\sigma,v\sigma)\smallsetminus \vars(u,v)\).

\begin{figure*}[ht]
  \centering
  \begin{align}
    \nonumber
    & \text{If \(\C = (\Phi, \Df, \Eqfst)\), \(\mu = \mgu(\Eqfst) \neq \bot\) and \(n = |\dom(\Phi)|\):} \\[2mm]
    \tag{\mbox{\textsc{s-Then}}} \label{rule:s-then}
    & (\multi {\IfP\ u = v\ \ThenP\ P\ \ElseP\ Q} \cup \P, \C)
      \sstep {\tau} (\multi {P} \cup \P, (\Phi,\Df,\Eqfst \wedge \sigma)) \\
    \tag*{\text{\small if \(\sigma \in \mguR(u\mu \eqs v\mu)\)}} \\
    \tag{\mbox{\textsc{s-Else}}} \label{rule:s-else}
    & (\multi {\IfP\ u = v\ \ThenP\ P\ \ElseP\ Q} \cup \P, \C)
      \sstep {\tau} (\multi {Q} \cup \P, (\Phi,\Df,\Eqfst \wedge \neg\mguR(u\mu \eqs v \mu))) \\
    %
    \tag{\mbox{\textsc{s-In}}} \label{rule:s-in}
    & (\multi {\InP {u} {x}.P} \cup \P, \C)
      \sstep {\InP {Y} {X}} (\multi {P} \cup \P, (\Phi,\Df \wedge X \dedfact x \wedge Y \dedfact y,\Eqfst \wedge \sigma)) \\
    \tag*{\text{\small if \(\quanti{Y}{n}\), \(\quanti{X}{n}\) and \(y\) are fresh and \(\sigma \in \mguR(y \eqs u \mu)\)}} \\
    \tag{\mbox{\textsc{s-Out}}} \label{rule:s-out}
    & (\multi {\OutP {u} {v}.P} \cup \P, \C)
      \sstep {\OutP {Y} {\ax_{n+1}}} (\multi {P} \cup \P, (\Phi \cup \{\ax_{n+1} \mapsto v \mu \sigma \norm\},\Df \wedge Y \dedfact y,\Eqfst \wedge \sigma)) \\
    \tag*{\text{\small if \(\quanti{Y}{n}\) and \(y\) are fresh and \(\sigma \in \mguR(y \eqs u \mu \wedge v\mu \eqs v\mu)\)}} \\
    \tag{\mbox{\textsc{s-Comm}}} \label{rule:s-comm}
    & (\multi {\OutP {u} {v}.P, \InP {w} {x}.Q} \cup \P, \C)
      \sstep {\tau} (\multi {P, Q\{x \mapsto v \mu \sigma\}} \cup \P, (\Phi,\Df,\Eqfst \wedge \sigma)) \\
    \tag*{\text{\small if \(\sigma \in \mguR(u \mu \eqs w \mu \wedge v\mu \eqs v\mu)\)}} \\
    %
    \tag{\mbox{\textsc{s-Par}}} \label{rule:s-par}
    & (\multi {P \mid Q} \cup \P, \C)
      \sstep {\tau} (\multi {P, Q} \cup \P, \C)
  \end{align}
  \caption{A symbolic semantics for the applied pi-calculus}
  \label{fig:semantics-symbolic}
\end{figure*}

The rule \eqref{rule:s-in} adds two deduction facts \(X \dedfact x\) and \(Y \dedfact y\) to \(\Df\), modelling that the input term and communication channel should be deducible by the adversary;
in particular the constraint \(\sigma \in \mguR(y \eqs u \mu)\) indicates that the term deduced by \(Y\) is effectively the channel \(u\).
The rule \eqref{rule:s-out} essentially follows the same logic, adding a fresh deduction fact and a constraint indicating that the channel is deducible.
We assume an implicit alpha renaming of bound variables so that each appear only once in the process:
this prevents reference conflicts in \(\Df\) when applying the rule \eqref{rule:s-in}.
Let us also point out that several rules introduce constraints of the form \(\mguR(u,u)\):
we recall that this substitution is not always \(\top\), but is the most general substitution \(\sigma\) ensuring that \(u\sigma\) is a message.
As in the concrete semantics, a \emph{symbolic trace} is then a finite sequence of transitions 
\[(\P_0,\C_0) \sstep{\alpha_1} \cdots \sstep{\alpha_n} (\P_n,\C_n)\] 
which may be referred to as \((\P_0,\C_0) \Sstep {\tr} (\P_n,\C_n)\) if \(\tr\) is obtained by removing the \(\tau\)'s from the word \(\alpha_1 \cdots \alpha_n\).
For simplicity the plain process \(P\) may be interpreted as the symbolic process \((\multi{P},(\emptyset,\top,\top))\).

\begin{example}
  We consider again the example of the private authentication protocol.
  We recall the process of the agent \(B\) receiving the communication, writing \(\pks_X,\sks_X\) instead of \(\pk(\sk(X)),\sk(X)\), and \(t = \adec(x,s)\):
  \[\begin{array}{l@{\ }l@{\qquad}l@{\ }l}
    B(s, p, n, r) = & \InP c x.\\
    & \IfP\, \snd(t) = p\, \ThenP\\
    & \phantom{\ElseP}\, \OutP c {\aenc(\pair {\fst(t), n, \pk(s)},r,p)}\\
    & \ElseP\, \OutP c {\aenc(n,r,\pk(s))}
  \end{array}\]
  and use a frame \(\Phi_0 = \{\ax_1 \mapsto \pks_A, \ax_2 \mapsto \pks_B, \ax_3 \mapsto  \aenc(\pair{N_A,\pks_A}, r_A, \pks_B)\}\), containing public keys and the connection request sent by \(A\).
  We give in Figure \ref{fig:semantics-symbolic-ex} a tree of all symbolic executions of \(B\) (we only write the constraints added at each step).
  
  \begin{figure}[ht]
    \centering
    \includegraphics[width=0.85\textwidth]{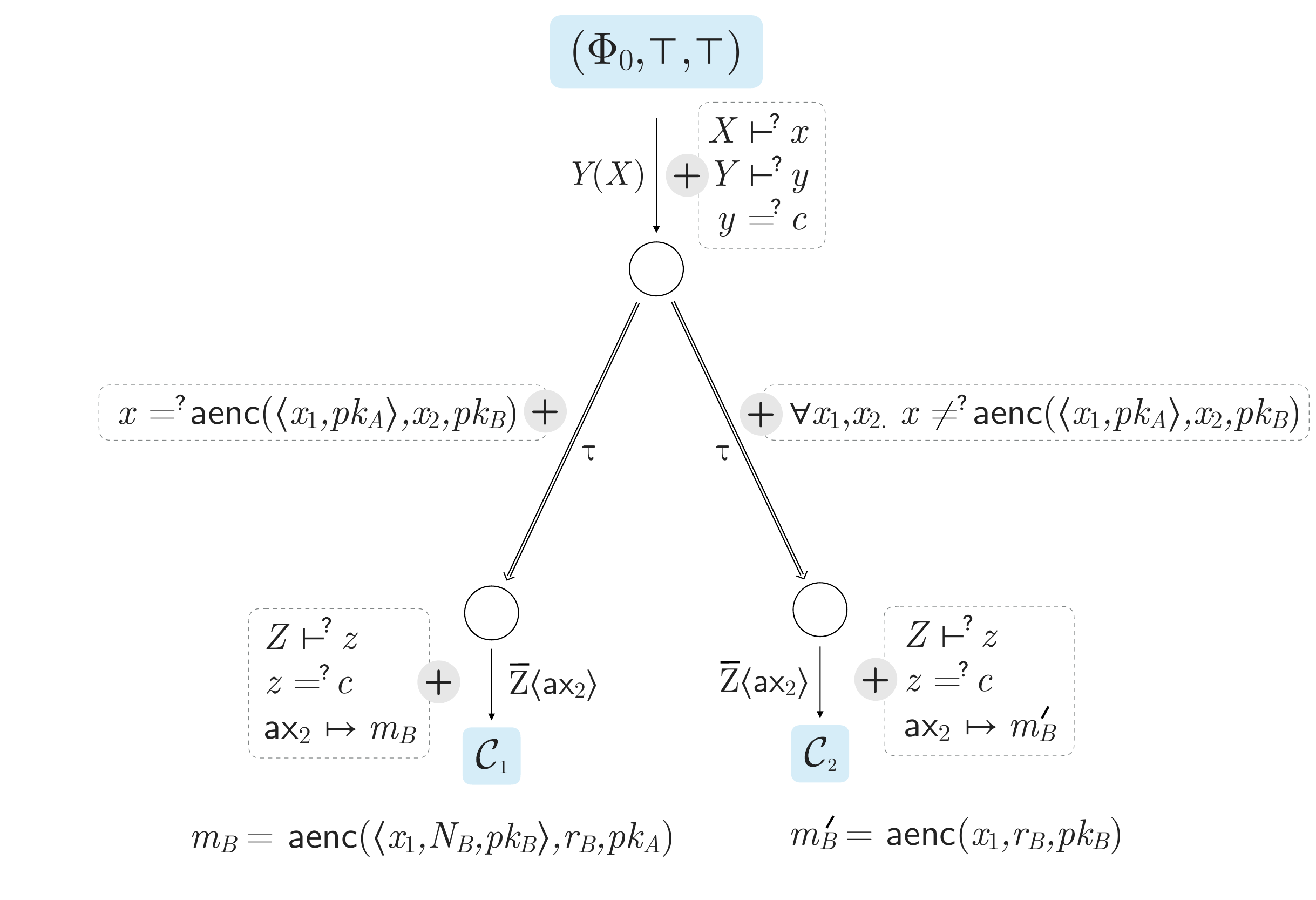}
    \caption{Tree of all constraint systems reachable by executing \(B\) symbolically}
    \label{fig:semantics-symbolic-ex}
  \end{figure}

  Intuitively, the branch of the constraint system \(\C_1\) abstracts the set of concrete traces where \(B\) accepts the connection, and the branch of \(\C_2\) those where \(B\) refuses it.
  Typically in the traces of the branch \(\C_1\) the attacker forwards the message of \(A\) or forges one pretending to be \(A\);
  this is formally expressed by the fact that \(\mgs(\C_1) = \{\Sigma_0 \cup \Sigma_{\mathsf{fwd}},\ \Sigma_0 \cup \Sigma_{\mathsf{att}}\}\) where:
  \[\Sigma_0 = \{Y \mapsto d, Z \mapsto d\} \qquad
    \Sigma_{\mathsf{fwd}} = \{X \mapsto \ax_3\} \qquad
    \Sigma_{\mathsf{att}} = \{X \mapsto \aenc(\pair{x_1,\ax_1}, x_3, \ax_2)\}
    \qedhere\]
\end{example}

\paragraph{Soundness and completeness}
Similar symbolic semantics have been developed in the context of protocol analysis~\cite{B07,CCD13}.
The general approach is to abstract the (infinite) set of concrete traces by the finite set of symbolic traces and to study the solutions of the resulting constraint systems.
A typical example is that the following statements are equivalent:
\begin{enumerate}
  \item \emph{Weak secrecy of the term \(u\) in \(P\)}: for all traces \(P \Cstep {\tr} (\P,\Phi)\), \(u\) is not deducible from \(\Phi\)
  \item for all symbolic traces \(P \Sstep{\tr} (\P,\C)\), the system \((\Phi(\C),\Df(\C) \wedge X \dedfact x, \Eqfst(\C) \wedge x \eqs u)\) has no solution, where \(\quanti{X}{n}\) and \(x\) are fresh, \(n = |\dom(\Phi)|\)
\end{enumerate}
(Recall that for notational convenience the plain process \(P\) may be interpreted as the symbolic process \((\multi{P},(\emptyset,\top,\top))\).)
This reduces weak secrecy (for a bounded number of sessions) to the decidability of whether a constraint system has a solution.
Similar approaches have been developed in~\cite{B07,CCD13} to decide equivalence properties for some classes of processes.
They rely on a connection between the symbolic and concrete semantics, under the form of two properties:
\begin{enumerate*}
  \item \emph{soundness}:
  applying to a symbolic trace a solution of its final constraint system leads to a concrete trace; and
  \item \emph{completeness}:
  all concrete traces are obtained by applying a solution to a symbolic one.
\end{enumerate*}
They are formalised below, the proof following from a straightforward induction on the length of the traces.

\begin{proposition}[soundness and completeness of the symbolic semantics] \label{prop:symbolic-sound-complete}
  Let \((\P, \C)\) be a symbolic process.
  Then we have:
  \begin{enumerate}
    \item \emph{Soundness:}
    for all symbolic traces \((\P,\C) \Sstep{\tr_s} (\Q,\C')\) and \((\Sigma,\sigma) \in \Sol(\C')\), there exists a concrete trace of the form \((\P\sigma, \Phi(\C)\sigma\norm) \Cstep{\tr_s\Sigma} (\Q\sigma, \Phi(\C')\sigma\norm)\)
    \item \emph{Completeness:}
    for all symbolic processes \((\P,\C)\), \((\Sigma,\sigma) \in \Sol(\C)\), and for all concrete traces \((\P\sigma, \Phi(\C)\sigma \norm) \Cstep{\tr} (\Q, \Phi)\), there exists a symbolic trace
    \((\P,\C) \Sstep{\tr'} (\Q',\C')\) and \((\Sigma',\sigma') \in \Sol(\C')\)
    such that \(\Sigma \subseteq \Sigma'\), \(\Q = \Q'\sigma'\), \(\tr = \tr'\Sigma'\) and \(\Phi = \Phi(\C')\sigma'\norm\).
  \end{enumerate}
\end{proposition}


\subsection{The key tool: the partition tree} 

To decide trace equivalence and labelled bisimilarity, we introduce the novel notion of a \emph{partition tree} of two bounded processes \(P\) and \(Q\).
The point is to build a (finite) tree of all symbolic executions of \(P\) and \(Q\), grouping into the same nodes intermediary processes as follows:

\begin{enumerate}
  \item All processes of a same node should have a \emph{common, unique mgs}.
  Since one symbolic process alone may already have several most general solutions, the node is parametrised by a restricting predicate \(\pi\) on second-order solutions (recall Example \ref{ex:mgs}).
  \item When applying the mgs of a node to all of the processes it contains (and instantiating the potential remaining variables by fresh distinct constants), the resulting frames are \emph{statically equivalent}.
  Conversely, all reachable symbolic processes that would verify this property should be in the node as well.
\end{enumerate}

A branch of this tree therefore represents the set of all equivalent traces of \(P\) and \(Q\) taking a given sequence of visible actions.
Taking profit of this observation we will show that whenever \(P\) and \(Q\) are not trace equivalent or labelled bisimilar, a witness of non-equivalence can be exhibited using the tree.
Formally its nodes are modelled by \emph{configurations} that consist of sets \(\Gamma\) of symbolic processes sharing a unique mgs and statically equivalent solutions.

\begin{definition}[configuration] \label{def:configuration}
  A \emph{configuration} is a pair \((\Gamma,\pi)\) where \(\Gamma\) is a set of symbolic processes and \(\pi\) a predicate on second-order substitutions.
  We also require that:
  \begin{enumerate}
    \item \label{it:configuration-pred-dom}
    the predicate \(\pi\) is defined on \(\vars[2](\Gamma)\), that is, for all \(\Sigma\), \(\pi(\Sigma)\) \textit{iff} \(\pi(\Sigma_{|\vars[2](\Gamma)})\);
    \item \label{it:configuration-unique-mgs}
    for all \((\P,\C) \in \Gamma\), \(|\mgs[\pi](\C)| = 1\);
    \item \label{it:configuration-same-solutions}
    for all \((\P_1,\C_1),(\P_2,\C_2) \in \Gamma\), if \((\Sigma,\sigma_1) \in \Sol[\pi](\C_1)\) then there exists \(\sigma_2\)
    such that \((\Sigma,\sigma_2) \in \Sol[\pi](\C_2)\) and \(\Phi(\C_1) \sigma_1 \StatEq \Phi(\C_2) \sigma_2\).
  \end{enumerate}
\end{definition}

The predicate \(\pi\) can typically be described using second-order (dis)equations.
We then consider trees with nodes labelled by configurations and edges by visible symbolic actions (i.e., not \(\tau\)).
Given a node \(n\) of such a tree, we write \(\Gamma(n)\) and \(\pi(n)\) the components of the corresponding configuration, and \(n \xrightarrow{\alpha} n'\) to express that \(n'\) is a child node of \(n\) through an edge labelled by the symbolic action \(\alpha\).
By definition of a mgs, the points \ref{it:configuration-unique-mgs} and \ref{it:configuration-same-solutions} of Definition \ref{def:configuration} above ensure that all symbolic processes in \(\Gamma(n)\) have the same set of second-order variables,
written \(\vars[2](n)\), and a common and unique mgs, written \(\mgs(n)\).

\begin{definition}[partition tree] \label{def:partition-tree}
  A \emph{partition tree} of two bounded processes \(P\) and \(Q\) is a tree \(T\) whose nodes are labelled by configurations and edges by visible symbolic actions, and that verifies the following properties.
  First of all \(P,Q \in \Gamma(\rootf(T))\) and \(\pi(\rootf(T)) = \top\), where \(\rootf(T)\) denotes the root node of the tree.
  Then for all nodes \(n\) of \(T\), \((\P,\C) \in \Gamma(n)\) and visible symbolic actions \(\alpha\):

  \begin{enumerate}
    \item \label{it:PT-silent}
    \textit{Closure by \(\tau\)-transition}:
    if \((\P,\C) \Sstep{\tau} (\P',\C')\) and \(\Sol[\pi(n)](\C') \neq \emptyset\) then \((\P',\C') \in \Gamma(n)\).

    \item \label{it:PT-child-concrete-derivation}
    \textit{All symbolic transitions are reflected in the tree:}
    if \((\P,\C) \Sstep{\alpha} (\P',\C')\) and \((\Sigma,\sigma) \in \Sol[\pi(n)](\C')\)
    then there exists an edge \(n \xrightarrow{\alpha} n'\) in \(T\) such that \((\P',\C') \in \Gamma(n')\) and \((\Sigma',\sigma) \in \Sol[\pi(n')](\C')\) for some \(\Sigma'\) that coincides with \(\Sigma\) on \(\vars[2](n)\).
  \end{enumerate}

  \noindent Moreover for all edges \(n \xrightarrow{\alpha} n_c\) of \(T\) and \((\P_c, \C_c) \in \Gamma(n_c)\):

  \begin{enumerate}[resume]
    \item \label{it:PT-monotonic}
    \textit{Predicates are refined along branches:} for all \(\Sigma\), if \(\Sigma\) verifies \(\pi(n_c)\) then it verifies \(\pi(n)\).

    \item \label{it:PT-parent-concrete-derivation}
    \textit{Nodes are maximal:}
    if \((\Sigma,\sigma) \in \Sol [\pi(n)](\C)\), \((\Sigma_c,\sigma_c) \in \Sol [\pi(n_c)](\C_c)\) and \(\Sigma \subseteq \Sigma_c\),
    then \(\Gamma(n_c)\) contains all symbolic processes \((\P',\C')\) such that \((\P,\C) \Sstep {\alpha} (\P',\C')\) and, for some substitution \(\sigma'\), \((\Sigma_c,\sigma') \in \Sol(\C')\) and \(\Phi(\C_c) \sigma_c \StatEq \Phi(\C') \sigma'\).
  \end{enumerate}

  \noindent
  The set of partition trees of \(P\) and \(Q\) is written \(\ptree(P,Q)\).
\end{definition}

The set \(\ptree(P,Q)\) is infinite (at least because arbitrarily many processes can be put in the root configuration) but our decision procedures only require to construct one, arbitrary partition tree.
The children \(n'\) of a node \(n\) represent the sets of processes, grouped w.r.t. static equivalence, reachable by one transition from a process of \(n\).
Item \ref{it:PT-child-concrete-derivation} ensures that all cases are covered, that is, for all symbolic transitions from \(n\) and all solutions \(\Sigma\), at least one child \(n'\) should contain the resulting symbolic process.
Note that we do not impose that \(\Sigma\) verifies \(\pi(n')\), but that there exists another solution \(\Sigma'\) computing the same first-order terms that does.
This more permissive approach will allow us, when generating partition-tree nodes in Section~\ref{sec:ptree}, to use families of predicates \(\pi\) that only consider solutions of a certain form (which therefore requires to prove that any deducible term can be computed by a recipe of this form).
Item \ref{it:PT-parent-concrete-derivation} then formalises that the nodes are saturated under static equivalence:
if \(n'\) is a child of \(n\) and a symbolic transition \(A \Sstep {\alpha} B\) from a process \(A \in \Gamma(n)\) may result into a process statically equivalent to a process \(C \in \Gamma(n')\) then \(B\) should be in \(\Gamma(n')\) as well.

\begin{example}
  Let us draw a partition tree corresponding to an anonymity analysis in the private authentication protocol, simplified for readability.
  We consider the following light version of the role of the process \(B\) accepting a connection from an agent \(X\), removing the identification nonces \(N_A,N_B\) from the protocol and replacing the decoy message by a fresh name \(r\):
  \[\begin{array}{l@{\ }l}
    B_X = & \InP {c} {x}. \IfP\ \radec(x,\sks_B) = \pks_X\ \ThenP\ \OutP {c} {\aenc(\okfun,r,\pks_X)}\ \ElseP\ \OutP {c} {r}
  \end{array}\]
  We consider a 3-agent scenario (\(A,B,C\)) where \(A\) has already emitted \(\aenc(\pks_A,r_A,\pks_B)\) to initiate a communication with \(B\). The security property we study is whether the identity of \(B\)'s accepted recipient remains anonymous.
  That is we want to prove \(P \approx Q\) where
  \begin{align*}
    P & = C[B_A] & Q & = C[B_C] & C[R] & = \OutP {c} {\pks_A}.\ \OutP {c} {\pks_B}.\ \OutP {c} {\pks_C}.\ \OutP {c} {\aenc(\pks_A, r_A,\pks_B)}.\ R
  \end{align*}
  
  \begin{figure}[ht]
    \centering
    \includegraphics[width=0.92\textwidth]{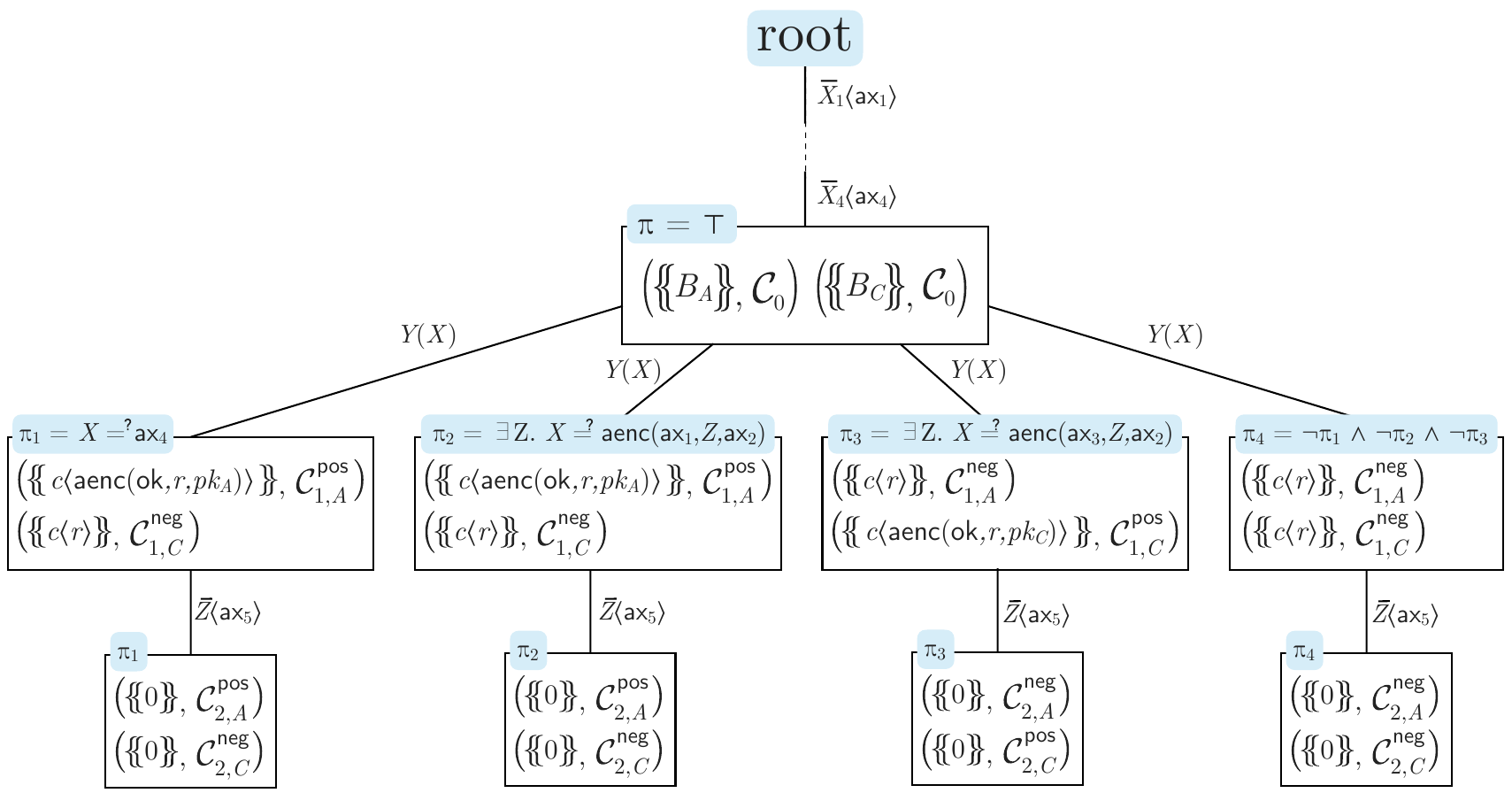}
    \caption{A simplified partition tree of \(P\) and \(Q\)}
    \label{fig:partition-tree}
  \end{figure}

  The partition tree in Figure \ref{fig:partition-tree} has been lightened for readability: 
  if a node contains two symbolic processes \(A_s,B_s\) such that \(A_s \sstep{\tau} B_s\), then \(A_s\) is omitted from the node (as it contains less constraints than \(B_s\) anyway).
  The configuration at the root of the tree only contains \(P\) and \(Q\).
  After the four initial outputs of the context \(C\), we reach the constraint system \(\C_0\) defined by:
  \begin{align*}
    \Phi(\C_0) & = \{\ax_1 \mapsto \pks_A, \ax_2 \mapsto \pks_B, \ax_3 \mapsto \pks_C, \ax_4 \mapsto \aenc(\pks_A, r_A,\pks_B)\} \\
    \Df(\C_0) & = X_1 \dedfact x_1 \wedge X_2 \dedfact x_2 \wedge X_3 \dedfact x_3 \wedge X_4 \dedfact x_4 \\
    \Eqfst(\C_0) & = x_1 \eqs c \wedge x_2 \eqs c \wedge x_3 \eqs c \wedge x_4 \eqs c
  \end{align*}
  The next step is the first one inducing a non-trivial case analysis.
  This node has four children for the adversary to compute the input \(\InP {c}{x}\):
  \begin{enumerate*}[label=\(\pi_{\arabic*}\)]
    \item forwards the message of \(A\),
    \item forges a message pretending it is from \(A\),
    \item forges a message pretending it is from \(C\),
    \item any other case.
  \end{enumerate*}
  The choice of these 4 cases is guided by the conditional \(\IfP\ \radec(x,\sks_B) = \pks_X\) (where $X=A$ or $X=C$) that is evaluated on the input.
  Choice $\pi_1$ results in the positive branch in both $B_A$ and $B_C$, as it corresponds to an honest execution. Choice $\pi_2$ results in choosing the positive branch in $B_A$ and the negative branch in $B_C$, while $\pi_3$ does the opposite. Choice \(\pi_4\) leads to the negative branch in all cases by construction (as it is the negation of the 3 previous cases). 
  
  More precisely we write \(\Phi(\C_{1,X}^{\poslab}) = \Phi(\C_{1,X}^{\neglab}) = \Phi(\C_0)\), \(\Df(\C_{1,X}^{\poslab}) = \Df(\C_{1,X}^{\neglab}) = \Df(\C_0) \wedge Y \dedfact y\) and
  \begin{align*}
    \Eqfst(\C_{1,X}^{\poslab}) & = \Eqfst(\C_0) \wedge y \eqs c \wedge x \eqs \aenc(\pks_X,x',\pks_B) \\
    \Eqfst(\C_{1,X}^{\neglab}) & = \Eqfst(\C_0) \wedge y \eqs c \wedge \forall x'.\, x \neqs \aenc(\pks_X,x',\pks_B)
  \end{align*}
  Then the final transitions simply execute the resulting outputs, i.e. \(\C_{2,X}^s\), \(s \in \{\poslab,\neglab\}\), is obtained by adding \(Z \dedfact z\) and \(z \eqs c\) to \(\C_{1,X}^s\).
  Since a ciphertext is indistinguishable from a nonce, the two outputs always end up in the same nodes;
  that is, all leaves contain at least one process originated from \(P\) and at least one from \(Q\), which is how we prove trace equivalence.
  The situation would be different with a rewrite rule such as \(\testaenc(\aenc(x,y,\pk(z))) \to \okfun\);
  a partition tree of \(P\) and \(Q\) with this extended rewriting system can be found in Figure \ref{fig:partition-tree-attack}.

  \begin{figure}[ht]
    \centering
    \includegraphics[width=0.95\textwidth]{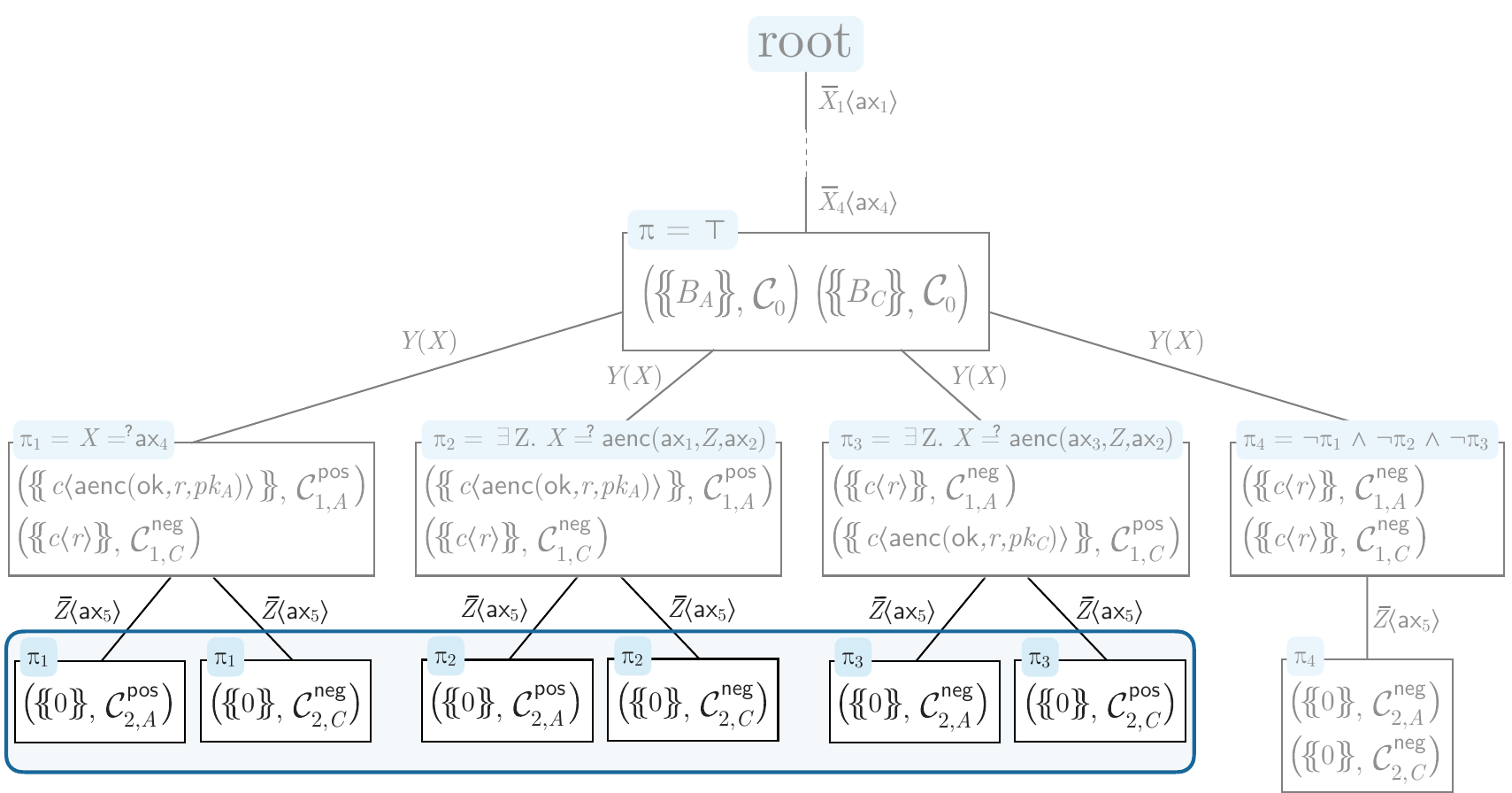}
    \caption{Partition tree with the rewriting system extended with \(\testaenc(\aenc(x,y,\pk(z))) \to \okfun\)}
    \label{fig:partition-tree-attack}
  \end{figure}

  We highlighted the part differing from the previous tree.
  Essentially some leaf nodes have been split in two due to the enhanced capabilities of the adversary to disprove static equivalence, inducing a violation of trace equivalence.
  For example the leftmost leaf's mgs is
  \[\{X \mapsto \ax_4, X_1 \mapsto c, \ldots, X_4 \mapsto c, Y \mapsto c, Z \mapsto c\}\]
  which corresponds to an attack trace where the attacker forwards the message of \(A\) and observes whether the response of \(B\) is a ciphertext, which reveals whether \(B\) accepts connections from \(A\) or not.
\end{example}

In the remaining of the section we formalise how to decide trace equivalence and labelled bisimilarity of two processes, given a partition tree and the mgs of each of its nodes.
For that we will rely on the following notion of reduction, characterising symbolic traces viewed as branches of a partition tree:

\begin{definition}[partition-tree trace]
  Given  a partition tree  \(T\)
  we write \((\P,\C), n \tstep{\alpha} (\P',\C'), n'\) when:
  \begin{enumerate}
    \item \(n\) and \(n'\) are nodes of \(T\) such that \((\P,\C) \in \Gamma(n)\) and \((\P',\C') \in \Gamma(n')\); and
    \item if \(\alpha = \tau\) then \(n = n'\), otherwise \(n \xrightarrow{\alpha} n'\) and \((\P,\C) \sstep{\alpha} (\P',\C')\).
  \end{enumerate}
  For convenience this notion is to be understood up to alpha renaming of the variables of the symbolic action \(\alpha\).
  We write \(A^s_0,n_0 \Tstep{\tr} A^s_p,n_p\) instead of \(A^s_0,n_0 \tstep{\alpha_1} \cdots \tstep{\alpha_p} A^s_p,n_p\) if \(\tr\) is the word obtained after removing \(\tau\) symbols from \(\alpha_1 \cdots \alpha_p\).
  If \(P\) is a plain process we may also write \(P \Tstep {\tr} A_s,n\) instead of \((\multi{P}, (\emptyset,\top,\top)), \rootf(T) \Tstep{\tr} A_s,n\).
\end{definition}

\subsection{Decision procedures for equivalence} \label{sec:ptree-eq}
In this section, we assume that we managed to compute a partition tree \(T \in \ptree(P_1,P_2)\) (in particular, that there exists one).
We then describe how to derive a decision procedure for trace equivalence and labelled bisimilarity from \(T\).

\paragraph{Trace equivalence}
As hinted in our various examples, deciding trace equivalence can be reduced to an analogue notion of equivalence using the (finite) transition relation \(\tstep{\alpha}\) instead of the concrete semantics \(\cstep{\alpha}\).
This is formalised by the following theorem:

\begin{theorem}[restate=thmTraceEquivPtree,name={partition-tree-based characterisation of trace equivalence}] \label{thm:trace-equiv-ptree}
  Whenever \(T \in \ptree(P_1,P_2)\),
  the following points are equivalent:
  \begin{enumerate}
    \item \label{it:trace-equiv-ptree-incl}
    \(P_1 \TraceIncl P_2\)
    \item \label{it:trace-equiv-ptree-trace}
    for all partition-tree traces \(P_1 \Tstep{\tr} (\P_1,\C_1),n\), we have \(P_2 \Tstep{\tr} (\P_2,\C_2),n\)
  \end{enumerate}
\end{theorem}

The proof of this result mostly follows from a combination of the soundness and completeness of the symbolic semantics, with two technical lemmas generalising the properties of the partition tree from edges to branches.
The detailed statements and proofs can be found in Appendix~\ref{app:decision-proc-trace-from-ptree}.

\paragraph{Simulations}
In the case of trace equivalence, a witness that \(A \not \TraceEq B\) was simply a trace of \(A\) or \(B\) that has no equivalent trace in the other process.
The case of labelled bisimilarity is however more involved.
Using vocabulary borrowed from game theory, the definition of bisimilarity can be seen as a \emph{prover-disprover game}:
at each state of the game the disprover chooses a transition from one of the two processes and the prover answers by choosing a transition of the same type from the other process (plus some potential \(\tau\)-transitions).
The disprover wins the game if they manage to reach a state with non-statically-equivalent processes or if the prover cannot answer to one of the moves.
A witness of non-equivalence is thus a winning strategy for the disprover.
We formalise this below,
recalling that if \(\alpha\) is an action, we write \(\bar{\alpha} = \alpha\) if \(\alpha \neq \tau\) and \(\bar{\alpha} = \epsilon\) if \(\alpha = \tau\).

\begin{definition}[witnesses] 
A \emph{witness of non-bisimilarity} \(\witness\) is a set of pairs \((A_0,A_1)\) verifying the following two conditions:
\begin{enumerate}
  \item \(A_0\) and \(A_1\) are ground extended processes such that \(A_0 \StatEq A_1\)
  \item there exists \(b \in \{0,1\}\) and a transition \(A_b \cstep{\alpha} A_b'\) such that for all traces \(A_{1-b} \Cstep{\bar{\alpha}} A_{1-b}'\), either
  \(A_0' \not\StatEq A_1'\) or \((A_0',A_1') \in \witness\).
\end{enumerate}
We say that in addition that \(\witness\) is a \emph{witness of non-simulation} if the above two conditions can always be satisfied with \(b = 0\).
We say that \(\witness\) is a witness for \((A_0,A_1)\) if \((A_0,A_1) \in \witness\).
\end{definition}

Note that the witness can be seen as a relation corresponding to the negation of the definition of bisimilarity (\Cref{def:(bi)simulation}) minus the static equivalence, i.e. $\not\LabBis \setminus \not\StatEq$.

\begin{proposition}[witness-based characterisation of labelled bisimilarity] \label{prop:concrete-witness}
If \(A_0 \StatEq A_1\) then:
\begin{enumerate}
  \item \(A_0 \not\LabBis A_1\) \textit{iff} there exists a witness of non-bisimilarity \(\witness\) for \((A_0,A_1)\)
  \item \(A_0 \not\Simu A_1\) \textit{iff} there exists a witness of non-simulation \(\witness\) for \((A_0,A_1)\)
\end{enumerate}
\end{proposition}
\begin{proof}
We only give the proof in the case of \(\LabBis\), as the proof for \(\Simu\) is analogue.
First, we observe that \(A_0 \not \LabBis A_1\) \textit{iff} there exists a binary relation \(\disim\) on ground extended processes such that \(A_0 \disim A_1\) and, for all \((B_0, B_1) \in \disim\), either
\begin{enumerate*}
  \item \(B_0 \not \StatEq B_1\), or
  \item there exists \(b \in \{0,1\}\) and a transition \(B_b \cstep{\alpha} B_b'\) such that for all traces \(B_{1-b} \Cstep{\bar{\alpha}} B_{1-b}'\),
  \(B_0' \disim B_1'\).
\end{enumerate*}
Let us call such a relation \(\disim\) a \emph{labelled attack on \((A_0,A_1)\)}.
Since processes are bounded there exist no infinite sequences of transitions and for all $A,B$, \(A \not \LabBis B\) therefore straightforwardly rephrases to the existence of a labelled attack \(\disim\) such that \(A \disim B\).
It then suffices to observe that
\begin{enumerate}
  \item If \(\disim\) is a labelled attack on \((A_0,A_1)\) then \(\disim\ \smallsetminus \not\StatEq\) is a witness for \((A_0,A_1)\).
  \item If \(\witness\) is a witness for \((A_0,A_1)\) then \(\witness\ \cup \not\StatEq\) is a labelled attack on \((A_0,A_1)\). \qedhere
\end{enumerate}
\end{proof}

We now define a symbolic variant of the notion of witness that can be constructed within a partition tree \(T\).
In essence, a symbolic witness may be seen as a winning strategy for the disprover in a bisimulation game limited to the finite transition relation \(\tstep{}\).

\begin{definition}[symbolic witnesses]
A \emph{symbolic witness of non bisimilarity} \(\witness_s\) w.r.t. a partition tree \(T\) is a finite tree whose nodes \(N\) are labelled by tuples \((A_0,n)\) or \((A_0,A_1,n)\) with \(n\) a node of \(T\) and \(A_0,A_1 \in \Gamma(n)\).
We also require that if \(N\) is labelled \((A_0,A_1,n)\), there exist \(b \in \{0,1\}\) and a transition \(A_b,n \tstep{\alpha} A'_b, n'\) (possibly \(\alpha = \tau\)) such that:
\begin{enumerate}
  \item If \(A_{1-b}\) is not reducible by \(\Tstep{\bar{\alpha}}\) then \(N\) has a unique child labelled \((A_b',n')\);
  \item otherwise the children of \(N\) are the nodes labelled \((A'_0,A'_1,n')\), \(A_{1-b}, n \Tstep{\bar{\alpha}} A'_{1-b}, n'\).
\end{enumerate}
We say that \(\witness_s\) is a \emph{witness of non-simulation} if the above two conditions can always be satisfied with \(b = 0\).
We say that \(\witness_s\) is a symbolic witness for \((A_0,A_1,n)\) when \(\rootf(\witness_s)\) is labelled by \((A_0,A_1,n)\).
\end{definition}

However purely symbolic witnesses do not exhibit consistent proofs of non-equivalence in general.
Indeed, while a concrete execution fixes the effective value of an input \(x\) at the moment it is performed, a symbolic execution records constraints on \(x\) all along the execution.
Rephrasing, the symbolic semantics puts the prover at a disadvantage in the game, since they have to answer to the disprover's input actions without knowing the values of the input terms.
Symbolic witnesses inducing invalid winning strategies for the disprover will be discarded by their absence of \emph{solutions} in the following sense:


\begin{definition}[solution of a symbolic witness] \label{def:solution-witness}
Let \(\witness_s\) be a symbolic witness.
A \emph{solution} of \(\witness_s\) is a function \(\fsol\) that maps nodes of \(\witness_s\) to ground second-order substitutions such that for all nodes \(N\) labelled \((A_0,n)\) or \((A_0,A_1,n)\),
\begin{enumerate}
  \item for all \(A_b = (\P,\C)\), \((\fsol(N),\sigma) \in \Sol[\pi(n)](\C)\) for some \(\sigma\);
  \item for all children nodes \(N_1,N_2\) of \(N\), \(\fsol(N) \subseteq \fsol(N_1) = \fsol(N_2)\).
\end{enumerate}
We denote \(\Sol(\witness_s)\) the set of solutions of \(\witness_s\).
\end{definition}


\begin{theorem}[restate=thmLabBisPtree,name={partition-tree-based characterisation of labelled bisimilarity}] \label{thm:ptree-lab-bis}
If \(T \in \ptree(P_1,P_2)\):
\begin{enumerate}
  \item \(P_1 \LabBis P_2\) \textit{iff} for all symbolic witnesses of non-bisimilarity \(\witness_s\) for \((P_1,P_2,\rootf(T))\), we have \(\Sol(\witness_s) = \emptyset\)
  \item \(P_1 \Simu P_2\) \textit{iff} for all symbolic witnesses of non-simulation \(\witness_s\) for \((P_1,P_2,\rootf(T))\), we have \(\Sol(\witness_s) = \emptyset\)
\end{enumerate}
\end{theorem}

The proof, although technical, simply connects the symbolic witnesses to concrete ones using the soundness and completeness of the symbolic semantics as well as the properties of the partition tree, following similar ideas as the analogue proof for trace equivalence.
The detailed proof can be found in Appendix~\ref{app:decision-proc-bisim-from-ptree}.

Assuming one has computed a partition tree \(T \in \ptree(P_1,P_2)\) and the mgs of each of its nodes, since there are finitely-many possible symbolic witnesses, Theorem \ref{thm:ptree-lab-bis} yields a decision procedure for the labelled bisimilarity of \(P_1\) and \(P_2\) provided one can decide whether a given symbolic witness has a solution.
For that we rely on a simple, bottom-up unification of the mgs' appearing in the witness;
details can be found in Section~\ref{sec:witness-complexity} where we study more precisely the complexity of partition-tree-based decision procedures.

\subsection{Generating partition trees (with  a constraint-solving oracle)} \label{sec:overview-ptree}


In this section we detail the skeleton of the procedure for computing a partition tree of two plain processes \(P_1\) and \(P_2\).
The description is modular in that most of the technical details, in particular the modelling of the node predicates and how we obtain the expected properties of the tree, are abstracted by a \emph{constraint-solving oracle} that we detail in the next sections.
This section should therefore be seen as the overview of the whole algorithm for deciding equivalence properties, which gives enough insight to discuss our implementation.

The algorithm generates the nodes of the tree top-down, that is, from the root to the leaves.
We outline the procedure in Figure \ref{fig:overview-ptree}.

\begin{figure}[ht]
  \centering
  \includegraphics[width=0.99\textwidth]{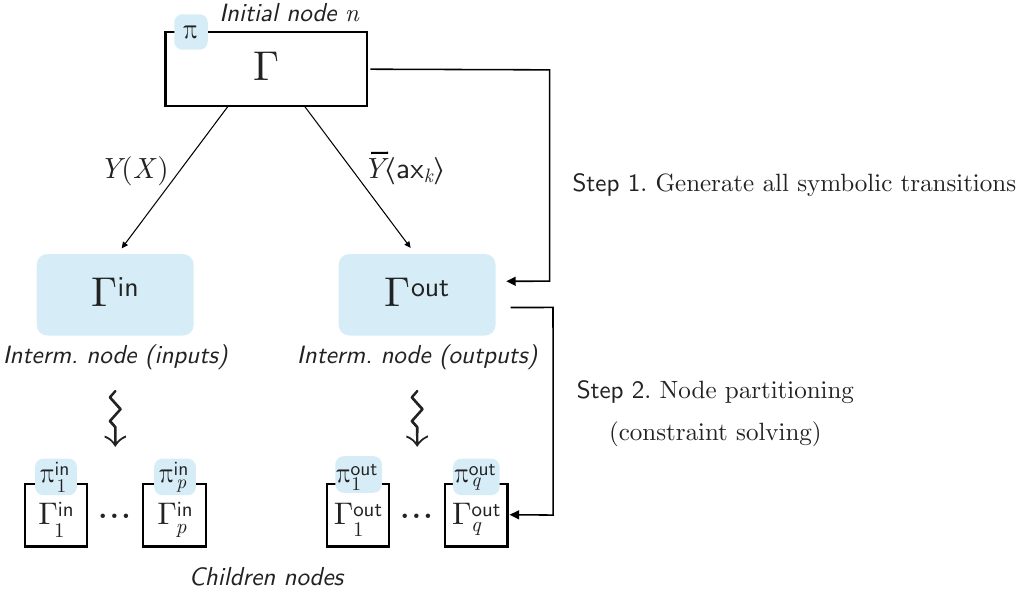}
  \caption{Computing the subtree of a partition tree rooted in a node \(n\)}
  \label{fig:overview-ptree}
\end{figure}

Let us now describe the algorithm to compute \(T \in \ptree(P_1,P_2)\) in more details, up to the technical developments detailed in the next sections.
\begin{enumerate}
  \item First, we initiate a root containing \(P_1\) and \(P_2\) and saturate the configuration by \(\tau\) transitions.
  That is, we consider the set of symbolic processes
  \[\Gamma(\rootf(T)) = \left\{(\P,\C) \mid P_i \Sstep{\tau} (\P,\C), i \in \{1,2\}, \Sol(\C) \neq \emptyset\right\}\]
  Note that the constraint systems \(\C\) involved in this definition do not contain deduction facts, which makes the decision of the emptiness of \(\Sol(\C)\) relatively straightforward.
  Using the terminology of the later Section~\ref{sec:ptree}, using \emph{simplification rules} permits to put the constraints into a simple form where the existence of a solution is trivial to decide.
  \item Then let us assume we already constructed a node \(n\) of the tree using this algorithm, in particular the corresponding configuration \((\Gamma(n),\pi(n))\).
  To compute the children of \(n\) we first enumerate all symbolic transitions from processes of \(\Gamma(n)\), separating input and output actions.
  That is, we compute the two sets
  \begin{align*}
    \Gamma^\inp & = \left\{B \mid A \in \Gamma(n), A \Sstep{\InP{Y}{X}} B\right\} &
    \Gamma^\outp & = \left\{B \mid A \in \Gamma(n), A \Sstep{\OutP{Y}{\ax_p}} B\right\}
  \end{align*}
  \item \(\Gamma^\inp\) and \(\Gamma^\outp\) are two intermediary sets that do not satisfy yet the father-child properties of the partition tree.
  For that we use a constraint-solving algorithm detailed in Chapter \ref{sec:ptree} (\emph{simplification rules} again, but also \emph{case distinction rules}) that will partition \(\Gamma^\inp\) and \(\Gamma^\outp\) to gather symbolic processes with statically-equivalent solutions and remove those with no solutions.
  This constraint solving results into a sequence of configurations
  \begin{align*}
    (\Gamma_1^\inp,\pi_1^\inp), \ldots, (\Gamma_p^\inp,\pi_p^\inp) & &
    (\Gamma_1^\outp,\pi_1^\outp), \ldots, (\Gamma_q^\outp,\pi_q^\outp)
  \end{align*}
  that will label the children of \(n\).
  The procedure is then carried out recursively from these child nodes until no more symbolic transitions are available.
\end{enumerate}

In Section~\ref{sec:ptree} we detail the missing parts of this procedure that take the form of constraint-solving rules, in the context of constructor-destructor subterm convergent theories.
Note that the approach is modular in that the proofs we have carried so far are independent of the assumptions on the rewriting system:
generalising the results of Section~\ref{sec:ptree} will automatically result in the decidability of trace equivalence and labelled bisimilarity of bounded processes for the extended class of theories.

\subsection{Implementation and performances} \label{sec:implem}

\paragraph{The \deepsec prover}
Building on the procedure's structure described above and the internal solver developed in the next sections, we have implemented a prototype in OCaml, called \deepsec (DEciding Equivalence Properties in SECurity protocols).
The user specifies a rewriting system (that is checked to be constructor-destructor and subterm convergent by the tool), two bounded processes, and the tool verifies whether they are trace equivalent.
If not, a concrete attack trace is returned in a dedicated graphical interface;
we refer to the \deepsec's website for development credits, tutorials and details on practical usage~\cite{website}:
\begin{center}
  \url{https://deepsec-prover.github.io/}
\end{center}

The tool's specification language implements the grammar presented in Section~\ref{sec:model}, including some syntax extensions for non-deterministic choice, private function symbols, a restricted form of patterned \(\LetP\) bindings, as well as bounded replication \(\BangP[n] P\) defined as \(n\) parallel copies of \(P\).
These additional primitives should mostly be seen as syntactic sugar, although the native integration allowed specific optimisations compared to encodings within the initial calculus.
The syntax and structure of \deepsec's input files are similar to the widely used \proverif tool~\cite{manual-proverif} to make it easier for new users to discover and use it.

\paragraph{Trace equivalence vs (bi)simulation}

  The tool currently only implements the trace equivalence procedure as it is rather efficient. Following \Cref{thm:trace-equiv-ptree}, the procedure for checking trace equivalence between $P_1$ and $P_2$ consists in generating the partition tree and checking that each node contains symbolic constraint systems both from $P_1$ and $P_2$. As different branches of the partition tree are independent from one another, the implementation only requires to store in memory the current branch that is being verified, instead of the whole partition tree. On the other hand, the procedure for checking (bi)simulation both requires to compute and store in memory the full partition tree. In addition, the procedure also requires guessing a symbolic witness, which would be extremely inefficient. A natural follow up to our work would be to explore ways of effectively implementing the decision procedure for (bi)simulation that would avoid these two main hurdles.

\paragraph{Partial order reductions}
The tool also implements \emph{partial order reductions} (POR), an optimisation technique for protocol analysis developed by Baelde et al.~\cite{BDH15}.
The basic idea is to discard part of the state space that is redundant but this optimisation is only sound when processes are \emph{action-determinate}, as defined in~\cite{BDH15}.
Although we omit here the definition of determinacy for simplicity, let us mention that not using private channels and assigning a different channel name to each parallel process is a simple, syntactic way to ensure this property.
This is however not always possible---typically when looking at some anonymity or unlinkability properties.
Typically, the private authentication protocol used as a running example can be modelled as a determinate process, but not the Helios and BAC protocols (due to private channels or because this introduces artificial violations of the equivalence property).

In practice, \deepsec automatically detects action-determinate processes and activates the POR, which drastically reduces the number of symbolic executions that need to be considered.
We also go further and allow to verify a refined equivalence, \emph{equivalence by session}, that allows to use similar POR techniques without the restriction to determinate processes.
This contribution is however out of the scope of this article;
details can be found in~\cite{CKR19} and our experimental results presented below only include the base POR of~\cite{BDH15}.

\paragraph{Distributing the computation}
The main task of \deepsec is to generate a partition tree and, as we explained, this is done using a top-down approach.
This task can be distributed as computing a given node of the tree can be done independently of its sibling nodes.
However, some engineering is needed to avoid heavy communication overhead due to task scheduling.
Indeed, the partition tree is not a balanced tree and we do not know in advance which branches will be larger than others.
Because of this, we do not directly compute and return the children of each node in the most straightforward manner, but proceed in two steps:
\begin{enumerate}
  \item \label{it:parall-step-1} We start with a breadth-first generation of the partition tree.
  The number of pending nodes will gradually grow until eventually exceeding a threshold parameter \(n\).
  \item \label{it:parall-step-2} Each available core focuses on one of these nodes, computes the whole subtree rooted in this node in a depth-first manner and, when this the task is completed, is assigned to a new node until none remain.
\end{enumerate}
If some cores become idle for too long in Step \ref{it:parall-step-2} (because the number of cores exceeds the number of non-completed nodes), we perform a \emph{new round}, that is, we interrupt the working nodes and restart this two-step procedure on incomplete nodes.
Although doing so wastes some proof work, this improves performances for particularly unbalanced trees.
Note that parallelisation is also supported by other automated analysers such as \akiss~\cite{CCC16}, but \deepsec goes one step further as it is able to distribute the computation not only on multiple cores of a given machine but also clusters of computers.

\paragraph{Benchmarks}
We performed extensive benchmarks to compare \deepsec against other tools that verify equivalence properties for a bounded number of sessions:
\akiss~\cite{CCC16}, \apte~\cite{C14}, \satequiv~\cite{CDD17} and \spec~\cite{TNH16}.
Experiments are carried out on Intel Xeon 3.10GHz cores, with 40Go of memory.
We distributed the computation on 20 cores for \akiss and \deepsec as they support parallelisation---unlike the others which therefore use a single core.
The results are summarised in Table \ref{fig:bench}
with the following symbol conventions:

\begin{center}
  \begin{tabular}{cl}
    \verified & analysis terminates and equivalence holds \\
    \attacksimple & analysis terminates and an attack is found \\
    \outofmemory & analysis aborted due to memory overflow (Out of Memory)\\
    \outoftime & analysis aborted due to timeout (12 hours) \\
    \unable & the tool is not expressive enough to analyse the protocol
  \end{tabular}
\end{center}

We first analysed strong secrecy and anonymity for several classical authentication protocols.
The \deepsec tool clearly outperforms \akiss, \apte, and \spec.
The \satequiv tool becomes more efficient, when the number of sessions significantly increases.

To put more emphasis on the broad scope we also include analyses of unlinkability and anonymity properties for a number of other protocols.
This includes the Private authentication protocol used as a running example, BAC~\cite{P04} and the Helios voting protocol~\cite{A08}.
In addition we study a simplified version of the AKA protocol deployed in 3G telephony networks without XOR~\cite{AMR12}, the Passive Authentication protocol implemented in the European passport~\cite{P04}, as well as the Prêt-à-Voter protocol (PaV)~\cite{RS06}.
Note that, while PaV is a priori in the scope of \akiss, it failed to produce a proof: \akiss only approximates trace equivalence of non-determinate processes and finds a false attack here.
Finally we note that BAC, PaV and Helios protocols are not action-determinate and therefore do not benefit from the POR optimisation, which explains the much higher verification times when increasing the sessions.
Nevertheless, as exemplified by some examples, attacks may be found very efficiently, as it generally does not require to explore the entire state space.

\begin{table}
  \centering
  \small
  \renewcommand{\arraystretch}{0.945} 
  \setlength\tabcolsep{6pt}
  \scalebox{1}{%

  \begin{tabular}{cccc@{}cc@{}cc@{}cc@{}cc@{}c}
    \specialrule{1.2pt}{1pt}{1pt}
    \multicolumn{1}{c}{} & \multicolumn{2}{c}{\bf Protocol (\# of roles)} &\multicolumn{2}{c}{\akiss} & \multicolumn{2}{c}{\apte} & \multicolumn{2}{c}{\spec} & \multicolumn{2}{c}{\satequiv} & \multicolumn{2}{c}{\deepsec} \\
    \specialrule{1.2pt}{1pt}{1pt}
    \light & \multirow{6}{*}{Denning-Sacco} & 3 &\verified &$<$1s&\verified &$<$1s&\verified &11s &\verified &$<$1s&\verified &$<$1s\\
    \light & & 6 &\verified &$<$1s 	&\verified & 1s& \multicolumn{2}{c}{\outofmemory} &\verified &$<$1s &\verified & $<$1s\\
    \light & & 7 &\verified & 6s &  \verified & 3s & & &\verified &  $<$1s  & \verified & $<$1s\\
    \light & & 10 &\multicolumn{2}{c}{\outofmemory} &\verified & 9m49 & & &\verified &  $<$1s  &\verified &  $<$1s\\
    \light & & 12 & & &\multicolumn{2}{c}{\outoftime} & &  &\verified  & $<$1s  & \verified &$<$1s\\
    \light & & 29 & & & & & & &\verified  & $<$1s  &\verified  & 1s\\
    \cline{2-13}
    \light&\multirow{7}{*}{Wide Mouth Frog} & 3 & \verified & $<$1s & \verified  &$<$1s & \verified  &5s &  \verified &$<$1s&\verified & $<$1s\\
    \light& & 6 & \verified & $<$1s &  \verified &$<$1s &\verified & 1h11m & \verified & $<$1s & \verified & $<$1s\\
    \light& & 7 & \verified & $<$1s & \verified & 1s & \multicolumn{2}{c}{\outofmemory} &  \verified &$<$1s&\verified & $<$1s\\
    \light& & 10 & \verified & 10s & \verified & 3m35 & & & \verified &$<$1s &\verified & 1s\\
    \light& & 12 & \verified & 22m16s &\multicolumn{2}{c}{\outoftime} & & &\verified &$<$1s  &\verified  & $<$1s\\
    \light& & 14 & \multicolumn{2}{c}{\outofmemory} & & & & &\verified &$<$1s  & \verified & $<$1s\\
    \light & & 23 & & & & & & & \verified &$<$1s &\verified  &1s\\
    \cline{2-13}
    \light &\multirow{7}{*}{Yahalom-Lowe} & 3 & \verified & $<$1s & \verified & $<$1s & \verified & 7s & \verified &$<$1s & \verified & $<$1s\\
    \light & & 6 & \verified & 2s & \verified & 41s & \multicolumn{2}{c}{\outofmemory} & \verified &$<$1s & \verified & $<$1s\\
    \light & & 7 & \verified & 42s & \verified & 34m38s & & &\verified &1s  & \verified &$<$1s \\
    \light & & 10 & \multicolumn{2}{c}{\outofmemory} & \multicolumn{2}{c}{\outoftime} & &&\verified & 1s  & \verified &$<$1s\\
    \light & & 12 & & && && & \verified &4s & \verified &2s\\
    \light & & 14 & & && && & \verified &7s & \verified &2s\\
    \light \multirow{-20}{*}{\begin{sideways}Strong secrecy\end{sideways}} & & 17 & & && && & \verified &12s &\verified &8s\\
    \hline
    \light & \multirow{6}{*}{\begin{tabular}{@{}c@{}}Private\\Authentication\end{tabular}} & 2 &  \verified & $<$1s &  \verified & $<$1s & \multicolumn{2}{c}{\multirow{6}{*}{\unable}} & \multicolumn{2}{c}{\multirow{6}{*}{\unable}} & \verified & $<$1s\\
    \light & & 4 & \verified & $<$1s &\verified &  1s &&  & & &\verified & $<$1s\\
    \light & & 6 & \verified & 21s &\verified &  4m18s & && & & \verified & $<$1s\\
    \light & & 8 & \multicolumn{2}{c}{\outofmemory} & \multicolumn{2}{c}{\outoftime} && & & & \verified & 1s\\
    \light & & 10 & & & & & && & &\verified &  2s\\
    \light \multirow{-6}{*}{\light\begin{sideways}Anonymity\end{sideways}} & & 15 & & & & && & & & \verified & 32s\\
    \hline
    \light & \multirow{2}{*}{\begin{tabular}{@{}c@{}}3G-AKA\end{tabular}}& 4 & \verified & 1m35s & \verified & 1h23m & \multicolumn{2}{c}{\multirow{2}{*}{\unable}} & \multicolumn{2}{c}{\multirow{2}{*}{\unable}} &  \verified &$<$1s\\
    \light & & 6 & \multicolumn{2}{c}{\outofmemory} & \multicolumn{2}{c}{\outoftime} && & & & \verified & 2s\\
    \cline{2-13}
    \light &\multirow{6}{*}{\begin{tabular}{@{}c@{}}Passive\\ Authentication\end{tabular}} & 4 & \verified &$<$1s & \verified & 1s & \multicolumn{2}{c}{\multirow{6}{*}{\unable}} & \multicolumn{2}{c}{\multirow{6}{*}{\unable}} &\verified &$<$1s\\
    \light & & 6 & \verified & 2m15s & \verified & 1m27s & && & &\verified &$<$1s\\
    \light & & 7 & \verified & 1h40m & \verified & 1m44s & && & &\verified &1s\\
    \light & & 9 & \multicolumn{2}{c}{\outoftime} &\verified &  2h08m && & & &\verified &$<$1s\\
    \light & & 15 & & &\multicolumn{2}{c}{\outoftime} & && & & \verified & 9s\\
    \light & & 21 & & && & && & & \verified & 15s\\
    \cline{2-13}
    \light & \multirow{2}{*}{\begin{tabular}{@{}c@{}}BAC\end{tabular}} & 4 & \multicolumn{2}{c}{\outofmemory} & \attacksimple &38m56s & \multicolumn{2}{c}{\multirow{2}{*}{\unable}} & \multicolumn{2}{c}{\multirow{2}{*}{\unable}} & \attacksimple& 1s\\
    \light \multirow{-10}{*}{\light\begin{sideways}Unlinkability\end{sideways}} & & 6 & & &\multicolumn{2}{c}{\outoftime} & & & & & \multicolumn{2}{c}{\outoftime}\\
    \hline
    \light & \begin{tabular}{@{}c@{}}Pr\^et-\`a-Voter\end{tabular} & 6 &\multicolumn{2}{c}{\unable}&\multicolumn{2}{c}{\unable} &\multicolumn{2}{c}{\unable} &\multicolumn{2}{c}{\unable} &\verified & 2s\\
    \cline{2-13}
    \light & \begin{tabular}{@{}c@{}}Helios Vanilla\end{tabular} & 6 & \attacksimple &47s & \attacksimple &$<$1s & \multicolumn{2}{c}{\unable} &\multicolumn{2}{c}{\unable} & \attacksimple &$<$1s\\
    \cline{2-13}
    \light & \multirow{3}*{Helios ZKP (vote swap)} & 10 & \multicolumn{2}{c}{\multirow{3}{*}{\outofmemory}} & \multicolumn{2}{c}{\multirow{3}{*}{\unable}} & \multicolumn{2}{c}{\multirow{3}{*}{\unable}} & \multicolumn{2}{c}{\multirow{3}{*}{\unable}} & \verified & <1s\\
    \light & & 11 & & & & & & & & & \verified & 7m 24s\\
    \light \multirow{-5}{*}{\light\begin{sideways}Ballot privacy\end{sideways}} & & 12 & & & & & & & & & \verified & 1h 38m\\
    \hline
  \end{tabular}}

  \caption{Performances of \deepsec (20 cores) against other protocol analysers}
  \label{fig:bench}
\end{table}

\section{Generation of the partition tree}
\label{sec:ptree}
  

In the previous sections, we detailed how to use the partition tree to derive decision procedures for equivalence properties.
We describe in this section a constraint solving procedure that may be used to generate one in practice.

\subsection{Extended constraint systems} \label{sec:ext-csys}
In order to carry out the constraint solving required to construct the partition tree, we extend constraint systems with components allowing to reason more finely about the \emph{attacker's knowledge}.
The notion of solution of constraint system is also extended to capture their expected properties.

\subsubsection{Knowledge base and formulas} \label{sec:formulas}

  \paragraph{New constraints}
    From now on we assume the existence of a rewriting system \(\R\) that is constructor-destructor and subterm convergent (we recall that the results of the previous sections did not rely on this assumption). All our definitions, lemmas and theorems will thus implicitly depend on this rewriting system.
    We introduce an extension of constraint systems with second-order constraints that serve key roles in the generation of the partition tree:

    \caseitem{\emph{Giving a finite representation of the deductive capabilities of the attacker}.}

      This takes the form of a \emph{knowledge base} \(\Solved\) which is a finite set of deduction facts.
      By relying on subterm convergence among others, our procedure will ensure that a term \(u\) is deducible \textit{iff} it can be deduced by applying constructor symbols to deduction facts of \(\Solved\), which makes deducibility easily decidable due to the constructor-destructor property.
      In particular we will only consider solutions that compute terms using entries of \(\Solved\) in this restricted manner.

    \caseitem{\emph{Giving a finite representation of the distinguishing capabilities of the attacker}.}

      This takes the form of a set of \emph{formulas} \(\USolved\) that is, in short, a finite representation of the term equalities that hold in the current frame.
      In particular static equivalence will be characterisable only from the formulas of \(\USolved\).

    \caseitem{\emph{Recording the constraints imposed on second-order solutions during the constraint solving}.}

      When computing most general solutions or performing case analyses on the form of solutions, we track the resulting effect on second-order solutions in a set \(\Eqsnd\) that is the second-order analogue of \(\Eqfst\).
      This is mostly how we model the predicates \(\pi\) that appear in the configurations in partition trees (Definition~\ref{def:configuration}).

    \medskip

    More formally we consider, in addition to deduction facts and second-order equations, a new atomic second-order constraint, \emph{equality facts} \(\xi \eqf \zeta\), \(\xi\) and \(\zeta\) second-order terms.
    Unlike second-order equations that model syntactic equalities, equality facts capture equalities modulo theory, that is, the fact that \(\xi\) and \(\zeta\) deduce the same first-order term.
    Concretely we extend the relation \(\models\) (Section~\ref{sec:mgs-def}) with
    \[(\Phi,\Sigma,\sigma) \models \xi \eqf \zeta \quad \mbox{iff} \quad \xi \Sigma \Phi \sigma\norm = \zeta \Sigma \Phi \sigma \norm\]
    We now define the constraints that are typically put in the set \(\USolved\).

    \begin{definition}[deduction formula, equality formula]
      A \emph{deduction} (resp. \emph{equality}) \emph{formula} is a constraint of the form \(\clause[S]{H}{(C_1 \wedge \ldots \wedge C_n)}\):
      \begin{enumerate}
        \item \(S\) is a set of (both first-order and second-order) variables;
        \item \(H\) is a deduction fact (resp. an equality fact);
        \item for all \(i \in \{1, \ldots, n\}\), \(C_i\) is either a deduction fact of the form \(X \dedfact t\), \(X \in \X[2]\), or a first-order syntactic equation \(u \eqs v\).
      \end{enumerate}
      A formula \(\psi\) is called \emph{solved} when it contains no hypotheses, i.e., \(\psi = (\clause[\emptyset]{H}{\top}) = H\).
      Given a formula \(\psi = \clause[S]{H}{\varphi}\), we denote by \(\Fhyp(\psi)\) the set of the syntactic equations appearing in the hypotheses \(\varphi\), and by \(\Df(\psi)\) the set of deduction facts in \(\varphi\).
    \end{definition}

    Intuitively, a formula captures a deduction or comparison that the attacker may perform and the premisses \(C_1, \ldots, C_n\) express conditions under which this is possible.
    Typically if the attacker observed a ciphertext \(\rsenc(m,r,k)\) (bound to an axiom \(\ax\)), we may express the deducibility of \(m\) through the formula
    \[\clause[X]{\rsdec(\ax,X) \dedfact m}{X \dedfact k}\]
    Another example is the following formula that expresses the tautology that two recipes deducing the same term should be equal in the sense of an equality fact:
    \[\clause[X,Y,z]{X \eqf Y}{(X \dedfact z \wedge Y \dedfact z)}\]
    This formula will serve as a generic placeholder when computing equality formulas during the constraint solving, that is, we will always add equality formulas obtained by substituting variables in the above formula.
    Although we consider arbitrary formulas such as the above two during the computation of the partition tree, note that only formulas of a certain shape will eventually be added in the set \(\USolved\) recording the attacker's distinguishing capabilities.
    We give more details about the invariants of the procedure in Appendix~\ref{app:invariants}, but we can mention for example that the formulas effectively recorded in \(\USolved\) will be of the form \(\clause{H}{\varphi}\), i.e., there are no universally-quantified variables, and \(\varphi\) only contains first-order equations.

  \paragraph{Extended constraint systems}
    We now formalise how we extend constraint systems to store the knowledge base, formulas, and to capture restrictions on the form of solutions.

    \begin{definition}[extended constraint system]
      A tuple \(\C^e = (\Phi, \Df, \Eqfst, \Eqsnd, \Solved, \USolved)\) is called an \emph{extended constraint system} where:
      \begin{enumerate}
        \item \((\Phi,\Df,\Eqfst)\) is a constraint system, although more general in that \(\Df\) may contain constraints of the form \(X \dedfact u\) or \(\forall X.\, X \ndedfact u\) where \(u\) may be an arbitrary constructor term;
        \item \(\Eqsnd\) is a set of second-order equations and constraints of the form \(\forall Y_1,\ldots, Y_k. \bigvee_{j=1}^p \xi_j \neqs \zeta_j\)
        \item \(\Solved\) is a set of deduction facts;
        \item \(\USolved\) is a set of deduction and equality formulas.
      \end{enumerate}
    \end{definition}

    As explained earlier, the set \(\Eqsnd\) gathers constraints to be satisfied by the second-order solutions of the system, \(\Solved\) is a finite representation of the attacker knowledge, and \(\USolved\) characterises the attacker capabilities to deduce and compare terms modulo theory.
    In particular the set \(\Eqsnd\) contains additional constraints to be satisfied by solutions while \(\Solved\) and \(\USolved\) are valid formulas that characterise potential attacker actions.
    For example, the (unsolved) deduction formulas in \(\USolved\) reason about potentially deducible terms:
    when such formula contains premisses, the procedure will perform a case analysis to distinguish cases where the hypotheses hold or not, leading to solved or trivial formulas, respectively.
    When a solved deduction formula is obtained this way, we add it to the knowledge base \(\Solved\) if \(u\) is not already deducible from it.

\subsubsection{(Most general) solutions} \label{sec:ext-mgs}

  We now define how the notion of solutions is lifted to extended constraint systems and how this embeds the predicates \(\pi\) used in the definition of partition-tree configurations.
  The definition of a solution \((\Sigma,\sigma)\) of \(\C^e\) follows three guidelines:
  \begin{enumerate*}
    \item it should be a solution in the usual sense and satisfy \(\Eqsnd(\C^e)\);
    \item the set of formulas \(\USolved(\C^e)\) plays no role in the definition of solutions:
    we will only prove invariants that this set verifies during our specific constraint-solving procedure (see Appendix~\ref{app:ptree}, Section \ref{app:invariants}); and
    \item all recipes used in the solution should have been constructed from the knowledge base \(\Solved(\C^e)\), \emph{uniformly} (that is, a same first-order term should not be deduced by different recipes in the solution).
  \end{enumerate*}
  In particular this requires a notion of \emph{consequence}, indicating that a recipe can be deduced from the knowledge base.

  \begin{definition}[consequence] \label{def:consequence}
    We define the set of \emph{consequences} of a set of deduction facts \(S\), denoted \(\conseq(S)\), as the set of pairs \((C[\xi_1,\ldots, \xi_n],C[u_1,\ldots, u_n])\)
    where \(C\) is a context built using \(\sigc \cup \sig_0\) and for all \(i \in \{1,\ldots, n\}\), \(\xi_i \dedfact u_i \in S\).
    We write \(\xi \in \conseq(S)\) if \(\exists t.\, (\xi,t) \in \conseq(S)\).
  \end{definition}

  We recall that by definition a deduction fact never has a constructor function symbol at its root (Definition \ref{def:constraint}):
  in particular if \(\xi \in \conseq(S)\), the context \(C\) in the above definition is unique.
  Writing \(\xi = C[\xi_1, \ldots, \xi_n]\) it is therefore possible to define unambiguously the set of \emph{consequential subterms} of \(\xi\)
  \[\stc(\xi,S) = \{ \xi_{|p} \mid p \text{ position of } C\}\]
  If \(R\) is a set of recipes we write \(\stc(R,S) = \bigcup_{\xi \in R} \stc(\xi,S)\).
  From this we can define solutions of extended constraint systems.

  \begin{definition}[solution of an extended constraint system] \label{def:ext-sol}
    A pair of substitutions \((\Sigma,\sigma)\) is a solution of \((\Phi, \Df, \Eqfst, \Eqsnd, \Solved, \USolved)\) if \((\Phi,\Sigma,\sigma) \models \Df \wedge \Eqfst \wedge \Eqsnd\) and the following two properties hold:
    \begin{enumerate}
      \item \emph{\(\Solved\)-Basis:}
      for all \(\xi \in \subterms[2](\im(\Sigma) \cup \Solved\Sigma)\),
      \(\msg(\xi\Phi\sigma)\) and \((\xi,\xi \Phi \sigma \norm) \in \conseq (\Solved \Sigma \sigma)\)
      \item \emph{Uniformity:}
      for all \(\xi,\xi' \in \stc(\im(\Sigma), \Solved\Sigma)\), \(\xi \Phi \sigma\norm = \xi' \Phi \sigma\norm\) implies \(\xi = \xi'\).
    \end{enumerate}
    The set of solutions of \(\C^e\) is written \(\Sol(\C^e)\) and \(\C^e\) is \emph{satisfiable} if \(\Sol(\C^e) \neq \emptyset\).
    We will denote by \(\bot\) an unsatisfiable extended constraint system.
    The notion of most general solution of \(\C^e\) is adapted in a straightforward way from the analogue for regular constraint systems.
  \end{definition}

  Intuitively when computing a node \(n\) of a partition tree, the extended constraint systems represent the predicate \(\pi(n)\):
  it will be defined so that given \((\P,\C) \in \Gamma(n)\) attached with \(\C^e\), we have \(\Sol[\pi(n)](\C) = \Sol(\C^e)\) (up to domain restriction).
  We detail this in Sections \ref{sec:cs-basics} and \ref{sec:correctness-proc}.

  \begin{example}
    Consider the extended constraint system \(\C^e\) defined by
    \begin{mathpar}
      \Phi = \{\ax_1 \mapsto \pair{k,x}\} \and
      \Df = \quanti{X}{0} \dedfact x \wedge \quanti{Y}{1} \dedfact y \and
      \Eqfst = y \eqs x \and
      \Eqsnd = \top \and
      \Solved = \ax_1 \dedfact \pair{k,x} \and
      \USolved = \Solved \wedge \fst(\ax_1) \dedfact k \wedge \snd(\ax_1) \dedfact x \wedge X \eqf \snd(\ax_1)
    \end{mathpar}
    This system involves an adversarial input \(x\) computable from an empty frame, which produces in response an output of \(\pair{k,x}\) for some name \(k\), and then the adversary inputs again \(y = x\).
    The set \(\USolved\), although not impacting the notion of solution, characterises here all successful operations that the attacker may perform in this situation:
    applying destructors to the term bound to \(\ax_1\) and observe that \(X\) and \(\snd(\ax_1)\) deduce the same term.

    We have for example \((\pair{X,\ax_1}, \pair{x,\pair{k,x}}) \in \conseq(\Solved \cup \Df)\).
    However the knowledge base is not \emph{saturated} in the sense that there are deducible terms \(u\), for example \(u = k\), such that there exist no recipes \(\xi\) such that \((\xi,u) \in \conseq(\Solved \cup \Df)\).
    In our procedure, the saturation is done by adding to \(\Solved\) all destructor applications that result into a non-consequence term.
    A saturated version of the constraint system would be
    \[\C^e_s = \C^e[\Solved \mapsto \Solved \wedge \fst(\ax_1) \dedfact k]\]
    Note that adding the deduction fact \(\snd(\ax_1) \dedfact x\) to the knowledge base is possible but redundant since \(x\) is already deducible from \(X\).
    The saturation
    ensures that for all \((\Sigma,\sigma)\) satisfying \(\Df(\C^e_s) \wedge \Eqfst(\C^e_s) \wedge \Eqsnd(\C^e_s)\), there exists \(\Sigma'\) such that \((\Sigma',\sigma) \in \Sol(\C^e_s)\), meaning that the requirement that solutions verify \(\Solved\)-basis can always be satisfied (which is key for satisfying the requirement that all symbolic transitions are reflected in the partition tree,
    recall Item~\ref{it:PT-child-concrete-derivation} of Definition~\ref{def:partition-tree}).
    Let us then consider
    \begin{align*}
      \Sigma & = \{X \mapsto a, Y \mapsto \snd(\ax_1)\} &
      \Sigma' & = \{X \mapsto a, Y \mapsto a\} &
      \sigma & = \{x \mapsto a, y \mapsto a\}
    \end{align*}
    Both \((\Sigma,\sigma)\) and \((\Sigma',\sigma)\) are solutions of the regular constraint system \((\Phi(\C^e_s), \Df(\C^e_s), \Eqfst(\C^e_s))\),
    but only \((\Sigma',\sigma)\) is a solution of \(\C^e\).
    This is because \(\Sigma\) does not verify uniformity:
    two different recipes \(a\) and \(\snd(\ax_1)\) are used to deduce the same first-order term \(a\).
    More generally we have \(\mgs(\C^e_s) = \{Y \mapsto X\}\).
    To obtain this result, the constraint-solving procedure for computing mgs', detailed in Section \ref{sec:mgs-gen}, will observe that \(X\) and \(Y\) deduce the same term and should therefore be unified to satisfy uniformity.
    A second-order equation \(X \eqs Y\) is thus added in \(\Eqsnd\), whose mgu is then the expected most general solution.
  \end{example}

  \begin{remark}[uniformity and complexity]
    In some sense enforcing that solutions are uniform ensures their minimality in terms of DAG size, by forcing identical recipes to be reused as much as possible when constructing the solution.
    This will be key for the complexity of our decision procedure, see Section~\ref{sec:termination}.
  \end{remark}

\subsubsection{\cs: the basics} \label{sec:cs-basics}
  Now we give details about the organisation of our constraint solver, detailed and proved correct in the next sections.
  As explained in Section \ref{sec:ext-mgs}, the goal of extended constraint system is to carry additional, structural information about solutions in a node \(n\), thus playing the role of the predicate \(\pi(n)\).
  More formally the procedure operates on:

  \begin{definition}[extended symbolic process, vector]
    An \emph{extended symbolic process} is a tuple \((\P,\C,\C^e)\) where \((\P,\C)\) is a symbolic process and \(\C^e\) an extended constraint system.
    We call a \emph{vector} a set of sets of extended symbolic processes \(\S = \{\Gamma_1, \ldots, \Gamma_n\}\).
    Each set \(\Gamma_i\) is called a \emph{component} of \(\S\).
  \end{definition}

  An extended symbolic process \((\P,\C,\C^e)\) induces a predicate \(\pi\) on the solutions of \((\P,\C)\) defined as follows:
  if \((\Sigma,\sigma) \in \Sol(\C)\), then \(\pi(\Sigma)\) holds \textit{iff} there exists \((\Sigma',\sigma') \in \Sol(\C^e)\) such that \(\Sigma \subseteq \Sigma'\) and \(\sigma \subseteq \sigma'\).
  In particular \(\Sol[\pi](\C) = \Sol(\C^e)\) (up to domain restriction).
  However, we recall that, in the definition of the node of a partition tree (configurations, Definition~\ref{def:configuration}), only \emph{one} common predicate \(\pi\) is used for all constraint systems \(\C\) of the configuration \(\Gamma\).
  For consistency, we therefore have to impose conditions ensuring that the predicates \(\pi\) corresponding to each \(\C \in \Gamma\) are all identical.
  Given a set of constraint systems \(\Gamma\) such that this property is not verified, the goal of the constraint-solving procedure is thus to refine \(\Gamma\) until obtaining a vector \(\S = \{\Gamma_1, \ldots, \Gamma_n\}\) such that

  \medskip
  \begin{enumerate}
    \item each component \(\Gamma_i\) can be used to model a partition-tree node, that is, a predicate \(\pi_i\) can be defined as above uniformly across all elements of \(\Gamma_i\);
    \item the underlying nodes verify the properties of the partition tree w.r.t. their father node \(\Gamma\).
  \end{enumerate}
  \medskip
  
  The procedure takes the form of various reduction relations that are used to refine a set of sets of extended symbolic processes, progressively, until reaching the final vector \(\S\):

  \medskip
  \begin{enumerate}
    \item A set of rules to compute \emph{most general solutions} (Section \ref{sec:mgs-gen}).
    \item A set of \emph{symbolic rules} (Section \ref{sec:symbolic-rules}) that formalise how to apply symbolic transitions to extended symbolic processes.
    \item Various sets of \emph{simplification rules} (Sections \ref{sec:first-simplification-rules}, \ref{sec:normalisation-rules} and \ref{sec:vector-rules}) that simplify vectors to remove unsatisfiable systems, or to split components that contain processes with non-statically-equivalent solutions.
    \item A set of \emph{case distinction rules} (Section \ref{sec:case-distinction-rules}) that refines the current vector based on case analyses to enforce the various properties of the partition tree (unique mgs in each component, maximal components w.r.t. static equivalence...).
  \end{enumerate}
  \medskip

  The overall procedure organising the above sets of rules into a complete algorithm to compute a partition tree is then detailed in Section \ref{sec:correctness-proc}.
  This can therefore be seen as the detailed version of the outline provided in Section \ref{sec:overview-ptree}.
  The main arguments for proving the correctness of the computation are also provided in Section~\ref{sec:correctness-proc};
  note however that these are only arguments of \emph{partial correctness}: 
  the termination of the procedure is later justified in Section~\ref{sec:termination}.

\subsection{\cs: computing most general solutions} \label{sec:mgs-gen}

\subsubsection{Applying solutions and unifiers} \label{sec:mgs-app}
  Because solutions \(\Sigma\) may introduce new second-order variables, their applications to a constraint system or a formula is not straightforward.
  Let for example \(\C^e = (\Phi, \Df, \Eqfst, \Eqsnd, \Solved, \USolved)\) where a variable \(\quanti{X}{k}\) is used to deduce a term \(u\), i.e. \((X \dedfact u) \in \Df\).
  Now say we want to consider the scenario where \(u\) is computed using a constructor \(\ffun/3\) and an entry of the knowledge base \((\xi \dedfact v) \in \Solved\) as a first argument, that is, we want to apply to \(\C^e\):
  \begin{align*}
    \Sigma & = \{X \rightarrow \ffun(\xi,X_1,X_2)\} &
    X_1,X_2 \text{ fresh}
  \end{align*}
  The raw application \(\C^e \Sigma\) has a flawed structure, in particular because the variables \(X_1\) and \(X_2\) would not be bound in the resulting system.
  To solve this issue we use a custom application mechanism that replaces \(X \dedfact u\) in \(\Df\) by \(X_1 \dedfact x_1,X_2 \dedfact x_2\), \(x_1,x_2\) fresh, and we add the equality \(u \eqs \ffun(v,x_1,x_2)\) to \(\Eqfst\) to express the logical link between \(X\) and \(X\Sigma = \ffun(\xi,X_1,X_2)\).
  %

  \begin{definition}[application of a substitution to an extended constraint system] \label{def:application_mgs_csys}
    Let \(\C^e = (\Phi, \Df, \Eqfst, \allowbreak\Eqsnd, \Solved, \USolved)\) and \(\Sigma\) be a substitution.
    We write \(\CApply{\Sigma}{\C^e}\) the constraint system:
    \[
      (\Phi, \Df', \Eqfst \wedge E_\Sigma, \Eqsnd \Sigma \wedge \Sigma_{|\vars[2](\C^e)}, \Solved\Sigma, \USolved\Sigma)
    \]
    where \(\Df' = (\Df \smallsetminus D_\dom) \cup D_\freshlab\) with the sets of:
    \begin{enumerate}
      \item deduction facts removed by the application of \(\Sigma\):
      \(D_\dom = \{ Y \dedfact u \in \Df \mid Y \in \dom(\Sigma)\}\)
      \item binding facts:
      \(D_\freshlab = \{ Y \dedfact y \mid Y \in \vars[2](\im(\Sigma_{|\vars[2](\C^e)})) \smallsetminus \vars[2](\C^e), y \text{ fresh}\}\)
      \item linking equations:
      \(E_\Sigma = \{ u \eqs v \mid Y \dedfact u \in D_\dom, (Y\Sigma,v) \in \conseq(\Solved\Sigma \cup \Df')\}\)
    \end{enumerate}
    By abuse of notation we may write \(\CApply{\Sigma}{S}\) for \((\P,\C,\CApply{\Sigma}{\C^e})\) if \(S = (\P,\C,\C^e)\).
  \end{definition}

  We will also use a similar mechanism for applying substitutions to formulas:

  \begin{definition}[application of a substitution to a formula]
    \label{def:application_mgs_formula}
    Let \(\C^e = (\Phi, \Df, \Eqfst, \Eqsnd, \Solved, \USolved)\),
    \(\psi = \clause[S]{H}{\varphi}\) be a formula,
    and \(\Sigma\) be a substitution.
    We denote \(\FApply{\Sigma}{\psi}{\C^e}\) (or \(\FApply{\Sigma}{\psi}{S}\) by abuse of notations if \(S = (\P,\C,\C^e)\)) the formula
    \[
        \clause[S']{H\Sigma}{(\Df' \wedge \Fhyp(\psi) \wedge E_\Sigma)}
    \]
    where \(\Df' = (\Df(\psi) \smallsetminus D_\dom) \cup D_\freshlab\), \(S' = (S \smallsetminus \dom(\Sigma)) \cup \vars[1](D_\freshlab)\) and:
    \begin{enumerate}
      \item \(D_\dom = \{ Y \dedfact u \in \Df(\psi) \mid Y \in \dom(\Sigma)\}\)
      \item \(D_\freshlab = \{ Y \dedfact y \mid Y \in \vars[2](\im(\Sigma)) \smallsetminus \vars[2](\C,\psi), y \text{ fresh}\}\)
      \item \(E_\Sigma = \left\{ u \eqs v \mid Y \dedfact u \in D_\dom,  (Y\Sigma,v) \in \conseq\left(\Solved \cup \Df \cup \Df'\right) \right\}\)
    \end{enumerate}
  \end{definition}

\subsubsection{Constraint-solving rules} \label{sec:mgs-rules}
  \paragraph{A complete example}
    By definition, the solutions of an extended constraint systems \(\C^e\) have to verify \(\Solved(\C^e)\)-basis, which means that in practice we only have to compute solutions constructed by applying constructors to the entries of the knowledge base and \(\Df\).
    Besides due to the uniformity requirement we can always unify two recipes that deduce the same first-order term.
    Putting everything together the most general solutions of an extended constraint system can then be computed with a simple transition system.
    Let us detail a complete example to illustrate the mechanisms in play, before formalising the corresponding constraint-solving rules.

    \begin{example}
      Given \(k,r \in \Nall\), let us consider a situation where the attacker has observed the output of a hash \(\hfun(r)\), then inputs a term \(x\), receives in response a ciphertext \(\raenc(k,r,x)\) encrypted with \(x\), and finally inputs a term \(y\) that should verify the equation \(y \eqs \pair{k,\hfun(x)}\).
      This is modelled by the frame \(\Phi = \{\ax_1 \mapsto \hfun(r), \ax_2 \mapsto \raenc(k,r,x)\}\) and the constraints
      \begin{align*}
        \Df & = \quanti{X}{1} \dedfact x \wedge \quanti{Y}{2} \dedfact y &
        \Eqfst & = y \eqs \pair{k,\hfun(x)}
      \end{align*}
      At this point a saturated knowledge base should contain the two entries of the frame and one recipe indicating that decrypting \(\ax_2\) results in obtaining the name \(k\).
      \[\Solved = \ax_1 \dedfact \hfun(r) \wedge \ax_2 \dedfact \raenc(k,r,x) \wedge \radec(\ax_2,X) \dedfact k\]
      We consider that \(\Eqsnd = \top\) and we leave the set of formulas \(\USolved\) unspecified since it has no influence on solutions.
      First of all some \emph{simplification rules} will be applied to propagate the equations on \(x\) and \(y\) to the whole system;
      here it will apply \(\mgu(\Eqfst)\) to \(\Df\), resulting in
      \[\Df = X \dedfact x \wedge Y \dedfact \pair{k,\hfun(x)}\]
      The constraint-solving rules detailed in the remaining of this section consider all ways to compute recipes for \(X\) and \(Y\) from the knowledge base.
      For each of these recipes two cases arise:
      either
      \begin{enumerate*}
        \item it is picked directly from the knowledge base; or
        \item it starts with a constructor symbol.
      \end{enumerate*}
      This will correspond to the constraint-solving rules \eqref{rule:res} and \eqref{rule:cons}, respectively.
      Finally, to satisfy the uniformity property,
      the procedure unifies any second-order terms in the system that deduce the same first-order term (Rule \eqref{rule:conseq}).
      We keep on refining the case analysis with these three rules, removing branches yielding contradictions, until no more rules are applicable.
      The resulting systems will either have no solutions, or be in a so-called \emph{solved form} and have \(\mgu(\Eqsnd)\) as a unique mgs.
      Let us do it for our example:

      \caseitem{\emph{{case 1}: the recipe for \(Y\) has a constructor symbol at its root (only possible case)}}

        The constructor in question is necessarily the pair.
        Therefore we let two fresh second-order variables \(\quanti{Y_1}{2},\quanti{Y_2}{2}\) and apply the substitution \(\{Y \mapsto \pair{Y_1,Y_2}\}\) to the system (in the sense of Definition \ref{def:application_mgs_csys}).
        After simplification this leads to the updated second-order constraints:
        \begin{align*}
          \Df & = X \dedfact x \wedge Y_1 \dedfact k \wedge Y_2 \dedfact \hfun(x) &
          \Eqsnd & = Y \eqs \pair{Y_1,Y_2}
        \end{align*}

      \caseitem{\emph{{case 1.1}: the recipe for \(Y_1\) is \(\radec(\ax_2,X)\) from the knowledge base (only possible case)}}

        We thus apply the substitution \(\{Y_1 \mapsto \radec(\ax_2,X)\}\), resulting in the updated constraints:
        \begin{align*}
          \Df & = X \dedfact x \wedge Y_2 \dedfact \hfun(x) &
          \Eqsnd & = Y \eqs \pair{\radec(\ax_2,X),Y_2} \wedge Y_1 \eqs \radec(\ax_2,X)
        \end{align*}

      \caseitem{\emph{{case 1.1.1}: the recipe for \(Y_2\) is the entry \(\ax_1\) from the knowledge base}}

        We therefore apply the substitution \(\{Y_2 \mapsto \ax_1\}\), resulting in the updated constraints:
        \begin{align*}
          \Df & = X \dedfact r &
          \Eqsnd & = Y \eqs \pair{\radec(\ax_2,X),\ax_1} \wedge Y_1 \eqs \radec(\ax_2,X) \wedge Y_2 \eqs \ax_1 
        \end{align*}
        However the constraints on \(X\) are now unsatisfiable:
        the corresponding recipe can neither start with a constructor nor be an entry of the knowledge base.
        The constraints in this branch of the case analysis therefore have no solutions.

      \caseitem{\emph{{case 1.1.2}: the recipe for \(Y_2\) has a constructor symbol at its root}}

        The constructor in question is necessarily \(\hfun\).
        Similarly to case 1 we apply the substitution \(\{Y_2 \mapsto \hfun(Y_3)\}\) for some fresh variable \(\quanti{Y_3}{2}\) which results in the updated constraints:
        \begin{align*}
          \Df & = X \dedfact x \wedge Y_3 \dedfact x &
          \Eqsnd & = Y \eqs \pair{\radec(\ax_2,X),\hfun(Y_3)} \wedge Y_1 \eqs \radec(\ax_2,X) \wedge Y_2 \eqs \hfun(Y_3)
        \end{align*}
        Then we observe that \(X\) and \(Y_3\) should be unified by uniformity because they deduce the same first-order term \(x\).
        We have \(\mgu(X \eqs Y_3) = \{Y_3 \mapsto X\}\) (we recall that \(\{X \mapsto Y_3\}\) is \emph{not} a valid second-order substitution because \(Y_3\) has a strictly greater type than \(X\)) which, after application to the system, results in the updated constraints:
        \begin{align*}
          \Df & = X \dedfact x &
          \Eqsnd & = Y \eqs \pair{\radec(\ax_2,X),\hfun(X)} \wedge Y_1 \eqs \radec(\ax_2,X) \wedge Y_2 \eqs \hfun(X) \wedge Y_3 \eqs X
        \end{align*}
        This will be a typical example of system in solved form.
        Since we considered all cases and only this branch was successful we conclude that the overall system has a unique mgs which is \(\mgu(\Eqsnd) = \{Y \mapsto \pair{\radec(\ax_2,X),\hfun(X)}\}\).
    \end{example}

  \paragraph{Formalisation}
    We will formalise the simplification rules in the next section and focus here on the main three rules \eqref{rule:conseq}, \eqref{rule:res} and \eqref{rule:cons} mentioned in the above example.
    For that we reason about a set \(\recipes(\C^e)\) that represents all recipes that are already used to constraint the solutions of \(\C^e\):
    \[\recipes(\C^e) =
      \stc(\im(\mgu(\Eqsnd(\C^e)), \Solved(\C^e) \cup \Df(\C^e))
      \cup \vars[2](\Df(\C^e))\]
    As we saw in the example, the mgs is gradually constructed ``within \(\Eqsnd\)'', in the sense that after normalising \(\C^e\) with the transition system defined in this section, it will have \(\mgu(\Eqsnd)\) as a unique mgs.
    In particular an invariant of our transition system is that \(\im(\mgu(\Eqsnd(\C^e)))\) is consequence of \(\Solved(\C^e)\) and \(\Df(\C^e)\), hence the notation \(\recipes(\C^e)\) is well defined.
    Formally speaking the transition system relies on three rules of the form
    \[\tag{\mbox{\(\star\)}} \label{rule:mgs-base}
      \C^e\ \xrightarrow{\Sigma}\ \CApply{\Sigma}{\C^e}\]
    for some substitution \(\Sigma\) and under various conditions capturing the possible ways to satisfy the constraints of \(\C^e\).
    The uniformity property is expressed by applying \eqref{rule:mgs-base} with
    \begin{equation}
      \tag{\mbox{\textsc{MGS-Conseq}}}
      \parbox{0.75\textwidth}{
        \(\Sigma = \mgu(\xi \eqs \zeta)\) for some \(\xi \in \recipes(\C^e) \cup \sig_0\), \(\zeta \in \recipes(\C^e)\), and provided
        \(\Sigma \neq \top\), \(\Sigma \neq \bot\), and \(\exists u. (\xi,u), (\zeta,u) \in \conseq(\Solved(\C^e) \cup \Df(\C^e))\)
      }
      \label{rule:conseq}
    \end{equation}
    The result is the unification in \(\C^e\) of the two second-order terms \(\xi\) and \(\zeta\) that deduce the same term \(u\).
    It then remains to add rules that express how each term \(u\), \((X \dedfact u) \in \Df(\C^e)\), can be constructed by the adversary from the knowledge base.
    When Rule \eqref{rule:conseq} is not applicable we thus apply \eqref{rule:mgs-base} under one of the following two conditions.
    The first one expresses that \(u\) is computed by directly using an entry from the knowledge base:
    \begin{equation}
      \tag{\mbox{\textsc{MGS-Res}}}
      \parbox{0.75\textwidth}{
        \(\Sigma = \mgu(X \eqs \xi) \neq \bot\) where, for some \(u \notin \X\), there exist deduction facts \((X \dedfact u) \in \Df(\C^e)\) and \((\xi \dedfact v) \in \Solved(\C^e)\)
      }
      \label{rule:res}
    \end{equation}
    Then the last rule expresses that the computation of \(u\) starts by applying a constructor \(\ffun\):
    \begin{equation}
      \tag{\mbox{\textsc{MGS-Cons}}}
      \parbox{0.75\textwidth}{
        \(\Sigma = \{ X \rightarrow \ffun(X_1,\ldots, X_n)\}\) where \(\quanti{X_1}{k},\ldots, \quanti{X_n}{k}\) are fresh, and there exists a deduction fact \((\quanti{X}{k} \dedfact \ffun(u_1,\ldots, u_n)) \in \Df(\C^e)\)
      }
      \label{rule:cons}
    \end{equation}
    As said above we always apply Rule~\eqref{rule:conseq} in priority, that is, we add to the last two rules the condition that Rule~\eqref{rule:conseq} cannot be applied.
    This will be crucial in particular when studying the complexity of the procedure in Section~\ref{sec:termination}.

\subsubsection{First set of simplification rules} \label{sec:first-simplification-rules}
  To effectively compute most general solutions, the above three rules are applied repeatedly, but some \emph{simplification rules} are also used in between.
  Their role is to put the constraint systems in a simpler form and in particular to detect the unsatisfiable systems.
  Other simplification rules, serving different purposes, will be introduced later in the procedure.
  The rules here are of two kinds:
  \begin{enumerate}
    \item \emph{simplification rules for formulas} that simply compute mgu and simplify the hypotheses of formulas; and
    \item \emph{simplification rules for mgs'} that apply mgu to the rest of the system, and detect unsatisfiability through contradictions or violations of uniformity.
  \end{enumerate}

  \paragraph{Simplification rules for formulas}
    We first introduce basic simplification rules for formulas that will be used even outside of the computation of most general solutions.
    We define five sets rules in Figure \ref{fig:normalisation_formula} that apply on constraints of \(\Eqfst\), \(\Eqsnd\) and \(\USolved\).

    \begin{figure}[ht]
      \[
        \begin{array}{l@{\quad}l}
          \mbox{\em Misc.} &
            \neg \top \simpl \bot \qquad \qquad
            \neg \bot \simpl  \top \qquad \qquad
            \varphi \wedge \top \simpl \varphi \qquad
            \varphi \wedge \bot \simpl \bot \\[3mm]
          \text{\em Universal vars.} & \clause[S \cup
          \{x\}]{H}{(x \eqs u \wedge \varphi)} \simpl
          \clause[S]{H\sigma}{\varphi\sigma} \qquad \text{if}\ \sigma =
          \mgu(x \eqs u) \neq \bot \\
          & \clause[S \cup \{x\}]{H}{\varphi}
          \simpl \clause[S]{H}{\varphi} \qquad \text{if}\ x \notin
          \vars[1](\varphi) \\[3mm]
          \mbox{\em 1\textsuperscript{st} order eq.} &
            u \eqs v \simpl \mgu(u \eqs v) \\[3mm]
          \mbox{\em 1\textsuperscript{st} order diseq.} &
          \forall S.\, \phi \simpl
              \left\{\begin{array}{ll}
                \forall S.\, \displaystyle \bigvee_{x \in \dom(\sigma)} x \neqs x\sigma & \mbox{with } \sigma = \mgu(\neg\phi) \neq \bot\\
                \top & \mbox{if } \mgu(\neg\phi) = \bot
              \end{array}\right. \\\\
          %
          \mbox{\em 2\textsuperscript{nd} order diseq.} &
            \forall S.\, \phi \simpl
              \left\{\begin{array}{ll}
                \forall S \cup S'.\hspace{-2mm} \displaystyle \bigvee_{X \in \dom(\Sigma)} \hspace{-3mm} X \neqs X\Sigma
                  & \mbox{with } \Sigma = \mgu(\neg\phi) \neq \bot \\[-4mm]
                  & \text{and } S' = \vars[2](\im(\Sigma)) \smallsetminus \vars[2](\phi)\\[2mm]
                \top & \mbox{if } \mgu(\neg\phi) = \bot
              \end{array}\right.
        \end{array}
      \] 
      \caption{Simplification rules on formulae}
      \label{fig:normalisation_formula}
    \end{figure}

    We recall that, in the case of the simplification of second order disequations, the computation of mgu may introduce new variables to match arities (see Section~\ref{sec:unification}), hence the need for the universal quantified variables \(S'\).
    No rules are needed for second-order equations in the context of our decision procedure, since Rules \eqref{rule:conseq}, \eqref{rule:res} and \eqref{rule:cons} already apply mgu to the entire system.
    The simplification rules are lifted to extended constraint systems \(\C^e\) in the natural way, by applying the simplifications to all formulas of \(\Eqfst(\C^e)\), \(\Eqsnd(\C^e)\) and \(\USolved(\C^e)\).

  \paragraph{Simplification rules for MGS}
    In addition of the rules of Figure \ref{fig:normalisation_formula} we define a couple of other rules specific to the computation of most general solutions.
    First of all the rule
    \begin{equation}
      \label{rule:unifEqfst_simpl}
      \tag{\mbox{\textsc{MGS-Unif}}}
      (\Phi, \Df, \Eqfst \wedge x \eqs u, \Eqsnd, \Solved, \USolved)
      \quad  \simpl \quad
        (\Phi\sigma, \Df\sigma, \Eqfst\sigma \wedge x \eqs u, \Eqsnd, \Solved\sigma, \USolved\sigma)
    \end{equation}
    where \(x \in \vars[1](\Eqfst, \Df, \Phi, \Solved, \USolved) \smallsetminus \vars(u)\) and \(\sigma = \{x \mapsto u\}\),
    propagates first-order mgu in the whole system.
    We also consider the following rule discarding a system with no solutions
    \begin{equation}
      \label{rule:mgs-unsat}
      \tag{\mbox{\textsc{MGS-Unsat}}}
      \C^e \simpl \bot
    \end{equation}
    where either of the following three conditions is satisfied:
    \begin{enumerate}
      \item \label{it:mgs-unsat}
      \(\Eqfst = \bot\)
      \item \label{it:mgs-not-uniform}
      there exist \(\xi,\zeta \in \recipes(\C^e)\) such that \((\xi,u),(\zeta,u) \in \conseq(\Solved(\C^e) \cup \Df(\C^e))\) and,
      writing \(\Sigma = \mgu(\xi \eqs \zeta)\), either \(\Sigma = \bot\) or \(\Eqsnd \Sigma \simpl^* \bot\) with the rules of Figure \ref{fig:normalisation_formula}
      \item \label{it:mgs-private}
      there exist \((\forall \quanti{X}{i}.\, X \ndedfact u) \in \Df(\C^e)\) and \(\xi \in \recipeset_i\) such that \((\xi,u) \in \conseq(\Solved(\C^e) \cup \Df(\C^e))\)
    \end{enumerate}

    The first condition captures trivially unsatisfiable systems, the second one systems with no uniform solutions, and the third one exhibits a public channel that has been used for an internal communication (which is forbidden by the semantics).
    Since the whole set of simplification rules (Figure \ref{fig:normalisation_formula} and the above two) is convergent modulo renaming of variables, we denote \(\C \simplnorm\) a normal form of the extended constraint system \(\C\) w.r.t. \(\simpl\).

\subsubsection{Overall procedure and correctness} \label{sec:mgs-proc}
  \paragraph{Description of the procedure}
    The point of the transition systems above is to transform an extended constraint system into a form where it has a unique mgs.
    More formally:

    \begin{definition}[solved extended constraint system]
      An extended constraint system \(\C^e\) is in \emph{solved form} if \(\C^e \neq \bot\), \(\C^e\) is irreducible w.r.t. \(\simpl\) and \(\xrightarrow{\Sigma}\), and all deduction facts in \(\Df(\C^e)\) have  variables as first-order terms. 
    \end{definition}

    Intuitively for such constraint systems, \(\mgu(\Eqsnd(\C^e))\) is the unique mgs of \(\C^e\).
    Note however that this method for computing mgs' is only correct under some invariants of our overall procedure.
    Typically, since second-order equations are not handled by simplification rules,
    if \(\Eqsnd\) contains two contradictory equations \(X \eqs a \wedge X \eqs b\) for two constants \(a \neq b\), our procedure would fail to detect the contradiction.
    If we define the reduction relation \(\SimpStep{\Sigma}\) as the reflexive transitive closure of the composition of relations \(\simpStep{\sigma}\),
    then under the invariants of the procedure we compute a set of most general solutions of \(\C^e\) as the set
    \(
      \{ \Sigma_{|\vars[2](\C^e)} \mid \C^e \SimpStep{\Sigma} \C^{e\prime}, \C^{e\prime} \text{ solved}\}\,.
    \)
    
    \begin{remark}[notation for extended symbolic processes]
      For convenience we often abuse notations and, if \(S = (\P,\C,\C^e)\) is an extended symbolic process, we write \(\mgs(S)\) instead of \(\mgs(\C^e)\) or say that \(S\) is in solved form.
    \end{remark}

  \paragraph{Correctness arguments}
    As mentioned earlier this procedure is only correct under some additional properties verified all along Algorithm~\ref{alg:ptree}.
    For the sake of precision we make explicit mention to these two invariants, \(\PredWellFormed(\C^e)\) and \(\PredCorrectFormula(\C^e)\).
    They are formally defined in Appendix~\ref{app:invariants} with a proof that they are preserved during the whole computation of the partition tree, but knowing their exact definition is not necessary to understand the results of this section.
    The core correctness arguments can be decomposed into following propositions, and are derived from the results proved in Appendix~\ref{app:mgs}.
    The first one states that when an extended constraint system cannot be reduced any more then its set of most general solutions is either empty or a singleton:

    \begin{proposition}[restate={propCorrectMgsSolved},name={mgs of an irreducible system}] \label{prop:correct-mgs-solved}
      Let \(\C^e\) be an extended constraint system that is irreducible w.r.t. \(\simpl\) and \(\xrightarrow{\Sigma}\), and such that the invariants \(\PredWellFormed(\C^e)\) and \(\PredCorrectFormula(\C^e)\) hold.
      Then
      \begin{enumerate}
        \item if \(\C^e\) is in solved form then \(\mgs(\C^e) = \{\mgu(\Eqsnd(\C^e))\}\)
        \item otherwise \(\mgs(\C^e) = \emptyset\)
      \end{enumerate}
    \end{proposition}

    The second argument is that applying the mgs constraint-solving rules is correct w.r.t. the solutions of the initial system.

    \begin{proposition}[restate={propSoundMgsStep},name={soundness of one step of the mgs constraint solving}] \label{prop:sound-mgs-step}
      Let \(\C^e\) be an extended constraint system with \(\C^e = \C^e \simplnorm\).
      If \(\C^e \simpStep{\Sigma} \C^{e\prime}\) and \((\Sigma',\sigma) \in \Sol(\C^{e\prime})\) then \((\Sigma'_{|\vars[2](\C^e)}, \sigma_{|\vars[1](\C^e)}) \in \Sol(\C^e)\).
    \end{proposition}

    Finally the last argument formalises than all solutions can be expressed as a sequence of mgs constraint-solving transitions.

    \begin{proposition}[restate={propCompleteMgsStep},name={completeness of one step of the mgs constraint solving}] \label{prop:complete-mgs-step}
      Let \(\C^e\) be an extended constraint system such that \(\C^e \simplnorm = \C^e\) and the invariants \(\PredWellFormed(\C^e)\) and \(\PredCorrectFormula(\C^e)\) hold.
      We also assume that at least one mgs constraint-solving rule is applicable to \(\C^e\).
      Then for all \((\Sigma,\sigma) \in \Sol(\C^e)\), there exist a constraint-solving transition \(\C^e \simpStep{\Sigma_0} \C^{e\prime}\) and \(\Sigma \subseteq \Sigma'\), \(\sigma \subseteq \sigma'\) such that \((\Sigma',\sigma') \in \Sol(\C^{e\prime})\).
    \end{proposition}

    Together these three results give the partial correctness of the procedure, that is, the correctness of the computation when it terminates.
    The termination is studied in Section~\ref{sec:termination}:

    \begin{theorem}[partial correctness of mgs computation]
      Let \(\C^e\) be an extended constraint system such that \(\PredWellFormed(\C^e)\) and \(\PredCorrectFormula(\C^e)\) hold.
      Then, assuming there exist no infinite sequences of \(\simpStep{}\) reductions from \(\C^e\), we have
      \[\mgs(\C^e) = \{\Sigma_{|\vars[2](\C^e)} \mid \C^e \simplnorm \SimpStep{\Sigma} \C^{e\prime}, \C^{e\prime} \text{ solved}\}\]
    \end{theorem}

    \begin{proof}
      Since a set of mgs' of \(\C^e\simplnorm\) is also a set of mgs' of \(\C^e\), we assume without loss of generality that \(\C^e \simplnorm = \C^e\).
      Let us write \(S = \{\Sigma_{|\vars[2](\C^e)} \mid \C^e \SimpStep{\Sigma} \C^{e\prime}, \C^{e\prime} \text{ solved}\}\) and prove that \(S\) is a set of mgs' of \(\C^e\simplnorm\).
      By the termination assumption, we can reason by well-founded induction on the reduction relation \(\simpStep{}\) from \(\C^e\).
      Using such an induction we can prove the two requirements of the definition, that is:
      \begin{enumerate}
        \item that all \(\Sigma \in S\) are solutions of \(\C^e\) after replacing their second-order variables by fresh constants (base case: Proposition \ref{prop:correct-mgs-solved};
        inductive case: soundness, i.e., Proposition \ref{prop:sound-mgs-step}).
        \item that all solutions of \(\C^e\) are instances of a substitution of \(S\)
        (base case: Proposition \ref{prop:correct-mgs-solved} again;
        inductive case: completeness, i.e., Proposition \ref{prop:complete-mgs-step}). \qedhere
      \end{enumerate}
    \end{proof}

\subsection{\cs: symbolic and simplification rules}

\subsubsection{Symbolic rules} \label{sec:symbolic-rules}

  The symbolic rules simply apply the transitions of the symbolic semantics to extended symbolic processes, adding the corresponding constraints to both the symbolic process and the extended constraint system.
  In that sense most rules are close to identical to those of the symbolic semantics (Section \ref{sec:symbolic-semantics}).
  Typically the analogue of the rule \eqref{rule:s-in} is:
  \begin{equation}
    \tag{\mbox{\textsc{E-In}}} \label{rule:e-in}
    (\multi {\InP {u} {x}.P} \cup \P, \C, \C^e)
      \sstep {\InP {Y} {X}} (\multi {P} \cup \P, \mathit{incr}(\C), \mathit{incr}(\C^e))
  \end{equation}
  where, if \(\D \in \{\C,\C^e\}\), \(\mathit{incr}(\D) = \D[\Df \mapsto \Df \wedge X \dedfact x \wedge Y \dedfact y, \Eqfst \mapsto \Eqfst \wedge \sigma]\)
  with \(\quanti{Y}{n}\), \(\quanti{X}{n}\) and \(y\) fresh (\(n\) size of the domain of the frame of \(\C\)), and \(\sigma \in \mguR(y \eqs u \mu)\), \(\mu = \mgu(\Eqfst(\C)) \neq \bot\).
  The only rule that is not a trivial extension of the symbolic semantics is the one for outputs that puts a deduction fact in \(\USolved\) to model the additional capability this offers to the attacker:
  \begin{equation}
    \tag{\mbox{\textsc{E-Out}}} \label{rule:e-out}
    (\multi {\OutP {u} {v}.P} \cup \P, \C, \C^e)
      \sstep {\OutP {Y} {\ax_{n+1}}} (\multi {P} \cup \P, \mathit{incr}(\C), \mathit{incr}(\C^e)[\USolved \mapsto \USolved \wedge \ax_{n+1} \dedfact v \sigma \norm])
  \end{equation}
  where, if \(\D \in \{\C,\C^e\}\), \(\mathit{incr}(\D) = \D[\Phi \mapsto \Phi \cup \{\ax_{n+1} \mapsto v \mu \sigma \norm\}, \Df \mapsto \Df \wedge Y \dedfact y, \Eqfst \mapsto \Eqfst \wedge \sigma]\)
  with \(\quanti{Y}{n}\) and \(y\) fresh (\(n\) size of the domain of the frame of \(\C\)), and \(\sigma \in \mguR(y \eqs u \mu \wedge v\mu \eqs v\mu)\), \(\mu = \mgu(\Eqfst(\C)) \neq \bot\).
  We omit the definition of the remaining rules corresponding to the other symbolic transitions, all being constructed similarly to \eqref{rule:e-in} by copying the new constraints of \(\C\) into \(\C^e\).

\subsubsection{Normalisation rules} \label{sec:normalisation-rules}

  We define a new set of simplification rules, called \emph{normalisation rules}, that operate on extended constraint systems.
  Similarly to the simplification rules for most general solutions introduced in Section \ref{sec:first-simplification-rules} they propagate first-order unifiers across the system and replace unsatisfiable systems by \(\bot\).
  They also rely on the computation of mgs' of Section \ref{sec:mgs-gen}, for example to identify and remove trivial constraints such as formulas with unsatisfiable hypotheses.
  They are defined in Figure \ref{fig:normalisation_constraint_systems} and commented below (in particular regarding the definition of \(\receq\)).

  \begin{figure}[ht]
    \begin{align}
      \tag{\mbox{\textsc{Norm-Unif}}}
      \label{rule:unifEqfst_norm}
      \C^e & \simpl \C^{e\prime} && \text{if \(\C^e \simpl \C^{e\prime}\) by rule \eqref{rule:unifEqfst_simpl}}\\
      \tag{\mbox{\textsc{Norm-no-MGS}}}
      \label{rule:uniform}
      \C^e & \simpl \bot && \mbox{if \(\mgs(\C^e) = \emptyset\)}\\
      \tag{\mbox{\textsc{Norm-Diseq}}}
      \label{rule:Disequation removal 1}
      \C^e[\Eqfst \mapsto \Eqfst \wedge \forall \tilde{x}. \phi]
      & \simpl
      \C^e && \mbox{if \(\mgs(\C^e[\Eqfst \mapsto \Eqfst \wedge \neg\phi]) = \emptyset\)}\\
      \tag{\mbox{\textsc{Norm-Formula}}}
      \label{rule:Disequation removal 2}
      \C^e[\USolved \mapsto \USolved \wedge \psi]
      & \simpl \C^e && \mbox{if \(\mgs(\C^e[\Eqfst \mapsto \Eqfst \wedge \Fhyp(\psi)]) = \emptyset\)}\\
      \tag{\mbox{\textsc{Norm-Dupl}}}
      \label{rule:Removal of unsolved formula}
      \C^e[\USolved \mapsto \USolved \wedge \psi]
      & \simpl \C^e && \mbox{if \(\exists \psi' \in \USolved\), \(\psi' \receq \psi\) and \(\psi'\) solved}
    \end{align}
    \caption{Normalisation rules on extended constraint systems}
    \label{fig:normalisation_constraint_systems}
  \end{figure}

  We recall that we also write \(\C^e \simpl \C^{e\prime}\) if a constraint of \(\Eqfst(\C^e)\), \(\Eqsnd(\C^e)\) or \(\USolved(\C^e)\) can be simplified using one of the simplification rules on formulas (Figure \ref{fig:normalisation_formula}).
  The relation \(\simpl\) can be lifted to sets of (sets of) extended constraint systems or symbolic processes in the natural way.
  Let us now comment on the rules of Figure \ref{fig:normalisation_constraint_systems}.
  Rule \eqref{rule:unifEqfst_norm} uses the same rule as in the mgs constraint solving to propagate first-order unifiers to the whole system.
  The next three rules exploit the existence of a most general solution of the constraint system to simplify some constraints:

  \medskip
  \begin{enumerate}
    \item Rule \eqref{rule:uniform} checks whether the constraint system is unsatisfiable, i.e., does not have a most general solution, and in this case transforms it into \(\bot\).
    \item Rule \eqref{rule:Disequation removal 1} similarly removes a disequation \(\forall \tilde{x}. \phi\) in \(\Eqfst\) when it does not effectively restrict the solutions:
    for that we require the constraint system not to have solutions that contradict the disequation.
    \item Analogously Rule \eqref{rule:Disequation removal 2} removes a formula with unsatisfiable hypotheses.
    The fact that we only consider the equations among the hypotheses (recall that \(\Fhyp{\psi}\) omits the hypotheses of \(\psi\) that are deduction facts) is due to an invariant of our procedure.
    We will indeed ensure that formulae are only added to the set \(\USolved\) after all deduction facts have been removed from hypotheses by appropriate solving.
  \end{enumerate}
  \medskip

  Finally Rule \eqref{rule:Removal of unsolved formula} removes an unsolved deduction or equality formula \(\psi\) from \(\USolved\) when it is subsumed by another formula \(\psi'\).
  This is formalised by the following notion of equivalence:

  \begin{definition}[head equivalence of formulas]
    Let \(\psi = \clause{H}{\varphi}\) and \(\psi' = \clause{H'}{\varphi'}\) be two formulas.
    We say that \(\psi\) and \(\psi'\) are \emph{head equivalent}, written \(\psi \receq \psi'\), if for some \(\xi,\zeta,u,u'\) either \(H = H' = (\xi \eqf \zeta)\), or \(H = (\xi \dedfact u)\) and \(H' = (\xi \dedfact u')\).
  \end{definition}

  That is, two formulas are head equivalent if their heads have the same second-order terms (but may differ on their first-order terms), which means they model the same attacker action.
  In particular if \(\psi \receq \psi'\) and \(\psi'\) is solved (namely has no hypotheses any more) then the formula \(\psi\) is already implied by \(\psi'\) which is why Rule \eqref{rule:Removal of unsolved formula} can remove it from \(\USolved\).

\subsubsection{Vector-simplification rules} \label{sec:vector-rules}

  We now define simplification rules that focus on vector, thus called \emph{vector-simplification rules}.
  They are described in Figure \ref{fig:normalisation_vector} and focus among other things on adding formulas and entries in the knowledge base.
  This has to be done concurrently on an entire vector component to ensure that the same attacker actions can be performed in all of its elements, that is, that they have statically-equivalent solutions.
  The rules assume that the constraint systems have been normalised by the normalisation rules (see Figure \ref{fig:normalisation_constraint_systems}), and one of them uses our custom notation for applying a substitution \(\Sigma\) to a formula (Section \ref{sec:mgs-app}, Definition \ref{def:application_mgs_formula}).
  Finally, for the sake of succinctness, if \(S = (\P,\C,\C^e)\) is an extended symbolic process we refer as \(\Phi(S), \Eqfst(S), \Eqsnd(S), \ldots\) to the corresponding components of \(\C^e\).

  \begin{figure}[ht]
    \begin{equation}
      \tag{\mbox{\textsc{Vect-rm-Unsat}}}
      \label{rule:vectorbot}
      \S \cup \{ \Gamma \cup \{ (\P,\C,\bot) \} \}
      \simpl
      \S \cup \{ \Gamma \}
    \end{equation}
    \begin{equation}
      \tag{\mbox{\textsc{Vect-Split}}}
      \label{rule:vector-split solved}
      \S \cup \{\Gamma\} \simpl \S \cup \{\Gamma^+, \Gamma^-\}
    \end{equation}
    if \(\Gamma^+,\Gamma^-\) is a partition of \(\Gamma\) and there exists a formula \(\psi\) such that
    \begin{enumerate}
      \item \(\forall S \in \Gamma^+,\,\exists \psi' \in \USolved(S),\, \psi \receq \psi' \text{ and } \psi'\) solved; and
      \item \(\forall S \in \Gamma^-,\, \forall \psi' \in \USolved(S),\, \psi \not\receq \psi'\)
    \end{enumerate}

    \begin{equation}
      \tag{\mbox{\textsc{Vect-add-Conseq}}}
      \label{rule:vector-solve}
      \S \cup \left\{ \Gamma \right\}
      \simpl
      \S \cup \left\{ \{ S[\Solved \mapsto \Solved \wedge \xi \dedfact u_{S}] \mid S \in \Gamma\} \right\}
    \end{equation}
    if for all \(S \in \Gamma\), \(S\) is solved, \(\xi \dedfact u_{S} \in \USolved(S)\) and for all second-order term \(\zeta\), \((\zeta,u_{S}) \notin \conseq(\Solved(S) \cup \Df(S))\)

    \begin{equation}
      \tag{\mbox{\textsc{Vect-add-Formula}}}
      \label{rule:vector-consequence}
      \S \cup \left\{ \Gamma \right\}
      \simpl
      \S \cup \left\{ \{ S[\USolved \mapsto \USolved \wedge \FApply{\Sigma}{\psi}{S})] \mid S \in \Gamma\} \right\}
    \end{equation}
    if \(\psi = \clause[X,Y,z]{X \eqf Y}{(X \dedfact z \wedge Y \dedfact z)}\) and \(\Sigma = \{X \mapsto \xi, Y \mapsto \zeta\}\), and
    for all \(S \in \Gamma\),
    \begin{enumerate}
      \item \(S\) is solved
      \item \(\USolved(S)\) contains a formula of the form \(\xi \dedfact u_{S}\).
      Besides, there should exist \(S \in \Gamma\), such that \((\zeta,u_S) \in \conseq(\Solved(S) \cup \Df(S))\).
      \item for all \((\clause{\zeta_1 \eqf \zeta_2}{\varphi}) \in \USolved(S)\), \(\zeta_1 \neq \xi\) and \(\zeta_2 \neq \xi\)
    \end{enumerate}
    \caption{Vector-simplification rules for sets of sets of extended symbolic processes}
    \label{fig:normalisation_vector}
  \end{figure}

  Rule \eqref{rule:vectorbot} removes \(\bot\) elements from the vector.
  Rule \eqref{rule:vector-split solved} splits a component whenever a common solution would yield statically inequivalent frames.
  More specifically, the rule separates the constraint systems in \(\Gamma^+\) in which a given recipe always yields a message (resp. an equality always holds) from the constraint systems in \(\Gamma^-\) in which the same recipe would never yield a message (resp. the same equality would never hold).
  This is characterised by the fact that a deduction (resp. equality) formula is solved in some constraint systems and not in the others. 
  Rule \eqref{rule:vector-solve} adds a solved deduction formula from \(\USolved(S)\) to \(\Solved(S)\) when this formula is solved in the entire component \(\Gamma\) and the new knowledge-base entries are not redundant with existing ones.
  Finally, when an equality fact \(\xi \eqf \zeta\) should hold in one constraint system, Rule \eqref{rule:vector-consequence} adds it to the entire component \(\Gamma\) (with appropriate hypotheses).
  Observe that we use in this rule the placeholder formula \(\psi = \clause[X,Y,z]{X \eqf Y}{(X \dedfact z \wedge Y \dedfact z)}\), introduced in Section \ref{sec:formulas}, stating that two recipes deducing the same term should verify an equality fact.

\subsection{\cs: case distinction rules} \label{sec:case-distinction-rules}

Our case distinction rules take the form of a transition system on vectors \(\S\) of extended symbolic processes similarly to the vector-simplification rules.
There are three different rules, each operating in a similar manner:
given a vector \(\S \cup \{\Gamma\}\), all rules perform a transformation of the following form on one component \(\Gamma\):
\[\tag{\mbox{\(\star\star\)}} \label{rule:case-base}
  \S \cup \{\Gamma\} \rightarrow \S \cup \{\Gamma^+,\Gamma^-\}\]
where \(\Gamma^+\) (the \emph{positive branch}) is intuitively obtained by applying a mgs \(\Sigma\) on each symbolic processes of \(\Gamma\) and \(\Gamma^-\) (the \emph{negative branch}) by adding the formula \(\neg\Sigma\) to each symbolic process \(S \in \Gamma\), where
\begin{align*}
  \neg\Sigma & = \forall S. \bigvee_{X \in \dom(\Sigma)} X \neqs X\Sigma &
  \mbox{with } S = \vars[2](\im(\Sigma)) \smallsetminus \vars[2](\Gamma)
\end{align*}
Intuitively this refines the component \(\Gamma\) by considering the cases where \(\Sigma\) is a solution or not.
After that, normalising the refined components \(\Gamma^+\) and \(\Gamma^-\) with the simplification rules---in particular Rules \eqref{rule:vectorbot} and \eqref{rule:vector-split solved}---will discard impossibles cases and separate processes with newly-found non-statically-equivalent solutions.
The three case distinction rules \eqref{rule:satisfiable}, \eqref{rule:equality} and \eqref{rule:rewrite} are presented in the next sections by specifying how \(\Gamma^+\) and \(\Gamma^-\) are computed from \(\Gamma\).
They are applied using a particular strategy defined by the following ordering on rules (where < means ``has priority over''):
\[\ref{rule:satisfiable} < \ref{rule:equality} < \ref{rule:rewrite}\]

Note that this ordering is mostly arbitrary: only the minimality of \ref{rule:satisfiable} will be needed in Section~\ref{sec:termination} for complexity.
The other inequalities are only there to reduce the number of cases to be considered in proofs.

\subsubsection{Rule Sat}

  The first rule focuses on satisfiability:
  its goal is to separate extended constraint systems of \(\Gamma\) that do not have the same solutions.
  For example if we have \(S = (\P,\C,\C^e) \in \Gamma\) and \(\Sigma \in \mgs(\C^e)\), all other symbolic processes \(S' \in \Gamma\) should also have a solution that is an instance of \(\Sigma\) (and if not, the component \(\Gamma\) should be split to separate \(S\) and \(S'\)).
  In particular this ensures that when this rule cannot be applied any more, all extended constraint system in \(\Gamma\) share a common, unique mgs (in particular they are in solved form).
  The same mechanism can be used to consider the solutions \(\Sigma\) making trivial some disequations of \(\Eqfst\) or hypotheses of some formulas in \(\USolved\).
  In particular the normalisation rules defined earlier in Section \ref{sec:normalisation-rules} will then handle the now trivial or unsatisfiable constraints.
  All this can be formalised as an instance of \eqref{rule:case-base} with:

  \begin{mdframed}[style=margindefstyle]
    For all $\Sigma$, $\Gamma$,
    \begin{equation}
      \tag{\text{\textsc{Sat}}}
      \label{rule:satisfiable}
      \begin{array}{r@{\ }l}
        \Gamma^+ & = \{\CApply{\Sigma}{S} \mid S \in \Gamma\} \\
        \Gamma^- & = \{S[\Eqsnd \mapsto \Eqsnd \wedge \neg\Sigma] \mid S \in \Gamma\}
      \end{array}
    \end{equation}
    where there exists \(S \in \Gamma\) such that either
    \begin{enumerate}
      \item \label{it:rule-sat-mgs}
      \(S\) not solved and \(\Sigma \in \mgs(S)\); or otherwise
      \item \label{it:rule-sat-hyp}
      there exists \(\psi \in \USolved(S)\) not solved and \(\Sigma \in \mgs(S[\Eqfst \mapsto \Eqfst \wedge \Fhyp(\psi)])\); or
      \item \label{it:rule-sat-diseq}
      \(\Eqfst(S)\) contains a disequality \(\psi = \forall \tilde{x}.\phi\) and \(\Sigma \in \mgs(S[\Eqfst \wedge \psi \mapsto \Eqfst]\mgu(\neg\phi))\).
    \end{enumerate}
  \end{mdframed}

\subsubsection{Rule Eq}

  The second case distinction rule focuses on the static equivalence between solutions of extended constraint systems.
  More specifically, the rule \eqref{rule:equality} checks whether an entry \(\xi_1 \dedfact u_1\) of one knowledge base of \(\Gamma\) can deduce the same term as another recipe \(\xi_2\) consequence of \(\Solved\).
  The rule is formalised as an instance of \eqref{rule:case-base} with
  \begin{mdframed}[style=margindefstyle]
    For all $\Sigma$, $\Sigma_0$, $\Gamma$,
    \begin{align}
      \tag{\text{\textsc{Eq}}}
      \label{rule:equality}
      \begin{array}{r@{\ }l}
        \Gamma^+ & = \{\CApply{\Sigma}{S}[\USolved \mapsto \USolved \wedge \FApply{\Sigma_0\Sigma}{\psi}{\CApply{\Sigma}{S}}] \mid S \in \Gamma\} \\
        \Gamma^- & = \{S[\Eqsnd \mapsto \Eqsnd \wedge \neg\Sigma] \mid S \in \Gamma\}
      \end{array}
    \end{align}
    if there exist \(S \in \Gamma\), such that \(\Sigma \in \mgs(S[\Eqfst \mapsto \Eqfst \wedge H_E, \Df \mapsto \Df \wedge H_D])\),
    where \(H_E\) and \(H_D\) are, respectively, the sets of equations and deduction facts of the hypotheses of \(\FApply {\Sigma_0} {\psi} {S}\), and:
    \begin{enumerate}
      \item \label{it:rule-equality-1}
      either \(\Sigma_0 = \{ X\rightarrow \xi_1, Y\rightarrow \xi_2 \}\) for some \((\xi_1 \dedfact u_1),(\xi_2 \dedfact u_2) \in \Solved(S)\) and
      for all \((\clause{H}{\varphi}) \in \USolved(S)\), \(H \neq (\xi_1 \eqf \xi_2)\); or
      \item \label{it:rule-equality-2}
      \(\Sigma_0 = \{ X \rightarrow \xi_1, Y \rightarrow \ffun(X_1,\ldots,X_n)\}\) for some \((\xi_1 \dedfact u_1) \in \Solved(S)\) and
      \(\ffun/n \in \sigc\) with \(\quanti{X_1}{k},\ldots, \quanti{X_n}{k}\) fresh and for all \((\clause{\zeta_1 \eqf \zeta_2}{\varphi}) \in \USolved(S)\),
      \(\zeta_b = \xi_1\) implies \(\rootf(\zeta_{1-b}) \neq \ffun\).
    \end{enumerate}
    where \(k = |\dom(\Phi(S))|\) and \(\psi = \clause[X,Y,z]{X \eqf Y}{(X \dedfact z \wedge Y \dedfact z)}\) with \(\quanti{X}{k},\quanti{Y}{k},z\) fresh variables.
  \end{mdframed}

  Similarly to Rule \eqref{rule:vector-consequence} the rule uses the generic equality formula
  \(\psi\)
  and the hypotheses of \(\FApply{\Sigma_0}{\psi}{S}\) express that \(X \Sigma_0\) and \(Y \Sigma_0\) deduce the same term.
  Since a recipe consequence of \(\Solved\) can either be coming from a deduction fact in \(\Solved\) or be a recipe with a constructor symbol at its root, we consider the two cases \ref{it:rule-equality-1} and \ref{it:rule-equality-2} each with the appropriate instantiation \(\Sigma_0\) of the placeholders \(X\) and \(Y\).
  The side requirements that head-equivalent formulas should not already be present in \(\USolved(S)\) are simply here for termination purpose, thus avoiding infinite aggregation of redundant formulas.

\subsubsection{Rule Rew}
  The third case distinction rule focuses on saturating the knowledge base.
  For example when outputting a term \(u\), the corresponding symbolic rule \eqref{rule:e-out} will add a deduction fact \(\ax_n \dedfact u\) to \(\USolved\);
  the rule \eqref{rule:rewrite} will apply rewrite rules on \(u\) to determine whether new messages can be learned by the attacker.
  Typically if \(u = \pair{u_1,u_2}\) the following actions will happen:
  \begin{enumerate}
    \item after \(\ax_n \dedfact u\) has been added to \(\USolved\) by the symbolic rule \eqref{rule:e-out}, it will be copied to the knowledge base \(\Solved\) by the simplification rules \eqref{rule:vector-solve} (assuming \(u\) is not already deducible from any knowledge base of the component)
    \item after that, the case-distinction rule \eqref{rule:rewrite} will add the two deduction facts \(\fst(\ax_n) \dedfact u_1\) and \(\snd(\ax_n) \dedfact u_2\) to \(\USolved\), which may in turn be transferred to \(\Solved\) as well.
  \end{enumerate}

  More precisely, given a deduction fact \(\xi_0 \dedfact u_0\), the rule checks whether one may apply a rewrite rule \(\ell \rightarrow r\) to \(u_0\), which may require to first apply a context on \(u_0\) (for example if \(u_0 = \hfun(a)\) and \(\ell = \ffun(\gfun(\hfun(x)))\)).
  For that we introduce a notion of \emph{skeleton} of \(\ell\).

  \begin{definition}[rewriting skeleton] \label{def:skeleton}
    Let \(p\) be a position of a first-order term \(\ell\).
    A \emph{skeleton for \((\ell,p)\)} is a tuple \((\xi,t,D)\) such that
    \(\xi \in \termset(\sig \cup \sig_0 \cup \X[2])\),
    \(t \in \termset(\sig \cup \sig_0 \cup \X[1])\),
    \(D\) is a set of deduction facts and
    \[(\rootf(\getpos{\xi}{q}), \rootf(\getpos{t}{q})) = \left\{
      \begin{array}{ll}
        (\rootf(\getpos{\ell}{q}), \rootf(\getpos{\ell}{q})) & \mbox{for any strict prefix \(q\) of \(p\)}\\
        (X_{q},x_q) & \mbox{for any other position \(q\) of \(\xi\)}
      \end{array}
    \right.\]
    where the set of variables \(X_{q}\) (resp. \(x_q\)), \(q\) a position of \(\xi\) that is not strict prefix of $p$, are fresh pairwise distinct second-order (resp. first-order) variables,
    and \(D\) is the set of all the deduction facts \(X_q \dedfact x_q\).
    The set of all such skeletons (which, notably, are all identical up to variable renaming but may therefore differ on the second-order-variable types) is written \(\Skel{\ell}{p}\).
  \end{definition}


  For a skeleton \((\xi,t,D) \in  \Skel{\ell}{p}\), the recipe \(\xi\) represents the context that the attacker will apply on top of the deduction fact \(\xi_0 \dedfact u_0\) at the position \(p\) to obtain the left-hand side \(\ell\).
  The term \(t\) represents the corresponding generic term on which the rewrite rule will be applied.
  Finally \(D\) is the set of deduction facts linking the variables of \(\xi\) and \(t\).
  \medskip

  Consider now a component \(\Gamma\), a symbolic process \(S \in \Gamma\), a deduction fact \((\xi_0 \dedfact u_0) \in \Solved(S)\), and a context \(C\).
  The first role of Rule \eqref{rule:rewrite} is to saturate the knowledge base, that is, to deduce the new term \(C[u_0] \norm\) using \(C[\xi_0]\).
  However after adding the new deduction fact \(C[\xi_0] \dedfact C[u_0] \norm\) to \(\USolved(S)\), a head-equivalent formula should be added to all other symbolic processes of \(\Gamma\) whenever it is possible, so that the vector-simplification rule \eqref{rule:vector-split solved} (which separates processes with non-statically-equivalent solutions) only separates \(S' \in \Gamma\) from \(S\) if \(C[\xi_0]\) yields a valid message in \(S\) but not in \(S'\).
  Yet the behaviour of a destructor symbol may be described by multiple rewrite rules:
  the rewrite rule used to normalise \(C[u_0]\) may therefore not be the same as the one used to normalise the term deduced by \(C[\xi_0]\) in \(S'\).
  Because of this we have to add to \(\USolved(S')\) all formulas corresponding to using all possible rewrite rules.
  For that we consider the following set of generic formulas:
  \[
    \RewF{\xi}{\ell \rightarrow r}{p} =
      \left\{\clause[S]{\xi \dedfact r'}{(D \wedge \mgu(\ell' \eqs t))}\ \left|
        \begin{array}{l}
          \ell' \rightarrow r' \in \R\\
          (\xi,t,D) \in \Skel{\ell}{p}\\
          S = \vars(D,\ell')
        \end{array}\right.\right\}
  \]
  Let us now give a complete example to illustrate all these notions.
  The goal is to detail what formulas will be added to \(\USolved\) by Rule \eqref{rule:rewrite} on a concrete case as the actual definition of the rule is quite technical and hard to read---although the intuition behind it is rather simple.

  \begin{example}
    Consider a rewriting system defined by a binary symbol \(\hfun\) and the two rewrite rules
    \begin{align*}
      \otherfun(\hfun(x,y),x) & \to y &
      \otherfun(\hfun(x,y),y) & \to x
    \end{align*}
    These two rewrite rules give access to either argument of \(\hfun\) assuming the other one is known.
    Now consider a component \(\Gamma\) containing two extended symbolic processes \(S_1\) and \(S_2\) with the respective frames, given \(k,k',s,s' \in \Nall\):
    \begin{align*}
      \Phi(S_1) & = \{\ax_1 \mapsto \hfun(k,k'), \ax_2 \mapsto k\} &
      \Phi(S_2) & = \{\ax_1 \mapsto \hfun(s,s'), \ax_2 \mapsto s'\}
    \end{align*}
    They are are statically equivalent, even if the recipe \(\otherfun(\ax_1,\ax_2)\) is not normalised using the same rewrite rule in \(\Phi(S_1)\) and \(\Phi(S_2)\).
    We assume that \(\Solved(S_1) = \ax_1 \dedfact \hfun(k,k') \wedge \ax_2 \dedfact k\) and \(\Solved(S_2) = \ax_1 \dedfact \hfun(s,s') \wedge \ax_2 \dedfact s'\).
    We describe the application of Rule \eqref{rule:rewrite} that uses
    the rewrite rule \((\ell \to r) = (\otherfun(\hfun(x,y),x) \to y)\) to deduce a new term in \(S_1\) by putting \(\ax_1\) at the position of \(\hfun(x,y)\) in \(\ell\), i.e., position $1$.
    To begin the rule considers a skeleton \((\xi,t,D) \in \Skel{\ell}{1}\):
    \begin{align*}
      \xi & = \otherfun(X_1,X_2) &
      t & = \otherfun(x_1,x_2) &
      D & = X_1 \dedfact x_1 \wedge X_2 \dedfact x_2
    \end{align*}
    The set \(\RewF{\xi}{\ell \rightarrow r}{1}\) therefore contains the following two formulas (normalised by the simplification rules on formulae):
    \begin{align*}
      \psi_1 & = \clause[X_1,X_2,x,y]{\otherfun(X_1,X_2) \dedfact y}{\left(X_1 \dedfact \hfun(x,y) \wedge X_2 \dedfact x\right)} \\
      \psi_2 & = \clause[X_1,X_2,x,y]{\otherfun(X_1,X_2) \dedfact x}{\left(X_1 \dedfact \hfun(x,y) \wedge X_2 \dedfact y\right)}
    \end{align*}
    They are only generic formulas indicating that when the left side of a rewrite rule can be computed then the right side can be computed as well.
    Since our goal is to apply the rewrite rule \(\ell \to r\) where the deduction fact \(\ax_1 \dedfact \hfun(k,k')\) is used to deduce the subterm \(\getpos{\ell}{1}\), we have to replace the variable at position 1 in \(\xi\), namely \(X_1\), by \(\ax_1\) in these formulas.
    We do this by applying the substitution \(\Sigma_0 = \{X_1 \mapsto \ax_1\}\) to \(\psi_1,\psi_2\) (in the sense defined in Section \ref{sec:mgs-app}, again normalised by simplification rules):
    \begin{align*}
      \FApply{\Sigma_0}{\psi_1}{S_1} & = \clause[X_2]{\otherfun(\ax_1,X_2) \dedfact k'}{X_2 \dedfact k} \\
      \FApply{\Sigma_0}{\psi_2}{S_1} & = \clause[X_2]{\otherfun(\ax_1,X_2) \dedfact k}{X_2 \dedfact k'}
    \end{align*}
    Rule \eqref{rule:rewrite} will then compute most general solutions to instantiate \(X_2\) in a way that satisfies the hypotheses of these formulas.
    In the case of the first formula we have the unique solution \(\mgs(S_1[\Df \mapsto \Df \wedge X_2 \dedfact k]) = \{\Sigma\}\) where \(\Sigma = \{X_2 \mapsto \ax_2\}\);
    the algorithm will therefore add the following deduction fact to \(\USolved(S_1)\):
    \[\FApply{\Sigma_0\Sigma}{\psi_1}{\CApply{\Sigma}{S_1}} = \otherfun(\ax_1,\ax_2) \dedfact k'\]
    On the contrary, the algorithm could not have added the second formula: 
    we have \(\mgs(S_1[\Df \mapsto \Df \wedge X_2 \dedfact k']) = \emptyset\), meaning that no solutions satisfy its hypotheses.
    Then we are almost done:
    as explained earlier it only remains to add a head-equivalent formula to \(\USolved(S_2)\), if any, so that the vector-simplification rule \eqref{rule:vector-split solved} does not split \(\Gamma\) if the recipe \(\otherfun(\ax_1,\ax_2)\) yields a valid message in both \(S_1\) and \(S_2\).
    That is, we should add to \(\USolved(S_2)\):
    \begin{align*}
      \FApply{\Sigma_0\Sigma}{\psi_1}{\CApply{\Sigma}{S_2}} & = (\clause{\otherfun(\ax_1,\ax_2) \dedfact s'}{s' \eqs s}) \\
      \FApply{\Sigma_0\Sigma}{\psi_2}{\CApply{\Sigma}{S_2}} & = (\clause{\otherfun(\ax_1,\ax_2) \dedfact s}{s' \eqs s'})
    \end{align*}
    This time the situation is reversed compared to \(S_1\):
    the first formula has unsatisfiable hypotheses (and will therefore be discarded at the next round of normalisation rules by Rule \eqref{rule:Disequation removal 2}) and the second one will be simplified to \(\otherfun(\ax_1,\ax_2) \dedfact s\).
    This illustrates why we add one formula for each rewrite rule:
    should we have only considered \(\psi_1\), we would have missed the head-equivalent formula \(\otherfun(\ax_1,\ax_2) \dedfact s\) in \(S_2\), resulting in the incorrect conclusion that \(\Phi(S_1) \not\StatEq \Phi(S_2)\).
  \end{example}

  \newcommand{\USolvedFunc}[1]{\mathfrak{F}(#1)}

  Let us now formalise Rule \eqref{rule:rewrite} in full generality.
  It can now be defined as an instance of \eqref{rule:case-base} under the following conditions:
  \begin{mdframed}[style=margindefstyle]
    For all $\Sigma$, $\Gamma$,
    \begin{align}
      \tag{\text{\textsc{Rew}}}
      \label{rule:rewrite}
      \begin{array}{r@{\ }l}
        \Gamma^+ & = \{\CApply{\Sigma}{S}[\USolved \mapsto \USolved \wedge \USolvedFunc{S}] \mid S \in \Gamma\} \\
        \Gamma^- & = \{S[\Eqsnd \mapsto \Eqsnd \wedge \neg\Sigma] \mid S \in \Gamma\}
      \end{array}
    \end{align}
    if there exist \(S \in \Gamma\), \(\ell\rightarrow r \in \R\),
    \(p\) position of \(\ell\),
    \(\xi \in \termset(\sig \cup \Xsndi[=]{k})\) with \(k = |\dom(\Phi(S))|\),
    \(\psi_0 \in \RewF{\xi}{\ell \rightarrow r}{p}\), \((\xi_0 \dedfact u_0) \in \Solved(S)\), and a function $\mathfrak{F}$ from subsets of $\Gamma$ to constraints such that the following conditions are met:
    \begin{enumerate}
      \item \label{it:rule-rewrite-1}
      \(p \neq \epsilon\) and \(\getpos {\ell} {p} \notin \X[1]\)
      \item \label{it:rule-rewrite-3}
      \(\Sigma_0 = \{ \getpos{\xi}{p} \rightarrow \xi_0\}\)
      and \(\Sigma \in \mgs(S[\Df \mapsto \Df \wedge \Df(\psi_1), \Eqfst \mapsto \Eqfst \wedge \Fhyp(\psi_1)])\) if \(\psi_1 = \FApply{\Sigma_0}{\psi_0}{S}\)
      \item \label{it:rule-rewrite-5}
      \(\Sigma_1\) is an injection from \(\im(\Sigma) \setminus \vars[2](\psi_1)\) to fresh constants, and for any such injection \(\Sigma_1'\), we have \(\FApply{\Sigma_0 \Sigma \Sigma_1'}{\psi_0}{\CApply{\Sigma}{S}} \notin \USolved(S)\)
      \item \label{it:rule-rewrite-2}
      for all \(S' \in \Gamma\), \(\USolvedFunc{S'} = \{\FApply{\Sigma_0 \Sigma \Sigma_1} {\psi} {\CApply{\Sigma}{S'}} \mid \psi \in \RewF{\xi}{\ell \rightarrow r}{p}\}\)
    \end{enumerate}
  \end{mdframed}

  The conditions are rather technical but simply capture the steps of the example.
  Commenting the requirements of the rule, Item \ref{it:rule-rewrite-1} ensures that the rewriting is not performed at a trivial position.
  Items~\ref{it:rule-rewrite-3} and \ref{it:rule-rewrite-5} describe the formulas added to each symbolic process of \(\Gamma\):
  just as in the example they are obtained by choosing one \(S \in \Gamma\), computing a formula \(\psi_0\) (\(\psi_1\) in the example) corresponding to applying one given rewrite rule \(\ell \to r\) in \(S\), replacing the position \(p\) of \(\ell\) by an entry of \(\Solved(S)\) by applying \(\Sigma_0\) and computing a solution \(\Sigma\) of the hypotheses of the resulting formula.
  Finally, in Item~\ref{it:rule-rewrite-2}, all other symbolic processes \(S' \in \Gamma\) receive the formulas of \(\USolvedFunc{S'}\), each attempting to apply a rewrite rule to yield a valid message with the same recipe as in \(\psi_0\). 

  Note that the computation of the mgs \(\Sigma\) may leave some second-order variables \(X_1, \ldots, X_n\) unconstrained because they do not need to be instantiated in a particular way to obtain a solution.
  This is where Item~\ref{it:rule-rewrite-5} come into play, replacing these pending variables by fresh constants.

\subsection{All in all: computing a partition tree} \label{sec:correctness-proc}

\paragraph{Overall procedure}
  We make reference to the various constraint-solving relations defined in the previous sections using the following notations:
  \begin{mathpar}
    \simplifstep \text{ : simplification rules on formulas} \and
    \normstep \text{ : normalisation rules} \and
    \vectstep \text{ : vector-simplification rules} \\
    \satstep \text{ : \eqref{rule:satisfiable} case-distinction rule} \and
    \eqstep \text{ : \eqref{rule:equality} case-distinction rule} \and
    \rewstep \text{ : \eqref{rule:rewrite} case-distinction rule}
  \end{mathpar}
  All these transition relations are interpreted as binary relations on vectors.
  We recall that we call a \emph{component} a set \(\Gamma\) of extended symbolic processes and a \emph{vector} a set \(\S\) of components, and that all \((\P,\C,\C^e) \in \Gamma\) induce a predicate \(\pi\) on second-order solutions of \(\C\) such that
  \[\Sol[\pi](\C) = \{(\Sigma_{|\vars[2](\C)}, \sigma_{|\vars[1](\C)}) \mid (\Sigma,\sigma) \in \Sol(\C^e)\}\]
  We therefore propose in Algorithm \ref{alg:ptree} a procedure to compute the partition tree of two bounded plain processes, where the nodes are labelled by components instead of regular partition-tree configurations;
  in particular the proof of correctness of this algorithm has to
  justify that the above predicate \(\pi\) can be defined uniformly
  across the entire nodes of the computed tree.

  \newcommand\nodegen{\mathsf{generateSubtree}}
  \newcommand\rootgen{\mathsf{generateRoot}}
  \newcommand\applysimpl{\mathsf{applySimpl}}
  \newcommand\applycase{\mathsf{applyCase}}
  \newcommand\Ifkw[2]{\textbf{if}\ #1\ \textbf{then}\ #2}
  \newcommand\ElseIfkw[2]{\textbf{else if}\ #1\ \textbf{then}\ #2}
  \newcommand\Elsekw[1]{\textbf{else}\ #1}
  \begin{algorithm}
    \caption{Computation of the partition tree, with nodes labelled with components}
    \label{alg:ptree}

    \algocomment{Application of simplification rules, as much as possible}\;
    \Def {\(\applysimpl(\S : \mathsf{Vector}) : \mathsf{Vector}\)} {
      \Ifkw {\(\S \simplifstep \S'\)} {%
        \(\returnkw\ \applysimpl(\S')\)\;
      }
      \ElseIfkw {\(\S \normstep \S'\)} {%
        \(\returnkw\ \applysimpl(\S')\)\;
      }
      \ElseIfkw {\(\S \vectstep \S'\)} {%
        \(\returnkw\ \applysimpl(\S')\)\;
      }
      \Elsekw {%
        \(\returnkw\ \S\)\;
      }
    }\;

    \algocomment{Application of case-distinction rules, with simplification rules in between}\;
    \Def {\(\applycase(\S : \mathsf{Vector}) : \mathsf{Vector}\)} {
      \Ifkw {\(\S \satstep \S'\)} {%
        \(\returnkw\ \applycase(\applysimpl(\S'))\)\;
      }
      \ElseIfkw {\(\S \eqstep \S'\)} {%
        \(\returnkw\ \applycase(\applysimpl(\S'))\)\;
      }
      \ElseIfkw {\(\S \rewstep \S'\)} {%
        \(\returnkw\ \applycase(\applysimpl(\S'))\)\;
      }
      \Elsekw {%
        \(\returnkw\ \S\)\;
      }
    }\;

    \algocomment{Generates the subtree rooted on a node labelled by the component \(\Gamma\)}\;
    \Def {\(\nodegen(\Gamma : \mathsf{Component}) : \mathsf{Tree}\)} {
      \(\Gamma_{\inp} \leftarrow \{S' \mid S \Sstep{\InP{Y}{X}} S', S \in \Gamma\}\)\;
      \(\Gamma_{\outp} \leftarrow \{S' \mid S \Sstep{\OutP{Y}{\ax}} S', S \in \Gamma\}\)\;
      \If {\(\Gamma_\inp = \Gamma_\outp = \emptyset\)} {%
        \(\returnkw\) a tree reduced to its root, labelled \(\Gamma\)\;
      }
      \Else{} {%
        \(\S \leftarrow \applycase(\applysimpl(\{\Gamma_\inp,\Gamma_\outp\}))\)\;
        \(T \leftarrow\) tree with root \(\Gamma\) and children \(\nodegen(\Gamma')\) for each \(\Gamma' \in \S\)\;
        \(\returnkw\ T\)\;
      }
    }\;

    \algocomment{Generates the root and then an entire partition tree of \(P_1\) and \(P_2\)}\;
    \Def {\(\ptree(P_1,P_2 : \mathsf{Processes}) : \mathsf{Tree}\)} {
      \(\Gamma_1 \leftarrow \{S \mid P_1 \Sstep{\epsilon} S\}\)\;
      \(\Gamma_2 \leftarrow \{S \mid P_2 \Sstep{\epsilon} S\}\)\;
      \(\{\Gamma\} \leftarrow \applysimpl(\{\Gamma_1 \cup \Gamma_2\})\) \algocomment{in the root, simplification rules never split the vector}\;
      \(\returnkw\ \nodegen(\Gamma)\)\;
    }
  \end{algorithm}

\paragraph{Correctness arguments}

  To conclude we mention that the core arguments justifying that Algorithm \ref{alg:ptree} effectively generates a partition tree can be found in Appendix~\ref{app:ptree-proof}.
  Technically, most of the theorem statements rely on a collection of invariants, with a proof of their preservation at each step of the procedure (called with ``\(\predlab\)'' names such as \(\PredWellFormed\), \(\PredCorrectFormula\),...).

\section{Termination and complexity}
\label{sec:complex}


\subsection{Preliminaries}

  \begin{enumerate}
    \item In Section \ref{sec:termination} we prove that Algorithm \ref{alg:ptree} uses a \emph{finite number of constraint-solving rules} to compute each branch of the partition tree.
    This proves the all considered security relations (trace equivalence and inclusion, simulation, (bi)similarity) to be decidable.
    \item For complexity purposes we then refine this result in Section \ref{sec:exp-mgs}:
    we prove that Algorithm~\ref{alg:ptree} applies at most an \emph{exponential number of rules} and that the nodes of the resulting partition tree have \emph{most general solutions of exponential (DAG) size}.
    \item Relying on these bounds, we show that two processes are not equivalent \textit{iff} there exists a \emph{non-equivalence witness of exponential size} (as defined in Section~\ref{sec:ptree-eq}).
    This shows the security relations to be decidable in co\nexp time.
    \item Finally we show in Section~\ref{sec:deepsec-hardness} that the security relations are \emph{co\nexp hard}.
    We also provide a complexity analysis in the pure pi-calculus.
    All in all:
  \end{enumerate}

  \begin{theorem}[restate=thmDeepsecConexp,name={complexity of equivalences}] \label{thm:deepsec-conexp}
    For bounded processes, the problems \TraceEquiv, \TraceInclus, \Simulation, \Similarity, and \Bisimilarity are co\nexp complete for constructor-destructor subterm convergent theories.
    Besides, in the pure pi calculus, \TraceEquiv and \TraceInclus are \polyh{2} complete, and \Simulation, \Similarity and \Bisimilarity are \pspace complete.
  \end{theorem}

  \paragraph{Notations}
  We also introduce some notations that will be used in most incoming sections.
  We recall that we study complexity w.r.t. the DAG size of terms (which provides stronger results compared to complexity bounds w.r.t. the tree size of terms);
  in particular the DAG size of a substitution \(\sigma\) is \(\dagsize{\sigma} = |\subterms(\im(\sigma))|\), hence the many occurrences of subterm sets below.
  Given an extended constraint system \(\C^e = (\Phi,\Df,\Eqfst,\Eqsnd,\Solved,\USolved)\) we write
  \begin{align*}
    \tag{first-order mgu}
    \mu^1 & = \mgu(\Eqfst) \\
    \tag{second-order mgu}
    \mu^2 & = \mgu(\Eqsnd) \\
    \tag{first-order terms}
    \terms[1] & = \subterms[1](\im(\Phi\mu^1),\im(\mu^1),\Solved \mu^1, \Df\mu^1) \\
    \tag{second-order terms}
    \terms[2] & = \subterms[2](\im(\mu^2),\Solved, \Df) \\
    \tag{solution recipes}
    \recipes & = \stc(\im(\mu^2),\ \Solved \cup \Df) \cup \vars[2](\Df)
  \end{align*}
  When the extended constraint system is not clear from context we write explicitly \(\mu^1(\C^e)\), \(\mu^2(\C^e)\), \(\terms[1](\C^e)\),...
  Intuitively \(\mu^i\) are the mgu's of the equations of \(\Eq[i]\) and we recall in particular that \(\mgs(\C^e) = \{\mu^2\}\) when \(\C^e\) is solved (Section~\ref{sec:mgs-proc}, Proposition~\ref{prop:correct-mgs-solved}).
  The other notations assume \(\mu^1 \neq \bot\) (if \(\mu^1 = \bot\), \(\C^e\) will be discarded by the normalisation rule \eqref{rule:mgs-unsat} anyway).
  The sets \(\terms[1]\) and \(\terms[2]\) respectively represent the first-order and second-order terms appearing in the system, while \(\recipes \subseteq \terms[2]\) models the set of recipes used to build the solution of \(\C^e\) (i.e., \(\mu^2\)) from \(\Solved \cup \Df\).
  We recall that it is the same set as the one used when defining the constraint-solving rules for most general solutions (Section~\ref{sec:mgs-rules}).

  \begin{remark}[uniformity of second-order terms across components]
    Due to an invariant of the procedure (\(\PredStruct\) formalised in Appendix \ref{app:ptree}, Section \ref{app:invariants}), we know that all extended constraint systems in a component \(\Gamma\) have the same second-order structure.
    Here this means that \(\mu^2(\C^e_1) = \mu^2(\C^e_2)\) and \(\terms[2](\C^e_1) = \terms[2](\C^e_2)\) for any \((\P_1,\C_1,\C^e_1),(\P_2,\C_2,\C^e_2) \in \Gamma\).
    For this reason we may write \(\mu^2(\Gamma)\) or \(\terms[2](\Gamma)\) instead of \(\mu^2(\C^e)\) or \(\terms[2](\C^e)\) for some arbitrary \((\P,\C,\C^e) \in \Gamma\).
  \end{remark}

\subsection{Termination of the constraint solving}
\label{sec:termination}

\subsubsection{Termination of the computation of most general solutions} \label{sec:termination-mgs}

We first study the termination of the procedure for computing most general solutions provided in Section~\ref{sec:mgs-gen}.
The proof mostly relies on the following measure that characterises the set of first-order terms of an extended constraint system \(\C^e\) that are not used in \(\mgs(\C^e)\), that is, that are not deduced by any recipe \(\xi \in \recipes(\C^e)\):
\[\measureNC(\C^e) =
  \{ t \in \terms[1] \mid t \not\in \X[1] \wedge \forall \xi \in \recipes(\C^e) \smallsetminus \X[2],\, (\xi,t) \notin \conseq(\Solved\mu^1 \cup \Df\mu^1)\}
\]
The simplification rules for mgs do not affect this value (except if \(\C^e\) is replaced by \(\bot\) by \eqref{rule:mgs-unsat}).
The application of \eqref{rule:conseq} will ensure that \(\measureNC\) is at least non-in\-creasing, while \eqref{rule:res} and \eqref{rule:cons} make it strictly decreasing.
We summarise this as the following proposition, proved in Appendix~\ref{app:termination};
we recall that, similarly to the correctness arguments in Section~\ref{sec:ptree}, the statements makes reference to some procedure invariants formalised in Appendix~\ref{app:invariants}:

\begin{proposition}[restate=propMgsDecrease,name={decrease of unused first-order terms during constraint solving}] \label{prop:mgs-decrease}
  Let \(\C^e\) be an extended constraint system such that \(\C^e \simplnorm = \C^e\) and the invariants \(\PredWellFormed(\C^e)\) and \(\PredCorrectFormula(\C^e)\) hold.
  Then let \(\C^e \simpStep{\Sigma} \C^{e\prime} \neq \bot\).
  If this transition is derived with:
  \begin{enumerate}
    \item Rule \eqref{rule:conseq}: \(|\measureNC(\C^{e\prime})| \leqslant |\measureNC(\C^e)|\)
    \item Rules \eqref{rule:res} or \eqref{rule:cons}: \(|\measureNC(\C^{e\prime})| < |\measureNC(\C^e)|\)
  \end{enumerate}
\end{proposition}

Using this proposition we can then easily prove the computation of the set of most general solutions to be terminating, and actually to give an upper bound on its cardinality:

\begin{theorem}[restate=propMgsSize,name={termination for most general solutions}]
  There exist no infinite sequences of transitions w.r.t. \(\simpStep{}\).
  Besides if \(\C^e\) is an extended constraint system such that the invariants \(\PredWellFormed(\C^e)\) and \(\PredCorrectFormula(\C^e)\) hold, we have
  \[|\mgs(\C^e)| \leqslant (|\Solved(\C^e)| + 1)^{|\measureNC(\C^e)|}\]
\end{theorem}

\begin{proof}
  First of all we observe that consecutive applications of Rule \eqref{rule:conseq} are terminating, since applying this rule strictly decrease the cardinality of the set \(m(\C^e)\) of parameters \((\xi,\zeta)\) the rule can be applied with.
  Combining this with Proposition \ref{prop:mgs-decrease} we obtain that if \(\C^e \simpStep{\Sigma} \C^{e\prime} \neq \bot\) then \(\C^e < \C^{e\prime}\) w.r.t. the lexicographic composition of \(\measureNC\) and \(m\).

  Besides consecutive applications of Rule \eqref{rule:conseq} are also confluent by unicity of mgu's.
  For the same reason the applications of Rules \eqref{rule:res} or \eqref{rule:cons} can be performed on one deterministically-chosen deduction fact \((X \dedfact u) \in \Df\).
  We therefore obtain \(\mgs(\C^e) = \{\Sigma_{|\vars[2](\C^e)} \mid \C^e \SimpStep{\Sigma} \C^{e\prime} \text{ normalised}, \C^{e\prime} \text{ solved}\}\)
  where a reduction \(\C^e \SimpStep{\Sigma} \C^{e\prime}\) is said to be \emph{normalised} when all applications of Rule \eqref{rule:conseq} and the choice of deduction facts in Rules \eqref{rule:res} or \eqref{rule:cons} are done in a fixed, deterministic way.
  Since for any \(\C^e\), there are at most \(|\Solved(\C^e)|\) normalised applications of Rule \eqref{rule:res} and 1 normalised application of Rule \eqref{rule:cons}, we deduce by Proposition \ref{prop:mgs-decrease} that \(|\mgs(\C^e)| \leqslant (|\Solved(\C^e)|+1)^{|\measureNC(\C^e)|}\).
\end{proof}

\subsubsection{Termination of the computation of partition trees}
To bound the number of rule applications in Algorithm \ref{alg:ptree} we define a well-founded measure that decreases after each case-distinction, simplification, normalisation and vector-simplification rules.
More precisely the rule applications are always of the form
\begin{align*}
  \S \cup \{\Gamma\} & \rightarrow \S \cup \{\Gamma_1, \ldots, \Gamma_p\} &
  p \in \{1,2\}
\end{align*}
and we show that for all \(i \in \eint{1}{p}\), \(\Gamma > \Gamma_i\) w.r.t. to a well-founded measure on components (under the invariants of the procedure defined in Appendix \ref{app:ptree}).
This therefore bounds the number of rule applications to compute a given branch of the partition tree.
The measure in question is a tuple of 9 integer components that is ordered w.r.t. the lexicographic ordering.

\paragraph{Measure 1:  sizes of the processes}

  As first element of the measure, we compute a maximum on the sizes of the processes in the multisets \(\P\), that is,
  \defcomp{1}{
    \max_{(\P,\C,\C^e) \in \Gamma} \sum_{R \in \P} \dagsize{R}
  }
  Notice that this stays unchanged for any simplification or case distinction rules but strictly decreases when applying the extended symbolic transitions.

\paragraph{Measure 2:  Number of constraint systems}

  The third element of the measure considers the number of extended symbolic processes in the set, i.e., \(|\Gamma|\), that may increase only when applying a symbolic transition;
  however it strictly decreases when applying the simplification rules \eqref{rule:vector-split solved} and \eqref{rule:vectorbot}.
  Moreover it also strictly decreases for the positive branch of Rule \eqref{rule:satisfiable} when applied with the case \ref{it:rule-sat-diseq} of its application conditions.
  In such a case, we consider a disequation \(\psi\) and a mgs \(\Sigma\) of \(\C^e_j\) that does not satisfy \(\psi\), which will lead to at least one \(S \in \Gamma\) being discarded by the simplification rule \eqref{rule:vectorbot}.
  \defcomp{2}{|\Gamma|}

\paragraph{Measure 3:  Number of terms not consequence}

  Given \(\C^e\) an extended symbolic constraint system, let us consider the following set representing the set of terms that are not consequence of \(\Solved(\C^e)\) and \(\Df(\C^e)\):
  \[
    \setSDF(\C^e) = \{ t \in \terms[1] \mid \forall \xi, (\xi,t) \notin \conseq(\Solved(\C^e) \cup \Df(\C^e))\}
  \]
  Typically it corresponds to the terms that are not deducible by the attacker but could potentially be (because the knowledge base is not saturated yet).
  In fact, when the simplification Rule \eqref{rule:vector-solve} is applied, i.e., when a deduction fact \(\xi \dedfact u\) is added to \(\Solved(\C^e)\), the term \(u\) is necessarily a subterm of the frame by the invariant \(\PredWellFormed(\C^e)\).
  Moreover by definition of Rule \eqref{rule:vector-solve} we know that \(u\) is not already consequence, meaning that the size of \(\setSDF(\C^e)\) will strictly decrease.
  Finally the case distinction rules never increase the number of elements of \(\setSDF(\C^e)\):
  indeed they all consist of applying substitutions \(\Sigma\) that are most general solutions of some systems having \(\Solved(\C^e)\) as their knowledge base, hence their first-order terms are consequence by \(\Solved\)-basis.
  All in all we choose the following component:
  \defcomp{3}{\min_{(\P,\C,\C^e) \in \Gamma} |\setSDF(\C^e)|}


\paragraph{Measure 4:  Number of unsolved extended constraint systems}

  We recall that the aim of Rule \eqref{rule:satisfiable}, case \ref{it:rule-sat-mgs} of its application conditions, is to put extended constraint systems in solved form (that is, in a form where they trivially have \(\mu^2\) as a unique mgs).
  If
  \defcomp{4}{|\{(\P,\C,\C^e) \in \Gamma \mid \C^e \text{ unsolved}\}|}
  then this measure is strictly decreasing when applying the rule in question.
  Once a system has a unique mgs, instantiating its second variables does not change this fact and the other case distinction rules are therefore non-increasing w.r.t. this measure.


\paragraph{Measure 5:  Applicability of Rule REW}

  The next element represents the number of applications of Rule \eqref{rule:rewrite} that are still possible.
  Typically, we consider all the parameters of the rule \eqref{rule:rewrite} (the deduction facts from \(\Solved\), the rewrite rule, etc...) on which the rule would be applied with a most general solutions that does not already corresponds to a deduction fact in \(\USolved\).
  If \(\C^e\) is an extended constraint system we therefore consider \(\setRew(\C^e)\) the set of tuples \((\psi,\ell \rightarrow r, p, \psi_0, \Sigma)\) that satisfy all the application conditions of Rule \eqref{rule:rewrite}, and
  \defcomp{5}
    {\sum_{(\P,\C,\C^e) \in \Gamma} |\setRew(\C^e)|}
  By definition \(|\setRew(\C^e)|\) strictly decreases for at least one \((\P,\C,\C^e) \in \Gamma\) (and non-increasing for the others) when applying Rule \eqref{rule:rewrite}.
  Then let \((\P,\C,\C^e) \in \Gamma\):
  the other case distinction rules \eqref{rule:satisfiable} and \eqref{rule:equality} do not increase \(|\setRew(\C^e)|\).
  Indeed if we consider one of their applications \(\C^e \simpStep{\Sigma'} \CApply{\Sigma'}{\C^e}\),
  we have
  \[\setRew(\CApply{\Sigma'}{\C^e}) = \{(\psi\Sigma',\ell \rightarrow r, p, \psi_0, \Sigma\Sigma') \mid (\psi,\ell \rightarrow r, p, \psi_0, \Sigma) \in \setRew(\C^e)\}\,.\]

  Note however that \(|\setRew(\C^e)|\) may increase by application of Rule \eqref{rule:vector-solve} since \(|\Solved(\C^e)|\) will increase;
  yet the measure is already decreasing by the component \(\compon[3]\).

\paragraph{Measure 6:  Number of unsolved deduction formulas}

  We recall that Rule \eqref{rule:satisfiable}, case \ref{it:rule-sat-hyp} of its application conditions, applies a most general solution to remove the hypotheses of one formula \(\psi \in \USolved(\C^e)\) for some \((\P,\C,\C^e) \in \Gamma\)
  (the formula becomes solved in the positive branch, and is removed by \eqref{rule:Disequation removal 2} in the negative branch).
  This rule application is therefore strictly decreasing w.r.t. the measure
  \defcomp{6}
    {|\{\psi \in \USolved(\C^e) \mid (\P,\C,\C^e) \in \Gamma, \psi \text{ unsolved deduction formula}\}|}
  We only consider deduction formulas \(\psi\) for this component.
  In particular the only rule that may increase this measure (i.e., generate unsolved deduction formulas) is \eqref{rule:rewrite} which is already decreasing w.r.t. the previous component of the measure.

\paragraph{Measure 7:  Applicability of Rule EQ}

  Similarly to the analogue component for Rule \eqref{rule:rewrite}, we now define the next component that bounds the maximal number of possible applications of Rule \eqref{rule:equality}.
  The application conditions stipulate that it can be applied either
  \begin{enumerate}
    \item \label{it:measure-eq-1}
    on two deduction facts of \(\Solved(\C^e_i)\), or
    \item \label{it:measure-eq-2}
    on one deduction fact of \(\Solved(\C^e_i)\) in combination with a construction function symbol.
  \end{enumerate}
  Even if the application conditions also consider a mgs \(\Sigma\), the number of applications of Rule \eqref{rule:equality} will not depend on their number;
  this is intuitively because after applying the rule with one arbitrary mgs \(\Sigma\), the conditions forbid any later applications with identical parameters except \(\Sigma\).
  Formally consider for example the case \ref{it:measure-eq-1} (case \eqref{it:measure-eq-2} follows the same reasoning).
  The rule is applied on two deduction facts \((\xi_1 \dedfact u_1),(\xi_2 \dedfact u_2) \in \Solved(\C^e)\).
  Thus, an equality formula with \(\xi_1\Sigma \eqf \xi_2\Sigma\) as head will be added in \(\USolved(\CApply{\Sigma}{\C^e})\).
  However, in further applications of the rule, the condition that
  ``for all \((\clause{H}{\varphi}) \in \USolved(\CApply{\Sigma}{\C^e})\), \(H \neq (\xi_1 \eqf \xi_2)\)''
  will prevent a new application with the same (up to instantiation of \(\Sigma\)) deductions facts from \(\Solved(\CApply{\Sigma}{\C^e})\).


  We therefore conclude that the rule \eqref{rule:equality} can be applied only once per pair of deduction facts in \(\Solved\) and once per deduction fact in \(\Solved\) and function symbol in \(\sigc\).
  If \(\C^e\) is an extended constraint system we therefore consider \(\setEq(\C^e)\) the set of pairs \((\psi,\psi') \in \Solved(\C^e)^2\) or \((\psi,\ffun) \in \Solved \times \sigc\) that satisfy all the application conditions of the rule \eqref{rule:equality}, and
  \defcomp{7}
    {\sum_{(\P,\C,\C^e) \in \Gamma} |\setEq(\C^e)|}

\paragraph{Measure 8:  Number of unsolved equality formulas}

  We now introduce the analogue of Component 7 for equality formulas, that is,
  \defcomp{8}
    {|\{\psi \in \USolved(\C^e) \mid (\P,\C,\C^e) \in \Gamma, \psi \text{ unsolved equality formula}\}|}
  As before Rule \eqref{rule:satisfiable} makes this measure decrease in the case \ref{it:rule-sat-hyp} of its application conditions.
  On the contrary unsolved equality formulas can be generated by two rules:
  the case distinction rule \eqref{rule:equality} or the vector-simplification rule \eqref{rule:vector-consequence}.

\paragraph{Measure 9: Remaining most general solutions}

  So far, every time we showed that one of the previous element of the measure (strictly) decrease by application of a case distinction rule, we always focused on the positive branches of case-distinction rules.
  The negative branches on the contrary only add recipe disequations to the system, which does not increase any the previous components of the measure \emph{but} strictly decreases the number of most general solutions we can compute for the same instance of the rule.
  For example, if \(\Sigma \in \mgs(\C^e)\) then \(|\mgs(\C^e)| > |\mgs(\C^e[\Eqsnd \wedge \neg\Sigma])|\).
  Hence it suffices to consider the last component:
  \defcomp{9}
    {\left|\left\{\Sigma \mid
      \begin{array}{@{\ }l}
        \text{there exists a case distinction rule applicable} \\
        \text{from \(\C^e\) with parameter \(\Sigma\), \((\P,\C,\C^e) \in \Gamma\)}
      \end{array}\right\}\right|}

\paragraph{Conclusion}
  This gives the termination of the algorithm for computing \(T \in \ptree(P,Q)\).
  We study more precisely the Components 1 to 9 of the measure in Appendix \ref{app:termination} and prove that they can all be bound by an exponential in \(\dagsize{P,Q,\R}\) (with \(\R\) the rewriting system, implicitly including the signature).
  Hence:

  \begin{theorem}[termination for partition trees] \label{thm:termination-ptree}
    For all \(P,Q\) plain processes, Algorithm \ref{alg:ptree} terminates with arguments \(P,Q\).
    Moreover each branch of the resulting tree is generated by applying at most an exponential number (in \(\dagsize{P,Q,\R}\)) of rules, not counting the negative branches of case distinction rules.
  \end{theorem}

\subsection{Bounding the size of most general solutions}
\label{sec:exp-mgs}

\subsubsection{Overall approach}
\paragraph{Objective}
  We now focus on the theoretical complexity of the decision problems \TraceEquiv, \Bisimilarity, \Simulation,\ldots
  Our goal for now is to prove that they are all decidable in co\nexp and the core argument to achieve this is to prove the theorem:

  \begin{theorem}[size of most general solutions] \label{thm:mgs-size}
    If \(T\) is a partition tree of \(P,Q\) (w.r.t. a rewriting system \(\R\)) generated by Algorithm \ref{alg:ptree}, then for all nodes \(n\) of \(T\), \(\dagsize{\mgs(n)}\) is exponential in \(\dagsize{P,Q,\R}\).
  \end{theorem}

  We will detail in Section \ref{sec:witness-complexity} how to derive a co\nexp decision procedure for equivalence properties by using this result.
  To bound the size of most general solutions we rely on the results previously established in Section \ref{sec:termination-mgs}:
  in the final partition tree \(\mgs(n) = \mu^2(\Gamma(n))\) and it therefore suffices to prove that for all nodes, \(\dagsize{\mu^2(\Gamma(n))}\) is exponential in \(\dagsize{P,Q,\R}\).
  However we will instead study the easier-to-track bound:
  \[|\terms[2](\Gamma(n))| \geqslant |\subterms(\im(\mu^2(\Gamma(n))))| = \dagsize{\mu^2(\Gamma(n))}\]

\paragraph{Evolution of second-order terms}
  Let us now consider each constraint-solving rule and determine how \(\terms[2]\) evolves along the components along a branch of the partition tree.
  \begin{enumerate}
    \item \emph{Symbolic rules}:
    only Rules \eqref{rule:e-in} and \eqref{rule:e-out} increase the size of \(\terms[2]\) by adding at most two new second order variables.
    \item \emph{Simplification, normalisation, vector-simplification rules}:
    only Rule~\eqref{rule:vector-solve} may increase the size of \(\terms[2](\Gamma)\).
    Indeed, it transfers a deduction fact from \(\USolved\) in \(\Solved\) for each extended constraint systems in the current component \(\Gamma\).
    \item \emph{Case distinction rules}:
    the positive branches of these rules increase the size of \(\terms[2]\) whereas the negative branches leave it unchanged.
  \end{enumerate}
  It therefore suffices to prove the following result to obtain Theorem \ref{thm:mgs-size}:

  \begin{proposition}[evolution of second-order terms in partition trees] \label{prop:mgs-size-step}
    If \(\S \cup \{\Gamma\} \rightarrow \S \cup \S'\) where \(\S' = \{\Gamma'\}\) is obtained by Rule \eqref{rule:vector-solve} or \(\Gamma' \in \S'\) is the positive branch of a case distinction rule,
    then \(\dagsize{\terms[2](\Gamma')} - \dagsize{\terms[2](\Gamma)}\) is bounded by a polynomial in \(\dagsize{P,Q,\R}\).
  \end{proposition}

  Indeed we recall that by Theorem \ref{thm:termination-ptree}, we already know that each branch of the partition tree is obtained after applying at most an exponential number of rules (negative branches of case distinction rules excluded).
  Hence we obtain the expected exponential bound on \(\terms[2]\) when combined with the above proposition.
  The remaining of Section \ref{sec:exp-mgs} is dedicated to its proof.

\subsubsection{Bounding the increase of the second-order terms}
\paragraph{When applying a mgs}
  We first study the growth of \(\terms[2](\C^e)\) when applying a mgs to \(\C^e\), which means proving Theorem \ref{thm:mgs-size} in the case of Rule \eqref{rule:satisfiable}.
  Similarly to our previous results on most general solutions (Section \ref{sec:termination-mgs}), our bounds depend on \(\measureNC(\C^e)\) the number of first-order terms of \(\C^e\) that are not already used in the solution, i.e., in \(\mu^2\).
  We also recall that by Proposition \ref{prop:mgs-decrease}, this measure is non-increasing when applying any of the mgs simplification and constraint-solving rules, and is even strictly decreasing in the case of Rule \eqref{rule:res} and \eqref{rule:cons}.
  Let us now show that its growth is actually inverted compared to \(\terms[2]\), that is, how much \(\terms[2]\) increases can be bounded by how much \(\measureNC\) decreases:

  \begin{proposition}[evolution of second-order terms when applying mgs]
    \label{prop:evol-mgs}
    For all extended processes \(\C^e\) that verify the invariants \(\PredWellFormed(\C^e)\) and \(\PredCorrectFormula(\C^e)\), we have
    \[\forall \Sigma \in \mgs(\C^e),\
      |\terms[2](\CApply{\Sigma}{\C^e})| \leq |\terms[2](\C^e)| + |\sig| \times (|\measureNC(\C^e)| - |\measureNC(\CApply{\Sigma}{\C^e})|)\]
  \end{proposition}

  \begin{proof}
    We assume \(|\sig| > 0\) by convention.
    It suffices to prove this property when replacing \(\CApply{\Sigma}{\C^e}\) by \(\C^{e\prime}\) for \(\C^e \rightarrow \C^{e\prime} \neq \bot\) obtained by a mgs simplification or constraint-solving rules.
    We perform a case analysis on the rule in question.

    \caseitem{\emph{case 1: simplification rule on formulas}}

      The simplification rules on formulas only affect first-order terms and second-order disequations and we therefore have \(\terms[2](\C^{e\prime}) = \terms[2](\C^e)\).

    \caseitem{\emph{case 2: mgs simplification rule}}

      We only need to consider Rule \eqref{rule:unifEqfst_simpl}.
      Since it only affects first-order terms, the reasoning is identical to the previous case.

    \caseitem{\emph{case 3: mgs constraint-solving rule}}

      Rules \eqref{rule:conseq} and \eqref{rule:res} apply a second-order substitution \(\Sigma = \mgu(\xi \eqs \zeta)\) to \(\C^e\) for some \(\xi,\zeta \in \terms[2](\C^e)\).
      In particular we deduce that \(|\terms[2](\C^{e\prime})| \leqslant |\terms[2](\C^e)|\) and the conclusion thus follows from the fact that \(\measureNC(\C^e) - \measureNC(\C^{e\prime}) \geqslant 0\) by Proposition \ref{prop:mgs-decrease}.
      Finally the only rule that increases \(\terms[2](\C^e)\) is the last one, \eqref{rule:cons}, that generates \(n\) fresh second-order variables for some constructor symbol \(\ffun/n\).
      In particular \(|\terms[2](\C^{e\prime})| \leqslant |\terms[2](\C^e)| + n\), hence the result since \(\measureNC(\C^e) - \measureNC(\C^{e\prime}) > 0\) by Proposition \ref{prop:mgs-decrease}.
  \end{proof}

  In particular we obtain Proposition \ref{prop:mgs-size-step} for Rule \eqref{rule:satisfiable} case \ref{it:rule-sat-mgs}, provided we manage to prove that \(\measureNC(\C^e)\) is bounded by a polynomial.
  We explain in Appendix \ref{app:termination} how to extend the argument to Rules \eqref{rule:equality}, \eqref{rule:rewrite} and \eqref{rule:vector-solve}.

\paragraph{Bound of unused terms}
  To conclude let us establish the polynomial bound on \(|\measureNC(\C^e)|\).
  In order to do so we explore the relation between \(\C\) and \(\C^e\) in \((\P,\C,\C^e) \in \Gamma\).
  Intuitively \(\measureNC(\C^e)\) always has less elements then \(\measureNC(\C)\) because
  \begin{enumerate}
    \item the symbolic rules always add the same constraints to \(\C\) and \(\C^e\), ensuring that \(\measureNC(\C^e)\) increases at most as much as \(\measureNC(\C)\) by these rules
    \item the other rules leave \(\C\) untouched and do not make \(\measureNC(\C^e)\) increase.
  \end{enumerate}

  \begin{proposition}[approximation of unused terms] \label{prop:approx-unused}
    For all \((\P,\C,\C^e) \in \Gamma\), 
    \[|\measureNC(\C^e)| \leqslant \dagsize{\Phi(\C)\mu^1(\C), \mu^1(\C)}\]
  \end{proposition}

  \begin{proof}
    Considering \(\C\) instead of \(\C^e\), we have the trivial approximation
    \[|\measureNC(\C)| \leqslant |\terms[1](\C)| = \subterms(\Phi(\C)\mu^1(\C), \mu^1(\C)) = \dagsize{\Phi(\C)\mu^1(\C), \mu^1(\C)}\,.\]
    It therefore suffices to prove that \(|\measureNC(\C^e)| \leqslant |\measureNC(\C)|\).
    For that we show that the inequality \(|\measureNC(\C^e)| \leqslant |\measureNC(\C)|\) is preserved when applying any of the constraint-solving rules.

    \caseitem{\emph{case 1: symbolic rules}}

      These rules add the same constraints to \(\C^e\) and \(\C\) (up to an additional deduction fact added to \(\USolved(\C^e)\) in the case of Rule \eqref{rule:e-out}, but this does not affect \(\measureNC(\C^e)\)).
      In particular since \(\Eqsnd(\C) = \Solved(\C) = \emptyset\), if we consider an instance \((\P,\C,\C^e) \sstep{\alpha} (\P',\C',\C^{e\prime})\) of a symbolic rule we therefore have
      \[|\measureNC(\C^{e\prime})| - |\measureNC(\C^e)| \leqslant |\measureNC(\C')| - |\measureNC(\C)|\]
      which gives the expected result.

    \caseitem{\emph{case 2: simplification, normalisation, vector-simplification rules}}

      By definition these rules only affect \(\C^e\) and leave \(\C\) untouched, hence the conclusion since these rules do not increase \(\measureNC(\C^e)\).

    \caseitem{\emph{case 3: case distinction rules}}

      Let us consider \(\CompatibleSubs(\C^e)\) the set of substitutions \(\Sigma\) such that the notation \(\CApply{\Sigma}{\C^e}\) is well defined, that is, such that
      \begin{enumerate}
        \item if \(\dom(\Sigma) \subseteq \vars[2](\Df(\C^e))\)
        \item for all \(X \in \dom(\Sigma)\), there exists \(t\) such that \((X\Sigma,t) \in \conseq(\Solved(\C^e) \cup \Df')\) where
        \(\Df' = \{Y \dedfact u \in \Df(\C^e) \mid Y \notin \dom(\Sigma)\} \cup D_\freshlab\) with
        \[D_\freshlab = \{Y \dedfact y \mid Y \in \vars[2](\im(\Sigma_{|\vars[2](\C^e)})) \smallsetminus \vars[2](\C^e), y \text{ fresh}\}\]
      \end{enumerate}
      This is intuitively the set of substitutions \(\Sigma\) whose image is constructed from \(\Solved(\C^e)\), up to the new variables of \(D_\freshlab\) introduced by \(\Sigma\).
      In particular we have for all \(\Sigma \in \CompatibleSubs(\C^e)\), \(|\measureNC(\CApply{\Sigma}{\C^e})| \leqslant |\measureNC(\C^e)|\) (which follows in more details from Proposition \ref{prop:trans-conseq} in Appendix \ref{app:ptree}), hence the conclusion.
  \end{proof}

  This relation allows to eventually reduce the problem to give a polynomial bound on \(\Phi(\C)\) and \(\mu^1(\C)\) which are only affected by symbolic rules (we recall that the other constraint-solving rules do not modify \(\C\)).
  All in all this concludes the proof of the expected polynomial bound:

  \begin{corollary}[polynomial evolution of second-order terms]
    For all extended processes \(\C^e\) that verify the invariants \(\PredWellFormed(\C^e)\) and \(\PredCorrectFormula(\C^e)\), we have
    \[\forall \Sigma \in \mgs(\C^e),\
      |\terms[2](\CApply{\Sigma}{\C^e})| \leq |\terms[2](\C^e)| + 9\dagsize{P,Q,\R}^3\]
  \end{corollary}

  \begin{proof}
    We agree on the convention that \(\dagsize{P}\), \(\dagsize{Q}\) and \(\dagsize{\R}\) are strictly positive.
    By Propositions \ref{prop:evol-mgs} and \ref{prop:approx-unused}, it suffices to prove that for all symbolic traces \(P \Sstep{\tr} (\P,\C)\), we have that 
    \(\dagsize{\Phi(\C)\mu^1(\C), \mu^1(\C)} \leqslant 9\dagsize{P,\R}^2\) (which, as we will see, is a very rough approximation).
    For that a quick induction on the length of \(\tr\) allows to construct a set of \(|\tr|\) variables \(Y = \{y_i\}_{i=1}^{|\tr|}\) and finite set of equations \(S\) such that
    \begin{enumerate}
      \item \label{it:poly-step-1}
      \(\mu^1(\C) \in \mguR(S)\)
      \item \label{it:poly-step-2}
      for all \((u \eqs v) \in S\), \(u\) (resp. \(v\)) is either a subterm of a term appearing in \(P\) or a variable of \(Y\)
      \item \label{it:poly-step-3}
      for all terms \(u \in \im(\Phi(\C)\mu^1(\C))\), there exists \(u_0\) subterm of a term appearing in \(P\) such that \(u_0 \mu^1(\C) = u\)
    \end{enumerate}
    The variables of \(Y\) model the fresh channel variables introduced when executing \eqref{rule:s-in} or \eqref{rule:s-out} transitions, and the set of equations \(S\) collects the equality tests performed during the trace and how each variable of \(P\) is instantiated by \(\Eqfst\) (including by private communications).
    Independently from this, by induction on a straightforward algorithm to compute mgu modulo theory, we have if \(\R\) is constructor-destructor subterm convergent
    \begin{align}
      |\subterms(\sigma)|
        \nonumber
        & \leqslant |\subterms(S)| + \dagsize{\R} \times |\{ t \in \subterms(S) \mid \rootf(t) \in \sigd \}| \\
        \label{eqn:mgu-modulo-bound} \tag{\mbox{\(\mathcal{E}\)}}
        & \leqslant 2|\subterms(S)| \times \dagsize{\R}
    \end{align}
    Altogether we therefore obtain
    \begin{align*}
      \dagsize{\Phi(\C)\mu^1(\C), \mu^1(\C)}
        & \leqslant |\subterms(P)| + 2 \dagsize{\mu^1(\C)}
          & & \text{(by \ref{it:poly-step-3})} \\
        & \leqslant \dagsize{P} + 4|\subterms(S)| \times \dagsize{\R}
          & & \text{(by \ref{it:poly-step-1} and \eqref{eqn:mgu-modulo-bound})} \\
        & \leqslant \dagsize{P} + 4(\dagsize{P} + |\tr|)\dagsize{\R}
          & & \text{(by \ref{it:poly-step-2})} \\
        & \leqslant 9 \dagsize{P,\R}^2 & & \qedhere
    \end{align*}
  \end{proof}

\subsection{Complexity upper bounds for equivalence properties} \label{sec:witness-complexity}

\subsubsection{Complexity of trace equivalence}

The goal of this section is to prove the following theorem:

\begin{theorem}[complexity of trace equivalence] \label{thm:trace-equiv-complexity}
  \TraceEquiv are co\nexp for bounded processes and constructor-destructor subterm convergent theories.
\end{theorem}

The proof relies on the following arguments that were developed in previous sections:
\begin{enumerate}
  \item \emph{charactering trace inclusion with partition trees:}
  Theorem~\ref{thm:trace-equiv-ptree}
  \item \emph{existence of a mgs of exponential size:}
  Theorem \ref{thm:mgs-size}
  \item \emph{soundness and completeness of the symbolic semantics:}
  see Proposition~\ref{prop:symbolic-sound-complete}
\end{enumerate}
Using these ingredients we prove the core property:

\begin{proposition}[witness of non-trace equivalence of exponential size] \label{prop:trace-witness-exp}
  Let \(P_1,P_2\) be two plain processes w.r.t. a constructor-destructor subterm convergent rewriting system \(\R\).
  The following points are equivalent:
  \begin{enumerate}
    \item \label{it:trace-witness-1}
    \(P_1 \not\TraceIncl P_2\)
    \item \label{it:trace-witness-2}
    there exists a trace \(t : P_1 \Cstep{\tr} A_1\) such that \(\dagsize{t}\) is exponential in \(\dagsize{P,Q,\R}\) and for all \(P_2 \Cstep{\tr} A_2\), \(A_1 \not\StatEq A_2\).
  \end{enumerate}
\end{proposition}

\begin{proof}
  The proof of \ref{it:trace-witness-2}\(\Rightarrow\)\ref{it:trace-witness-1} is trivial and we therefore focus on \ref{it:trace-witness-1}\(\Rightarrow\)\ref{it:trace-witness-2}.
  Let us assume that \(P_1 \not\TraceIncl P_2\), and let \(T \in \ptree(P_1,P_2)\) the partition tree computed by Algorithm \ref{alg:ptree}.
  By Theorem \ref{thm:trace-equiv-ptree} we obtain a partition-tree trace \(P_1 \Tstep{\tr} (\P,\C),n\) such that there exist no traces of the form \(P_2 \Tstep{\tr} (\P',\C'),n\).
  But by Theorem \ref{thm:mgs-size} we know that the (DAG) size of \(\mgs(n)\) is of exponential in \(\dagsize{P,Q,\R}\), which gives a solution \((\Sigma,\sigma) \in \Sol[\pi(n)](\C)\) of exponential size as well by definition of a mgs.

  Let us then consider the trace \(t : P_1 \Cstep{\tr\Sigma} (\P\sigma, \Phi(\C)\sigma \norm)\) (that exists by soundness of the symbolic semantics) and show that it satisfies the conditions of \ref{it:trace-witness-2}.
  It is indeed of exponential DAG size.
  Besides assume by contradiction that there exists a trace \(P_2 \Cstep{\tr\Sigma} (\Q,\Psi)\) such that \(\Phi(\C)\sigma \StatEq \Psi\).
  By using the completeness of the symbolic semantics and the properties of the partition tree (Lemma~\ref{lem:PT-parent-concrete-derivation}), we would obtain a symbolic process \((\P',\C')\) such that \(P_2 \Tstep{\tr} (\P',\C'),n\), yielding a contradiction.
\end{proof}

To obtain a decidability result we also use the following result on static equivalence from \cite{AC06}:

\begin{proposition}[witness of non-static equivalence of polynomial size]
  If two frames \(\Phi\) and \(\Psi\) are not statically equivalent w.r.t. a subterm convergent rewriting system \(\R\), there exist two recipes \(\xi\) and \(\zeta\) such that \(\dagsize{\xi,\zeta}\) is polynomial in \(\dagsize{\Phi,\Psi,\R}\), \(\xi \Phi \norm = \zeta \Phi\norm\) and \(\xi \Psi\norm \neq \zeta \Psi\norm\).
\end{proposition}

Wrapping everything together we obtain the following \nexp decision procedure for non-trace equivalence:

\begin{enumerate}
  \item Given two processes \(P_1,P_2\), guess an integer \(i \in \eint{1}{2}\) and a trace \(P_1 \Cstep{\tr} (\P,\Phi)\) of exponential size.
  In particular, although \(|\dom(\Phi)| \leqslant |\tr|\), the sizes of the terms in \(\im(\Phi)\) may be exponential as well.
  \item For each of the exponentially-many traces of the form \(t : P_2 \Cstep{\tr} (\Q,\Psi)\), guess two recipes \(\xi_t,\zeta_t\) of exponential size.
  \item if for one such trace \(t\) we do not have \(\xi_t \Phi\norm = \zeta_t \Phi\norm \Leftrightarrow \xi_t \Psi\norm = \zeta_t \Psi\norm\), conclude that \(P_1 \not\TraceEq P_2\).
\end{enumerate}

\subsubsection{Complexity of labelled bisimilarity}

The goal of this section is to prove the following theorem:

\begin{theorem}[complexity of labelled bisimilarity]
  \Bisimilarity is co\nexp for bounded processes and constructor-destructor subterm convergent theories.
\end{theorem}

Similarly to trace equivalence we build on the results of the previous sections, this time using the characterisation of labelled bisimilarity based on symbolic witnesses (Theorem~\ref{thm:ptree-lab-bis}).
Given a partition tree with most general solutions of exponential size, our goal is therefore to derive from it a symbolic witness of non-equivalence and a solution of this witness (Definition~\ref{def:solution-witness}), both of exponential size as well.

\begin{proposition}[witness of non-labelled bisimilarity of exponential size]
  Let \(P_1,P_2\) be two plain processes w.r.t. a constructor-destructor subterm convergent rewriting system \(\R\).
  The following points are equivalent:
  \begin{enumerate}
    \item \label{it:bis-witness-1}
    \(P_1 \not\LabBis P_2\)
    \item \label{it:bis-witness-2}
    there exists a witness \(\witness\) for \((P_1,P_2)\) such that \(\dagsize{\witness}\) is exponential in \(\dagsize{P,Q,\R}\).
  \end{enumerate}
\end{proposition}

\begin{proof}
  The proof of \ref{it:trace-witness-2}\(\Rightarrow\)\ref{it:trace-witness-1} is trivial and we therefore focus on \ref{it:trace-witness-1}\(\Rightarrow\)\ref{it:trace-witness-2}.
  Let us assume that \(P_1 \not\LabBis P_2\), and let \(T \in \ptree(P_1,P_2)\) the partition tree computed by Algorithm \ref{alg:ptree}.
  By Theorem \ref{thm:ptree-lab-bis} we obtain a symbolic witness \(\witness_s\) for \((P_0,P_1,\rootf(T))\) such that \(\Sol(\witness_s) \neq \emptyset\), and it suffices to prove that there exists a solution of \(\witness_s\) of exponential size (where the size of a solution \(\fsol\) is \(\sum_{N \in \dom(\fsol)} \dagsize{\fsol(N)}\)).
  More precisely we construct by induction on \(\witness_s\) a function \(f\) mapping the nodes of \(\witness_s\) to second-order substitutions (not necessarily ground) such that:
  \begin{enumerate}
    \item \((f_{\mgs} f) \in \Sol(\witness_s)\), where \(f_{\mgs}(S,n) = \mgs(n)\) and the notation \(f = gh\) is defined by \(f(N) = g(N) h(N)\) for all nodes \(N\) of \(\witness_s\)
    \item for all \(\fsol \in \Sol(\witness_s)\), there exists \(f'\) such that \(\fsol = f_{\mgs} f f'\)
  \end{enumerate}
  In particular since \(f_{\mgs}\) is of exponential size by Theorem \ref{thm:mgs-size}, it suffices to ensure that \(f\) is of exponential size as well.

  \caseitem{\emph{case 1: \(\witness_s\) is reduced to a leaf \(N\)}.}

    Then it suffices to choose \(f(N) = \id\).

  \caseitem{\emph{case 2: \(\witness_s\) has a root labelled \((S,n)\) and children \(N_1, \ldots, N_p\) labelled \((S_1,n'), \ldots, (S_p,n')\)}}

    Let us write \(S = \{A_0,A_1\}\) with, by definition, a symbolic trace \(A_0 \sstep{\alpha} A_0'\) such that each trace \(A_1 \Sstep{\bar{\alpha}} A_1^i\) corresponds to a child \(S_i = \{A_0',A_1^i\}\).
    We apply the induction hypothesis to the children to obtain their respective functions \(f_1, \ldots, f_p\).
    We recall that \(\Sol(\witness_s) \neq \emptyset\) by hypothesis and that all solutions \(\fsol\) verify \(\fsol(N_1) = \cdots = \fsol(N_p)\);
    thus, since by induction hypothesis all solutions of \(\witness_s\)  are instances of \(f_{\mgs} f_i\), we obtain:
    \[\mgu(\mgs(n') f_1(N_1)\rho_1 \wedge \ldots \wedge \mgs(n') f_p(N_p)\rho_p) \neq \bot\]
    for \(\rho_1, \ldots, \rho_p\) fresh variables renamings of \(\im(f_1(N_1)), \ldots, \im(f_p(N_p))\), respectively.
    In particular, assuming without loss of generality that all the \(f_i(N_i)\) have the same domain \((\vars[2](n') \smallsetminus \dom(\mgs(n'))) \cup \im(\mgs(n'))\), we can write
    \[\Sigma = \mgu(f_1(N_1)\rho_1 \wedge \ldots \wedge f_p(N_p)\rho_p) \neq \bot\]
    Note that this \(\mgu\) is only polynomially bigger than each \(f_i(N_i)\).
    Since \(\mgs(n')\) is an instance of \(\mgs(n)\), we also let \(\Sigma_0\) such that \(\mgs(n') = \mgs(n) \Sigma_0\).
    We then conclude the proof by defining \(f\) as follows:
    \begin{enumerate}
      \item \(f(\rootf(\witness_s)) = (\Sigma_0 \Sigma)_{|\vars[2](n)}\)
      \item for all \(i \in \eint{1}{p}\), for all nodes \(N\) in the subtree of \(\witness_s\) rooted in \(N_i\), \(f(N) = f_i \Sigma\). \qedhere
    \end{enumerate}
\end{proof}

\subsection{Complexity lower bounds} \label{sec:deepsec-hardness}

We prove in this section the complexity lower bounds stated in Theorem~\ref{thm:deepsec-conexp}.

\subsubsection{Extensions of the calculus} \label{sec:tools encoding}
We first introduce useful syntax extensions that can be encoded in the original calculus.
We point out that that using these encodings does not affect the complexity of deciding the related decision problems, since they rely on polynomial-size encodings.

\paragraph{Internal non-deterministic choice}

  A first classical operator is the \emph{non-deterministic choice}:
  \(P + Q\) is a process that can be executed either as \(P\) or as \(Q\).
  Its operational semantics can therefore be described by adding the following rule to those of Figure \ref{fig:semantics}:
  \begin{align}
    \tag{\mbox{\textsc{Choice}}}
    \label{rule:choice}
    (\multi {P+Q} \cup \P, \Phi) & \cstep {\silent} (\multi {R} \cup \P, \Phi) & \text{\small if \(R \in \{P,Q\}\)}
  \end{align}
  This reduction can easily be encoded as an internal communication on a fresh private channel.
  We formalise it by a process transformation \(\sem{\cdot}\):
  \begin{align}
    \label{translation choice}
    \sem {P+Q} &\eqdef \OutP{s}{s} \mid \InP{s}{x}.\sem{P} \mid \InP{s}{y}. \sem{Q} & \mbox{where \(s\in\Nall\) and \(x,y\in\X[1]\) are fresh}
  \end{align}
  and all other cases of the syntax are handled as homomorphic extensions of \(\sem{\cdot}\). 
  As for the parallel operator we will sometimes use the big operator \(\sum\) assuming right-associativity. The correctness of this translation with respect to \(\TraceEq\) and \(\LabBis\) will be stated later on in this section.

  We also introduce the \(\guessBinary{x}\) construct which non-deterministically assigns either \(0\) or \(1\) to \(x\). \(\guessBinary{x}.P\) silently reduces to either \(P\{ x \mapsto \0\}\) or \(P\{ x \mapsto \1\}\) and \(\guessBinary{\vec{x}}.P\) is defined as \(\guessBinary{x_1}.\guessBinary{x_2}\ldots\guessBinary{x_n}.P\) where \(\vec x = x_1,\cdots,x_n\).
  Formally, we extend the operational semantics with the rule
  \begin{align*}
    %
    \tag{\sc Choose-0}\label{rule:choose-0}
    (\P\cup\multi{\guessBinary{x}.P}, \Phi) &\cstep{\epsilon} (\P\cup\multi{P\{ x \mapsto \0\}}, \Phi)\\
    %
    \tag{\sc Choose-1}\label{rule:choose-1}
    (\P\cup\multi{\guessBinary{x}.P}, \Phi) &\cstep{\epsilon} (\P\cup\multi{P\{ x \mapsto \1\}}, \Phi)\\
  \end{align*}
  and define
  \[
    \sem{\guessBinary{y}.P} \eqdef (\OutP{d}{0} + \OutP{d}{1}) \mid \InP{d}{y}.\sem{P} \quad\mbox{with \(d \in \Nall\) is fresh}
  \]

\paragraph{Boolean circuits and formulae}
\label{sec:simplify section circuit}

  Complete problems in complexity theory often involve boolean formulae (e.g., \sat or \qbf). The ability to evaluate boolean formulae, or boolean circuits in general, within the applied \(\pi\)-calculus is therefore crucial. We can implement such a feature by the means of private channels and internal communication: each edge of a boolean circuit \(\Gamma\) indeed mimics a channel transmitting a boolean over a network (Figure \ref{fig:circuit conversion}). 

  \begin{figure}[ht]
    \centering

    \begin{tikzpicture}
      [
        auto,
        or/.style={or gate US,draw}
      ]

      \newcommand\size{0.8}

      \node[or,fill=cyan!50] (OR) at (0,0) {\(\vee\)};
      \node (Arrow) at (2,0) {\scalebox{2}{\(\rightsquigarrow\)}};

      \node[draw,fill=cyan!50] (Inp) at (3.5,0) {};
      \node (Par1) at (5.5,0) {\scalebox{2}{\(|\)}};
      \node[draw,fill=cyan!50] (Comp) at (7.5,0) {};
      \node (Par2) at (9.5,0) {\scalebox{2}{\(|\)}};
      \node[draw,fill=cyan!50] (P) at (11.5,0) {\(P(x,y)\)};

      \draw[-] (OR) edge[sloped] node[above] {\scalebox{\size}{\(c_1\)}} (-1,0.3);
      \draw[-] (OR) edge[sloped] node[below] {\scalebox{\size}{\(c_2\)}} (-1,-0.3);
      \draw[-] (OR) edge[sloped] node[above] {\scalebox{\size}{\(c_3\)}} (1,0.3);
      \draw[-] (OR) edge[sloped] node[below] {\scalebox{\size}{\(c_4\)}} (1,-0.3);

      \draw[->] (Inp) edge[sloped] node[above] {\scalebox{\size}{\(\OutP{c_1}a\)}} (5,0.3);
      \draw[->] (Inp) edge[sloped] node[below] {\scalebox{\size}{\(\OutP{c_2}b\)}} (5,-0.3);

      \draw[<-] (Comp) edge[sloped] node[above] {\scalebox{\size}{\(\InP {c_1} x\)}} (6,0.3);
      \draw[<-] (Comp) edge[sloped] node[below] {\scalebox{\size}{\(\InP{c_2} y\)}} (6,-0.3);
      \draw[->] (Comp) edge[sloped] node[above] {\scalebox{\size}{\(\OutP{c_3}{x\vee y}\)}} (9,0.3);
      \draw[->] (Comp) edge[sloped] node[below] {\scalebox{\size}{\(\OutP{c_4}{x\vee y}\)}} (9,-0.3);

      \draw[<-] (P) edge[sloped] node[above] {\scalebox{\size}{\(\InP{c_3}x\)}} (10,0.3);
      \draw[<-] (P) edge[sloped] node[below] {\scalebox{\size}{\(\InP{c_4}y\)}} (10,-0.3);
    \end{tikzpicture}

    \caption{Simulation of an OR-gate within the applied \(\pi\)-calculus}
    \label{fig:circuit conversion}
  \end{figure}
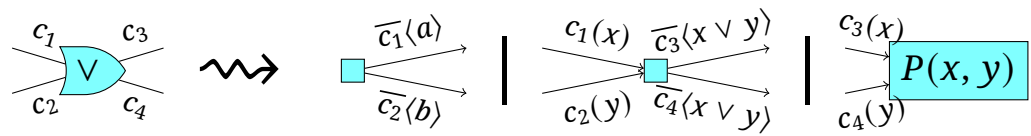

  The essence of circuits lies in so-called {\it logical gates} which are boolean functions with at most two inputs. We consider the fan-out of gates to be possibly more than one in order to model the fact that wires of a gate can be split and connected to the input of different gates.
  Formally, we assume without loss of generality that the gate has at most two (identical) outputs, to be given as input to other gates. Logical gates usually range over the constants \(\0\) and \(\1\) and the predicates \(\wedge\), \(\vee\) and \(\neg\) with the usual truth tables but we may use other common operators such as \(=\). From that a {\it boolean circuit} is an acyclic graph of logical gates: each input (resp. output) of a gate is either isolated or connected to a unique output (resp. input) of another gate, which defines the edges of this graph.

  Such a circuit \(\Gamma\) with \(m\) isolated inputs and \(n\) isolated outputs thus models a boolean function \(\Gamma:\B^m\rightarrow\B^n\) (where \(\B=\{\0,\1\}\)).  We write \((c_1,c_2,g,c_3,c_4)\in\Gamma\) to state that \(g:\B^2\rightarrow\B\) is a gate of \(\Gamma\) whose inputs are passed through edges \(c_1\) and \(c_2\) and whose output is sent to edges \(c_3\) and \(c_4\). This notation is naturally lifted to other in-outdegrees.

\paragraph{Embedding into the calculus}
  The syntax of plain processes is now extended with the construction \(x_1,\cdots,x_n\leftarrow\Gamma(b_1,\cdots,b_m).P\) where \(\Gamma:\B^m\rightarrow\B^n\) is a circuit, \(x_1,\ldots,x_n\) variables and \(b_1,\cdots,b_m\) terms. We fix two distinct terms \(\0,\1 \in \sig_0\) to model \(\B\) within the calculus, and the labelled operational semantics is extended with the rule:
  \begin{align*}
    \tag{\textsc{Valuate}}\label{rule:valuate}
    (\P\cup\multi{\vec x \leftarrow \Gamma(\vec b).P}, \Phi) &\cstep{\epsilon}
    (\P\cup\multi{P\{ \vec{x} \mapsto \Gamma(\vec{b}\norm)\}}, \Phi) & \mbox{if \(\msg(\vec b)\) and \(\vec{b}\norm \subseteq \B\)}
  \end{align*}

  Now we have to extend the definition of \(\sem\cdot\) (previous subsection) to handle the new operator. For simplicity we only consider the case where gates have two inputs and two outputs: handling lower arities is straightforward. If \((c_1,c_2,g,c_3,c_4)\in\Gamma\), we first define:

  \[\sem{c_1,c_2,g,c_3,c_4}\eqdef\InP{c_1}x.\InP{c_2}y.\prod_{b,b'\in\B}\IfP~x=b~\ThenP~\IfP~y=b'~\ThenP~(\OutP{c_3}{g(b,b')}\mid\OutP{c_4}{g(b,b')})\]

  where \(c_1,c_2,c_3,c_4\in\Nall\) (assuming that different circuits in a process do not share edges). To sum it up, we simply see circuit edges as private channels and simulate the logical flow of the gate. It is then easily extended:
  \begin{align*}
    \sem{\evalFormula{\vec x}{\Gamma(\vec b)}.P} & \eqdef
    \left(\prod_{k=1}^m \OutP {c_{i_k}} {b_k}\right) \mid
    \left(\prod_{(c_1,c_2,g,c_3,c_4)\in\Gamma} \sem {c_1,c_2,g,c_3,c_4}\right) \mid
    \InP{c_{o_1}}{x_1} \ldots \InP {c_{o_n}} {x_n}. \sem {P}
  \end{align*}

  where \((c_{i_k})_{k=1}^m\) (resp. \((c_{o_k})_{k=1}^n\)) are the isolated input (resp. output) edges of \(\Gamma\).
  Note that when \(b\) and \(b'\) are fixed booleans, \(g(b,b')\) denotes the boolean obtained from the truth table of \(g\): we emphasise that \(g\) is {\it not} a function symbol of the signature \(\sig\).

  \begin{remark}[simplifying assumption]
    \label{rem:at least one gate}
    We assume that every input of a circuit goes through at least one gate and every circuit has at least one output. This is to avoid irrelevant side cases in proofs.
  \end{remark}

\paragraph{Correctness of the translation}

  Now we dispose of an extended syntax and semantics as well as a mapping \(\sem\cdot\) removing the new constructors from a process.
  The correctness of this translation is proven in Appendix \ref{app:termination}:

  \begin{proposition}[correctness of the encodings]
    \label{prop:encodings correct}
    Let \(\TraceEq^+\) and \(\LabBis^+\) be the notions of trace equivalence and labelled bisimilarity over the extended calculus (the flag \(^+\) being omitted outside of this lemma). For all extended processes \(A = (\P,\Phi)\), the translation \(\sem A = (\sem \P, \Phi) = (\multi {\sem P ~|~ P \in \P}, \Phi)\) can be computed in polynomial time, \(A\TraceEq^+\sem A\) and \(A\LabBis^+\sem A\).
  \end{proposition}

  \begin{remark}[stability of common fragments]
    As the finite and pure fragments of the applied \(\pi\)-calculus are closed under \(\sem\cdot\), sums and circuits can be safely used within any intersection of such fragments.
    The encoding does not use else branches either.
  \end{remark}

\subsubsection{Lower bounds in the pure fragment} \label{sec:lower-pure}

  We now use the above tool to state our reductions, first, in the pure pi calculus.
  
  \paragraph{Trace equivalence}

  To show that trace equivalence is \polyh{2}-hard we proceed by a reduction from \qbf[2], that is, the problem of deciding, given \(\varphi\) a boolean formula whose variables are partitioned into \(\{\vec x\}\cup\{\vec y\}\), whether \(\forall \vec{x}.\exists \vec{y}. \varphi(\vec x,\vec y) = \1\).
  Our goal is to thus to construct two processes \(A\) and \(B\) such that:
  \begin{align}
    A \TraceEq B 
    & & \textit{iff}
    & & \forall \vec{x}.\exists \vec{y}. \varphi(\vec x,\vec y) = \0 \label{tr pure hard}
  \end{align}

  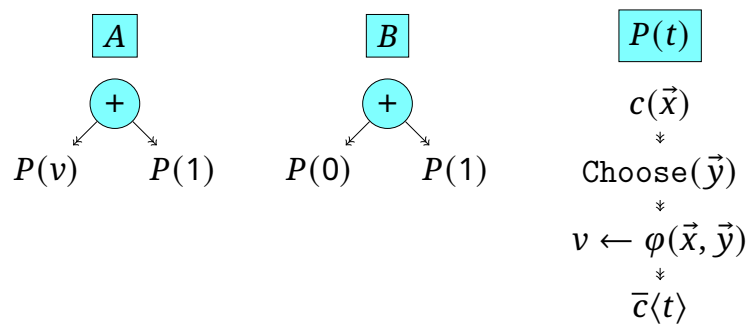
\begin{figure}[!ht]
    \centering
    \scalebox{0.9}
    {
      \begin{tikzpicture}
        [
          op/.style={draw,circle,fill=cyan!50,minimum size=1.6em},
          test/.style={draw,rectangle},
          proc/.style={},
          title/.style={draw,rectangle,fill=cyan!50}
        ]
        \newcommand\size{0.9}

        \node(TitleP) [title] at (8,0) {\(P(t)\)};
        \node(Input) [proc] at (8,-1) {\(\InP c {\vec x}\)};
        \node(Guess) [proc] at (8,-2) {\(\guessBinary{\vec y}\)};
        \node(Comp) [proc] at (8,-3) {\(\evalFormula v {\varphi(\vec x,\vec y)}\)};
        \node(Send) [proc] at (8,-4) {\(\OutP c t\)};

        \node(TitleA) [title] at (0,0) {\(A\)};
        \node(SumCircleA) [op] at (0,-1) {};
        \node(SumA) [proc] at (0,-1) {\(+\)};
        \node(LeftA) [proc] at (-1,-2) {\(P(v)\)};
        \node(RightA) [proc] at (1,-2) {\(P(\1)\)};

        \node(TitleB) [title] at (4,0) {\(B\)};
        \node(SumCircleB) [op] at (4,-1) {};
        \node(SumB) [proc] at (4,-1) {\(+\)};
        \node(LeftB) [proc] at (3,-2) {\(P(\0)\)};
        \node(RightB) [proc] at (5,-2) {\(P(\1)\)};

        \draw[->>] (Input) -- node[auto] {} (Guess);
        \draw[->>] (Guess) -- node[auto] {} (Comp);
        \draw[->>] (Comp) -- node[auto] {} (Send);

        \draw[->>] (SumCircleA) -- node[auto] {} (LeftA);
        \draw[->>] (SumCircleA) -- node[auto] {} (RightA);
        \draw[->>] (SumCircleB) -- node[auto] {} (LeftB);
        \draw[->>] (SumCircleB) -- node[auto] {} (RightB);
      \end{tikzpicture}
    }
    \caption{Schematic definition of \(A\) and \(B\)}
    \label{fig:eqtr pure hardness}
  \end{figure}

  Consider three distinct names \(c,\0,\1\in\sig_0\) (the last two modelling booleans for the syntax extension of circuit evaluation, see Section~\ref{sec:tools encoding}). 
  Processes \(A\) and \(B\) are depicted in Figure~\ref{fig:eqtr pure hardness}. 
  Intuitively, the process \(P(t)\) (where \(t\) is a term which may depend on the variables bound by \(P\)) gets a valuation of \(\vec x\) from the attacker, internally chooses a valuation of \(\vec y\), computes the value of \(\varphi(\vec x,\vec y)\) using rule \ref{rule:valuate}, and outputs \(t\). 
  From that it is quite easy to see that \(A\) and \(B\) have the same set of traces \textit{iff} for all valuation of \(\vec x\), there exists a valuation of \(\vec y\) such that \(\varphi(\vec x,\vec y)=\0\).
  This reduction is formalised and proved in Appendix~\ref{app:lower-pure}.

  \begin{theorem}
    In the pure pi-calculus, \TraceEquiv and \TraceInclus are \polyh{2}-hard for bounded positive processes.
  \end{theorem}

  Note that the hardness for \TraceInclus is directly implied by the hardness of \TraceEquiv.
  This is evidenced by the reduction that, for all processes \(P,Q\), \(P \TraceIncl Q\) \textit{iff} \(P + Q \TraceEq Q\).


  \paragraph{Simulations}

  We now prove that labelled bisimilarity is \pspace-hard for the positive pure pi calculus by reduction from \qbf.
  This is more involved as \qbf allows arbitrary quantifier alternation.  Let \(\varphi\) be a boolean formula whose variables are partitioned into \(\{x_1,\ldots, x_n\}\cup\{y_1,\ldots, y_n\}\) for some \(n\in\mathbb N\). We construct (in polynomial time in the size of \(\varphi\) and \(n\)) two processes \(A\) and \(B\) such that:
  \begin{align}
    \label{eqn:reduction pspace hard}
    A \LabBis B
    & & \textit{iff}
    & & A \Simi B
    & & \textit{iff}
    & & \forall x_1 \exists y_1 \ldots \forall x_n \exists y_n.~\varphi(x_1,\ldots, x_n,y_1,\ldots, y_n) = \0
  \end{align}

  Both \qbf and labelled bisimilarity may be seen as bisimulation games: an attacker plays the \(\exists\)-quantifiers (selects a transition in a process) whereas a defender responds with the \(\forall\)-quantifiers (tries to find a similarly-labelled sequence of transitions in the other process). The role of \(A\) and \(B\)  is to implement this intuitive connection: the attacker moves will be simulated by public inputs \(\InP c {x_i}\) and the defender responses by instructions \(\guessBinary{z_i}.\InP c {y_i}\). The structure of \(A\) and \(B\) is then designed to constrain the moves of the two players so that the winning condition of the attacker is exactly \(\exists x_1\forall y_1\ldots\exists x_n \forall y_n.~\varphi(\vec x,\vec y)=1\).

  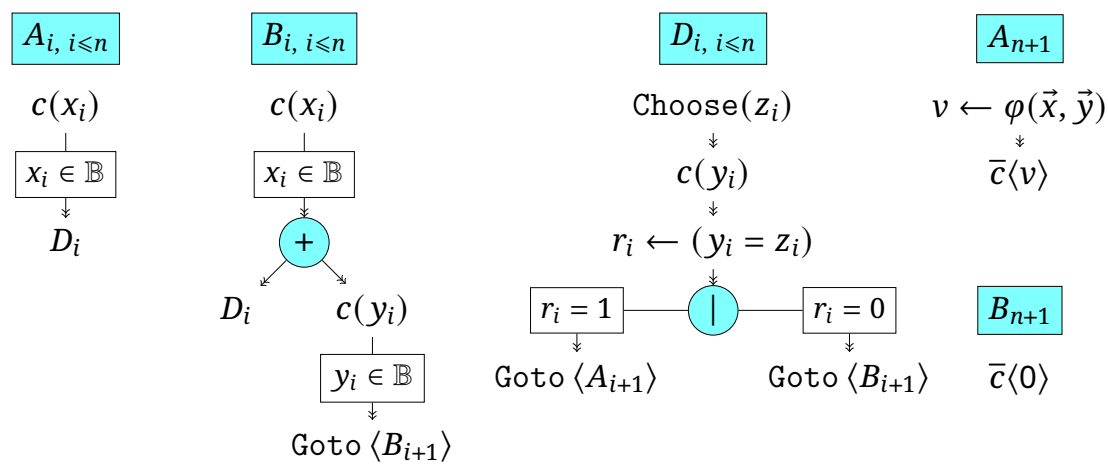
\begin{figure}[!ht]
    \centering
    \scalebox{0.9}
    {
      \begin{tikzpicture}
        [
          op/.style={draw,circle,fill=cyan!50,minimum size=1.6em},
          test/.style={draw,rectangle},
          proc/.style={},
          title/.style={draw,rectangle,fill=cyan!50}
        ]
        \newcommand\size{0.9}

        \node(TitleAi) [title] at (0.5,0) {\(A_{i,~i\leqslant n}\)};
        \node(InputA) [proc] at (0.5,-1) {\(\InP c {x_i}\)};
        \node(TestBoolA) [test] at (0.5,-2) {\scalebox{\size}{\(x_i\in\B\)}};
        \node(ProcDiA) [proc] at (0.5,-3) {\(D_i\)};

        \node(TitleBi) [title] at (4,0) {\(B_{i,~i\leqslant n}\)};
        \node(InputB) [proc] at (4,-1) {\(\InP c {x_i}\)};
        \node(TestBoolB) [test] at (4,-2) {\scalebox{\size}{\(x_i\in\B\)}};
        \node(SumBCircle) [op] at (4,-3) {};
        \node(SumB) [proc] at (4,-3) {\(+\)};
        \node(ProcDiB) [proc] at (3,-4) {\(D_i\)};
        \node(InputB2) [proc] at (5,-4) {\(\InP c {y_i}\)};
        \node(TestBoolB2) [test] at (5,-5) {\scalebox{\size}{\(y_i\in\B\)}};
        \node(ProcBi) [proc] at (5,-6) {\(\call{B_{i+1}}\)};

        \node(TitleDi) [title] at (10,0) {\(D_{i,~i\leqslant n}\)};
        \node(GuessD) [proc] at (10,-1) {\(\guessBinary{z_i}\)};
        \node(InputD) [proc] at (10,-2) {\(\InP c {y_i}\)};
        \node(CompDiff) [proc] at (10,-3) {\(\evalFormula{r_i}{(y_i=z_i)}\)};
        \node(ParaCircle) [op] at (10,-4) {};
        \node(Para) [proc] at (10,-4) {\(|\)};
        \node(TestD1) [test] at (8,-4) {\scalebox{\size}{\(r_i=1\)}};
        \node(TestD0) [test] at (12,-4) {\scalebox{\size}{\(r_i=0\)}};
        \node(CallA) [proc] at (8,-5) {\(\call{A_{i+1}}\)};
        \node(CallB) [proc] at (12,-5) {\(\call{B_{i+1}}\)};

        \node(TitleAend) [title] at (14.5,0) {\(A_{n+1}\)};
        \node(CompAend) [proc] at (14.5,-1) {\(\evalFormula v {\varphi(\vec x,\vec y)}\)};
        \node(SendAend) [proc] at (14.5,-2) {\(\OutP c v\)};

        \node(TitleBend) [title] at (14.5,-4) {\(B_{n+1}\)};
        \node(SendBend) [proc] at (14.5,-5) {\(\OutP c \0\)};

        \draw[-] (InputA) -- node[auto] {} (TestBoolA);
        \draw[->>] (TestBoolA) -- node[auto] {} (ProcDiA);

        \draw[-] (InputB) -- node[auto] {} (TestBoolB);
        \draw[->>] (TestBoolB) -- node[auto] {} (SumBCircle);
        \draw[->>] (SumBCircle) -- node[auto] {} (ProcDiB);
        \draw[->>] (SumBCircle) -- node[auto] {} (InputB2);
        \draw[-] (InputB2) -- node[auto] {} (TestBoolB2);
        \draw[->>] (TestBoolB2) -- node[auto] {} (ProcBi);

        \draw[->>] (GuessD) -- node[auto] {} (InputD);
        \draw[->>] (InputD) -- node[auto] {} (CompDiff);
        \draw[->>] (CompDiff) -- node[auto] {} (ParaCircle);
        \draw[-] (ParaCircle) -- node[auto] {} (TestD1);
        \draw[-] (ParaCircle) -- node[auto] {} (TestD0);
        \draw[->>] (TestD1) -- node[auto] {} (CallA);
        \draw[->>] (TestD0) -- node[auto] {} (CallB);

        \draw[->>] (CompAend) -- node[auto] {} (SendAend);
      \end{tikzpicture}
    }
    \caption{Schematic definition of \(A_i\) and \(B_i\)}
    \label{fig:eqobs pure hardness}
  \end{figure}

  \(A\) and \(B\) are defined inductively by processes \(A_i\), \(B_i\) and \(D_i\), depicted in Figure~\ref{fig:eqobs pure hardness},  structured in a way that, in a bisimulation game:
  \begin{enumerate}
  \item the attacker chooses the instance of \(x_i\);
  \item the defender chooses the instance of \(z_i\) and can force the attacker to instantiate \(y_i\) with the same value (the attacker not doing so allows for a trivial victory of the defender).
  \end{enumerate}
  The intermediary processes \(\call{A_i}\) and \(\call{B_i}\) intuitively formalise value passing from one index to another, in order to avoid an exponential blowup when encoding \(n\) nested tests.
  Their precise definition and the correctness of the reduction is formalised in Appendix~\ref{app:lower-pure}.
  As before, the hardness of the pre-order follows from the hardness of its symmetric closure.

  \begin{theorem}
    In the pure pi-calculus, \Simulation, \Similarity and \Bisimilarity are \pspace-hard for bounded positive processes.
  \end{theorem}

\subsubsection{Reduction of SuccinctSAT to process equivalence}

We now show that, when cryptographic primitives are modelled by a destructor subterm convergent rewrite system, \TraceEquiv, \TraceInclus, \Simulation, \Similarity and \Bisimilarity are co\nexp hard by reducing SuccinctSAT to process equivalence. Consider an instance of \sucsat, \(\Gamma\), with \(m+2\) inputs and \(n+1\) outputs and we design \(\sig\), \(\R\) subterm destructor and \(A\) and \(B\) positive processes such that, for any equivalence relation \({\approx} \in \{\Simi,\TraceEq,\LabBis\}\), \(A\not\approx B\) \textit{iff} \(\sem\Gamma_\varphi\) is satisfiable.

\paragraph{Term algebra}
  Terms are built over the following signature:
  \begin{align*}
    \sig \eqdef~~& \0,~\1, & & \mbox{(booleans \(\B\))}\\
               & \Node/2,~\pi/2, & & \mbox{(binary trees)}\\
               & \hfun/2, & & \mbox{(one-way binary hash)}\\
               & \hNode/2,~\hBool/2,~\invN/1,~\invB/1 & & \mbox{(testable binary hashes)}
  \end{align*}

  We equip this term algebra with the rewriting system \(E\) containing the following rules modelling subtree extraction (for binary trees) and argument testing (for hashes):
  \begin{align*}
    \pi(\Node(x,y),\0)\rightarrow~&x
    & \pi(\Node(x,y),\1)\rightarrow~&y\\
    \invN(\hNode(\Node(x,y),z))\rightarrow~&\1
    & \invB(\hBool(\0,z))\rightarrow~&\1 & \invB(\hBool(\1,z))\rightarrow~&\1
  \end{align*}

  In particular \(\R\) is subterm and destructor, the destructor symbols being \(\pi\), \(\invN\) and \(\invB\).
  We will also use a shortcut for recursive subtree extraction:
  if \(\ell\) is a finite sequence of first-order terms, the notation \(\recpos t \ell\) is inductively defined by:
  \begin{align*}
    \recpos t \epsilon &\eqdef t & \recpos t {b \cdot \ell} &\eqdef\recpos{\pi(t,b)}\ell
  \end{align*}

\paragraph{Core of the reduction}
  Let us give the intuition behind the construction before diving into the formalism.
  Recall that we are studying a formula in CNF \(\sem\Gamma_\varphi\) with \(2^n\) variables and \(2^m\) clauses. In particular, given a valuation of its \( 2^n\) variables, we can verify in non-deterministic polynomial time in \(n,m\) that it falsifies \(\sem\Gamma_\varphi\):
  \begin{enumerate}
    \item guess an integer \(i\in\eint 0 {2^m-1}\) as a sequence of \(m\) bits;
    \item obtain the three literals of the \(i\)\textsuperscript{th} clause of \(\sem\Gamma_\varphi\) (requiring three runs of the circuit \(\Gamma\)) and verify that the valuation falsifies the disjunction of the three literals.
  \end{enumerate}

  \noindent This non-deterministic verification is the essence our reduction. In the actual processes:
  \begin{enumerate}
    \item a process \(\CheckTree x\) checks that \(x\) is a correct encoding of a valuation, that is, that \(x\) is a complete binary tree of height \(n\) whose leaves are booleans;
    \item a process \(\CheckSat x\) implements the points {\it 1.} and {\it 2.} above.
  \end{enumerate}

  All of this is then formulated as equivalence properties within \(A\) and \(B\) (see the intermediary lemmas in the next paragraph for details). Intuitively, we want to express the following statement by equivalence properties: ``{\it for all term \(x\), either \(x\) is not an encoding of a valuation or falsifies a clause of \(\sem\Gamma_\varphi\)}''.
  A schematised definition is proposed in Figure \ref{fig:nexp reduction}.

  \begin{figure}[!ht]
    \centering
    \scalebox{0.80}
    {
      \begin{tikzpicture}
        [
          op/.style={draw,circle,fill=cyan!50,minimum size=1.6em},
          test/.style={draw,rectangle},
          proc/.style={},
          title/.style={draw,rectangle,fill=cyan!50}
        ]
        \newcommand\size{0.9}

        \node(TitleA) [title] at (-1.5,0) {\(A\)};
        \node(InputA) [proc] at (-1.5,-1) {\(\InP c x\)};
        \node(SumCircleA) [op] at (-1.5,-2) {};
        \node(SumA) [proc] at (-1.5,-2) {\(+\)};
        \node(CallA) [proc] at (-3,-3) {\(\CheckSat x\)};
        \node(VerifA) [proc] at (0,-3) {\(\CheckTree x\)};

        \node(TitleB) [title] at (-1.5,-4.1) {\(B\)};
        \node(InputB) [proc] at (-1.5,-5.1) {\(\InP c x\)};
        \node(SumCircleB) [op] at (-1.5,-6.1) {};
        \node(SumB) [proc] at (-1.5,-6.1) {\(+\)};
        \node(CallB) [proc] at (-3,-7.1) {\(\CheckSat x\)};
        \node(VerifB) [proc] at (0,-7.1) {\(\CheckTree x\)};
        \node(OutputB1) [proc] at (-1.5,-8.1) {\(\OutP c {\hfun(\0,s)}\)};
        \node(OutputB2) [proc] at (-1.5,-9.1) {\(\OutP c {\hfun(\1,s)}\)};

        \node(TitleP) [title] at (4,0) {\(\CheckSat x\)};
        \node(ChooseP) [proc] at (4,-1) {\(\guessBinary{p_1,\ldots,p_m}\)};
        \node(Lit1) [proc] at (4,-2) {\(\evalFormula{b_1,\ell_1}{\Gamma(\vec p,\0,\1)}\)};
        \node(Lit2) [proc] at (4,-3) {\(\evalFormula{b_2,\ell_2}{\Gamma(\vec p,\1,\0)}\)};
        \node(Lit3) [proc] at (4,-4) {\(\evalFormula{b_3,\ell_3}{\Gamma(\vec p,\1,\1)}\)};
        \node(Eval) [proc] at (4,-5.8) {\(\evalFormula v {\left(\begin{array}{l}b_1=\recpos x {\ell_1}\\~\vee~ b_2=\recpos x{\ell_2}\\~\vee~b_3=\recpos x {\ell_3}\end{array}\right)}\)};
        \node(OutputP1) [proc] at (4,-7.6) {\(\OutP c {\hfun(v,s)}\)};
        \node(OutputP2) [proc] at (4,-8.6) {\(\OutP c {\hfun(\1,s)}\)};

        \node(TitleQ) [title] at (10,0) {\(\CheckTree x\)};
        \node(SumCircleQ) [op] at (10,-1) {};
        \node(SumQ) [proc] at (10,-1) {\(+\)};
        \node(VerifNode) [proc] at (8.2,-2) {\(\sum_{i=0}^{n-1}\)};
        \node(ChooseVN1) [proc] at (8.2,-3) {\(\guessBinary{p_1,\ldots,p_i}\)};
        \node(OutputVN1) [proc] at (8.2,-4) {\(\OutP c {\hNode(\recpos x {\vec p},s)}\)};
        \node(OutputVN2) [proc] at (8.2,-5) {\(\OutP c {\hfun(\1,s)}\)};
        \node(Invi)[fill,circle,scale=0.3] at (11.8,-2) {};
        \node(VerifBool) [proc] at (11.8,-3) {\(\guessBinary{p_1,\ldots,p_n}\)};
        \node(OutputVB1) [proc] at (11.8,-4) {\(\OutP c {\hBool(\recpos x {\vec p},s)}\)};
        \node(OutputVB2) [proc] at (11.8,-5) {\(\OutP c {\hfun(\1,s)}\)};

        \draw[->>] (InputA) -- node[auto] {} (SumCircleA);
        \draw[->>] (SumCircleA) -- node[auto] {} (CallA);
        \draw[->>] (SumCircleA) -- node[auto] {} (VerifA);

        \draw[->>] (InputB) -- node[auto] {} (SumCircleB);
        \draw[->>] (SumCircleB) -- node[auto] {} (CallB);
        \draw[->>] (SumCircleB) -- node[auto] {} (VerifB);
        \draw[->>] (SumCircleB) -- node[auto] {} (OutputB1);
        \draw[->>] (OutputB1) -- node[auto] {} (OutputB2);

        \draw[->>] (ChooseP) -- node[auto] {} (Lit1);
        \draw[->>] (Lit1) -- node[auto] {} (Lit2);
        \draw[->>] (Lit2) -- node[auto] {} (Lit3);
        \draw[->>] (Lit3) -- node[auto] {} (Eval);
        \draw[->>] (Eval) -- node[auto] {} (OutputP1);
        \draw[->>] (OutputP1) -- node[auto] {} (OutputP2);

        \draw[->>] (SumCircleQ) -- node[auto] {} (VerifNode);
        \draw[-] (SumCircleQ) -- node[auto] {} (Invi);
        \draw[->>] (Invi) -- node[auto] {} (VerifBool);
        \draw[-,dotted] (VerifNode) -- node[auto] {} (ChooseVN1);
        \draw[->>] (ChooseVN1) -- node[auto] {} (OutputVN1);
        \draw[->>] (OutputVN1) -- node[auto] {} (OutputVN2);
        \draw[->>] (VerifBool) -- node[auto] {} (OutputVB1);
        \draw[->>] (OutputVB1) -- node[auto] {} (OutputVB2);
      \end{tikzpicture}
    }
    \caption{Informal definition of \(A\) and \(B\)}
    \label{fig:nexp reduction}
  \end{figure}
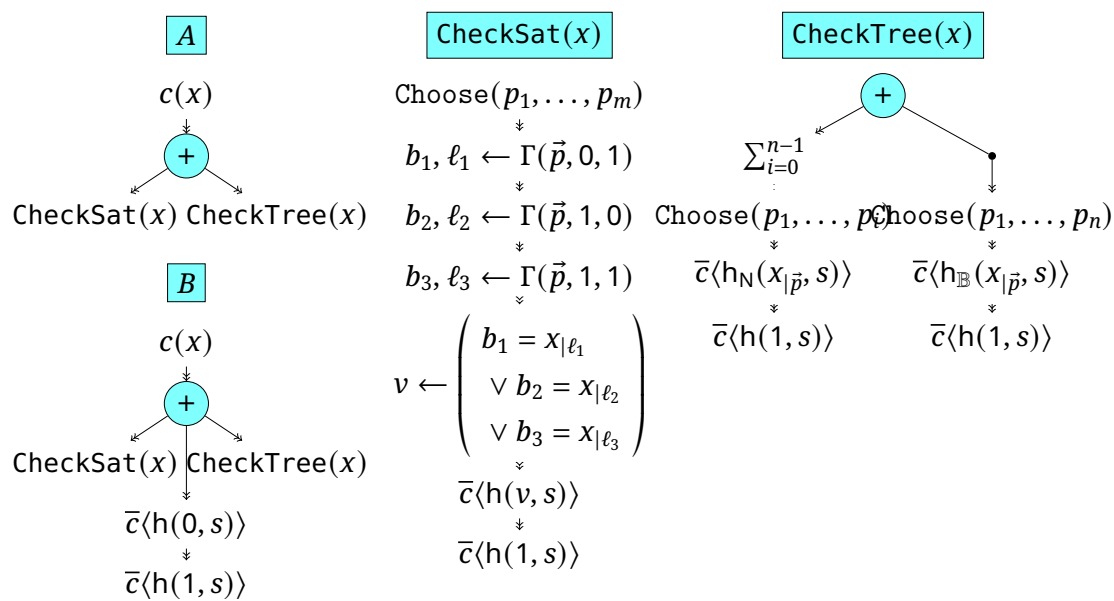

\paragraph{Formal construction}
  Let us now define the processes depicted in Figure \ref{fig:nexp reduction} properly; note that all the proofs about the correctness of this construction are relegated to Appendix \ref{app:termination} but we still state several intermediary lemmas in order to highlight the proof structure.
  But first of all, let us give a name to a frame which is at the core of our reduction:
  \[\Phi_\0=\{\ax_1\mapsto\hfun(\0,s),~\ax_2\mapsto\hfun(\1,s)\}\]

  \(\Phi_\0\) is reached after executing the central branch of \(B\) and everything is about knowing under which conditions a frame statically equivalent to \(\Phi_\0\) can be reached in \(A\). Let us define the processes themselves now. We fix \(s\in\Nall\) and define, if \(x\) is a protocol term:
  \begin{align*}
    \CheckTree x \eqdef~& \displaystyle\sum_{i=0}^{n-1}\left(~\guessBinary{p_1,\ldots,p_i}.~\OutP c {\hNode(\recpos x{p_1\cdots p_i},s)}.~\OutP c {\hfun(\1,s)}~\right)\\
    & + \guessBinary{p_1,\ldots,p_n}.~\OutP c {\hBool(\recpos x{p_1\cdots p_n},s)}.~\OutP c {\hfun(\1,s)}
  \end{align*}

  \begin{proposition}[restate=propCorrectnessChecktree,name={correctness of the tree checker}]
    \label{prop:correctness checktree}
    Let \(x\) be a message which is not a complete binary tree of height \(n\) with boolean leaves.
    Then there exists a reduction \(\CheckTree x\Cstep {\epsilon} (\multi P,\emptyset)\) such that \(P \LabBis \OutP c {\hfun(\0,s)}.\, \OutP c {\hfun(\1,s)}\).
  \end{proposition}

  \noindent
  Now let us move on to \(\CheckSat x\). This process binds a lot of variables:
  \begin{enumerate}
    \item \(\vec p=p_1,\ldots,p_m\) models the non-deterministic choice of a clause number in \(\eint 0 {2^m-1}\);
    \item \(b_i,\ell_i\), \(i\in\eint 1 3\), where \(\ell_i\) is a sequence of \(n\) variables, model the literals of the clause chosen above (\(b_i\) is the negation bit and \(\ell_i\) the identifier of the variable);
    \item \(v\) stores whether the chosen clause is satisfied by the valuation modelled by \(x\).
  \end{enumerate}
  \[\begin{array}{rl}
    \CheckSat x\eqdef~&\guessBinary{\vec p}.\\
    &\evalFormula{b_1,\ell_1}{\Gamma(\vec p,\0,\1)}.\\
    &\evalFormula{b_2,\ell_2}{\Gamma(\vec p,\1,\0)}.\\
    &\evalFormula{b_3,\ell_3}{\Gamma(\vec p,\1,\1)}.\\
    &\evalFormula v {(b_1=\recpos x{\ell_1} \vee~ b_2=\recpos x{\ell_2} \vee~ b_3=\recpos x{\ell_3})}.\\
    &\OutP c{\hfun(v,s)}.\OutP c{\hfun(\1,s)}
  \end{array}\]

  \begin{proposition}[restate=propCorrectnessChacksat,name={correctness of the sat checker}]
    \label{prop:correctness checksat}
    Let \(x\) be a complete binary tree of height \(n\) whose leaves are booleans, and \(\val_x\) be the valuation mapping the variable number \(i\) of \(\sem\Gamma_\varphi\) to \(\recpos x {p_1 \cdots p_n} \in \B\) where \(p_1 \cdots p_n\) is the binary representation of \(i\) (i.e., \(i = \sum_{k=1}^np_k2^{k-1}\)).
    If \(\val_x\) does not satisfy \(\sem\Gamma_\varphi\) then there exists \(\CheckSat x \Cstep \epsilon P\) such that \(P \LabBis \OutP c {\hfun(\0,s)}.\, \OutP c {\hfun(\1,s)}\).
  \end{proposition}

  We can finally wrap up everything by defining \(A\) and \(B\) and stating the last part of the correctness theorem.
  We recall that all the proofs can be found in Appendix \ref{app:termination}.
  \begin{align*}
    A\eqdef~&\InP c x.(\CheckSat x ~+~ \CheckTree x)\\
    B\eqdef~&\InP c x.(\CheckSat x ~+~ \CheckTree x ~+~ \OutP c {\hfun(\0,s)}.\OutP c {\hfun(\1,s)})
  \end{align*}

  \begin{proposition}[restate=propCorrectnessCheckall,name={correctness of the reduction}]
    \label{prop:correctness checkall}
    For any equivalence relation \({\approx} \in \{\Simi,\TraceEq,\LabBis\}\), \(\sem\Gamma_\varphi\) is satisfiable \textit{iff} \(A\not\approx B\).
  \end{proposition}

  As a conclusion we obtain the co\nexp hardness of equivalence properties (and their respective pre-orders as a consequence) for constructor-destructor subterm convergent theories.
  This is stated by the theorem below, which additionally puts an emphasis on the fact that the rewriting system used in our reduction is constant, that is, it does not depend on \(\Gamma\).

  \begin{theorem}[restate=thmDeepsecHardness,name={hardness of equivalences}] \label{thm:conexp-finite}
    There exists a fixed constructor-destructor subterm convergent rewriting system \(\R\) such that the decision problems \(\R\)--\TraceEquiv, \(\R\)--\TraceInclus, \(\R\)--\Simulation, \(\R\)--\Similarity and \(\R\)--\Bisimilarity are co\nexp hard for bounded positive processes.
  \end{theorem}

\section{Conclusion and future work}\label{sec:conclusion}


In this article we have studied automated verification of equivalence properties, encompassing both theoretical and practical aspects. 
We provide tight complexity results for static equivalence, trace equivalence and labelled (bi)similarity (as well as their respective pre-orders), summarised in~Table~\ref{fig:summary}.
In particular we show that deciding trace equivalence and labelled (bi)similarity for a bounded number of sessions is co\nexp complete for subterm convergent destructor rewrite systems.
Finally, we implement the procedure for deciding trace equivalence in the \deepsec prototype.
As demonstrated through an extensive benchmark (Table \ref{fig:bench}), our tool is broad in scope and efficient compared to other tools.

Our work opens several directions for future work. 
It would be interesting to lift the restriction of subterm convergent equational theories to allow for more cryptographic primitives.
Similarly, we plan to avoid the restriction to destructor rewrite systems to more general ones.
Also, in recent work~\cite{CCK-csf22} it was shown that labelled similarity is the same relation as \emph{may testing equivalence} in presence of a probabilistic adversary which motivates the extension of our implementation beyond trace equivalence.
The presented procedure for (bi)similarity is however highly non-deterministic and a naive implementation would certainly be inefficient.

Finally, we also plan to extend the \deepsec tool with support for other types of properties.
The extension provided in this article to simulation and other security relations shows the modularity of our core proof technique, the partition tree, for analysing security properties.
For example, a simplified version of the tree could be used to verify more classical (and simpler) \emph{trace properties}, which would significantly rise the scope and usability of \deepsec.
More generally, since navigation within the tree already allows to verify the complex notion of bisimilarity, we expect that the technique should scale to \emph{hyperproperties} in general.
There are few formalisms and results for such properties in the context of security protocols, but hyperlogics fitting our symbolic model have recently been introduced~\cite{BDM22}. 
They allow for example to model fine variants of equivalence relations to capture subtle hypotheses, and their combination to liveness or real-time properties.
We expect that our proof techniques would allow to study the decidability and complexity of a large fragment of such logics.


\printbibliography

\newpage 
\appendix


\section{Decision procedures using partition trees} \label{app:decision-proc-from-ptree}

  We detail in this appendix the technical proofs of correctness of the decision procedures for equivalences of Section~\ref{sec:ptree-eq}, assuming a partition tree \(T\) priorly constructed.

  \subsection{Trace equivalence} \label{app:decision-proc-trace-from-ptree}

    We prove in this section the following theorem:

    \thmTraceEquivPtree*

    The proof of this theorem relies on two technical lemmas extending the properties of the partition tree edges to its branches, i.e., from \(\tstep{}\) to \(\Tstep{}\).
    For example we can generalise as follows the fact that the nodes of the tree are labelled by maximal configurations, i.e., Definition~\ref{def:partition-tree}, Item~\ref{it:PT-parent-concrete-derivation}:

    \begin{lemma}[restate=PTParentConcreteDerivation,name={}] \label{lem:PT-parent-concrete-derivation}
      Assume that \((\P_1,\C_1),n \Tstep{\tr} (\P_1',\C_1'),n'\) and \((\P_2,\C_2) \Sstep {\tr} (\P_2',\C_2')\) with \((\P_2,\C_2) \in \Gamma(n)\).
      We also consider, for all \(i \in \{1,2\}\), a solution \((\Sigma',\sigma_i') \in \Sol[\pi(n')](\C_i')\) such that
      \(\Phi(\C_1') \sigma_1' \StatEq \Phi(\C_2') \sigma_2'\).
      Then we have \((\P_2,\C_2),n \Tstep {\tr} (\P_2',\C_2'),n'\).
    \end{lemma}

    \begin{proof}
      We proceed by induction on the length \(\tr\).
      The case \(\tr = \epsilon\) follows from the saturation of nodes under \(\tau\)-transition (Definition \ref{def:partition-tree}, Item \ref{it:PT-silent}).
      Otherwise we let, with \(\tr = \alpha \cdot \tilde {\tr}\),
      \begin{align*}
        (\P_1,\C_1),n & \Tstep{\alpha} (\tilde{\P_1},\tilde{\C_1}), \tilde{n} \Tstep{\tilde {\tr}} (\P_1',\C_1'),n' &
        (\P_2,\C_2) & \Sstep{\alpha} (\tilde{\P_2},\tilde{\C_2}) \Sstep{\tilde {\tr}} (\P_2',\C_2')
      \end{align*}
      We also consider the restrictions \(\Sigma = \Sigma'_{|\vars[2](n)}\) and \(\tilde{\Sigma} = \Sigma'_{|\vars[2](\tilde{n})}\).
      In particular \(\Sigma \subseteq \tilde{\Sigma}\) and there exist \(\sigma_2,\tilde{\sigma}_1,\tilde{\sigma}_2\) such that
      \begin{align*}
        (\Sigma,\sigma_2) & \in \Sol(\C_2) &
        (\tilde{\Sigma},\tilde{\sigma}_1) & \in \Sol(\tilde{\C}_1) &
        (\tilde{\Sigma},\tilde{\sigma}_2) & \in \Sol(\tilde{\C}_2)
      \end{align*}
      The hypothesis that \(\Phi(\C_1') \sigma_1' \StatEq \Phi(\C_2') \sigma_2'\) also implies that \(\Phi(\tilde{\C}_1) \tilde{\sigma}_1 \StatEq \Phi(\tilde{\C}_2) \tilde{\sigma}_2\).
      Besides since predicates are refined along branches (Definition \ref{def:partition-tree}, Item \ref{it:PT-monotonic}) and are defined on the variables of their configurations (Definition \ref{def:configuration}, Item \ref{it:configuration-pred-dom}),
      we know that \(\Sigma\) and \(\tilde{\Sigma}\) verify \(\pi(n)\) and \(\pi(\tilde{n})\), respectively.

      All in all we can use the maximality of the node \(\tilde{n}\) (Definition \ref{def:partition-tree}, Item \ref{it:PT-parent-concrete-derivation} applied to the edge \(n \xrightarrow{\alpha} \tilde{n}\)),
      which gives that \((\tilde{\P}_2,\tilde{\C}_2) \in \Gamma(\tilde{n})\).
      Hence \((\P_2,\C_2),n \Tstep {\alpha} (\tilde{\P}_2,\tilde{\C}_2),\tilde{n}\) by definition and
      the conclusion then follows from the induction hypothesis applied to the remaining of the traces.
    \end{proof}

    Combined with the soundness and the completeness of the symbolic semantics, this permits to prove one direction of Theorem \ref{thm:trace-equiv-ptree}:

    \begin{proof}[Proof of Theorem \ref{thm:trace-equiv-ptree},
      \ref{it:trace-equiv-ptree-incl}\(\Rightarrow\)\ref{it:trace-equiv-ptree-trace}.]
      Let us consider a trace \(P_1 \Tstep {\tr} (\P_1,\C_1), n\) and exhibit a trace \(P_2 \Tstep {\tr} (\P_2,\C_2), n\).
      We decompose the proof into the following steps:
      \begin{enumerate}
        \item By \emph{soundness} of the symbolic semantics we obtain a trace \(P_1 \Cstep{\tr \Sigma} (\P_1 \sigma_1, \Phi(\C_1) \sigma_1 \norm)\) for an arbitrary solution \((\Sigma,\sigma_1) \in \Sol(\C_1)\).
        \item By \emph{hypothesis \ref{it:trace-equiv-ptree-incl}} there exists a concrete trace \(P_2 \Cstep{\tr \Sigma} (\P,\Phi)\) such that \(\Phi \StatEq \Phi(\C_1) \sigma \norm\).
        \item By \emph{completeness} of the symbolic semantics we obtain a symbolic trace \(P_2 \Sstep {\tr'} (\P_2,\C_2)\) and \((\Sigma',\sigma_2) \in \Sol(\C_2)\) such that \(\tr \Sigma = \tr' \Sigma'\), \(\P_2 \sigma_2 = \P\) and \(\Phi(\C_2) \sigma_2 \norm = \Phi\).
        Due to the form of symbolic actions, we know that there exists a second-order-variable renaming \(\rho\) such that \(\tr = \tr' \rho\);
        in particular \(P_2 \Sstep {\tr} (\P_2,\C_2 \rho)\) and \((\Sigma,\sigma_2) \in \Sol(\C_2 \rho)\).
        \item By \emph{Lemma \ref{lem:PT-parent-concrete-derivation}} we therefore obtain that \(P_2 \Tstep {\tr} (\P_2,\C_2 \rho), n\), which gives the expected conclusion. \qedhere
      \end{enumerate}
    \end{proof}

    The second property of the partition tree we extend is the fact that symbolic transitions are reflected in the tree, i.e., Definition~\ref{def:partition-tree}, Item~\ref{it:PT-child-concrete-derivation}:

    \begin{lemma}[restate=PTChildConcreteDerivation,name={}] \label{lem:PT-child-concrete-derivation}
      Let \(n\) be a node of a partition tree \(T\) and \((\P,\C) \in \Gamma(n)\).
      If \((\P,\C) \Sstep {\tr} (\P',\C')\) and \((\Sigma,\sigma) \in \Sol[\pi(n)](\C')\) then there exist a node \(n'\) and a substitution \(\Sigma'\) such that \((\P,\C),n \Tstep {\tr} (\P',\C'),n'\)
      and \((\Sigma',\sigma) \in \Sol[\pi(n')](\C')\).
    \end{lemma}

    Note that unlike the definition of partition tree, we do not require that \(\Sigma'\) coincides with \(\Sigma\) on \(\vars[2](n)\).
    This additional requirement would not make the lemma false but is unnecessary to prove Theorem~\ref{thm:trace-equiv-ptree}.
    The lemma is proved by induction on \(\tr\) below:
    
    \begin{proof}
      We proceed by induction on the length of \(\tr\).
      If \(\tr = \epsilon\) it suffices to choose \(n = n'\) and the conclusion immediately follows.
      Otherwise let us decompose the symbolic trace into
      \begin{align*}
        (\P,\C) & \Sstep{\tilde{\tr}} (\tilde{\P},\tilde{\C}) \Sstep {\alpha} (\P',\C') &
        \tr & = \tilde{\tr} \cdot \alpha
      \end{align*}
      Note that \((\Sigma_{|\vars[2](\tilde{n})}, \sigma_{|\vars[1](\tilde{n})}) \in \Sol(\tilde{\C})\), and \(\Sigma_{|\vars[2](\tilde{n})}\) verifies \(\pi(n)\) by definition of a configuration
      (since \(\Sigma\) verifies it and has the same restriction to \(\vars[2](\Gamma(n))\) as \(\Sigma_{|\vars[2](\tilde{n})}\)).
      By induction hypothesis we therefore obtain \(\tilde{n},\tilde{\Sigma}\) such that \((\P,\C),n \Tstep {\tilde{\tr}} (\tilde{\P},\tilde{\C}),\tilde{n}\)
      and \((\tilde{\Sigma},\sigma_{|\vars[1](\tilde{n})}) \in \Sol[\pi(\tilde{n})](\tilde{\C})\).
      Let us then consider the extension
      \[\tilde{\Sigma}^e = \tilde{\Sigma} \cup \Sigma_{|\vars[2](n') \smallsetminus \vars[2](\tilde{n})}\]
      To conclude the proof it suffices to apply the Item \ref{it:PT-child-concrete-derivation} of Definition \ref{def:partition-tree} to the symbolic transition \((\tilde{\P},\tilde{\C}) \Sstep {\alpha} (\P',\C')\)
      and the solution \((\tilde{\Sigma}^e,\sigma)\);
      what remains to prove is therefore that we effectively have  \((\tilde{\Sigma}^e,\sigma) \in \Sol[\pi(\tilde{n})](\C')\).
      First of all we indeed have by construction \(\dom(\tilde{\Sigma}^e) = \vars[2](\C')\) and \(\dom(\sigma) = \vars[1](\C')\).
      We also know that \(\tilde{\Sigma}^e\) satisfies the predicate \(\pi(\tilde{n})\) because \(\tilde{\Sigma} = \tilde{\Sigma}^e_{|\vars[2](\tilde{n})}\) satisfies it.
      The first-order solution \(\sigma\) satisfies the constraints of \(\Eqfst(\C')\) since \((\Sigma,\sigma) \in \Sol(\C')\) by hypothesis.
      Finally we let \(\varphi \in \Df(\C')\) and prove that \((\Phi(\C'),\tilde{\Sigma}^e,\sigma) \models \varphi\):

      \caseitem{\emph{case 1:} \(\varphi = (X \dedfact x) \in \Df(\tilde{\C})\)}

        The conclusion follows from the fact that \((\tilde{\Sigma},\sigma_{|\vars[1](\tilde{n})}) \in \Sol(\tilde{\C})\).

      \caseitem{\emph{case 2:} \(\varphi = (X \dedfact x) \in \Df(\C') \smallsetminus \Df(\tilde{\C})\)}

        The conclusion follows from the fact that \((\Sigma,\sigma) \in \Sol(\C')\).

      \caseitem{\emph{case 3:} \(\varphi = \forall X.\, X \ndedfact x\)}

        We have to prove that \(x \sigma\) is not deducible from the frame \(\Phi(\C') \sigma\), which is a consequence from the fact that \((\Sigma,\sigma) \in \Sol(\C')\).
    \end{proof}

    Using again the soundness and completeness of the symbolic semantics, we can finally derive the other direction of Theorem~\ref{thm:trace-equiv-ptree}.

    \begin{proof}[Proof of Theorem~\ref{thm:trace-equiv-ptree},
      \ref{it:trace-equiv-ptree-trace}\(\Rightarrow\)\ref{it:trace-equiv-ptree-incl}.]
      Let us consider a trace \(P_1 \Cstep {\tr} (\P,\Phi)\) and exhibit a trace \(P_2 \Cstep {\tr} (\Q,\Psi)\) such that \(\Phi \StatEq \Psi\).
      We decompose the proof into the following steps:
      \begin{enumerate}
        \item By \emph{completeness} of the symbolic semantics we obtain a symbolic trace \(P_1 \Sstep {\tr_s} (\P_1,\C_1)\) and \((\Sigma,\sigma_1) \in \Sol(\C)\) such that \(\tr_s \Sigma = \tr\), \(\P_1 \sigma_1 = \P\) and \(\Phi(\C_1) \sigma_1 \norm = \Phi\).
        \item By \emph{Lemma \ref{lem:PT-child-concrete-derivation}} we then obtain a partition-tree trace \(P_1 \Tstep {\tr_s} (\P_1,\C_1),n\) and \(\Sigma'\) such that \((\Sigma',\sigma_1) \in \Sol[\pi(n)](\C_1)\).
        \item By \emph{hypothesis \ref{it:trace-equiv-ptree-trace}} there also exists a partition-tree trace \(P_2 \Tstep {\tr_s} (\P_2,\C_2),n\).
        By definition of a configuration we also know that there exists \(\sigma_2\) such that \((\Sigma',\sigma_2) \in \Sol[\pi(n)](\C_2)\) and \(\Phi(\C_1) \sigma_1 \StatEq \Phi(\C_2) \sigma_2\).
        \item By \emph{soundness} of the symbolic semantics applied to we then obtain a concrete trace \(P_2 \Cstep {\tr_s \Sigma'} (\Q,\Psi)\) with \(\Q = \P_2 \sigma_2\) and \(\Psi = \Phi(\C_2) \sigma_2 \norm \ \StatEq\ \Phi(\C_1) \sigma_1 \norm\ = \Phi\).
      \end{enumerate}
      However we may have \(\tr_s\Sigma' \neq \tr\) and, to conclude the proof, we prove that \(P_2 \Cstep {\tr} (\Q,\Psi)\) as well.
      For that it suffices to prove that \(\tr \Psi \norm = \tr_s \Sigma' \Psi\norm \), that is, although the recipes or \(\tr\) and \(\tr_s\Sigma'\) are different they produce the same first-order terms.
      Since \(\Phi\) and \(\Psi\) are statically equivalent, \(\tr = \tr_s \Sigma\) and \(\Phi = \Phi(\C_1) \sigma_1 \norm\), it suffices to prove that \(\tr_s \Sigma \Phi(\C_1) \sigma \norm = \tr_s \Sigma' \Phi(\C_1) \sigma\norm \).
      Let \(X \in \vars[2](\tr_s)\).
      A quick look at the rules of the symbolic semantics shows that there exists a deduction fact \((X \dedfact x) \in \Df(\C_1)\).
      In particular, since \((\Sigma,\sigma_1)\) and \((\Sigma',\sigma_1)\) are both solutions of \(\C_1\) we have
      \(X \Sigma \Phi(\C_1) \sigma_1 \norm = x \sigma_1 \norm = X \Sigma' \Phi(\C_1) \sigma_1 \norm\),
      hence the conclusion.
    \end{proof}

  \subsection{Simulations} \label{app:decision-proc-bisim-from-ptree}

    In this section, we now prove the main theorem at the basis of the decision procedure for simulations and its variants:

    \thmLabBisPtree*
  
    We only prove the case of labelled bisimilarity, as the proof for simulation is analogue.
    For that, we prove the following technical lemma by induction on the structure of the (symbolic) witness; 
    this lemma is a stronger version of the theorem for the purpose of managing the induction invariant.

    \begin{lemma} \label{lem:lab-bis-ptree}
      Let \(n\) be a node of a partition tree \(T\) and \(A_0,A_1 \in \Gamma(n)\).
      We let \(A_i = (\P_i,\C_i)\) and \(\Sigma,\sigma_0,\sigma_1\) such that \((\Sigma,\sigma_i) \in \Sol[\pi(n)](\C_i)\).
      If \(A_i^c = (\P_i \sigma_i,\Phi(\C_i) \sigma_i \norm)\), the following points are equivalent:
      \begin{enumerate}
        \item \label{it:lab-bis-ptree-equiv}
        \(A_0^c \not \LabBis A_1^c\)
        \item \label{it:lab-bis-ptree-witness}
        there exist a symbolic witness \(\witness_s\) for \((A_0,A_1,n)\) and a solution \(\fsol \in \Sol(\witness_s)\) such that \(\fsol(\rootf(\witness_s)) = \Sigma\)
      \end{enumerate}
    \end{lemma}

    To prove this lemma we first observe that, by definition of a configuration (Definition~\ref{def:configuration}), \(A_0^c \StatEq A_1^c\) because these two processes are obtained by instanciating two symbolic processes from a same node \(n\) with a common solution \(\Sigma\).
    We then prove the two directions separately.

    \medskip

    \begin{bigproof}[Proof of Lemma \ref{lem:lab-bis-ptree}, \ref{it:lab-bis-ptree-equiv}\(\Rightarrow\)\ref{it:lab-bis-ptree-witness}]
      We prove the result by induction on \(|\P_0,\P_1|\).
      The conclusion is immediate if \(|\P_0,\P_1| = 0\) as it yields a contradiction:
      the multisets \(\P_0\) and \(\P_1\) can only contain null processes and the fact that \(A_0^c \StatEq A_1^c\) justifies that \(A_0^c \LabBis A_1^c\).
      Otherwise we let by Proposition \ref{prop:concrete-witness} a witness \(\witness\) of \((A_0^c,A_1^c)\).
      Thus, by definition, there exist \(b \in \{0,1\}\) and a transition \(A_b^c \cstep{\alpha} A_b^{\prime c} = (\Q, \Phi)\)
      such that for all traces \(A_{1-b}^c \Cstep{\bar{\alpha}} A_{1-b}^{\prime c}\) such that \(A_0^{\prime c} \StatEq A_0^{\prime c}\),
      we have \((A_0^{\prime c}, A_1^{\prime c}) \in \witness\)
      (and therefore \(A_0^{\prime c} \not \LabBis A_1^{\prime c}\) by Proposition \ref{prop:concrete-witness}).
      Let us now construct a symbolic witness \(\witness_s\) of \((A_0,A_1,n)\) and a suitable solution \(\fsol\).

      \caseitem{\emph{case 1:} \(\alpha \neq \tau\)}

        By completeness of the symbolic semantics (Proposition \ref{prop:symbolic-sound-complete}) applied to the transition \(A_b^c \cstep{\alpha} A_b^{\prime c}\), we let a symbolic transition \(A_b \sstep {\alpha_s} A_b' = (\Q_s,\C)\) and a solution \((\Sigma', \sigma') \in \Sol (\C)\)
        such that \(\Sigma \subseteq \Sigma'\),
        \(\alpha = \alpha_s \Sigma'\), \(\Q = \Q_s \sigma'\) and \(\Phi = \Phi(\C) \sigma' \norm\).
        Note that by hypothesis \(\Sigma\) verifies \(\pi(n)\) and, therefore, so does its extension \(\Sigma'\) (since by definition predicates are stable by domain extension, recall Definition \ref{def:configuration}).
        Then since the symbolic transition \(A_b \sstep {\alpha_s} A_b'\) is reflected in \(T\) (in the sense of Definition \ref{def:partition-tree}, Item \ref{it:PT-child-concrete-derivation}),
        we obtain a transition \(A_b, n \tstep {\alpha_s} A_b', n'\) and \(\Sigma''\) such that \((\Sigma'',\sigma') \in \Sol [\pi(n')] (\C)\) and \(\Sigma_{|\vars[2](n)}'' = \Sigma_{|\vars[2](n)}' \quad (= \Sigma)\).

        \caseitem{\emph{case 1a}: there exist no \(A_{1-b}'\) such that \(A_{1-b},n \Tstep {\alpha_s} A_{1-b}', n'\)}

          Then we define \(\witness_s\) to be the tree whose root is labelled \((A_0,A_1,n)\) and that has a unique child labelled \((A_b',n')\).
          We then consider \(\fsol\) mapping the child to \(\Sigma''\) and the root to \(\Sigma''_{|\vars[2](n)} = \Sigma\), which is a solution of \(\witness_s\).

        \caseitem{\emph{case 1b}: otherwise}

          In this case we define \(\witness_s\) as follows.
          Its root is labelled \((A_0,A_1,n)\) and its children are all the nodes labelled \((A_0',A_1',n')\), with \(A_{1-b},n \Tstep {\alpha_s} A_{1-b}', n'\).
          For each such node, as explained in the beginning of the proof we have \(A_0^{\prime c} \not\LabBis A_1^{\prime c}\) which permits to apply the induction hypothesis with the solution \(\Sigma''\).
          This gives a symbolic witness rooted in this node and \(\fsol\) a solution mapping this node to \(\Sigma''\).
          Let us write more explicitly these witnesses \(\witness_s^1, \ldots, \witness_s^p\) and \(\fsol^1, \ldots, \fsol^p\) the corresponding solutions.
          To conclude it then suffices to choose \(\witness_s^1, \ldots, \witness_s^p\) as the children of the root of \(\witness_s\), and \(\fsol\) maps the root of \(\witness_s\) to \(\Sigma''_{|\vars[2](n)} = \Sigma\) and each node \(n\) of \(\witness_s^i\) to \(\fsol^i(n)\).

        \caseitem{\emph{case 2:} \(\alpha = \tau\)}

          Analogue to case 1 in the simpler case where \(n = n'\) and \(\Sigma = \Sigma' = \Sigma''\).
          Note also that the analogue of case 1a cannot arise.
    \end{bigproof}

    \begin{bigproof}[Proof of Lemma \ref{lem:lab-bis-ptree}, \ref{it:lab-bis-ptree-witness}\(\Rightarrow\)\ref{it:lab-bis-ptree-equiv}]
        We construct a concrete witness \(\witness\) of \((A_0^c, A_1^c)\) as follows:
        \[\witness = \left\{\begin{array}{r|l}
          \multirow{2}*{%
            \(\left((\P_0 \sigma_0,\Phi(\C_0) \sigma_0 \norm),
            (\P_1 \sigma_1,\Phi(\C_1) \sigma_1 \norm)\right)\)}
            & N \text{ node of \(\witness_s\) labelled } ((\P_0,\C_0),(\P_1,\C_1),n), \\
            & \forall i \in \{0,1\}, (\fsol(N), \sigma_i) \in \Sol(\C_i)
        \end{array}\right\}\]
        The fact that all \((B_0,B_1) \in \witness\) verify \(B_0 \StatEq B_1\) follows from Definition \ref{def:configuration}.
        Then let us consider \(((\P_0 \sigma_0,\Phi(\C_0) \sigma_0 \norm),
        (\P_1 \sigma_1,\Phi(\C_1) \sigma_1 \norm)) \in \witness\) using the notations of the construction of \(\witness\) above.
        By definition of a symbolic witness there exists \(b \in \{0,1\}\) and a transition \((\P_b,\C_b), n \tstep{\alpha} (\P_b',\C_b'), n'\) such that:

        \caseitem{\emph{case 1:} \(N = \rootf(\witness_s)\) has a unique child \(N'\) labelled \(\{(\P_b',\C_b')\},n'\)}

          Then consider the concrete transition \((\P_b \sigma_b,\Phi(\C_b) \sigma_b \norm) \cstep{\alpha \fsol(N')} (\P_b' \sigma_b',\Phi(\C_b') \sigma_b' \norm)\) obtained by soundness of the symbolic semantics (Proposition \ref{prop:symbolic-sound-complete})
          where \((\fsol(N'),\sigma_b') \in \Sol[\pi(n')](\C_b')\).
          By completeness of the symbolic semantics (which is possible to apply since \(\fsol(N) \subseteq \fsol(N')\) by definition of a solution of a symbolic witness) and maximality of the node \(n'\) (Definition \ref{def:partition-tree}, Item \ref{it:PT-parent-concrete-derivation}), there cannot exist any concrete trace of the form
          \[(\P_{1-b} \sigma_{1-b},\Phi(\C_{1-b}) \sigma_{1-b} \norm) \cstep{\alpha \fsol(N')} (\P,\Phi)\]
          such that \(\Phi \StatEq \Phi(\C_b') \sigma_b'\), hence the conclusion.

        \caseitem{\emph{case 2:} the children of \(N = \rootf(\witness_s)\) are all the nodes \(N'\) labelled \(((\P_0',\C_0'), (\P_1',\C_1'),n')\), where \((\P_{1-b},\C_{1-b}), n \Tstep{\alpha} (\P_{1-b}',\C_{1-b}'), n'\) (and there is at least one such child)}

          Let \(N'\) be an arbitrary child of \(N\), labelled \(((\P_0',\C_0'), (\P_1',\C_1'),n')\) with the above notations.
          As in the previous case we consider the concrete transition obtained by soundness of the symbolic semantics, \((\P_b \sigma_b,\Phi(\C_b) \sigma_b \norm) \cstep{\alpha \fsol(N')} (\P_b' \sigma_b',\Phi(\C_b') \sigma_b' \norm)\).
          Then let us consider a trace of the form
          \begin{align*}
            (\P_{1-b} \sigma_{1-b},\Phi(\C_{1-b}) \sigma_{1-b} \norm) \Cstep{\alpha \fsol(N')} (\P,\Phi) = A & &
            \Phi \StatEq \Phi(\C_b') \sigma_b'
          \end{align*}
          Our goal is to prove that \((A, (\P_b \sigma_b,\Phi(\C_b) \sigma_b \norm)) \in \witness\).
          Using the completeness of the symbolic semantics and the maximality of \(n'\) as in the previous case, we obtain a partition-tree trace \((\P_{1-b}, \C_{1-b}),n \Tstep{\alpha} (\P_{1-b}'',\C_{1-b}''), n'\) and \((\fsol(N'),\sigma_{1-b}'') \in \Sol[\pi(n')](\C_{1-b}'')\)
          with \(\P = \P_{1-b}'' \sigma_{1-b}''\) and \(\Phi = \Phi(\C_{1-b}'')\sigma_{1-b}'' \norm\).
          By hypothesis there therefore exists a node \(N''\), a child of \(N\), labelled \(((\P_b',\C_b'),(\P_{1-b}'',\C_{1-b}''),n')\).
          The conclusion then follows from the fact that \(\fsol(N') = \fsol(N'')\) by definition of a solution of a symbolic witness.
    \end{bigproof}


\section{Correctness of the generation of partition trees} \label{app:ptree}

\subsection{Invariants of the procedure} \label{app:invariants}

In this section we present some additional properties that are verified all along the procedure by the nodes of the partition tree under construction.
Such nodes are sets of extended symbolic processes (i.e. tuples \((\P,\C,\C^e)\) with \(\P\) a process, \(\C\) a constraint system and \(\C^e\) an extended constraint system).
Understanding the technical details of these invariants is not necessary to understand the algorithm itself, however most of our subprocedure (e.g. the generation of most general solutions) are only correct in their context.

\paragraph{Invariant 1: Well-formedness}
  The first invariant is about the shape of the extended constraint systems.
  Two important properties are that
  all equations of \(\Eqfst\) and \(\Eqsnd\) are trivially satisfiable (they are essentially of the form \(x \eqs u\) where \(x\) appears nowhere else in the constraint system) and that those of \(\Eqsnd\) only use terms that can be constructed from the knowledge base (i.e. they are consequences of \(\Solved \cup \Df\)).

  \begin{definition} \label{def:well-formed}
    We define the predicate \(\PredWellFormed\) on extended constraint systems as follows;
    we have that \(\PredWellFormed((\Phi, \Df, \Eqfst, \Eqsnd, \Solved, \USolved))\) holds when
    \begin{itemize}
      \item Variables in \(\Solved\) and \(\USolved\): \(\vars[2](\Solved,\USolved) \subseteq \vars[2](\Df)\)
      \item Equation: \(\mgu(\Eqn[i]) \neq \bot\), \(\dom(\mgu(\Eqn[i])) \cap \vars[i](\Df) = \emptyset\), \(\vars[i](\im(\mgu(\Eqn[i]))) \subseteq \vars[i](\Df)\).
      \item Solution is consequence: \(\im(\mgu(\Eqsnd)) \subseteq \conseq(\Solved \cup \Df)\).
      \item Shape of \(\Solved\): For all \(\psi = (\xi \dedfact u) \in \Solved\), \(\psi \in \USolved\), \(u \notin \X[1]\), \(u \in \subterms(\Phi)\) and \(\subterms(\xi) \subseteq \conseq(\Solved \cup \Df)\).
      \item Shape of \(\USolved\): For all \(\clause[S]{H}{\varphi} \in \USolved(\C)\), \(S\) is empty and \(\varphi\) only contains syntactic equations as hypothesis, i.e. no deduction facts.
      Moreover \(\strsubterms[2](H) \subseteq \conseq(\Solved \cup \Df)\),
      and if \(H = \xi \dedfact u\) then either \(u \in \subterms(\Phi)\) or there exists a recipe \(\zeta\) such that \((\zeta,u) \in \conseq(\Solved \cup \Df)\).
    \end{itemize}
  \end{definition}

  This invariant is lifted to sets of extended constraint systems in the natural way, i.e. \(\PredWellFormed(S)\) holds \textit{iff} for all \(\C \in S\), \(\PredWellFormed(\C)\) holds.

\paragraph{Invariant 2: Formula soundness}
  The second invariant states that any substitution that satisfies the deduction facts of \(\Df\) and the equalities of \(\Eqfst\) also satisfies all formulas of \(\Solved \cup \USolved\).
  This means that the procedure only adds correct formulas to the constraint system, sometimes under some hypothese for the formulas of \(\USolved\).

  \begin{definition} \label{def:invariant_correctness_formula}
    We define the predicate \(\PredCorrectFormula\) on extended constraint systems as follows;
    we have that \(\PredCorrectFormula((\Phi, \Df, \Eqfst, \Eqsnd, \Solved, \USolved))\) holds when for all \(\psi \in \Solved \cup \USolved\), for all substitutions \(\Sigma,\sigma\),
    if
    \[(\Phi,\Sigma,\sigma) \models \Df' \wedge
      \equality{\Eqfst} \qquad \text{with}\ \Df' = \{\psi' \in \Df \mid \vars[2](\psi') \subseteq \vars[2](\psi)\}\]
    where \(\equality{\Eqfst}\) if the set of equations of \(\Eqfst\), then \((\Phi,\Sigma,\sigma) \models \psi\).
  \end{definition}

\paragraph{Invariant 3: Formula completeness}
  Given a formula \(\psi = \clause{H}{\varphi}\) in \(\USolved\), the soundness invariant above states that when some substitutions satisfy \(\varphi\) then they also satisfy the head \(H\).
  The next invariant can be seen as a kind of converse statement:
  when some substitutions satisfy the head \(H\), there exists a formula \(\psi' \in \USolved\) (not necessarily the same) that has the same head and whose hypotheses are satisfied.
  This means that the procedure always covers all cases when generating the potential hypotheses of a given head \(H\).

  \begin{definition} \label{def:invariant_complete_formula}
    In this definition, we say that \((\Phi,\Sigma,\sigma)\) \emph{weakly satisfies} a head \(H\) when:
    \begin{itemize}
      \item if \(H = \xi \dedfact u\) then \(\msg(\xi \Sigma \Phi \sigma)\), i.e. the recipe \(\xi\) leads to a valid message
      \item if \(H = \xi \eqf \zeta\) then \((\Phi,\Sigma,\sigma)\) satisfies \(H\) in the usual sense, i.e. the recipes \(\xi,\zeta\) lead to the same valid message.
    \end{itemize}
    Given a set of extended processes \(\Gamma\), the predicate \(\PredCompleteFormula(\Gamma)\) holds when for all \(\C^e_1,\C^e_2 \in \Gamma\), for all \((\clause{H}{\varphi}) \in \USolved(\C^e_1)\),
    for all \((\Sigma,\sigma) \in \Sol(\C^e_2)\), if  \((\Phi(\C^e_2),\Sigma,\sigma)\) weakly satisfies \(H\)
    then there exists \((\clause{H'}{\varphi'}) \in \USolved(\C^e_2)\) such that \(H' \receq H\) and \((\Phi(\C^e_2),\Sigma,\sigma) \models \varphi'\).
  \end{definition}

\paragraph{Invariant 4: Knowledge-base saturation}
  The next invariant states that the knowledge base \(\Solved\) is saturated, i.e. that we do not miss solutions by imposing that they are constructed from \(\Solved\) (see Definition \ref{def:ext-sol}).

  \begin{definition} \label{def:invariant_consequence}
    We define the predicate \(\PredConseq\) on extended constraint systems as follows.
    We have that \(\PredConseq(\C)\) holds when, considering the minimal \(k\) such that \(\vars[2](\Df(\C)) \subseteq \Xsndi{k}\),
    for all \(\ffun/n \in \sigd\), for all \((\xi_1,u_1),\ldots, (\xi_n,u_n) \in \conseq(\Solved(\C))\) such that \(\xi_1,\ldots, \xi_n \in \recipeset_k\),
    if \(\ffun(u_1,\ldots, u_n) \norm\) is a constructor term then there is \(\xi \in \recipeset_k\) such that \((\xi,\ffun(u_1,\ldots, u_n) \norm) \in \conseq(\Solved(\C) \cup \Df(\C))\).
  \end{definition}

\paragraph{Invariant 5: Preservation of solutions}
  Finally the last invariant states that all the constraint solving performed on the additional data of extended constraint systems (\(\Eqsnd\), \(\Solved\), \(\USolved\)) preserves the solutions.
  That is, in an extended symbolic process \((\P,\C,\C^e)\), where \(\C\) is the constraint system obtained by only collecting constraints during the execution of the process \(\P\) (i.e. without additional constraint solving), all solutions of \(\C^e\) are solutions of \(\C\).

  \begin{definition} \label{def:invariant_solution}
    We define the predicate \(\PredSymb\) defined on extended symbolic processes where \(\PredSymb((\P,\C,\C^e))\) holds when
    \begin{itemize}
      \item for all \(i \in \N\), \(\vars[2](\C) \subseteq \Xsndi{i}\) iff \(\vars[2](\C^e) \subseteq \Xsndi{i}\).
      \item for all \((\Sigma,\sigma) \in \Sol(\C^e)\), \((\Sigma_{|\vars[2](\C)},\sigma_{|\vars[1](\C)}) \in \Sol(\C)\).
    \end{itemize}
  \end{definition}

  We restrict the substitutions to the variables of \(\C\) since our extended constraint system may introduce new variables (e.g. by applying most general solutions) but all these variables are uniquely defined by the instantiation of the variables of \(\C\).

\paragraph{Invariant 6: Component structure}
  Finally we state an invariant on components stating that all of their constraint systems have the same second-order structure.
  This invariant is preserved during the procedure by the fact that the constraint-solving rules that modify \(\Solved\) or \(\Eqsnd\) are always applied to the entire components.

  \begin{definition} \label{def:invariant_structure}
    We define the predicate \(\PredStruct\) defined on sets of extended symbolic processes where \(\PredSymb(\Gamma)\) holds when for all \((\P_1,\C_1,\C^e_1),(\P_2,\C_2,\C^e_2) \in \Gamma\),
    \begin{itemize}
      \item \(\dom(\Phi(\C^e_1)) = \dom(\Phi(\C^e_2))\)
      \item \(\vars[2](\C^e_1) = \vars[2](\C^e_2)\)
      \item \(\{\xi \mid (\xi \dedfact u) \in \Solved(\C^e_1)\} = \{ \xi \mid (\xi \dedfact u) \in \Solved(\C^e_2)\}\)
    \end{itemize}
  \end{definition}

\paragraph{Overall invariant}
  As we mentioned, all invariants are lifted to sets of extended symbolic processes in the natural way if needed.
  We refer as
  \[\PredAll(\Gamma) =
    \PredWellFormed(\Gamma) \wedge
    \PredCorrectFormula(\Gamma) \wedge
    \PredCompleteFormula(\Gamma) \wedge
    \PredConseq(\Gamma) \wedge
    \PredSymb(\Gamma) \wedge
    \PredStruct(\Gamma)\]
  the invariant of the whole procedure on the nodes of the partition tree (i.e. sets of extended symbolic processes).

\subsection{Preservation of the invariants} \label{app:invariants-preserve}

In this section we prove that the invariants of the procedure stated in Section \ref{app:invariants} are preserved all along the procedure.
First of all we prove the case of the case distinction rules:

\begin{lemma}
  Let \(\S\) be a set of set of extended symbolic processes such that \(\PredAll(\S)\).
  Let \(\S \rightarrow \S'\) by applying one case distinction rule and then normalising the result with the simplification rules.
  We have that \(\PredAll(\S')\).
\end{lemma}

For the proof we let \(\Gamma' \in \S'\) and write \(\Gamma \in \S\) the set from which \(\Gamma'\) is originated, i.e. \(\Gamma'\) is one of the sets obtained after applying one case distinction rule to \(\Gamma\) (either the positive or the negative branch) and then normalising with the simplification rules.

\begin{proof} (Preservation of \(\PredWellFormed\).)
  Let \(\C^{e\prime} = (\Phi, \Df, \Eqfst, \Eqsnd, \Solved, \USolved)\) for some \((\P',\C',\C^{e\prime}) \in \Gamma'\).
  We consider each item of the definition of the predicate (Definition \ref{def:well-formed}) and show that \(\C^{e\prime}\) verifies them.

  \caseitem{\emph{property 1 (Variables in \(\Solved\) and \(\USolved\)):}
  \(\vars[2](\Solved,\USolved) \subseteq \vars[2](\Df)\).}

    We first observe that this property is preserved by all simplification rules:
    therefore it sufficies to prove that it is preserved by application of case distinction rules.
    For that we also observe that if \(\C^e\) is an extended constraint system verifying this property then for all second-oder substitutions \(\Sigma\), \(\CApply{\Sigma}{\C^e}\) verifies it as well (which is precisely why we consider this notation rather than a raw application \(\C^e \Sigma\)).
    This is sufficient for getting the expected result for Rule \eqref{rule:satisfiable}.
    The case of the negative branches of Rules \eqref{rule:rewrite} and \eqref{rule:equality} are trivial.
    Regarding their positive branches, we only treat the case of Rule \eqref{rule:equality} since the treatment of \eqref{rule:rewrite} can be obtained by using a similar reasoning on each formulas added by the rule (the sets \(\USolved_0\) in the notations of Rule \eqref{rule:rewrite}).
    The rule \eqref{rule:equality} does not introduce elements in \(\Solved\) and the only second-oder variables introduced in \(\USolved\) that are not trivially added to \(\Df\) are from recipes \(\xi_1,\xi_2\) such that \(\xi_1 \dedfact u_1, \xi_2 \dedfact u_2 \in \Solved(\C^e)\) for some \((\P,\C,\C^e) \in \Gamma\).
    In particular by hypothesis \(\vars[2](\xi_1,\xi_2) \subseteq \vars[2](\Df(\C^e)) \subseteq \vars[2](\Df(\C^{e\prime}))\), hence the conclusion.

  \caseitem{\emph{property 2 (first and second-order equations):}
  we have \(\mgu(\Eqn[i]) \neq \bot\), \(\dom(\mgu(\Eqn[i])) \cap \vars[i](\Df) = \emptyset\), and \(\vars[i](\im(\mgu(\Eqn[i]))) \subseteq \vars[i](\Df)\).}

    The fact that \(\mgu(\Eqn[i]) \neq \bot\) is a direct consequence of the facts that \(\C^{e\prime} \neq \bot\) and that \(\C^{e\prime}\) is already normalised by the simplification rules of Figure \ref{fig:normalisation_constraint_systems},
    in particular Rules \ref{rule:unifEqfst_norm} and the rules inherited from Figure \ref{fig:normalisation_formula}.
    The remaining properties are simple invariants of the mgu's that are straightforward to verify.

  \caseitem{\emph{property 3 (Solution is consequence):}
  \(\im(\mgu(\Eqsnd)) \subseteq \conseq(\Solved \cup \Df)\).}

    This property comes from the fact that \(\mgu(\Eqsnd)\) is only modified in the positive branches of the case distinction rules by applying a mgs \(\Sigma\) to the extended constraint systems \(\C^e \in \Gamma\).
    A quick induction on the constraint solving relation for computing mgs \(\simpl\) show that \(\im(\Sigma) \subseteq \conseq(\Solved(\C^e) \cup \Df(\C^e))\), hence the conclusion.

  \caseitem{\emph{property 4 (Shape of \(\Solved\)):}
  for all \(\psi = (\xi \dedfact u) \in \Solved\), \(\psi \in \USolved\), \(u \notin \X[1]\), \(u \in \subterms(\Phi)\) and \(\subterms(\xi) \subseteq \conseq(\Solved \cup \Df)\).}

    The only rule adding a deduction fact to \(\Solved\) is the rule \eqref{rule:vector-solve};
    in particular this gives \(\psi \in \USolved\).
    The added deduction facts are of the form \(\xi \dedfact u\) where, for some \((\P,\C,\C^e) \in \Gamma\), \(\xi \dedfact u \in \USolved(\C^e)\) and \((\zeta,u) \notin \conseq(\Solved(\C^e) \cup \Df(\C^e))\) for any recipe \(\zeta\).
    The fact that \(\xi \dedfact u \in \USolved(\C^e)\) and \(\PredWellFormed(\Gamma)\) (property 5) justify that \(\subterms(\xi) \subseteq \conseq(\Solved \cup \Df)\);
    this justifies that \(u \in \subterms(\Phi)\) as well when taking into account that \(u\) is not a consequence of \(\Df(\C^e)\) (in particular it cannot be a ground contructor term without names otherwise it would even be consequence of the empty set).
    Finally the property \(u \notin \X[1]\) also follows from \((\zeta,u)\) not being a consequence of \(\Df(\C^e)\).

  \caseitem{\emph{property 5 (Shape of \(\USolved\)):}
  For all \(\clause{H}{\varphi} \in \USolved(\C^{e\prime})\), \(\varphi\) only contains syntactic equations as hypothesis, i.e. no deduction facts. Moreover \(\strsubterms[2](H) \subseteq \conseq(\Solved \cup \Df)\),
  and if \(H = \xi \dedfact u\) then either \(u \in \subterms(\Phi)\) or there exists a recipe \(\zeta\) such that \((\zeta,u) \in \conseq(\Solved \cup \Df)\).}

    The property that the hypotheses of the formula only contain syntactic equations can be obtained by a straightforward inspection of each case-distinction and simplification rules.
    Note in particular that whenever a formula \(\psi\) with deduction-fact hypotheses is considered (in Rules \eqref{rule:rewrite} and \eqref{rule:equality}), they are applied to a substitution \(\Sigma\) so that \(\FApply{\Sigma}{\psi}{\C^e}\) only has equations as hypotheses.
    As for the second part of property 5 we write
    \begin{enumerate*}
      \item \label{it:inv-wf-usolved-1}
      the property \(\strsubterms[2](H) \subseteq \conseq(\Solved \cup \Df)\) and
      \item \label{it:inv-wf-usolved-2}
      the rest (the property about the head terms of deduction formulas).
    \end{enumerate*}
    We perform a case analysis on the rule that added the formula to \(\Gamma'\).
    \begin{itemize}
      \item rule \eqref{rule:vector-consequence}:
      the proof of \ref{it:inv-wf-usolved-2} directly follows from \(\PredWellFormed(\Gamma)\) since the rule does not add a deduction formula.
      Regarding \ref{it:inv-wf-usolved-1},
      using the notations of the rule, the head of the formula is \(\xi \eqf \zeta\) where \(\xi \dedfact u \in \USolved\) for some \(u\), and \(\zeta \in \conseq(\Solved \cup \Df)\).
      In particular the conclusion follows from \(\PredWellFormed(\Gamma)\) (property 4 for the subterms of \(\xi\) and property 5 for the subterms of \(\zeta\))
      \item rule \eqref{rule:rewrite}:
      we first prove \ref{it:inv-wf-usolved-1}.
      Using the notations of the rule, all non-root positions of the head of a formula of \(\USolved_0\) are either variables \(X\) such that \(\Df\) contains a deduction fact \(X \dedfact x\) (generated by the skeleton), a position of \(\xi_0\), or a public function symbol (since the rewriting system is constructor, the left-hand sides \(\ell\) of the rewrite rules only contain a destructor at their roots).
      In particular all strict subterms of \(H\) are consequences of \(\Solved \cup \Df\).
      Now let us prove \ref{it:inv-wf-usolved-2}.
      By definition of the rule, \(u\) is a constructor term in normal form obtained after applying one rewrite rule \(\ell \to r\) at the root of \(C[u_0]\) for some context \(C\) (not containing names but possibly containing variables \(x\) such that \(X \dedfact x \in \Df\) for some \(X\)).
      We recall that the rewriting system is constructor-destructor and subterm convergent, which leaves two cases.
      The first is that \(r\) is a ground constructor term, and then \(u = r\) is trivially a consequence of \(\Solved \cup \Df\) for the recipe \(\zeta = r\).
      Otherwise \(r\) is a subterm of \(\ell\), meaning that \(u\) is a subterm of \(C[u_0]\).
      This implies that \(u\) is either a subterm of \(u_0\), a subterm of the context \(C\), or a term of the form \(C'[u_0]\) for some subcontext \(C'\) of \(C\).
      In all cases, since \(u_0 \in \subterms(\Phi)\) by \(\PredWellFormed(\Gamma)\) (property 4), this gives the expected conclusion.
      \item rule \eqref{rule:equality}:
      the proof of \ref{it:inv-wf-usolved-2} follows from \(\PredWellFormed(\Gamma)\) since the rule does not add a deduction formula.
      Regarding \ref{it:inv-wf-usolved-1},
      this follows from \(\PredWellFormed(\Gamma)\) (property 4).
      \qedhere
    \end{itemize}
\end{proof}

\begin{proof} (Preservation of \(\PredCorrectFormula\).)
  Let \(\C' = (\Phi, \Df, \Eqfst, \Eqsnd, \Solved, \USolved)\) for some \((\P',\C',\C^{e\prime}) \in \Gamma'\), \(\psi \in \Solved \cup \USolved\), and \((\Sigma,\sigma)\) such that
  \[(\Phi,\Sigma,\sigma) \models
    \underset {\vars[2](\psi') \subseteq \vars[2](\psi)} {\underset {\psi' \in \Df} \bigwedge} \hspace{-5mm} \psi' \quad \wedge
    \equality{\Eqfst}\]
  We have to prove that \((\Phi,\Sigma,\sigma) \models \psi\).
  Here we prove more precisely that the conclusion holds when applying one time any of the case-distinction or simplification rules.
  First of all we observe that this property is preserved by the application of a mgs to an extended constraint system.
  In particular this makes the conclusion immediate for all rules except the following:
  \begin{itemize}
    \item Rule \eqref{rule:rewrite} (positive branch):
    using the notations of the rule and recalling \(\PredCorrectFormula(\Gamma)\), the conclusion follows from the fact that if \(\xi_0 \Sigma \Phi \sigma \norm = u_0 \norm\), \(\ell \to r \in \R\) and \(C[u] = \ell \sigma'\) for some substitution \(\sigma'\), then \(C[\xi_0] \Sigma \Phi \sigma \norm = r \sigma' \norm\).

    \item Rule \eqref{rule:equality} (positive branch):
    using the notations of the rule and recalling \(\PredCorrectFormula(\Gamma)\), if \(\psi\) is the formula added to \(\USolved\) by this rule, we have \((\Phi,\Sigma,\sigma) \models \psi\) iff for some recipes \(\xi_1,\xi_2\), if \(\xi_1 \Sigma \Phi \sigma \norm = z\) and \(\xi_2 \Sigma \Phi \sigma \norm = z\)
    then \(\xi_1 \Sigma \Phi \sigma \norm = \xi_2 \Sigma \Phi \sigma \norm\).
    This naturally holds.

    \item Rule \eqref{rule:vector-solve}:
    since the element added to \(\Solved\) is originated from \(\USolved\), the conclusion follows from \(\PredCorrectFormula(\Gamma)\).

    \item Rule \eqref{rule:vector-consequence}:
    the reasoning is identical to the one for \eqref{rule:equality}. \qedhere
  \end{itemize}
\end{proof}

\begin{proof} (Preservation of \(\PredCompleteFormula\).)
  Let \({\C_1^e}',{\C_2^e}'\) for some \((\P_1',\C_1',{\C^e_1}'),(\P_2',\C_2',{\C^e_2}') \in \Gamma'\), \(\psi = \clause{H}{\varphi} \in \USolved({\C^e_1}')\) and
  \((\Sigma,\sigma) \in \Sol({\C^e_2}')\).
  We assume that \((\Phi({\C^e_2}'),\Sigma,\sigma)\) weakly satisfies \(H\).
  We want to prove that there exists \(\psi' = (\clause{H'}{\varphi'}) \in \USolved({\C^e_2}')\) such that \(H' \receq H\) and \((\Phi({\C^e_2}'), \Sigma, \sigma) \models \varphi'\).
  We prove the strenghtened property stating that the invariant is preserved after the application of each case-distinction and simplification rules.
  \begin{itemize}
    \item Rules \eqref{rule:unifEqfst_norm}, \eqref{rule:uniform}, \eqref{rule:Disequation removal 1}:
      the conclusion directly follows from \(\PredCompleteFormula(\Gamma)\).

    \item Rule \eqref{rule:Disequation removal 2}:
      by \(\PredCompleteFormula(\Gamma)\) we let \(\psi_0' = (\clause{H_0'}{\varphi_0'}) \in \USolved(\C^e_2)\) such that \(H_0' \receq H\) and \((\Phi(\C^e_2), \Sigma, \sigma) \models \varphi_0'\).
      The conclusion directly follows from this, except if \(\psi_0'\) is the formula removed by the rule, i.e. if \(\mgs(\C^e_2[\Eqfst \mapsto \Eqfst \wedge \varphi_0']) = \emptyset\).
      However this would yield a contradiction with \((\Phi(\C^e_2), \Sigma, \sigma_2) \models \varphi_0'\), hence the conclusion.

    \item Rule \eqref{rule:Removal of unsolved formula}:
      using the same notations as in the previous case, we let \(\psi_0'\) by \(\PredCompleteFormula(\Gamma)\) and the only non-trivial case is again the one where \(\psi_0'\) is removed from \(\C^e_2\) by the rule.
      This means that there exists \(\psi' \in \USolved({\C^e_2}')\) solved such that \(\psi' \receq H_0'\), hence the conclusion since \(\Fhyp(\psi') = \top\) and \(H_0' \receq H\).

    \item Rules \eqref{rule:vectorbot} and \eqref{rule:vector-split solved}:
      the conclusion follows from the \(\PredCompleteFormula(\Gamma)\) since \(\Gamma' \subseteq \Gamma\).

    \item Rule \eqref{rule:vector-solve}:
      directly follows from \(\PredCompleteFormula(\Gamma)\).

    \item Rule \eqref{rule:vector-consequence}:
      Let us write
      \[\Gamma' = \{ (\P,\C,\C^e[\USolved \mapsto \USolved \wedge \FApply{\Sigma}{\psi}{\C^e}]) \mid (\P,\C,\C^e) \in \Gamma \}\]
      with the notations and assumptions of the rule.
      If \(i \in \{1,2\}\) we write \(S_i = (\P,\C,\C^e_i)\).
      The only case that does not directly follows from \(\PredCompleteFormula(\Gamma)\) is the case where \(\psi = \FApply{\Sigma}{\psi}{\C^e_1}\).
      But since there exists a formula \(\xi \dedfact u_{S_2} \in \USolved(\C^e_2)\) by hypothesis, we know by \(\PredCorrectFormula(\Gamma)\) that \(\xi \Phi(\C^e_2) \Sigma \sigma_2 \norm = u_{S_2}\).
      In particular since \((\Phi(\C^e_2), \Sigma, \sigma_2)\) weakly satisfies \(\xi \eqf \zeta\), we can write \(u_{S_2}' = \zeta \Sigma \Phi(\C^e_2) \sigma_i \norm\) and have \(u_{S_2} = u_{S_2'}\).
      Since we also have by definition (after formula normalisation, see Figure \ref{fig:normalisation_formula})
      \[\FApply{\Sigma}{\psi}{\C^e_2} = (\clause{\xi \eqf \zeta}{u_{S_2} \eqs u_{S_2'}})\]
      we obtain the expected conclusion.

    \item Any case-distinction rule (negative branch) or Rule \eqref{rule:satisfiable}:
    follows from \(\PredCompleteFormula(\Gamma)\).

    \item Rule \eqref{rule:rewrite} (positive branch):
      using the notations of the rule, the only case that does not directly follow from the assumption \(\PredCompleteFormula(\Gamma)\) is the case where \(\psi \in \USolved_0\).
      The hypothesis \((\Phi(\C^e_2),\Sigma,\sigma)\) weakly models the head of \(\psi\) can then be rephrased as \(\msg(\xi' \Sigma \Phi(\C^e_2) \sigma)\) for some recipe \(\xi' = C[\xi_0]\) such that \(\rootf(C) \in \sigd\), by definition of \(\USolved_0\).
      Let us write \(u = \xi_0 \Sigma \Phi(\C^e_2) \sigma \norm\).
      Since the rewriting system is constructor-destructor and \(\rootf(C[u]) \in \sigd\), there exists at least one rewrite rule \(\ell' \to r' \in \R\) such that \(C[u]\) and \(\ell'\) are unifiable.
      In particular it sufficies to choose \(\psi'\) the formula of \(\USolved_0\) corresponding to picking this rule in the definition of \(\RewF{\xi}{\ell \to r}{p}\).

    \item Rule \eqref{rule:equality} (positive branch):
      easily follows from \(\PredCompleteFormula(\Gamma)\) since the rule only adds to each \((\P,\C,\C^e) \in \Gamma\) the same solved formula.
      \qedhere
  \end{itemize}
\end{proof}

\begin{proof} (Preservation of \(\PredConseq\).)
  Let \(\C' = (\Phi, \Df, \Eqfst, \Eqsnd, \Solved, \USolved)\) for some \((\P',\C',\C^{e\prime}) \in \Gamma'\) and \(k\) such that \(\vars[2](\Df(\C)) \subseteq \Xsndi{k}\).
  We also let \(\ffun/n \in \sigd\) and \((\xi_1,u_1),\ldots, (\xi_n,u_n) \in \conseq(\Solved(\C))\) such that \(\xi_1,\ldots, \xi_n \in \recipeset_k\),
  and \(\ffun(u_1,\ldots, u_n) \norm\) is a constructor term.
  We have to prove that there is \(\xi \in \recipeset_k\) such that \((\xi,\ffun(u_1,\ldots, u_n) \norm) \in \conseq(\Solved(\C) \cup \Df(\C))\).
  This can be justified by \(\PredWellFormed(\Gamma)\) (in particular the item ``shape of \(\Solved\)'').
  Indeed, by definition the term \(\ffun(u_1,\ldots, u_n)\) contains a single destructor, which is the symbol \(\ffun\) at its root.
  In particular if \(\ffun(u_1,\ldots, u_n) \norm\) is indeed a constructor protocol term \(u\), this has been obtained after applying a single rewriting rule at the root of \(\ffun(u_1,\ldots, u_n)\) since the rewriting system is constructor-destructor.
  Therefore, by subterm convergence, \(u\) is either a ground protocol term (without names, by definition of a rewrite rule) or a subterm of one of the \(u_i\)s.
  In the former case the conclusion is immediate, in the latter it follows from the invariant \(\PredWellFormed\).
\end{proof}

\begin{proof} (Preservation of \(\PredSymb\).)
  The preservation of this invariant is straightforward:
  the case distinction and simplification rules only modify the extended constraint systems by
  \begin{enumerate*}
    \item applying a mgs, or
    \item adding formulas to \(\Eqfst\), \(\Eqsnd\), \(\Solved\),..., or
    \item removing formulas from \(\USolved\), or
    \item removing trivially-false disequations from \(\Eqfst\).
  \end{enumerate*}
  In particular such operations restrict the set of solutions.
\end{proof}

We can then extend this property to the whole procedure by proving the preservation by symbolic rules.

\begin{lemma}
  Let \(\Gamma\) be a set of extended symbolic processes such that \(\PredAll(\{\Gamma\})\).
  We assume that no case-distinction or simplification rules can be applied to \(\Gamma\).
  We then let
  \begin{align*}
    \Gamma_\inp & = \{ (\P',\C',\C^{e\prime}) \mid (\P,\C,\C^e) \in \Gamma, (\P,\C,\C^e) \Sstep{\InP{Y}{X}} (\P',\C',\C^{e\prime})\}\\
    \Gamma_\outp & = \{ (\P',\C',\C^{e\prime}) \mid (\P,\C,\C^e) \in \Gamma, (\P,\C,\C^e) \Sstep{\OutP{Z}{\ax_{|\Phi(\C^e)|+1}}} (\P',\C',\C^{e\prime})\}
  \end{align*}
  and the set of set of symbolic processes \(\S\) obtained by normalising \(\{\Gamma_\inp,\Gamma_\outp\}\) with the simplification rules.
  Then \(\PredAll(\S)\).
\end{lemma}

\begin{proof}
  Most invariants are either easily seen to be preserved by application of any symbolic or simplification rules, or are a straightforward consequence of the fact that no simplification rules can be applied to \(\S\).
  The only substantial case is the Invariant 4 stating that the knowledge base is saturated.
  Let us then consider \((\P,\C,\C^e) \in \Gamma'\) for some \(\Gamma' \in \S\) and we let \(k\) the minimal index such that \(\vars[2](\Df(\C)) \subseteq \Xsndi{k}\).
  We then let \(\ffun/n\, \in \sigd\) and \((\xi_1,u_1), \ldots, (\xi_n,u_n) \in \conseq(\Solved(\C^e))\) such that \(u \norm\), with \(u = \ffun(u_1, \ldots, u_n)\) is a constructor term.
  We recall that no case-distinction rules can be applied to \(\{\Gamma\}\), in particular \eqref{rule:rewrite}.
  But since the rewriting system is constructor-destructor, \(\ffun \in \sigd\), and \(u \norm\) is a constructor term, this means that a rewrite rule is applicable at the root of \(u\).
  We distinguish two cases.

  \caseitem{\emph{case 1:} there exists a rule \(\ell \to r\) applicable at the root of \(u\), \(i \in \eint {1} {n}\) and \(p\) a position such that
  \(\getpos {\ell} {i \cdot p} \notin \X[1]\) and \((\getpos {\xi_i} {p} \dedfact v) \in \Solved(\C^e)\) for some \(v\).}

    We consider the instance of Rule \eqref{rule:rewrite} with the following parameters:
    the position \(i \cdot p\), the rule \(\ell \to r\), a recipe \(\xi \in \termset(\sig,\quanti{\X}{|\Phi(\C^e)|})\), the formula \(\psi_0\) obtained by considering the rule \(\ell \to r\) in \(\RewF{\xi}{\ell \to r}{i \cdot p}\),
    the deduction fact \(\getpos {\xi_i} {p} \dedfact v\), \(\Sigma_0 = \{\getpos{\xi}{i \cdot p} \mapsto \getpos {\xi_i} {p}\}\),
    and a mgs \(\Sigma\) such that the head of \(\FApply{\Sigma_0\Sigma}{\psi_0}{\CApply{\Sigma}{\C^e}}\) is of the form \(\zeta \dedfact u\norm\) for some recipe \(\zeta\).
    Since this instance of the rule is not applicable by hypothesis, we deduce that there already exists a solved formula \(\psi \in \USolved(\C^e)\) of the form \(\zeta' \dedfact u\norm\).
    Since no simplification rules are applicable neither, in particular \eqref{rule:vector-solve}, we deduce that there exists a recipe \(\xi'\) such that \((\xi',u\norm) \in \conseq(\Solved(\C^e) \cup \Df(\C^e))\).

  \caseitem{\emph{case 2:} otherwise}

    We let \(\ell \to r\) an arbitrary rewrite rule applicable at the root of \(u\).
    By hypothesis for all positions \(i \cdot p\) of \(\ell\) that are not variables we know that \(\getpos{\xi_i}{p}\) is not the recipe of a deduction fact from \(\Solved(\C^e)\);
    since \(\xi_i \in \conseq(\Solved(\C^e))\), \(\rootf(\getpos{\xi_i}{p})\) is therefore a constructor symbol.
    We rule out the immediate case where \(r\) is a ground term and only consider the one where \(r\) is a strict subterm of \(\ell\).
    Then, since the rewriting system is constructor-destructor, we can fix a ground constructor context \(C\) such that \(r = C[x_p,\ldots,x_q]\) for \(x_p, \ldots, x_p\) the variables numbered \(p\) to \(q\) of \(\ell\) (w.r.t. the lexicographic ordering on the positions of the term \(\ell\)).
    In particular the recipe \((C[\xi_p, \ldots, \xi_q], u \norm) \in \conseq(\Solved(\C^e))\).
\end{proof}

\subsection{Preliminary technical results}

In this section we prove some low level technical results that will be useful during the incoming proofs.

\paragraph{Preservation by application of substitutions}
\label{sec:properties_formulas}

  In this section, we show that applying substitution preserves in some cases the different notions we use in our algorithms, namely first-order and second-order equations, deduction and equality facts; and uniformity.


  \begin{lemma}
  \label{lem:apply_subst_clause}
  Let \(\psi\) be either a deduction fact, or an equality fact or a first-order equation or a second-order equation.
  For all ground frame \(\Phi\), for all substitutions \(\Sigma,\Sigma',\sigma,\sigma'\) if \(\dom(\Sigma) \cap \dom(\Sigma') = \emptyset\) and \(\dom(\sigma) \cap \dom(\sigma') = \emptyset\)
  then \((\Phi,\Sigma\Sigma',\sigma\sigma') \models \psi\) is equivalent to \((\Phi,\Sigma\Sigma',\sigma\sigma') \models \psi\Sigma\sigma\) and is equivalent to \((\Phi,\Sigma',\sigma') \models \psi\Sigma\sigma\).
  \end{lemma}

  \begin{proof}
    The proof of this lemma is done by case analysis on \(\psi\).

    \medskip

    \noindent\emph{Case \(u \eqs v\):}
    Consider \((\Phi,\Sigma\Sigma',\sigma\sigma') \models u \eqs v\).
    This is equivalent to \(u\sigma\sigma' = v\sigma\sigma'\).
    Since \(\vars(\Sigma) \cap \X[1] = \emptyset\), we deduce that \((\Phi,\Sigma\Sigma',\sigma\sigma') \models u \eqs v\) is equivalent to \(u\Sigma\sigma\sigma' = v\Sigma\sigma\sigma'\).
    This is also equivalent to \(u\Sigma\sigma\sigma\sigma' = v\Sigma\sigma\sigma\sigma'\) and so \((\Phi,\Sigma\Sigma',\sigma\sigma') \models u\Sigma\sigma \eqs v\Sigma\sigma\).
    Note that \(u\Sigma\sigma\sigma' = v\Sigma\sigma\sigma'\) is also equivalent to \((\Phi,\Sigma',\sigma') \models u\Sigma\sigma \eqs v\Sigma\sigma\).

    \medskip

    \noindent\emph{Case \(\xi \eqs \xi'\):}
    Similar to the previous case.

    \medskip

    \noindent\emph{Case \(\xi \dedfact u\):}
    Consider \((\Phi,\Sigma\Sigma',\sigma\sigma') \models \xi \dedfact u\).
    It is equivalent to \(\xi\Sigma\Sigma'\Phi\norm = u\sigma\sigma'\) and \(\msg(\xi\Sigma\Sigma'\Phi)\).
    But \((\xi \dedfact u)\Sigma\sigma = \xi\Sigma \dedfact u\sigma\).
    Moreover, by definition of substitution (in particular the acyclic property), we deduce that \(\xi\Sigma\Sigma'\Phi = \xi\Sigma\Sigma\Sigma'\Phi\) and \(u\sigma\sigma' = u\sigma\sigma\sigma'\). Hence, \(\msg(\xi\Sigma\Sigma'\Phi)\) is equivalent to \(\msg(\xi\Sigma\Sigma\Sigma'\Phi)\);
    and \(\xi\Sigma\Sigma'\Phi\norm = u\sigma\sigma'\) is equivalent to \(\xi\Sigma\Sigma\Sigma'\Phi\norm = u\sigma\sigma\sigma'\).
    Hence \((\Phi,\Sigma\Sigma',\sigma\sigma') \models \xi \dedfact u\) is equivalent to \((\Phi,\Sigma\Sigma',\sigma\sigma') \models \xi \dedfact u\Sigma\sigma\).
    Note that \(\msg(\xi\Sigma\Sigma'\Phi)\) and \(\xi\Sigma\Sigma'\Phi\norm = u\sigma\sigma'\) are also equivalent to \((\Phi,\Sigma',\sigma') \models \xi\Sigma \dedfact u\sigma\).

    \medskip

    \noindent\emph{Case \(\xi \eqf \xi'\):}
    Similar to previous case.
  \end{proof}


\paragraph{Properties on consequence of set of deduction facts}
\label{sec:properties_consequence}

  This section contains some results about the consequence relation when modifying a set of deduction facts.
  The first lemma is about the application of substitutions.


  \begin{lemma}
  \label{lem:consequence_protocol_terms_substitution}
  Let \(S\) be a set of solved deduction facts.
  For all substitutions \(\sigma\) of protocol terms, for all \((\xi,t) \in \conseq(S)\), \((\xi,t\sigma) \in \conseq(S\sigma)\).
  \end{lemma}

  \begin{proof}
  We know that \((\xi,t) \in \conseq(S)\) implies \(\xi = C[\xi_1,\ldots, \xi_n]\) and \(t = C[t_1,\ldots, t_n]\) for some public context \(C\) and for all \(i\), \((\xi_i \dedfact t_i \in S\).
  Hence \((\xi_i \dedfact t_i\sigma \in S\sigma\) which allows us to conclude.
  \end{proof}


  \begin{proposition}[transitivity of consequences] \label{prop:trans-conseq}
    Let \(S,S'\) be two sets of solved deduction facts.
    Let \(\varphi = \{X_i \dedfact u_i\}_{i=1}^n\) such that all \(X_i\) are pairwise distinct, let \((\xi,t) \in \conseq(S \cup \varphi)\) and \(\Sigma,\sigma\) be two substitutions.
    If for all \(i \in \eint{1}{n}\), \((X_i\Sigma,u_i\sigma) \in \conseq(S\Sigma\sigma \cup S')\) then \((\xi\Sigma,t\sigma) \in \conseq(S\Sigma\sigma \cup S')\).
  \end{proposition}

  \begin{proof}
    We prove this result by induction on \(|\xi|\).
    The base case (\(|\xi| = 0\)) is trivial as there are no terms of size \(0\) and we hence focus on the inductive step.
    We perform a case analysis on the hypothesis \((\xi,t) \in \conseq(S \cup \varphi)\).

    \caseitem{\emph{case 1: \(\xi = t \in \sig_0\)}}
      We directly have by definition that \((\xi\Sigma,t\sigma) \in \conseq(S\Sigma\sigma \cup S')\).

    \caseitem{\emph{case 2: there are \(\xi_1,t_1,\ldots, t_m,\xi_m\) and \(\ffun \in \sigc\) such that \(\xi = \ffun(\xi_1,\ldots, \xi_m)\), \(t = \ffun(t_1,\ldots, t_m)\) and for all \(i \in \eint{1}{m}\), \((\xi_i,t_i) \in \conseq(S \cup \varphi)\)}}
      By induction hypothesis we know that for all \(j \in \eint{1}{m}\), \((\xi_j\Sigma,t_j\sigma) \in \conseq(S\Sigma\sigma \cup S')\).
      Writing \(\xi\Sigma = \ffun(\xi_1\Sigma,\ldots, \xi_m\Sigma)\) and \(t\sigma = \ffun(t_1\sigma,\ldots,t_m\sigma)\), we conclude that \((\xi\Sigma,t\sigma) \in \conseq(S\Sigma\sigma \cup S')\).

    \caseitem{\emph{case 3: \(\xi \dedfact t \in S \cup \varphi\)}}
      If \((\xi \dedfact t) \in S\) then \((\xi\Sigma \dedfact t\sigma) \in S\Sigma\sigma\) and the result directly holds.
      Otherwise \(\xi \dedfact t \in \varphi\) and hence by hypothesis \((\xi\Sigma, t\sigma) \in \conseq(S\Sigma\sigma \cup S')\).
  \end{proof}

  The previous lemma showed that a consequence \((\xi,t)\) is preserved when applying some substitution \(\Sigma,\sigma\) under the right conditions.
  However, it is quite strong since we ensure that \(\xi\Sigma\) is consequence with \(t\sigma\).
  In some cases, we cannot guarantee that \(\xi\Sigma\) is consequence with \(t\sigma\) but with some other first-order term.
  This is the purpose of the next lemma.


  \begin{lemma}
    \label{lem:consequence_subtitution_recipe}
    Let \(S\), \(S'\) be two sets of solved deduction facts.
    Let \(\varphi = \{X_i \dedfact u_i\}_{i=1}^n\) such that all \(X_i\) are pairwise distinct.
    For all \(\Sigma\), for all \(\xi \in \conseq(S \cup \varphi)\), if for all \(i \in \{1, \ldots, n\}\), \(X_i\Sigma \in \conseq(S\Sigma \cup S')\) then \(\xi\Sigma \in \conseq(S\Sigma \cup S')\).
  \end{lemma}

  \begin{proof}
    We prove this result by induction on \(|\xi|\).
    The base case (\(|\xi| = 0\)) being trivial as there is no term of size \(0\), we focus on the inductive step.

    Since \(\xi\) is consequence of \(S \cup \varphi\), we know by definition that there exists \(t\) such that one of the following conditions hold:
    \begin{enumerate}
      \item \(\xi = t \in \sig_0\)
      \item there exists \(\xi_1,t_1,\ldots, t_m,\xi_m\) and \(\ffun \in \sigc\) such that \(\xi = \ffun(\xi_1,\ldots, \xi_m)\), \(t = \ffun(t_1,\ldots, t_m)\) and for all \(i \in \{1, \ldots,m\}\), \((\xi_i,t_i)\) is consequence of \(S \cup \varphi\).
      \item there exists \(t\) such that \(\xi \dedfact t \in S \cup \varphi\).
    \end{enumerate}
    In case 1, we directly have by definition that \((\xi\Sigma,t)\) is a consequence of \(S\Sigma \cup S'\).
    In case 2, by our inductive hypothesis on \(\xi_1,\ldots, \xi_m\), we have that for all \(j \in \{1,\ldots, m\}\), \(\xi_j\Sigma\) is a consequence of \(S\Sigma \cup S'\) hence there exists \(t'_1,\ldots, t'_m\) such that for all \(j \in \{1,\ldots, m\}\), \((\xi_j\Sigma,t'_j)\) is a consequence of \(S\Sigma \cup S'\).
    With \(\xi\Sigma = \ffun(\xi_1\Sigma,\ldots, \xi_m\Sigma)\) and \(t' = \ffun(t'_1,\ldots,t'_m)\), we conclude that \((\xi\Sigma,t')\) is consequence of \(S\Sigma \cup S'\).
    In case 3, if \(\xi \dedfact t \in S\) then \(\xi\Sigma \dedfact t \in S\Sigma\) and so the result directly holds.
    Else \(\xi \dedfact t \in \varphi\) and so by hypothesis \(\xi\Sigma \in  \conseq(S\Sigma \cup S')\).
  \end{proof}

  In the next lemma, we show that when a recipe is consequence of the sets of solved deduction formulas \(\Solved\Sigma\sigma\) where \((\Sigma,\sigma)\) is a solution of the constraint system, then all subterms of that recipe are also consequence of \(\Solved\Sigma\sigma\).
  This property is in fact guaranted by the fact that \(\Solved\) contains itself recipes consequence of itself.
  This is an important property that allows us to generate solutions that satisfy the uniformity property.


  \begin{lemma}
    \label{lem:consequence_subterms_uninstantiated}
    Let \(\C = (\Phi, \Df, \Eqfst, \Eqsnd, \Solved, \USolved)\) be an extended constraint system such that \(\PredWellFormed(\C)\).
    For all \(\xi \in \conseq(\Solved \cup \Df)\), \(\subterms(\xi) \subseteq \conseq(\Solved \cup \Df)\).
  \end{lemma}

  \begin{proof}
    Since \(\xi \in \conseq(\Solved \cup \Df)\), we know that \(\xi = C[\xi_1,\ldots, \xi_n]\) where \(C\) is a public context and \(\xi_1,\ldots, \xi_n\) are recipes of deduction facts from \(\Solved\) or \(\Df\).
    Hence since \(\xi' \in \subterms(\xi)\), we have that the position \(p\) of \(\xi'\) in \(\xi\) is either a position of \(C\) thus \(\xi' \in \in \conseq(\Solved \cup \Df)\) from the definition of consequence;
    or is a position of one of the \(\xi_i\) and thus we conclude by the predicate \(\PredWellFormed(\C)\).
  \end{proof}


  \begin{lemma}
  \label{lem:consequence_subterms}
    Let \(\C = (\Phi, \Df, \Eqfst, \Eqsnd, \Solved, \USolved)\) be an extended constraint system such that \(\PredWellFormed(\C)\).
    For all \((\Sigma,\sigma) \in \Sol(\C)\), for all \(\xi \in \conseq(\Solved\Sigma\sigma)\), \(\subterms(\xi) \subseteq \conseq(\Solved\Sigma\sigma)\).
  \end{lemma}


  \begin{lemma}
    \label{lem:consequence_implies_dedfact}
    Let \(S\) be a set of ground deduction facts. Let \(\Phi\) be a ground frame.
    Assume that for all \(\psi \in S\), \(\Phi \models \psi\). For all \((\xi,t) \in \conseq(S)\), \(\Phi \models \xi \dedfact t\).
  \end{lemma}

  \begin{proof}
    We prove this result by induction on \(|\xi|\).
    The base case being trivial, we focus on the inductive step. Since \((\xi,t)\) is consequence of \(S\) then one of the following properties holds:
    \begin{enumerate}
      \item \(\xi = t \in \sig_0\)
      \item there exist \(\xi_1,t_1,\ldots, \xi_n,t_n\) and \(\ffun \in \sigc\) such that \(\xi = \ffun(\xi_1,\ldots, \xi_n)\), \(t = \ffun(t_1,\ldots, t_n)\) and for all \(i \in \{1,\ldots, n\}\), \((\xi_i,t_i)\) is consequence of \(S\)
      \item \(\xi \dedfact t \in S\).
    \end{enumerate}
    In Case 1, the result trivially holds.
    In case two, a simple induction on \((\xi_1,t_1),\ldots, (\xi_n,t_n)\) allows us to conclude.
    In case 3, we know by hypothesis that \(\Phi \models \xi \dedfact t\) hence the result holds
  \end{proof}


\subsection{Correctness of most general solutions} \label{app:mgs}

In this section we prove the correctness of the constraint solving procedure for computing mgs' in Section \ref{sec:mgs-gen}.
We show that given an extended constraint system \(\C\), the Rules \eqref{rule:conseq}, \eqref{rule:res},
\eqref{rule:cons}, \eqref{rule:mgs-unsat}, and \eqref{rule:unifEqfst_simpl} allow to compute the most general solutions of \(\C\).


\begin{lemma}
  Let \(\C\) an extended constraint system such that \(\PredWellFormed(\C)\) and \(\PredCorrectFormula(\C)\).
  If any rule is applicable on \(\C\) then for all \((\Sigma,\sigma) \in \Sol(\C)\), there exists \(\C'\), \(\Sigma'\) such that \(\C \SimpStep{} \C'\) and \((\Sigma\Sigma',\sigma) \in \Sol(\C')\).
\end{lemma}

\begin{proof}
  First, assume that there exist \(\xi,\zeta \in R(\C)^2\) and \(u\) such that \(\xi \neq \zeta\) and \((\xi,u),(\zeta,u) \in \conseq(\Solved \cup \Df)\).
  By \(\PredCorrectFormula(\C)\) and Lemmas \ref{lem:consequence_implies_dedfact} and Proposition \ref{prop:trans-conseq},
  we deduce that \(\Phi\sigma \models \xi\Sigma \dedfact u\sigma \wedge \zeta\Sigma \dedfact u\sigma\).
  As such we have \(\xi\Sigma\norm = u\sigma = \zeta\Sigma\norm\).
  However by definition of a solution it implies that \(\xi\Sigma = \zeta\Sigma\).
  Thus there exists \(\Sigma' = \mgu(\xi \eqs \zeta)\) such that \(\Sigma' \neq \emptyset\) and \(\Sigma' \neq \bot\).

  In such a case, let us show that \((\Sigma,\sigma) \in \Sol(\C')\) with \(\C \rightarrow \C'\) by Rule \eqref{rule:conseq}.
  We already know that \(\Sigma \models \Eqsnd(\C)\) and since \(\xi\Sigma = \zeta\Sigma\) with \(\Sigma' = \mgu(\xi \eqs \zeta)\), we directly have that \(\Sigma \models \Eqsnd(\C)\Sigma' \wedge \Sigma'\).
  Moreover, \(\Solved(\C)\Sigma = \Solved(\C)\Sigma'\Sigma\).
  Hence, the two bullets of the definition of solutions is trivially satisfied by that fact that \((\Sigma,\sigma) \in \Sol(\C)\).
  Therefore, we conclude that \((\Sigma,\sigma) \in \Sol(\C')\).

  \smallskip

  Let us now consider the case where our assumption do no hold. Thus since we assume that at least one rule is applicable on \(\C\), there exists \(X \dedfact u \in \Df(\C)\) where \(u \notin \X[1]\). Let us do a case analysis on \(X\Sigma\) since \(X\Sigma \in \conseq(\Solved\Sigma\sigma)\) by definition of a solution.
  \begin{itemize}
    \item either \(X\Sigma \in \sig_0\): in such a case, we have \(\C \rightarrow \C'\) by Rule \eqref{rule:conseq} and we can prove similarly as in the previous case that \((\Sigma,\sigma) \in \Sol(\C')\);
    \item or \(X\Sigma = \ffun(\xi_1,\ldots,\xi_n)\) where \(\xi_i \in \conseq(\Solved\Sigma\sigma)\) for all \(i\):
    note that we know that \(X\Sigma\Phi\sigma\norm = u\sigma\).
    Hence \(u = \ffun(u_1,\ldots, u_n)\) for some \(u_1,\ldots, u_n\).
    We deduce that for all \(i\), \(\Phi\sigma \models \xi_i \dedfact u_i\sigma\).
    Thus, by considering \(\Sigma' = \{ X_i \mapsto \xi_i\}_{i=1}^n\), we can conclude that \(\C \rightarrow \C'\) by Rule \eqref{rule:cons} and \((\Sigma\Sigma',\sigma) \in \Sol(\C')\);
    \item or \(X\Sigma \dedfact u\sigma \in \Solved\Sigma\sigma\) (since once again \(X\Sigma\Phi\sigma\norm = u\sigma\)):
    thus there exists \(\xi \dedfact v \in \Solved\) such that \(\xi\Sigma = X\Sigma\) and \(u\sigma = v\sigma\).
    Hence \(\mgu{\xi}{X}\) exists and \(\sigma \models u \eqs v\).
    Thereofore, we can conclude that \(\C \rightarrow \C'\) by Rule \eqref{rule:res} and \((\Sigma,\sigma) \in \Sol(\C')\). \qedhere
  \end{itemize}
\end{proof}


\begin{lemma}
  Let \(\C \neq \bot\) an extended constraint system such that \(\C\simplnorm = \C\), \(\PredWellFormed(\C)\) and \(\PredCorrectFormula(\C)\). If \(\C \not\simpStep{}\) and \(\C\) is a solved extended constraint system then \(\mgs(\C) = \{\mgu(\Eqsnd)\}\).
\end{lemma}

\begin{proof}
  We know that for all \((\Sigma,\sigma) \in \Sol(\C)\), \(\Sigma \models \Eqsnd(\C)\) thus we directly obtain the existence of \(\Sigma'\) such that \(\Sigma = \mgu(\Eqsnd)\Sigma'\).
  Consider now the second bullet point of the definition of most general solutions.
  We know that \(\C\) is solved.
  Hence consider a fresh bijective renaming  \(\Sigma_1\) from \(\vars[2](\Sigma_0) \cup \vars[2](\C) \setminus \dom(\Sigma_0)\) to \(\sig_0\).
  Let us define \(\sigma_1 = \{x \mapsto X\Sigma_1 \mid X \dedfact x \Df(\C)\}\).
  Thanks to \(\PredWellFormed(\C)\), \(\PredCorrectFormula(\C)\) and Lemma \ref{lem:consequence_subterms}, \ref{lem:consequence_subtitution_recipe},
  and \ref{lem:consequence_implies_dedfact} that
  \((\Phi\mgu(\Eqfst(\C))\sigma_1,\mgu(\Eqsnd)\Sigma_1, \mgu(\Eqfst(\C))\sigma_1) \models \Df \wedge \Eqfst \wedge \Eqsnd\).
  Moreover, by Lemma \ref{lem:consequence_subterms}, we know that the first bullet of the definition of solution is satisfied.
  Finally, the second bullet is satisfied otherwise Rule \eqref{rule:mgs-unsat} would be applicable which contradict \(\C\simplnorm = \C\).
  Therefore, \((\mgu(\Eqsnd)\Sigma_1, \mgu(\Eqfst(\C))\sigma_1) \in \Sol(\C)\).
  We conclude that \(\mgs(\C) = \{\mgu(\Eqsnd)\}\).
\end{proof}


\begin{lemma}
  Let \(\C\) an extended constraint system such that \(\C\simplnorm = \C\), \(\PredWellFormed(\C)\) and \(\PredCorrectFormula(\C)\).
  If \(\C \not\simpStep{}\) and \(\C\) is not solved then \(\Sol(\C) = \emptyset\).
\end{lemma}

\begin{proof}
  Since \(\C\) is not solved, we have two possibilities:
  Either (a) all deduction facts in \(\Df\) are have variables as right hand term but not pairwise distinct.
  But in such a case Rule \eqref{rule:conseq} would be applicable which contradicts \(\C \not\simpStep{}\);
  or (b) there exists \((X \dedfact u) \in \Df(\C)\) such that \(u \notin \X[1]\).
  Since Rule \eqref{rule:conseq} is not applicable, we deduce that \(u \notin \sig_0\) and for all \(\xi,\zeta \in \recipes(\C) \setminus \{ X\}\), \((\xi,u) \notin \conseq(\Solved \cup \Df)\).
  But rule Rule \eqref{rule:cons} is not applicable therefore, we deduce that \(u \in \Nall\).

  Assume now that \(\Sol(\C) \neq \emptyset\) and so \((\Sigma,\sigma) \in \Sol(\C)\).
  Thus \(X\Sigma\Phi\norm = u\).
  By definition of a solution, we know that \((X\Sigma,u) \in \conseq(\Solved(\C)\Sigma\sigma)\).
  Since \(u \in \Nall\) it implies that there exists \((\xi \dedfact v) \in \Solved(\C)\) such that \(X\Sigma = \xi\Sigma\) and \(u = v\sigma\).
  Note that by \(\PredWellFormed(\C)\), we also have that \(v \notin \X[1]\) and so \(u = v\).
  In such a case, we obtain a contradiction with the fact the Rule \eqref{rule:res} is not applicable.
\end{proof}


\subsection{Correctness of the partition tree} \label{app:ptree-proof}

In this section we prove the correctness of the procedure generating the partition tree, using the invariants proved in Appendix \ref{app:invariants}.
Let us first start by noticing that the case distinction rules and simplification rules preserves the first order solutions of the extended constraint systems. This property is stated in the following lemma.

\begin{lemma} \label{lem:preservation_solutions}
  Let \(\S\) be a set of set of extended symbolic processes such that \(\PredAll(\S)\).
  Let \(\S \rightarrow \S'\) by applying only case distinction or simplifications rules (i.e. no symbolic transitions).
  Then:
  \begin{itemize}
    \item \emph{Soundness:}

    for all \(S \in \S\), for all \((\P,\C,\C^e) \in S\), for all \((\Sigma,\sigma) \in \Sol(\C^e)\),
    there exist \(S' \in \S'\), \((\P,\C,{\C^e}') \in S'\) and \((\Sigma',\sigma') \in \Sol({\C^e}')\) such that \(\sigma_{|\vars[1](\C)} = \sigma'_{|\vars[1](\C)}\)

    \item \emph{Completeness:}

    for all \(S' \in \S\), for all \((\P,\C,{\C^e}') \in S'\), for all \((\Sigma',\sigma') \in \Sol({\C^e}')\),
    there exist \(S \in \S\), \((\P,\C,\C^e) \in S\) and \((\Sigma,\sigma) \in \Sol(\C^e)\) such that \(\Sigma_{|\vars[2](\C)} = \Sigma'_{|\vars[2](\C)}\) and \(\sigma_{|\vars[1](\C)} = \sigma'_{|\vars[1](\C)}\)
  \end{itemize}
\end{lemma}

\begin{proof}
  We do a case analysis on the rule applied.

  \caseitem{\emph{case 1:} Normalisation rules (simplification rules of Figure \ref{fig:normalisation_constraint_systems})}

    First, let us notice the result directly hold for Rules \ref{rule:unifEqfst_norm}, \ref{rule:Disequation removal 2}, and \ref{rule:Removal of unsolved formula}.
    Indeed, Rule \ref{rule:Disequation removal 1} does not modify constraints on recipe and preserves the constraints on protocol terms.
    Moreover, Rule \ref{rule:Disequation removal 2},\ref{rule:Removal of unsolved formula} affect \(\USolved\) which do not impact the solutions of the extended constraint system.
    For Rule \ref{rule:uniform}, since \(\mgs(\C^e) = \emptyset\), we have by definition of most general unifiers that \(\Sol(\C^e) = \emptyset\) (otherwise the first bullet of the definition is contradicted).
    Hence the result holds since \(\Sol(\bot) = \emptyset\). Similarly, Rule \ref{rule:Disequation removal 2} checks whether the disequations \(\forall \tilde{x}.\phi\) is trivially true meaning that the rule preserves the solutions.

  \caseitem{\emph{case 2:} Simplification rules on partitions of extended symbolic processes (Figure \ref{fig:normalisation_vector})}

    The rule \ref{rule:vectorbot} only removes an extended symbolic process with an extended constraint systems having no solution hence the result holds.
    Rule \ref{rule:vector-split solved} splits a set of \(\S\) into two sets thus preserving the extended symbolic processes, and Rule \ref{rule:vector-consequence} only adds element in \(\USolved\) which do not impact the solutions of a constraint system.
    Therefore, for all these rules, the result hold.
    For Rule \ref{rule:vector-solve} however, the result is not direct since the rule adds an element in the set \(\Solved\) which has an impact on the solutions of a constraint system.
    However, we know from the application condition of the rule that the head protocol terms of the deduction facts added in \(\Solved_i\) are not consequence of \(\Solved_i \cup \Df_i\).
    But we also know that \(\PredConseq(\S)\) and \(\PredWellFormed(\S)\) hold hence it implies that the recipe \(\xi\) (see Figure \ref{fig:normalisation_vector}) contains \(\ax_{|\Phi_i|}\) and \(\vars[2](\Df_i) \cap \Xsndi{|\Phi_i|} = \emptyset\).
    Hence, \(\xi\) cannot appear in the second order solutions \(\C^e_i\) which allows us to conclude that the solutions are preserved.

  \caseitem{\emph{case 3:} Case distinction rules}

    The case of case distinction rules is straightforward.
    Indeed, by definition all rules \eqref{rule:satisfiable}, \eqref{rule:equality} and \eqref{rule:rewrite} always refine a set of extended symbolic process \(\Gamma\) into \(\Gamma_1,\Gamma_2\)
    where \(\Gamma_1\) is obtained by applying a substitution \(\Sigma\) to \(\Gamma\), and \(\Gamma_2\) by adding the constraint \(\neg \Sigma\) to \(\Gamma\).
    In particular this refinement preserves the solutions as expected.
\end{proof}

Now we can show that the static equivalence is preserved by application of the case distinction and simplification rules.

\begin{lemma}
  Let \(\S\) be a set of set of extended symbolic processes such that \(\PredAll(\S)\).
  Let \(\S \rightarrow \S'\) by applying only case distinction or simplifications rules (i.e. no symbolic transitions), \(\Gamma \in \S\), \((\P_1,\C_1,\C^e_1), (\P_2,\C_2,\C^e_2) \in \Gamma\),
  \((\Sigma,\sigma_1) \in \Sol(\C^e_1)\) and \((\Sigma,\sigma_2) \in \Sol(\C^e_2)\) such that \(\Phi(\C_1)\sigma_1 \StatEq \Phi(\C_2)\sigma_2\).

\end{lemma}

\begin{proof}
  Once again, let us consider the potential rule applied.

  \caseitem{\emph{case 1:} Simplification rules}

    The only non trivial case is Rule \ref{rule:vector-split solved} (the other ones do not refine the partition and the conclusion is therefore immediate).
    However by Lemma
    applied to \(\S\) we know that if a deduction fact occurs in constraint systems \(\C^e_1\) but no recipe equivalent formula can be found in the constraint system \(\C^e_2\), then no solution of \(\C^e_2\) can satisfy the head of the formula.
    Besides by \(\PredCorrectFormula(\S)\) we also know that all solutions of \(\C^e_1\) satisfy this deduction fact.
    Then since \((\Sigma,\sigma_1) \in \Sol(\C^e_1)\), \((\Sigma,\sigma_2) \in \Sol(\C^e_2)\) and \(\Phi(\C_1)\sigma_1 \sim \Phi(\C_2)\sigma_2\), we obtain a contradiction.
    Therefore, \(\C^e_1\) and \(\C^e_2\) are necessarily in the same set of \(\S'\).

  \caseitem{\emph{case 2:} Case distinction rules}

    Note that for case distinction rules, the proof is simple since each rule create a partition of the second-order solutions with respect to some mgs \(\Sigma_0\).
    Thus, assume w.l.o.g. that \((\Sigma',\sigma'_1) \in \Sol({\C^e_1}')\).

  \caseitem{\emph{case 2a:} negative branch of the rule}
    First consider that \(S'\) corresponds to branch in which we applied \(\neg \Sigma_0\).
    In such a case, since we already know that \(\Sigma'\) satisfies \(\neg \Sigma_0\) and no other constraint is added, we directly obtain from \((\Sigma,\sigma_2) \in \Sol(\C^e_2)\) that \((\Sigma',\sigma'_2) \in \Sol({\C^e_2}')\) (in this case, we even have \(\sigma_2 = \sigma'_2\)).

  \caseitem{\emph{case 2b:} positive branch of the rule}
    Now consider that \(S'\) corresponds to the branch in which we applied \(\Sigma_0\).
    In such a case, the application of \(\Sigma_0\) on \(\C^e_2\) regroups all the solution of \(\C^e_2\) that satisfies \(\Sigma_0\).
    Since we know that \((\Sigma,\sigma_2) \in \Sol(\C^e_2)\) and \({\Sigma'}_{|\vars[2](\C_1)} = {\Sigma}_{|\vars[2](\C_1)}\) which implies \({\Sigma'}_{|\vars[2](\C_2)} = {\Sigma}_{|\vars[2](\C_2)}\), the result holds.
\end{proof}

The previous two lemmas allow us to obtain the soundness and completeness properties of the partition tree.
Note that the monoticity of the second-order predicate
is also proved by~\ref{lem:preservation_solutions} (Completeness part) since a solution of a child constraint system is also a solution of parent one.
We now need to prove that all nodes of the partition tree are valid configurations.
For that we prove properties on extended constraint systems such that no more case distinction rules are applicable.

\begin{lemma} \label{lem:sat-non-applicable}
  Let \(\S\) be a set of sets of extended symbolic processes such that \(\PredAll(\S)\) and no instance of the rule \eqref{rule:satisfiable} or normalisation rules (i.e. the simplification rules of Figure \ref{fig:normalisation_constraint_systems}) are applicable.
  For all \(S \in \S\), for all \((\P,\C,\C^e) \in S\), writing \(\C^e = (\Phi, \Df, \Eqfst, \Eqsnd, \Solved, \USolved)\) we have that
  \begin{enumerate}
    \item \(\C^e\) is solved
    \item all formulas \(\psi \in \USolved\) are solved
    \item \(\Eqfst\) does not contain disequations
  \end{enumerate}
\end{lemma}

\begin{proof}
  First of all the non-applicability of Rule \eqref{rule:satisfiable} case \ref{it:rule-sat-mgs} gives that either \(\C^e\) is solved or \(\mgs(\C^e) = \emptyset\);
  due to normalisation rules not being applicable we deduce that \(\mgs(\C^e) \neq \emptyset\) meaning that \(\C^e\) is solved.
  Sinmilarly by the non applicability of case \ref{it:rule-sat-hyp} we know that for all \(\psi \in \USolved\), either \(\psi\) is solved or \(\mgs(\C^e[\Eqfst \mapsto \Eqfst \wedge \Fhyp(\psi)]) = \emptyset\).
  But since the normalisation rules are also not applicable, we know that \(\mgs(\C^e[\Eqfst \mapsto \Eqfst \wedge \Fhyp(\psi)]) \neq \emptyset\):
  therefore \(\psi\) is solved.
  Finally the non applicability of case \ref{it:rule-sat-diseq} and of the normalisation rules also gives us that \(\Eqfst\) is only composed of syntactic equations.
\end{proof}

\begin{lemma}
  Let \(\S\) be a set of sets of extended symbolic processes such that \(\PredAll(\S)\) and no instance of the rule \eqref{rule:satisfiable} or normalisation rules are applicable.
  For all \(S \in \S\), for all \((\P,\C,\C^e) \in S\), \(|\mgs(\C^e)| = 1\).
\end{lemma}

\begin{proof}
  Let us denote \(\C^e = (\Phi, \Df, \Eqfst, \Eqsnd, \Solved, \USolved)\).
  By Lemma \ref{lem:sat-non-applicable} we know that \(\C^e\) is solved, that all formulas \(\psi \in \USolved\) are solved, and that \(\Eqfst\) only contain equations.

  \caseitem{\emph{Step 1:} Construction of \((\Sigma,\sigma)\) such that \((\Phi,\Sigma,\sigma) \models \Df \wedge \Eqfst \wedge \Eqsnd\)}
    Since \(\C^e\) is solved, we deduce that all deduction facts in \(\Df = \{ X_i \dedfact x_i \}_{i=1}^n\) for some \(n\) and pairwise distinct \(x_i\)s and \(X_i\)s.
    Consider now the substitutions \(\Sigma_0 = \{ X_i \rightarrow n_i\}_{i=1}^n\) and \(\sigma_0 = \{ x_i \rightarrow n_i\}_{i=1}^n\) where the \(n_i\)s are pairwise distincts public names, i.e. \(n_i \in \sig_0\).
    Since no more normalisation rules are applicable, we know that the disequations in \(\Eqsnd\) not trivially unsatisfiable.
    Therefore by replacing the free variables of the disequations by names allow us to obtain that \(\Sigma_0\) the disequations of \(\Eqsnd\).
    By considering \(\Sigma = \mgu(\Eqsnd)\Sigma'\), we obtain that \(\Sigma \models \Eqsnd\).
    Moreover we proved that \(\Eqfst\) does not contain any disequations, we directly obtain that \(\mgu(\Eqfst)\sigma_0 \models \Eqfst\).
    Therefore, by defining \(\sigma = \mgu(\Eqfst)\sigma_0\), we obtain that \((\Phi,\Sigma,\sigma) \models \Df \wedge \Eqfst \wedge \Eqsnd\).

  \caseitem{\emph{Step 2:} Proof that \((\Sigma,\sigma)\) is a solution}
    To prove that \((\Sigma,\sigma)\) is an actual solution of \(\C^e\) it remains to prove that it verifies the additional required two conditions:
    \(\Solved\)-basis and uniformity. 
    Let us first prove the \(\Solved\)-basis, i.e. that for all \(\xi \in \subterms(\im(\Sigma)) \cup \strsubterms[2](\Solved \Sigma)\), \(\msg(\xi \Phi \sigma)\) and \((\xi,\xi\Phi\sigma) \in \conseq(\Solved \Sigma \sigma)\).
    The case \(\xi \in \strsubterms[2](\Solved \Sigma)\) directly follows from \(\PredWellFormed(\C^e)\).
    Let us therefore consider the case \(\xi \in \subterms(\im(\Sigma))\).
    Since \(\PredWellFormed(\C^e)\) holds we have that \(\im(\mgu(\Eqsnd)) \subseteq \conseq(\Solved \cup \Df)\);
    for the same reason we have that for all \(\zeta \dedfact u \in \Solved\), \(\subterms(\zeta) \subseteq \conseq(\Solved \cup \Df)\).
    Therefore by applying Lemma \ref{lem:consequence_subtitution_recipe}, and by a quick induction on the size of the recipe in \(\im(\Sigma)\), we obtain that \(\xi \in \conseq(\Solved\Sigma)\).
    Finally, by definition of consequence and since \(\PredCorrectFormula(\C^e)\) holds, we have \(\msg(\xi\Phi\sigma)\) hence the \(\Solved\)-basis.

    Let us now prove uniformity.
    We know that \(\C^e\) is solved which therefore means that for all recipes \(\xi,\zeta \in \stc(\im(\mgu(\Eqsnd)),\Solved \cup \Df)^2 \cup (\sig_0 \times \vars[2](\Df))\), \((\xi,u), (\zeta,u) \in \conseq(\Solved \cup \Df)\) implies \(\xi = \zeta\).
    Since \(\Sigma = \mgu(\Eqsnd)\Sigma_0\) we directly obtain that for all \(\xi,\zeta \in \stc(\Sigma,\Solved\Sigma)\), \((\xi,u),(\zeta,u) \in \conseq(\Solved\Sigma)\) implies \(\xi = \zeta\), which is exactly the uniformity.

  \caseitem{\emph{Step 3:} Unicity of the solution}
    This step is rather straightforward:
    considering that \(\Sigma = \mgu(\Eqsnd)\Sigma_0\) and any other solutions \((\Sigma',\sigma') \in \Sol(\C^e)\) satisfy \(\Sigma' \models \Eqsnd\), we deduce that \(\mgs(\C^e) = \{ \mgu(\Eqsnd) \}\) and so \(|\mgs(\C^e)| = 1\).
\end{proof}

Let us now show that all extended constraint systems in the set have the same solutions and that they are statically equivalent.

\begin{lemma}
  Let \(\S\) be a set of set of extended symbolic processes such that \(\PredAll(\S)\) and no instances of the rules \eqref{rule:satisfiable}, \eqref{rule:equality} or \eqref{rule:rewrite} or simplification rules are applicable.
  For all \(S \in \S\), for all \((\P_1,\C_1,\C^e_1), (\P_2,\C_2,\C^e_2) \in S\), if \((\Sigma,\sigma_1) \in \Sol(\C^e_1)\) then \((\Sigma,\sigma_2) \in \Sol(\C^e_2)\) and \(\Phi(\C^e_1)\sigma_1 \StatEq \Phi(\C^e_2)\sigma_2\).
\end{lemma}

\begin{proof}
  Since \eqref{rule:satisfiable} and normalisation rules are not applicable, we know by Lemma \ref{lem:sat-non-applicable} that all extended constraint systems \(\C^e \in S\) have a particular form, that is
  \begin{enumerate*}
    \item all deduction facts in \(\Df(\C^e)\) have pairwise distinct variables as right hand side; and
    \item \(\Eqfst(\C^e)\) only contain syntactic equations.
  \end{enumerate*}
  Moreover, we know that all extended constraint systems have the same structure.
  Therefore, if \((\Sigma,\sigma_1) \in \Sol(\C^e_1)\), we deduce that \(\Sigma \models \Eqsnd(\C^e_1)\) and for all \(\xi \in \subterms(\im(\Sigma))\),
  \(\xi \in \conseq(\Solved(\C^e_1)\Sigma)\), meaning that \(\Sigma \models \Eqsnd(\C^e_2)\) and \(\xi \in \conseq(\Solved(\C^e_2)\Sigma)\).
  Since the first order solutions are always completely defined by the second-order substitutions, we can build \(\sigma'_2\) such that for all \(X \dedfact x \in \Df(\C^e_2)\), \(X\Sigma(\Phi(\C^e_2)\sigma'_2)\norm = x\sigma'_2\).
  Moreover, since \(\PredCorrectFormula(\C^e_1)\) and \(\PredCorrectFormula(\C^e_2)\) both hold and since for all \(\xi \in \subterms(\im(\Sigma))\), \(\xi \in \conseq(\Solved(\C^e_2)\Sigma)\), we deduce that for all \(\xi \in \subterms(\im(\Sigma))\), \(\msg(\xi\Phi(\C^e_2)\sigma'_2)\).
  Note that we also need to satisfy the syntactic equations in \(\Eqfst\).
  However thanks to \(\PredWellFormed(\C^e_2)\) holding, we know that \(\dom(\mgu(\Eqfst(\C^e_2)) \cap \vars[1](\Df(\C^e_2)) = \emptyset\).
  Thus, we can build \(\sigma_2 = \mgu(\Eqfst(\C^e_2))\sigma'_2\) and obtain that \((\Sigma,\sigma_2) \models \Df(\C^e_2) \wedge \Eqfst(\C^e_2) \wedge \Eqsnd(\C^e_2)\).
  Note that by origination property of an extended constraint system, we have \(\Phi(\C^e_2)\sigma'_2 = \Phi(\C^e_2)\sigma_2\).
  Therefore, since we already prove that for all \(\xi \in \subterms(\im(\Sigma))\), \(\msg(\xi\Phi(\C^e_2)\sigma'_2)\) and \(\xi \in \conseq(\Solved(\C^e_2)\Sigma)\),
  it only remains to prove the second bullet point of Definition the definition of solutions extended constraint system to obtain that \((\Sigma,\sigma_2) \in \Sol(\C^e_2)\).

  To prove this it sufficies to prove that \(\Phi(\C^e_1)\sigma_1 \StatEq \Phi(\C^e_2)\sigma_2\):
  the conclusion will then follow since \((\Sigma,\sigma_1) \in \Sol(\C^e_1)\).
  Therefore we let recipes \(\xi,\xi'\) and show that:
  \begin{enumerate}[label=(\roman*)]
    \item \label{it:ptree-config-msg}
      \(\msg(\xi\Phi(\C^e_1)\sigma_1)\) iff \(\msg(\xi\Phi(\C^e_2)\sigma_2)\)
    \item \label{it:ptree-config-stateq}
      if \(\msg(\xi\Phi(\C^e_1)\sigma_1),\ \xi'\Phi(\C^e_1)\sigma_1)\) then \(\xi\Phi(\C^e_1)\sigma_1\norm = \xi'\Phi(\C^e_1)\sigma_1\norm\) iff \(\xi\Phi(\C^e_2)\sigma_2\norm = \xi'\Phi(\C^e_2)\sigma_2\norm\).
  \end{enumerate}
  We prove this by lexicographic induction on \((N(\xi,\xi'),\max(|\xi|,|\xi'|)\) where \(N(\xi\,\xi')\) is the number of subterms \(\zeta \in \subterms(\xi,\xi')\) such that \(\zeta \notin \conseq(\Solved(\C^e_1)\Sigma)\)
  (recall that since \(\C^e_1\) and \(\C^e_2\) have the same structure, we have \(\zeta \in \conseq(\Solved(\C^e_1)\Sigma)\) iff \(\zeta \in \conseq(\Solved(\C^e_2)\Sigma)\)).

  \caseitem{\emph{case 1:} \(N(\xi,\xi') = 0\) and \(\max(|\xi|,|\xi'|) = 0\)}
    Impossible since there exist no terms of size \(0\).

  \caseitem{\emph{case 2:} \(N(\xi,\xi') > 0\)}
  \caseitem{\emph{subgoal 2a}: Proof of \ref{it:ptree-config-msg}}
    Assume \(\msg(\xi\Phi(\C^e_1)\sigma_1)\).
    Let us also assume by contradiction that \(\neg \msg(\xi\Phi(\C^e_2)\sigma_2)\).
    Since we know that \(N(\xi,\xi') > 0\), there exists \(\zeta \in \subterms(\xi,\xi')\) such that \(\zeta \not \in \conseq(\Solved(\C^e_1)\).
    Without loss of generality we can consider that \(\zeta \in \subterms(\xi)\) (otherwise we can apply our inductive hypothesis on \(\xi\) twice since \(N(\xi,\xi')\) would be equal to \(0\) and so we would obtain a contradiction).
    Moreover, let us consider \(\zeta\) such that \(|\zeta|\) is minimal.
    Therefore, by definition of consequence, we deduce that \(\zeta = \gfun(\zeta_1,\ldots, \zeta_n)\) with \(\gfun \in \sigd\) and for all \(i \in \{1,\ldots,n\}\), \(\zeta_i \in \conseq(\Solved(\C^e_1))\).
    Since \(\msg(\xi\Phi(\C^e_1)\sigma_1)\) we also deduce that \(\gfun(\zeta_1,\ldots, \zeta_n)\Phi(\C^e_1)\sigma_1\norm\) is a protocol term.
    Therefore, there exist a rewrite rule \(\gfun(\ell_1,\ldots, \ell_n) \rightarrow r\) and a substitution \(\gamma\) such that \(\ell_i\gamma = \zeta_i\Phi(\C^e_1)\sigma_1\norm\) for all \(i = 1\ldots n\).

    Recall that the rule \eqref{rule:rewrite} is not applicable on \(\C^e_1\) and \(\C^e_2\).
    Therefore we can show that provided \(\neg \msg(\xi\Phi(\C^e_2)\sigma_2)\) and \(\gfun(\zeta_1,\ldots, \zeta_n)\Phi(\C^e_1)\sigma_1\norm\) is a protocol term then we necessarily have that there exists \(\zeta'_1,\ldots, \zeta'_n\) and \(u\)
    such that \(\gfun(\zeta'_1,\ldots, \zeta'_n) \dedfact u_1 \in \USolved(\C^e_1)\) and \(\zeta'_i\Sigma\Phi(\C^e_1)\sigma_1\norm = \zeta_i\Phi(\C^e_1)\sigma_1\norm\).
    Moreover, since the normalisation rules are also not applicable (in particular Rule \ref{rule:vector-split solved}), we deduce that there exists \(u_2\) such that \(\gfun(\zeta'_1,\ldots, \zeta'_n) \dedfact u_2 \in \USolved(\C^e_2)\). By \(\PredWellFormed(\C^e_1)\),
    we know that for all \(i \in \{1,\ldots, n\}\), \(\zeta'_i \in \conseq(\Solved(\C^e_1) \cup \Df(\C^e_1))\) and so \(\zeta'_i\Sigma \in \conseq(\Solved(\C^e_1)\Sigma)\).
    Moreover, by hypothesis on \(\zeta_i\), we know that \(\zeta_i \in \conseq(\Solved(\C^e_1)\Sigma)\).
    Thus, by applying our inductive hypothesis, we obtain that \(\zeta_i\Phi(\C^e_2)\sigma_2\norm = \zeta'_i\Sigma\Phi(\C^e_2)\sigma_2\norm\).
    Moreover, by \(\PredCorrectFormula(\C^e_2)\), we know that \(\gfun(\zeta'_1,\ldots, \zeta'_n)\Sigma\Phi(\C^e_2)\sigma_2\norm = u_2\sigma_2\) which is a protocol term.
    We conclude that \(\gfun(\zeta_1,\ldots, \zeta_n)\Sigma\Phi(\C^e_2)\sigma_2\norm\) is a protocol term and thus \(\msg(\xi\Phi(\C^e_2)\sigma_2)\) gives us a contradiction.

  \caseitem{\emph{subgoal 2b}: Proof of \ref{it:ptree-config-stateq}}
    Assume now that \(\xi\Phi(\C^e_1)\sigma_1\norm = \xi'\Phi(\C^e_1)\sigma_1\norm\), \(\msg(\xi\Phi(\C^e_1)\sigma_1)\) and \(\msg(\xi'\Phi(\C^e_1)\sigma_1)\).
    Let us once again take the smallest \(\zeta \in \subterms(\xi,\xi')\) such that \(\zeta \not \in \conseq(\Solved(\C^e_1)\).
    We already proved above that there exist \(u_1,u_2\), \(\gfun\), \(\zeta'_1,\ldots, \zeta'_n, \zeta_1,\ldots, \zeta_n\) such that:
    \begin{itemize}
      \item \(\zeta = \gfun(\zeta_1,\ldots, \zeta_n)\)
      \item \(\gfun(\zeta'_1,\ldots, \zeta'_n) \dedfact u_1 \in \USolved(\C^e_1)\)
      \item \(\gfun(\zeta'_1,\ldots, \zeta'_n) \dedfact u_2 \in \USolved(\C^e_2)\)
      \item for all \(i \in \{1,\ldots, n\}\), \(\zeta_i\Phi(\C^e_2)\sigma_2\norm = \zeta'_i\Sigma\Phi(\C^e_2)\sigma_2\norm\) and \(\zeta_i\Phi(\C^e_1)\sigma_1\norm = \zeta'_i\Sigma\Phi(\C^e_1)\sigma_1\norm\).
    \end{itemize}
    By \(\PredConseq(\C^e_1)\), we know that there exists \(\beta\) such that \((\beta,u_1) \in \conseq(\Solved(\C^e_1) \cup \Df(\C^e_1)\).
    However the normalisation Rule \ref{rule:vector-consequence} is not applicable on the set of extended symbolic processes.
    Thus, we deduce that there exists \(\beta'\) such that \((\beta',u_1) \in \conseq(\Solved(\C^e_1) \cup \Df(\C^e_1)\) and \(\gfun(\zeta'_1,\ldots, \zeta'_n) \eqf \beta' \in \USolved(\C^e_1)\).
    Once again due to the normalisation Rule \ref{rule:vector-split solved}, we obtain that \(\gfun(\zeta'_1,\ldots, \zeta'_n) \eqf \beta' \in \USolved(\C^e_2)\).
    But \(\PredCorrectFormula(\C^e_2)\) and \(\PredCorrectFormula(\C^e_1)\) hold meaning that \((\Phi(\C^e_2)\sigma_2,\Sigma,\sigma_2) \models \gfun(\zeta'_1,\ldots, \zeta'_n) \eqf \beta'\) and \((\Phi(\C^e_1)\sigma_1,\Sigma,\sigma_1) \models \gfun(\zeta'_1,\ldots, \zeta'_n) \eqf \beta'\).

    Note that if \(p\) is the position of \(\zeta\) in \(\xi\) then we have \(N(\replacepos{\xi}{p}{\beta'\Sigma}, \xi') < N(\xi,\xi')\).
    Thus by applying our inductive hypothesis, we obtain that \(	(\Phi(\C^e_2)\sigma_2,\Sigma,\sigma_2) \models \replacepos{\xi}{p}{\beta'\Sigma} \eqf \xi'\).
    Since \((\Phi(\C^e_2)\sigma_2,\Sigma,\sigma_2) \models \gfun(\zeta'_1,\ldots, \zeta'_n) \eqf \beta'\)
    and \((\Phi(\C^e_2)\sigma_2,\Sigma,\sigma_2) \models \gfun(\zeta'_1,\ldots, \zeta'_n) \eqf \gfun(\zeta_1,\ldots, \zeta_n)\),
    we conclude that \((\Phi(\C^e_2)\sigma_2,\Sigma,\sigma_2) \models \xi \eqf \xi'\).

  \caseitem{\emph{case 3:} \(N(\xi,\xi') = 0\) and \(\max(|\xi|,\xi'|) > 0\)}
    In such a case, we know that \(\xi, \xi' \in \conseq(\Solved(\C^e_1)\Sigma)\) and \(\xi, \xi' \in \conseq(\Solved(\C^e_2)\Sigma)\).
    By definition of consequence and by \(\PredCorrectFormula(\C^e_1)\) and \(\PredCorrectFormula(\C^e_2)\), we directly obtain that \(\msg(\xi\Phi(\C^e_1)\sigma_1)\) and \(\msg(\xi\Phi(\C^e_2)\sigma_2)\) (same thing for \(\xi'\)).
    Now assume that \((\Phi(\C^e_1)\sigma_1,\Sigma,\sigma_1) \models \xi \eqf \xi'\).
    Since both \(\xi,\xi'\) are consequences of \(\Solved(\C^e_1)\Sigma\), we deduce that:
    \begin{itemize}
      \item either \(\xi = \ffun(\xi_1,\ldots, \xi_n)\) and \(\xi' = \ffun(\xi'_1,\ldots,\xi'_n)\) with \(\ffun \in \sigc\) and \((\Phi(\C^e_1)\sigma_1,\Sigma,\sigma_1) \models \xi_i \eqf \xi'_i\) for all \(i\).
      Therefore, we can apply our inductive hypothesis on the \((\xi_i,\xi'_i)\)s to conclude.
      \item or \(\xi\Sigma,\xi'\Sigma \in \Solved(\C^e_1)\Sigma\):
      Since we know that the rule \eqref{rule:equality} is not applicable, it implies that \(\xi \eqf \xi' \in \USolved(\C^e_1)\) and so \(\xi \eqf \xi' \in \USolved(\C^e_2)\) thanks to the normalisation Rule \ref{rule:vector-split solved}.
      Since \(\PredCorrectFormula(\C^e_2)\) holds, we can conclude that \((\Phi(\C^e_2)\sigma_2,\Sigma,\sigma_2) \models \xi \eqf \xi'\).
      \item or \(\xi\Sigma \in \Solved(\C^e_1)\Sigma\) and \(\xi' = \ffun(\xi'_1,\ldots,\xi'_n)\) with \(\ffun \in \sigc\);
      Once again since the rule \eqref{rule:equality} is not applicable, we deduce that there exists \(\zeta'_1,\ldots, \zeta'_n\) such that \(\xi \eqf \ffun(\zeta'_1,\ldots, \zeta'_n) \in \USolved(\C^e_1)\).
      Note from \(\PredCorrectFormula(\C^e_1)\) that in such a case, \((\Phi(\C^e_1)\sigma_1,\Sigma,\sigma_1) \models \xi \eqf \ffun(\zeta'_1,\ldots, \zeta'_n)\)
      meaning that \((\Phi(\C^e_1)\sigma_1,\Sigma,\sigma_1) \models \xi'_i \eqf \zeta'_i\) for all \(i \in \{1,\ldots, n\}\).
      Since \(|\xi'_i \Phi(\C^e_1)\sigma_1\norm| < |\xi\Phi(\C^e_1)\sigma_1\norm|\), we can apply our inductive hypothesis on all \((\xi'_i,\zeta'_i)\) meaning that \((\Phi(\C^e_2)\sigma_2,\Sigma,\sigma_2) \models \ffun(\zeta'_1,\ldots, \zeta'_n) \eqf \xi'\).
      However, by the rule \ref{rule:vector-split solved} not being applicable,
      \(\xi \eqf \ffun(\zeta_1,\ldots, \zeta'_n) \in \USolved(\C^e_1)\) implies \(\xi \eqf \ffun(\zeta'_1,\ldots, \zeta'_n) \in \USolved(\C^e_2)\) and so by \(\PredCorrectFormula(\C^e_2)\),
      we obtain that \((\Phi(\C^e_2)\sigma_2,\Sigma,\allowbreak\sigma_2) \models \xi \eqf \ffun(\zeta'_1,\ldots, \zeta'_n)\) which allows us to conclude that \((\Phi(\C^e_2)\sigma_2,\Sigma,\sigma_2) \models \xi \eqf \xi'\).
      \qedhere
  \end{itemize}
\end{proof}


\section{Termination proof} \label{app:termination}

\subsection{For mgs}

  The termination of the computation of most general solutions mostly relied on the following result, yet to be proved:

  \propMgsDecrease*

  \begin{proof}
    Consider first the simplification rule \eqref{rule:unifEqfst_simpl} and the ones from Figure \ref{fig:normalisation_formula}.
    They typically apply protocol term substitutions on the constraint system (they also effect recipe disequations that are irrelevant in \(\measureNC(\C^e)\)).
    Note that the applied substitution is always generated from terms already in the constraint system.
    As such \(\mu^1(\C^e\norm) = \mu^1(\C^e)\) and so
    \(\Phi(\C^e\norm)\mu^1(\C^e\norm) = \Phi(\C^e)\mu^1(\C^e)\),
    \(\Solved(\C^e\norm)\mu^1(\C^e\norm) = \Solved(\C^e)\mu^1(\C^e)\) and
    \(\Df(\C^e\norm)\mu^1(\C^e\norm) = \Df(\C^e)\mu^1(\C^e)\).
    Thus, we directly obtain that \(|\measureNC(\C^e\norm)| \leqslant |\measureNC(\C^e)|\).

    Let us look at Rules \eqref{rule:conseq}, \eqref{rule:cons} and \eqref{rule:res} and let us consider \(\C^e \xrightarrow{\Sigma} \C'^e\).
    The rule \eqref{rule:conseq} does not modify the protocol terms of the constraint systems by apply a recipe substitution.
    However, we show an invariant on the constraint systems that any \(\xi,\zeta \in \stc(\C^e)\) are consequence of \(\Solved \cup \Df\) as well as any of their subterms (see Definition \ref{def:well-formed} in Appendix).
    Thus, we deduce from the definition of \(\stc(\C^e)\) that \(\stc(\C^e)\Sigma \subseteq \stc(\C'^e)\).
    To conclude that \(|\measureNC(\C'^e)| \leqslant |\measureNC(\C^e)|\), we rely on the technical Proposition \ref{prop:trans-conseq};
    in other words, if \((\xi,t) \in \conseq(\Solved(\C^e)\mu^1(\C^e) \cup \Df(\C^e)\mu^1(\C^e))\) then \((\xi\Sigma,t) \in \conseq(\Solved(\C'^e)\mu^1(\C'^e) \cup \Df(\C'^e)\mu^1(\C'^e))\).
    Since \(\stc(\C^e)\Sigma \subseteq \stc(\C'^e)\), we conclude that \(|\measureNC(\C'^e)| \leqslant |\measureNC(\C^e)|\).

    Applying the same reasoning for the rule \eqref{rule:res}, we can also show that \(\measureNC(C'^e) \leqslant \measureNC(\C^e)\).
    However, we can even show that this inequality is strict.
    Indeed, using the same the notation in the rule \eqref{rule:res}, this rule is only applied if \(\C^e = \C^e\norm\) and for all \(\xi \in \stc(\C^e) \setminus \{ X\}\), \((\xi,u) \not\in \conseq(\Solved \cup \Df)\).
    Note that \(\C^e = \C^e\norm\) implies that \(\Solved\mu^1 = \Solved\) and \(\Df\mu^1 = \Df\).
    Moreover, it also implies that \(u \in \measureNC{\C^e}\).
    However, in \(\C'^e\), we have that \((\xi,u\mu_1(\C'^e)) \in \conseq(\Solved(\C'^e)\mu^1(\C'^e) \cup \Df(\C'^e)\mu^1(\C'^e))\).
    Moreover, we show another invariant on the constraint system (see Definition \ref{def:well-formed} in Appendix) that ensures us that \(X \in \stc(\C^e)\) and so \(\xi \in \stc(\C'^e)\).
    Hence, we obtain that \(u\mu^1(\C'^e) \not\in \measureNC(\C'^e)\) allowing us to conclude that \(\measureNC(C'^e) < \measureNC(C^e)\).
    By applying the same reasoning, we can also show that \(\measureNC(C'^e) < \measureNC(C^e)\) when the rule \eqref{rule:cons} is applied.
  \end{proof}

\subsection{Exponential measure}

  Another argument left pending is that each component of the measure except the last one can be bounded by an exponential in the DAG size of the parameters of the problem.
  We give a bound for each of them, in particular relying on the bound on \(\measureNC\) proved in the body of the article.

  \begin{enumerate}
    \item \(\compon[1](\Gamma) \leqslant \dagsize{P,Q}\):

      by definition.



    \item \(\compon[2](\Gamma) \leqslant (\dagsize {P} \dagsize {E})^{\dagsize {P}} + (\dagsize {Q} \dagsize {E})^{\dagsize {Q}}\):

      The measure corresponds to the number of symbolic transitions possible from \(P\) and \(Q\) for a given symbolic trace, hence the bound.
      Notice that the part \(\dagsize {E}^{\dagsize {P}}\) is due to the computation of the most general unifiers modulo \(E\) in the symbolic transitions.

    \item \(\compon[3](\Gamma) \leqslant 9\dagsize{P,Q,E}^3\):

      It suffices to observe that \(\setSDF(\C^e) \leqslant \measureNC(\C^e)\) and to use the bound proved in the body of the article.

    \item \(\compon[4](\Gamma) \leqslant \compon[2](\Gamma)\):

      Trivial.

    \item \(\compon[5](\Gamma) \leqslant \compon[2](\Gamma) \times \dagsize{E}^{\dagsize{E}} \times (18\dagsize{P,Q,E})^{27\dagsize{P,Q,E}^3}\):

      Bounding the size of \(|\setRew(\C^e)|\) can easily be done:
      the number of \(\psi \in \Solved\) possible is bounded by \(|\Solved|\), itself bounded by \(|\setSDF(\C^e)|\).
      The number of rewrite rules, position \(p\) and \(\psi_0 \in \RewF{\xi}{\ell \rightarrow r}{p}\) only depends on the rewrite systems and can be bounded by \(\dagsize{E}^{\dagsize{E}}\).
      Note that the exponential comes mainly from the number of possible positions in \(\ell\).
      We already know that the number of most general solutions is bounded by \((|\Solved(\C^e)| + 1)^{\measureNC(\C^e)}\).
      Combining with all previous results, and with the rough approximation \(9\dagsize{P,Q,E}^3+1 \leqslant 18\dagsize{P,Q,E}^3\), we obtain the above bound.

    \item \(\compon[6](\Gamma) \leqslant |E| \times \compon[2](\Gamma)\):

      To bound this number, we need to recall that we always apply the case distinction rules with the priority ordering
      \eqref{rule:satisfiable} < \eqref{rule:rewrite}.
      Thus, when we apply a rule \eqref{rule:rewrite}, there is no unsolved deduction formula in any of the extended constraint systems (otherwise we should have applied the rule \eqref{rule:satisfiable}).
      It means this measure is bounded by the number of deduction formulas produced by one instance of \eqref{rule:rewrite}.
      By definition, we know that \(|\RewF{\xi}{\ell \rightarrow r}{p}| \leqslant |\R|\) (one formula per rewrite rule).
      Thus, the rule \eqref{rule:rewrite} generates at most \(|E| \times \compon[2](\Gamma)\) deduction formulas.

    \item \(\compon[7](\Gamma) \leqslant \compon[2](\Gamma) \times 2\dagsize{E}(\dagsize{P,Q})^2 (1 + \dagsize {E})^2\):

      The application conditions stipulate that the rule can be applied either (a) on two deduction facts of \(\Solved(\C^e_i)\), or (b) on one deduction fact of \(\Solved(\C^e_i)\) in combination with a construction function symbol.

      Note that even though the rule also consider the existence of a most general solution \(\Sigma \in \mgs(\C^e_i[\Eqfst \mapsto \Eqfst \wedge \Fhyp(\FApply {\Sigma_0} {\psi} {\C^e_i})])\), the number of applications of the rule \eqref{rule:equality} will not depend on the number of possible most general solutions.
      Indeed, consider the case (a) where the rule is applied on two deduction fact \((\xi_1 \dedfact u_1),(\xi_2 \dedfact u_2) \in \Solved(\C^e_i)\).
      Thus, an equality formula with \(\xi_1\Sigma \eqf \xi_2\Sigma\) as head will be added in \(\USolved(\CApply{\Sigma}{\C^e_i})\).
      However, in the application conditions of the rule, we also require that
      \emph{for all \((\clause[S]{H}{\varphi}) \in \USolved(\C^e_i)\), \(H \neq (\xi_1 \eqf \xi_2)\)}.
      Thus, a new application of the rule \eqref{rule:equality} on \(\CApply{\Sigma}{\C^e_i}\) with the same (up to instantiation of \(\Sigma\)) deductions facts from \(\Solved(\CApply{\Sigma}{\C^e_i})\) will be prevented.

      The same situation occurs in case (b) with the condition \emph{for all \((\clause[S]{\zeta_1 \eqf \zeta_2}{\varphi}) \in \USolved(\C^e_i)\), \(\zeta_1 = \xi_1\) implies \(\rootf(\zeta_2) \neq \ffun\)}.
      We therefore conclude that the rule \eqref{rule:equality} can be applied only once per pair of deduction facts in \(\Solved\) and once per deduction fact in \(\Solved\) and function symbol in \(\sigc\).

    \item \(\compon[8](\Gamma) \leqslant \compon[2](\Gamma)\):

      Unsolved equality formulas can be generated by two rules: the case distinction rule \eqref{rule:equality} or the simplification Rule \ref{rule:vector-consequence}.
      However, once again because of the priority order \eqref{rule:satisfiable} < \eqref{rule:equality}, the two rules cannot be triggered simultaneously and the rule \eqref{rule:equality} is only triggered when there is no unsolved equality formulas.
      Note that due to the condition \emph{\(\forall i. \forall (\clause[S]{\zeta_1 \eqf \zeta_2}{\varphi}) \in \USolved_i\), \(\zeta_1 \neq \xi\) or \(\zeta_2 \neq \xi\)} in Rule \ref{rule:vector-consequence},
      two instances of the Rule \ref{rule:vector-consequence} with different recipes \(\zeta\) (e.g. if \(u_1\) can be deducible with two different recipes) cannot be applied sequentially.
      Thus, at any given moment, there is at most one unsolved equality formula per extended constraint system of \(\Gamma\), hence the bound.
  \end{enumerate}

\subsection{Bounding the increase of second order terms}

  In Section \ref{sec:exp-mgs}, Proposition \ref{prop:evol-mgs}, we gave a bound on the increase of the size of most general solutions when applying the rule \eqref{rule:satisfiable}.
  We give here the arguments to extend to the other case distinction rules.
  For that it suffices to generalise this property to a more general set of substitution \(\Sigma\):

  \begin{definition}
    Let \(\C^e\) be an extended constraint system. Let \(\Sigma\) be a second-order substitution.
    We say that \(\Sigma \in \CompatibleSubs(\C^e)\) if \(\dom(\Sigma) \subseteq \vars[2](\Df(\C^e))\) and for all \(X \in \dom(\Sigma)\),
    \(X\Sigma \in \conseq(\Solved(\C^e) \cup \Df' \cup D_\Sigma)\) where
    \(\Df' = \{ X \dedfact u \in \Df(\C^e) \mid X \not\in \dom(\Sigma) \}\) and
    \(D_\Sigma = \{ X \dedfact x \mid x \text{ fresh and } X \in \vars[2](\Sigma) \setminus \vars[2](\C^e)\}\).
  \end{definition}

  Intuitively, \(\CompatibleSubs(\C^e)\) represents the recipe substitutions \(\Sigma\) that can be applied be applied to the constraint system \(\C^e\), i.e. \(\CApply{\Sigma}{\C^e}\),
  and such that the recipes in the of \(\Sigma\) would be consequence of \(\CApply{\Sigma}{\C^e}\).
  Note that \(\mgs(\C^e) \subseteq \CompatibleSubs(\C^e)\).

  By applying Proposition \ref{prop:trans-conseq}, we can show that:
  \begin{equation}
    \text{for all }\Sigma \in \CompatibleSubs(\C^e),
    |\measureNC(\CApply{\Sigma}{\C^e})| \leqslant |\measureNC(\C^e)|
    \label{term:compsubs}
  \end{equation}
  Note that in a set of symbolic processes two extended constraint systems \(\C^e_1, \C^e_2\) always have the same \emph{recipe structure} (Invariant \(\PredStruct\)),
  i.e. \(|\Phi(\C^e_1)| = |\Phi(\C^e_2)|\), \(\vars[2](\C^e_1) = \vars[2](\C^e_2)\) and
  \(\{ \xi \mid (\xi \dedfact u) \in \Solved(\C^e_1)\} = \{ \xi \mid (\xi \dedfact u) \in \Solved(\C^e_2)\}\).
  Thus, we deduce that \(\CompatibleSubs(\C^e_1) = \CompatibleSubs(\C^e_2)\).
  Therefore, we can conclude that for any simplification and case distinction rules, \(|\measureNC(\C^e)|\) never increase for all extended constraint systems in a set of extended symbolic processes.


\section{Proofs of complexity lower bounds}
\subsection{Advanced winning strategies}
\label{app:strategies}

  Before starting the proofs, we present some characterizations of observational (in)equivalence in order to make the incoming proofs easier to handle.

  \begin{remark}
    The results of this section (\ref{app:strategies}) also apply to the extended semantics of Section \ref{sec:tools encoding}.
  \end{remark}

  \paragraph{For the defender}
  \label{app:strategies def}

  The transitions of the semantics which are deterministic and silent are not essential to equivalence proofs as they do not interfere substantially with them. We introduce below a refined proof technique to rule them out.

  \begin{definition}[simplification]
    \label{def:simplification}
    A multiset of closed plain processes \(\S\) is silent in an extended process \(\process \P \Phi\) when for all transitions \(\process {\P \cup \S} \Phi \cstep \alpha \process \Q {\Phi'}\),
    it holds that \(\Q = \P' \cup \S\) with \(\process \P \Phi \cstep \alpha \process {\P'} {\Phi'}\) and \(\S\) silent in \(\process {\P'} {\Phi'}\). Then we define \(\silentstep\) (simplification relation) the relation on extended processes defined by the following inference rules:
    \begin{mathpar}
      \inferrule
        {\mbox{\(\S\) silent in \(\process \P \Phi\)}}
        {\process {\P \cup \S} \Phi \silentstep \process \P \Phi}[(S-sil)]\label{rule:s-sil}

      \inferrule
        {c \in \Nall \\ \msg t \\ c \notin \names{\P, \Phi}}
        {\process {\P \cup \multi {\OutP c t.P, \InP c x.Q}} \Phi \silentstep \process {\P \cup \multi {P, Q\{x \mapsto t\}}} \Phi}[(S-comm)]\label{rule:s-comm2}

      \inferrule
        {A \cstep \tau B ~ \mbox{by rules \textsc{Null}, \textsc{Par}, \textsc{Then}, \textsc{Else}}}
        {A \silentstep B}[(S-npte)]\label{rule:s-npte}
    \end{mathpar}
  \end{definition}

  In other words, we write \(A \silentstep B\) when \(B\) is obtained from \(A\) by removing some silent process or applying a deterministic (in the sense of the confluence lemma below) instance of the transition relation \(\cstep \tau\). We call \(\silentpistep\) the restriction of \(\silentstep\) to the rule \nameref{rule:s-npte}. Their reflexive transitive closures are denoted \(\silent\) and \(\silentpi\) respectively as usual.

  \begin{lemma}
    \label{lem:silent strong confluence}
    If \(A \silentstep B\) (by some rule \(\rho_{\mathsf{sil}}\) of the definition of \(\silentstep\)) and \(A \cstep \alpha C\) (by some rule \(\rho_{\mathsf c}\) of the semantics), then either \(B = C\) and \(\alpha = \silent\), or there exists \(D\) such that \(C \silentstep D\) (by rule \(\rho_{\mathsf{sil}}\)) and \(B \cstep \alpha D\) (by rule \(\rho_{\mathsf c}\)).
  \end{lemma}

  \begin{proof}
    We make a case analysis on the rule used to obtain the reduction \(A \silentstep B\):

    \begin{itemize}
      \item \case[: by rule \nameref{rule:s-sil}] 1

      Then we write \(A = \process {\P \cup \S} \Phi\), \(B = \process \P \Phi\).
      By definition of silent processes, the reduction \(A \cstep \alpha C\) hence gives \(C = \process {\P' \cup \S} {\Phi'}\) where \(\process \P \Phi \cstep \alpha \process {\P'} {\Phi'} = D\) and \(\S\) silent in \(D\). In particular \(D\) gives the expected conclusion.

      \item \case[: by rule \nameref{rule:s-comm2} or \nameref{rule:s-npte}] 2

      Then either \(B = C\) and the conclusion is immediate, or \(B \neq C\) and a quick analysis of the rules of the semantics gives \(\P,Q_B,Q_B',Q_C,Q_C',\Phi,\Phi'\) such that:
      \begin{align*}
        A & = \process {\multi {Q_B,Q_C} \cup \P} \Phi &
        B & = \process {\multi {Q_B',Q_C} \cup \P} \Phi &
        C & = \process {\multi {Q_B,Q_C'} \cup \P} {\Phi'} \\
        & & \mathrm{}~& \process {\multi {Q_B}} \Phi \silentstep \process {\multi {Q_B'}} \Phi &
        \mathrm{}~& \process {\multi {Q_C}} \Phi \cstep \alpha \process {\multi {Q_C'}} {\Phi'}
      \end{align*}
      and we conclude by choosing \(D = \process {\multi {Q_B',Q_C'} \cup \P} {\Phi'}\). \qedhere
    \end{itemize}
  \end{proof}

  \begin{corollary}
    \label{cor:silent confluence}
    If \(A \silent B\) and \(A \cstep \alpha C\) then either \(B \silent C\) and \(\alpha = \silent\), or there exists \(D\) such that \(C \silent D\) and \(B \cstep \alpha D\).
  \end{corollary}

  \begin{proof}
    By a straightforward induction on the number of steps of the reduction \(A \silent B\).
  \end{proof}

  \begin{corollary}
    \label{cor:npte convergence}
    \(\silentpistep\) is convergent.
  \end{corollary}

  \begin{proof}
    The termination of \(\silentpistep\) follows from the termination of the whole calculus. As for the local confluence (which sufficies by Newmann's lemma), we observe that by Lemma \ref{lem:silent strong confluence}, if \(A \silentpistep B\) and \(A \silentpistep C\) then either \(B = C\), or there is \(D\) such that \(B \silentpistep D\) and \(C \silentpistep D\): in particular, \(B \silentpi E\) and \(C \silentpi E\) for some \(E \in \{C,D\}\).
  \end{proof}

  In particular, all extended processes \(A\) have a unique normal form w.r.t. \(\silentpistep\) which will be written \(\procnorm A\). This notation is lifted to multiset of processes, writing \(\procnorm {\P}\) (which is consistent since \(\silentpistep\) does not modify the frame). With all of this, we eventually gathered all the ingredients to introduce our characterization of bisimilarity:

  \begin{definition}[bisimulation up to \(\silentstep\)]
    A symmetric relation \(\bisim\) on extended processes is then said to be bisimulation up to \(\silentstep\), or a bisimulation up to simplification, when:
    \begin{itemize}
      \item \(\bisim~\subseteq~\StatEq\);
      \item for all extended processes \(A,B\) such that \(A\bisim B\), and for all transitions \(A\cstep\alpha A'\), there exists \(B\Cstep\alpha B'\) such that \(A' \silent \bisim \silentrev B'\).
    \end{itemize}
  \end{definition}

  \begin{proposition}
    \label{prop:bisim up to}
    For all extended processes \(A\) and \(B\), \(A\LabBis B\) \textit{iff} there exists a bisimulation up to simplification \(\bisim\) such that \(A \silent \bisim \silentrev B\).
  \end{proposition}

  \begin{proof}
    The forward implication follows from the fact that \(\LabBis\) is a bisimulation up to simplification (by reflexivity of \(\silent\)). For the converse, let us consider \(\bisim\) a bisimulation up to \(\silentstep\) and prove that it is contained in \(\LabBis\). In order to do that, it sufficies to show that:
    \begin{itemize}
      \item \(\silent \bisim \silentrev\) is symmetric;
      \item \((\silent \bisim \silentrev)~ \subseteq~ \StatEq\);
      \item for all extended processes \(A,B\) such that \(A \silent \bisim \silentrev B\), if \(A\cstep\alpha A'\) then there exists \(B\Cstep\alpha B'\) such that \(A' \silent \bisim \silentrev B'\).
    \end{itemize}
    These three properties indeed justify that \((\silent \bisim \silentrev)~ \subseteq~ \LabBis\) by definition, hence the expected conclusion as \(\bisim \subseteq \silent \bisim \silentrev\) by reflexivity of \(\silent\). Yet it appears that the first two points directly follows from the properties of \(\bisim\) and the reflexivity of \(\silent\), and we thus only need to prove the third point. Let us therefore consider the following hypotheses and notations:
    \begin{align*}
      A \silent C & ~\bisim~ D \silentrev B &
      A & \cstep \alpha A'
    \end{align*}
    and let us exhibit \(B'\) such that \(B \Cstep \alpha B'\) and \(A' \silent \bisim \silentrev B'\). Let us consider the two cases induced by the application of Corollary \ref{cor:silent confluence}:

    \begin{itemize}
      \item \case 1 \(A' \silent C\) and \(\alpha = \silent\)

      Then we can choose \(B' = B\).

      \item \case 2 there exists \(C'\) such that \(A' \silent C'\) and \(C \cstep \alpha C'\)

      Consequently, since \(\bisim\) is a bisimulation up to simplification, there is \(D'\) such that \(D \Cstep \alpha D'\) and \(C' \silent \bisim \silentrev D'\). Then we remark that for all extended processes \(B_1, B_2, B_3\):
      \begin{itemize}
        \item a transition \(B_1 \silentstep B_2\) with rules \nameref{rule:s-npte} or \nameref{rule:s-comm2} implies \(B_1 \cstep \tau B_2\);
        \item if \(B_1 \silentstep B_2\) with rule \nameref{rule:s-sil} and \(B_2 \cstep \beta B_3\), then \(B_1 \cstep \beta \silentstep B_3\).
      \end{itemize}
      In particular since \(B \silent D \Cstep \alpha D'\), we have \(B \Cstep \alpha D'' \silent D'\) for some \(D''\). Hence the conclusion by choosing \(B' = D''\). \qedhere
    \end{itemize}
  \end{proof}

  \paragraph{For the attacker}
  \label{app:strategies atk}
  When taking the negation of labelled bisimilarity, we essentially obtain a set of rules for a game whose states are pairs of processes \((A,B)\): an attacker selects a transition and a defender answers by selecting a equivalently-labelled sequence of transitions in the other process.

  \begin{definition}[labelled attack]
    \label{def:labelled attack}
    A relation \(\disim\) on extended processes is called a labelled attack when for all \(A,B\) such that \(A\disim B\), it holds that:
    \begin{enumerate}
      \item either: \(A\not\StatEq B\)
      \item or: \(\exists A\cstep\alpha\Cstep\tr A',~\forall B\Cstep{\alpha.\tr} B',~A'\disim B'\)
      \item or: \(\exists B\cstep\alpha\Cstep\tr B',~\forall A\Cstep{\alpha.\tr} A',~A'\disim B'\)
    \end{enumerate}
  \end{definition}

  Note that labelled attacks are not the direct translation of the above intuition since they allow the attacker to choose several transitions in a row; this intuitively entails no loss of generality since it is equivalent to the attacker selecting some transitions non-adaptatively (i.e. independently of the answer of the defender). Here is the formal statement of correctness:

  \begin{proposition}
    \label{prop:eqobs labelled disim}
    For all extended processes \(A\) and \(B\), \(A\not\LabBis B\) \textit{iff} there exists a labelled attack \(\disim\) such that \(A\disim B\).
  \end{proposition}

  \begin{proof}
    The forward implication is immediate since \(\not\LabBis\) is a labelled attack (we can even choose \(\tr=\epsilon\) everytime). Let then \(\disim\) be a labelled attack such that \(A\disim B\) and let us prove that \(\disim\subseteq~\not\LabBis\). More precisely, we prove that \(\disim \subseteq ~ \disim' \subseteq ~ \not\LabBis\) for some relation \(\disim'\). We will construct \(\disim'\) in such a way that for all \(A,B\) extended processes, \(A\disim'B\) entails:
    \begin{itemize}
      \item[\((i)\)] either: \(A\not\StatEq B\);
      \item[\((ii)\)] or: \(\exists A\cstep\alpha A',~\forall B\Cstep\alpha B',~A'\disim B'\);
      \item[\((iii)\)] or: \(\exists B\cstep\alpha B',~\forall A\Cstep\alpha A',~A'\disim B'\)
    \end{itemize}
    The inclusion \(\disim'\subseteq~\not\LabBis\) is indeed clear if this property is verified, hence the expected conclusion provided such a relation \(\disim'\). We concretely define it as the smallest relation on extended processes saturated by the following inference rules:
    \begin{mathpar}
      \inferrule
      {A\disim B}
      {A\disim' B}[(Axiom)]\label{rule:axiom}

      \inferrule
      {A\disim'B\\A\cstep\alpha A'\cstep{\alpha'}\Cstep\tr A''\\\forall B\Cstep{\alpha.\alpha'.\tr}B'',A''\disim'B''\\B\Cstep\alpha B'}
      {A'\disim'B'}[(Dec-L)]\label{rule:split-l}

      \inferrule
      {A\disim'B\\B\cstep\alpha B'\cstep{\alpha'}\Cstep\tr B''\\\forall A\Cstep{\alpha.\alpha'.\tr}A'',A''\disim'B''\\A\Cstep\alpha A'}
      {A'\disim'B'}[(Dec-R)]\label{rule:split-r}
    \end{mathpar}
    In particular, note that \(\disim\subseteq\disim'\) thanks to the rule \nameref{rule:axiom}.
    As for the two other rules \nameref{rule:split-l} and \nameref{rule:split-r}, they intuitively decompose sequences \(\cstep\alpha\Cstep\tr\) into atomic transitions in order to switch from points {\it 2.} or {\it 3.} of Definition \ref{def:labelled attack} to points \((ii)\) or \((iii)\).

    Let then \(A\) and \(B\) be two extended processes such that \(A\disim'B\). We consider a proof-tree of \(A\disim'B\) in the inference system above and perform a case analysis on the rule at its root:

    \begin{itemize}

      \item \case[\nameref{rule:axiom}] 1 \(A\disim B\).

      As \(\disim\) is a labelled attack, we apply the case analysis of Definition \ref{def:labelled attack}:

      \begin{itemize}
        \item \case {1.a} \(A\not\StatEq B\).

        Then \((i)\) is satisfied.

        \item \case {1.b} there exists \(A\cstep{\alpha}A'\Cstep\tr A''\) such that \(A''\disim B''\) for all \(B\Cstep{\alpha.\tr}B''\).

        In particular, keeping in mind that \(\disim\subseteq\disim'\) due to the rule \nameref{rule:axiom}, we have \(A''\disim'B''\) for all \(B\Cstep{\alpha.\tr}B''\). Let us then show that the transition \(A\cstep{\alpha}A'\) satisfies \((ii)\).
        We therefore have to show that \(A'\disim'B'\) for all \(B\Cstep\alpha B'\). If \(A'=A''\) then the result follows from the hypothesis.
        Otherwise let us write \(A\cstep{\alpha}A'\cstep{\alpha'}\Cstep{\tr'}A''\) where \(\alpha'.{\tr'}=\tr\) and the rule \nameref{rule:split-l} justifies that \(A'\disim'B'\) for all \(B\Cstep{\alpha}B'\).

        \item \case {1.c} there exists \(B\cstep{\alpha}B'\Cstep\tr B''\) such that \(B''\disim A''\) for all \(A\Cstep{\alpha.\tr}A''\).

        Analogous, targeting \((iii)\) instead of \((ii)\) and replacing \nameref{rule:split-l} by \nameref{rule:split-r}.
      \end{itemize}

      \item \case[\nameref{rule:split-l}] 2 there are \(A_0\), \(A_1\), \(A_2\), \(B_0\), \(B_2\), \(\alpha\), \(\alpha'\), \(\tr\), such that \(A_0\disim'B_0\), \(B_0\Cstep\alpha B\),
      \(A_0\cstep\alpha A\cstep{\alpha'}A_1\Cstep\tr A_2\), and \(\forall~B_0\Cstep{\alpha.\alpha'.\tr}B_2,A_2\disim'B_2\).

      Let us show that the transition \(A\cstep{\alpha'}A_1\) satisfies \((ii)\). We therefore have to show that \(A_1\disim'B_1\) for all \(B\Cstep{\alpha'} B_1\). If \(A_1=A_2\) then the result follows from the hypothesis.
      Otherwise we write \(A\cstep{\alpha'}A_1\cstep{\alpha''}\Cstep{\tr'}A_2\) where \(\alpha''.\tr'=\tr\) and the rule \nameref{rule:split-l} justifies that \(A_1\disim'B_1\) for all \(B\Cstep\alpha B_1\).

      \item \case[\nameref{rule:split-r}] 3 Analogous to case 2.\qedhere
    \end{itemize}
  \end{proof}

\subsection{Correctness of the encodings (Section \ref{sec:tools encoding})}
\label{app:tools encoding}

  Now we prove that the translation \(\sem\cdot\) of the extended semantics is correct:

  \begin{lemma}
    \label{lem:encodings correct}
    Let \(\TraceEq^+\) and \(\LabBis^+\) be the notions of trace equivalence and labelled bisimilarity over the extended calculus (the flag \(^+\) being omitted outside of this lemma). For all extended processes \(A = \process \P \Phi\), the translation \(\sem A = \process {\sem \P} \Phi = \process{ \multi {\sem P \mid P \in \P} }\Phi\) can be computed in polynomial time, \(A\TraceEq^+\sem A\) and \(A\LabBis^+\sem A\).
  \end{lemma}

  But first of all, a (trivial) observation about the free variables of a translated process:

  \begin{lemma}
    For all plain processes \(P\) and all first-order substitution \(\sigma\), \(\sem{P\sigma}=\sem P\sigma\).
  \end{lemma}


  We will use this lemma implicitly in the remaining of this section. Besides, as we have the inclusion of relations \(\LabBis~\subseteq~\TraceEq\), we only need to prove the observational-equivalence statement of Lemma \ref{lem:encodings correct}. We recall that we use notations \(\procnorm A\) and \(\procnorm \P\) to refer to normal forms w.r.t. \(\silentpistep\) (see Corollary \ref{cor:npte convergence}).

  \begin{proposition}
    We consider \(\bisim\) the symmetric closure of:
    \[\{ (C, \procnorm {\sem C}) ~|~ \mbox{\(C\) extended process such that \(C = \procnorm C\)} \}\]
    \(\bisim\) is a bisimulation up to simplification.
  \end{proposition}

  \begin{proof}
    \(\bisim\) is symmetric by definition and is trivially included in \(\StatEq\). Let then \((A,B) \in \bisim\) and \(A \cstep \alpha A'\) and let us exhibit \(B'\) such that \(B \Cstep \alpha B'\) and \(A' \bisim B'\). We perform a case analysis on the rule triggerring the transition \(A \cstep \alpha A'\):

    \begin{itemize}
      \item \case[(rules \textsc{Null}, \textsc{Par}, \textsc{Then}, \textsc{Else})] 1

      This case cannot arise as \(A\) is in normal form w.r.t. \(\silentpistep\) by definition of \(\bisim\).

      \item \case[(rule \textsc{In})] 2 \(\alpha = \InP \xi \zeta\) for some \(\xi, \zeta \in \termset(\sig, \sig_0 \cup \dom \Phi)\) and:
      \begin{align*}
        A & = \process {\P \cup \multi {\InP u x.P}} \Phi  & & \mbox{with \(\msg u\), \(\msg {\xi \Phi}\) and \(\xi \Phi \norm = u \norm\)}\\
        A' & = \process {\P \cup \multi {P\{ x \mapsto \zeta \Phi \norm \}}} \Phi & & \mbox{with \(\msg {\zeta \Phi}\)}
      \end{align*}
      Then, by a case analysis on the hypothesis \(A \bisim B\):
      \begin{itemize}
        \item \case {2.a} \(B = \procnorm {\sem A}\)

        Then we can write:
        \begin{align*}
          B & = \process {\procnorm {\sem \P} \cup \multi {\InP u x.\sem P}} \Phi
        \end{align*}
        and we conclude by remarking that \(B \cstep {\InP \xi \zeta} B' = \process {\procnorm {\sem \P} \cup \multi {\sem P\{ x \mapsto \zeta \Phi \norm \}}} \Phi\) and:
        \[A' \silent \procnorm {A'} ~\bisim~ \procnorm {\sem {\procnorm {A'}}}
        = \procnorm {\sem {A'}} = \process {\procnorm {\sem \P} \cup \procnorm {\multi {\sem P\{ x \mapsto \zeta \Phi \norm \}}}} \Phi \silentrev B'\]

        \item \case {2.b} \(A = \procnorm {\sem B}\) (and \(B = \procnorm B\))

        Note that \(u = \xi \Phi \norm\) cannot be one of the names introduced by the translation \(\sem \cdot\): these names can indeed not appear in \(\Phi\) since they are chosen private and fresh and since the semantics cannot introduce new private names. In particular we can write:
        \begin{align*}
          B & = \process {\Q \cup \multi {\InP u x.Q}} \Phi  & & \mbox{with \(\sem Q = P\) and \(\procnorm {\sem \Q} = \P\)}
        \end{align*}
        and we conclude by writing \(B \cstep {\InP \xi \zeta} B' = \process {\Q \cup \multi {Q\{ x \mapsto \zeta \Phi \norm \}}} \Phi\) and:
        \[ A' \silent \process {\procnorm {\sem \Q} \cup \procnorm {\multi {\sem Q \{ x \mapsto \zeta \Phi \norm \}}}} \Phi = \procnorm {\sem {B'}} = \procnorm{\sem {\procnorm {B'}}} ~\bisim~ \procnorm {B'} \silentrev B'\]
      \end{itemize}

      \item \case[(rule \textsc{Out})] 3 \(\alpha = \OutP \xi {\ax_n}\) for some \(\xi, \in \termset(\sig, \sig_0 \cup \dom \Phi)\), \(\ax_n \in \AX\) and:
      \begin{align*}
        A & = \process {\P \cup \multi {\OutP u t.P}} \Phi  & & \mbox{with \(\msg u\), \(\msg t\), \(\msg {\xi \Phi}\) and \(\xi \Phi \norm = u \norm\)}\\
        A' & = \process {\P \cup \multi P} {\Phi'} & & \mbox{where \(\Phi' = \Phi \cup \{ \ax_n \mapsto t \norm\}\) and \(n = |\Phi| + 1\)}
      \end{align*}
      Then, by a case analysis on the hypothesis \(A \bisim B\):
      \begin{itemize}
        \item \case {3.a} \(B = \procnorm {\sem A}\)

        Then we can write:
        \begin{align*}
          B & = \process {\procnorm {\sem \P} \cup \multi {\OutP u t.\sem P}} \Phi
        \end{align*}
        and we conclude by remarking that \(B \cstep {\OutP \xi {\ax_n}} B' = \process {\procnorm {\sem \P} \cup \multi {\sem P}} \Phi\) and:
        \[A' \silent \procnorm {A'} ~\bisim~ \procnorm {\sem {\procnorm {A'}}}  = \procnorm {\sem {A'}} = \process {\procnorm {\sem \P} \cup \procnorm {\multi {\sem P}}} \Phi \silentrev B'\]

        \item \case {3.b} \(A = \procnorm {\sem B}\) (and \(B = \procnorm B\))

        For the same reason as in case 2.b, we can write:
        \begin{align*}
          B & = \process {\Q \cup \multi {\OutP u t.Q}} \Phi  & & \mbox{with \(\sem Q = P\) and \(\procnorm {\sem \Q} = \P\)}
        \end{align*}
        and we conclude by writing \(B \cstep {\OutP \xi {\ax_n}} B' = \process {\Q \cup \multi Q} \Phi\) and:
        \[ A' \silent \process {\procnorm {\sem \Q} \cup \procnorm {\multi {\sem Q}}} \Phi = \procnorm {\sem {B'}} = \procnorm{\sem {\procnorm {B'}}} ~\bisim~ \procnorm {B'} \silentrev B'\]
      \end{itemize}

      \item \case[(rule (\textsc{Comm}))] 4 \(\alpha = \silent\) and:
      \begin{align*}
        A & = \process {\P \cup \multi {\OutP u t.P, \InP v x.Q}} \Phi & & \mbox{with \(\msg u\), \(\msg v\), \(\msg t\) and \(u \norm = v \norm\)}\\
        A' & = \process {\P \cup \multi {P,Q\{ x \mapsto t \}}} \Phi
      \end{align*}
      Then, by a case analysis on the hypothesis \(A \bisim B\):
      \begin{itemize}
        \item \case {4.a} \(B = \procnorm {\sem A}\)

        Then we can write:
        \begin{align*}
          B & = \process {\procnorm {\sem \P} \cup \multi {\OutP u t.\sem P, \InP v x.\sem Q}} \Phi
        \end{align*}
        and we conclude by remarking that \(B \cstep \tau B' = \process {\procnorm {\sem \P} \cup \multi {\sem P, \sem Q\{ x \mapsto t\}}} \Phi\) and:
        \[A' \silent \procnorm {A'} ~\bisim~ \procnorm {\sem {\procnorm {A'}}}  = \procnorm {\sem {A'}} =
        \process {\procnorm {\sem \P} \cup \multi {\procnorm {\sem P}, \procnorm {\sem Q\{ x \mapsto t\}}}} \Phi \silentrev B'\]

        \item \case {4.b} \(A = \procnorm {\sem B}\) (and \(B = \procnorm B\)) and a term \(w\) such that \(w \norm = u \norm\) appears in \(B\) (syntactically)

        In particular \(\norm u\) is not a fresh name introduced by the translation \(\sem \cdot\) and we can therefore write:
        \begin{align*}
          B & = \process {\P' \cup \multi {\OutP u t.P', \InP v x.Q'}} \Phi & & \mbox{with \(\procnorm {\sem {\P'}} = \P\), \(\sem {P'} = P\) and \(\sem {Q'} = Q\)}
        \end{align*}
        and we conclude by writing \(B \cstep \tau B' = \process {\P' \cup \multi {P', Q'\{ x \mapsto t\}}} \Phi\) and:
        \[A' \silent \process {\procnorm {\sem {\P'}} \cup \multi {\procnorm {\sem {P'}}, \procnorm {\sem {Q'}\{ x \mapsto t\}}}} \Phi =
        \procnorm {\sem {B'}} = \procnorm{\sem {\procnorm {B'}}} ~\bisim~ \procnorm {B'} \silentrev B'\]

        \item \case {4.c} \(A = \procnorm {\sem B}\) (and \(B = \procnorm B\)) and there exists no term \(w\) appearing in \(B\) such that \(w \norm = u \norm\) (syntactically)

        Then \(u\) is a fresh name introduced by \(\sem \cdot\). We consider the two disjoint cases where it was introduced for the translation of a sum or a circuit:

        \begin{itemize}
          \item If \(u \in \Nall\) is a fresh name generated in order to translate a sum of \(B\), or rephrased more formally:
          \begin{align*}
            B & = \process {\Q \cup \multi {P' + Q'}} \Phi\\
            A & = \process {\procnorm {\sem \Q} \cup \multi {\OutP u u, \InP u x.\sem {P'}, \InP u x.\sem {Q'}}} \Phi & & \mbox{where \(x \in \X[1]\) but \(x \notin \vars {P',Q'}\)}\\
            A' & = \process {\procnorm {\sem \Q} \cup \multi {\sem R, \InP u x.\sem S}} \Phi & & \mbox{where \(R,S \in \multi {P', Q'}\), \(R \neq S\)}
          \end{align*}
          We let \(B' = \process {\Q \cup \multi R} \Phi\) and remark that \(B \cstep \tau B'\) by the rule \ref{rule:choice} (if \(R = P'\) and \(S = Q'\), or \(R = Q'\) and \(S = P'\)).
          Besides, let us observe that the name \(u\) does not appear in \(\procnorm {\sem \Q}\), \(\sem R\) nor \(\Phi\) by construction of \(\sem \cdot\) and that \(\InP u x.\sem S\) is therefore easily seen to be silent in \(A'' = \process {\procnorm {\sem \Q} \cup \multi {\sem R}} \Phi\).
          In particular it entails that \(A' \silentstep A''\) by the rule \nameref{rule:s-sil} which gives the conclusion:
          \[A' \silentstep A'' \silent \process {\procnorm {\sem \Q} \cup \procnorm {\multi {\sem R}}} \Phi = \procnorm {\sem {B'}} = \procnorm {\sem {\procnorm {B'}}} ~\bisim~ \procnorm {B'} \silentrev B'\]

          \item \(u \in \Nall\) is a fresh name generated in order to translate a \(\guessBinary x\): this case can be handle analogously to the previous one.

          \item If \(u \in \Nall\) is a fresh name generated in order to translate a circuit of \(B\), or formally:
          \begin{align*}
            B & = \process {\Q \cup \multi {\evalFormula {\vec x} {\Gamma(\vec b)}. P'}} \Phi\\
            A & = \process {\procnorm {\sem \Q} \cup \multi {\sem {\evalFormula {\vec x} {\Gamma(\vec b)}.P'}}} \Phi
          \end{align*}
          We call \((c_i)_i\) the private fresh names introduced by the translation \(\sem {\evalFormula {\vec x} {\Gamma(\vec b)}.P'}\), stressing that none of them appears in \(\procnorm {\sem \Q}\), \(\sem {P'}\) nor \(\Phi\). Here \(u \in \{c_i\}_i\) and we can therefore write:
          \begin{align*}
            A' & = \process {\procnorm {\sem \Q} \cup \multi {Q'}} \Phi & &
            & \mbox{where \(\process {\multi {\sem {\evalFormula {\vec x} {\Gamma(\vec b)}.P'}}} \Phi \cstep \alpha \process {\multi {Q'}} \Phi\)}
          \end{align*}
          If the sequence \(\vec b\) contains a term which is not a message or does not reduce to a boolean, then one easily obtain that \(\process {\multi {Q'}} \Phi \silentpi \process \S \Phi\) where \(\S\) is silent in \(\process {\procnorm {\sem \Q}} \Phi\) (for that we assume, w.l.o.g. that each input of \(\Gamma\) goes through at least one gate).
          Hence since \(\multi {\evalFormula {\vec x} {\Gamma(\vec b)}. P'}\) is also silent in \(\process \Q \Phi\), we conclude with \(B' = \process \Q \Phi\).

          Otherwise assume that \(\msg {\vec b}\) and \(\vec b \norm \subseteq \B\). Then one easily obtain by induction on the number of gates of \(\Gamma\) that \(\process {\multi {Q'}} \Phi \silentpi \process {\multi {\sem {P'}\{\vec x \mapsto \Gamma(\vec b)\}}} \Phi\) and we conclude by choosing
          \(B' = \process {\Q \cup \multi {P'\{\vec x \mapsto \Gamma(\vec b)\}}} \Phi\).
        \end{itemize}

        \item \case[(rules \ref{rule:choice}, \ref{rule:choose-0}, \ref{rule:choose-1} or \ref{rule:valuate})] 5

        The arguments of each of these cases are analogous to priorly-met subcases. \qedhere
      \end{itemize}
    \end{itemize}
  \end{proof}

  In particular note that \(A \silent \procnorm A ~\bisim~ \procnorm{ \sem {\procnorm A}} = \procnorm{\sem A} \silentrev \sem A\) for all extended processes \(A\), hence Lemma \ref{lem:encodings correct}.

  \subsection{Reductions in the pure calculus} \label{app:lower-pure}

    We now formalise and prove the correctness of the reductions intuited in Section~\ref{sec:lower-pure}.

    \paragraph{For trace equivalence}
    We define the processes \(P(t)\), \(A\) and \(B\) as follows.
    \begin{align*}
      P(t) & \eqdef \InP c{\vec x}.~\guessBinary{\vec{y}}.~\evalFormula{v}{\varphi(\vec x ,\vec y)}.~\OutP c t\\
      A & \eqdef P(v)+P(\1)\\
      B & \eqdef P(\0)+P(\1)
    \end{align*}

    \begin{proposition}[Reduction for trace equivalence]
      \(A \TraceEq B\) \textit{iff} \(\forall \vec{x}.\exists \vec{y}. \varphi(\vec x,\vec y) = \0\).
    \end{proposition}

    \begin{proof}
      We do the proof by double implication.
      \begin{itemize}
        \item[\(\Rightarrow\)] Suppose that \(A \not \TraceEq B\). By a quick case analysis, we obtain \(B \cstep \epsilon \process {\multi {P(0)}} \emptyset \Cstep \tr \process \emptyset {\{\ax_1 \mapsto \0\}}\)
        where \(\tr = \InP c {\vec t}. \OutP c {\ax_1}\) for some messages \(\vec t\), such that for all reduction \(A \Cstep \tr \process \C \Phi\) the frames \({\{\ax_1 \mapsto \0\}}\) and \(\Phi\) are not statically equivalent. In particular, for all \(\vec y \subseteq \B\),
        by choosing \(\Phi = \{\ax_1 \mapsto \varphi(\vec t,\vec y)\}\) reachable from \(A\), we obtain \(\varphi(\vec t, \vec y) \neq \0\), hence the result.
    
        \item[\(\Leftarrow\)] Conversely, suppose that exists exists \(\vec x \subseteq \B\) such that \(\varphi(\vec x, \vec y) = \1\) for all \(\vec y \subseteq \B\). Then the trace \(B \Cstep \epsilon \process \emptyset {\{\ax_1 \mapsto \0\}}\) cannot be matched in \(A\) and therefore \(A \not \TraceEq B\). \qedhere
      \end{itemize}
    \end{proof}

    \paragraph{For simulations}
    We recall that a graphical depiction of the processes has been provided in Section~\ref{sec:lower-pure}.
    We fix a family of private channels \((c_P)_P\subseteq\Nall\) indexed by processes \(P\) which will be used to simulate instructions \(\call P\). We use a shortcut \(\OutP d {\args t p}\) for an indexed sequence of terms \((t_i)_i\) to denote the sequence of outputs:
    \[\OutP d {\args t p}\eqdef\OutP d{t_1}\ldots\OutP d{t_p}\]
    and a similar notation for sequences of inputs. Then the {\tt Goto} feature is implemented as follows, allowing for passing and receiving program states through parallel processes:
    \begin{align*}
      \call{A_i}\eqdef~ &\OutP{c_{A_i}}{\args x i,\args y i} &
      \call{B_i}\eqdef~ &\OutP{c_{B_i}}{\args x i,\args y i} \\
      \getEnv{A_i}.P\eqdef~ &\InP{c_{A_i}}{\args x i,\args y i} &
      \getEnv{B_i}.P\eqdef~ &\InP{c_{B_i}}{\args x i,\args y i} \\
    \end{align*}

    Formally the processes \(A_i\), \(B_i\) and \(D_i\) are defined below. We stress out that \(A\) and \(B\) are closed (as required) but that \(A_i\), \(B_i\) and \(D_i\) are not in general. Fixing a public channel \(c\in\sig_0\), we write:
    \begin{align*}
      \forall i\leqslant n,~A_i\eqdef~&\InP c {x_i}.~\evalFormula{x_i}{x_i}.~D_i\\
      \forall i\leqslant n,~B_i\eqdef~&\InP c {x_i}.~\evalFormula{x_i}{x_i}.~(D_i+(\InP c {y_i}.~\evalFormula{y_i}{y_i}.~\call{B_{i+1}}))\\
      A_{n+1}\eqdef~&\evalFormula v {\varphi(\vec x,\vec y)}.~\OutP c v\\
      B_{n+1}\eqdef~&\OutP c \0\\
      D_i \eqdef~& \guessBinary{z_i}.~\InP c {y_i}.~\evalFormula{r_i}{(y_i=z_i)}.\\
      & ((\IfP~r_i=1~\ThenP~\call{A_{i+1}})\\
      & \quad\mid(\IfP~r_i=0~\ThenP~\InP c {y_i}.~\evalFormula{y_i}{y_i}.~\call{B_{i+1}}))
    \end{align*}

    As in the reduction for trace equivalence, the \(\guessBinary\alpha\) simulates non-deterministic choice among \(\B\); the construction \(\evalFormula\alpha\alpha\), which may seem useless, encodes the test \(\alpha \in \B\). Finally, we define \(A\) and \(B\) by putting the auxiliary processes in parallel and connecting the {\tt Goto}'s to the {\tt getEnv}'s: 
    \begin{align*}
      A \eqdef~& A_1\mid C & B\eqdef~&B_1\mid C & C\eqdef~&\displaystyle\prod_{i=2}^{n+1}(\getEnv{A_i}.A_i)~\mid~\prod_{i=2}^{n+1}(\getEnv{B_i}.B_i)
    \end{align*}
    \(A\) and \(B\) can be computed in time \(\mathcal O(n^2+|\varphi|)\) in a straightforward way.
    The formal proof of the reduction is detailed below.

    \begin{proposition}[Correctness of the reduction]
      The following statements are equivalent:
      \begin{enumerate}
        \item \label{it:lower-pure-bisim} \(A \LabBis B\)
        \item \label{it:lower-pure-simil} \(A \Simi B\)
        \item \label{it:lower-pure-formula} \(\forall x_1 \exists y_1 \ldots \forall x_n \exists y_n.~\varphi(x_1,\ldots, x_n,y_1,\ldots, y_n) = \0\)
      \end{enumerate}
    \end{proposition}

    \begin{proof}
      We recall that, again, our proof uses the advanced winning strategy framework presented in Section~\ref{app:strategies}.
      \begin{itemize}
        \item[\ref{it:lower-pure-bisim}\(\Rightarrow\)\ref{it:lower-pure-simil}]
        Follows from the inclusion \({\LabBis} \subset {\Simi}\).
        \item[\ref{it:lower-pure-formula}\(\Rightarrow\)\ref{it:lower-pure-bisim}]
        Suppose that \(\forall x_1\exists y_1 \ldots \forall x_n \exists y_n.~\varphi(x_1,\ldots, x_n,y_1,\ldots, y_n) = \1\). For convinience we use a notation for subprocess extraction: is \(\ell\) is a position of a process \(C\), then the subprocess of \(C\) at position \(\ell\) (which may not be closed) is denoted by \(C_{|\ell}\). Then writing in addition:
        \[
          C_i = \prod_{j = i + 1}^{n+1} (\getEnv {A_j}. A_j) \mid \prod_{j = i + 1}^{n+1} (\getEnv {B_j}.B_j)
        \]
        we define \(\bisim\) the smallest reflexive symmetric relation on closed extended processes such that:
        \begin{enumerate}[label=\emph{\arabic*.}]
          \item \({(A_i \mid C_i)}(\args x {i-1}, \args y {i - 1}) \bisim {(B_i \mid C_i)}(\args x {i-1}, \args y {i - 1})\)

          if
          \(\forall x_i \exists y_i \ldots \forall x_n \exists y_n, \varphi(\vec x, \vec y)\).

          \item \(({A_i}_{|\ell} \mid C_i)(\args x i, \args y {i - 1}) \bisim ({B_i}_{|\ell} \mid C_i)(\args x {i-1}, \args y {i - 1})\)

          if \(\ell \in \{\0,\0.\0\}\) and
          \(\exists y_i \ldots \forall x_n \exists y_n, \varphi(\vec x, \vec y)\).

          \item \(({B_i}_{|\0.\0.\1.\ell} \mid C_i)(\args x i, \args y i) \bisim ({D_i}_{|\0.\ell} \mid C_i)(\args x i, \args y i)\)

          if \(\ell \in \{\epsilon,\0\}\) and \(\forall x_{i+1} \exists y_{i+1} \ldots \forall x_n \exists y_n, \varphi(\vec x, \vec y)\).
        \end{enumerate}
        Then one can verify that \(\bisim\) is a bisimulation up to \(\silentstep\), and \(A \bisim B\) by hypothesis, hence the \(A \LabBis B\).
        \item[\ref{it:lower-pure-simil}\(\Rightarrow\)\ref{it:lower-pure-formula}]
        By contraposition, if we suppose that \(\exists x_1\forall y_1\ldots\exists x_n \forall y_n.~\varphi(x_1,\ldots, x_n,y_1,\ldots, y_n) = \1\), then one can define a labelled attack \(\disim\) (valid against simulation) such that \(B \disim A\).
        We omit the concrete construction as it is analogous to that of \(\bisim\) above; 
        all in all this gives the conclusion \(A \not \Simi B\). \qedhere
      \end{itemize}
    \end{proof}

  \subsection{Reduction in the full calculus}
  Before concretely proving the pending lemmas, let us introduce some notations and prove intermediary results about static equivalence.
  We fix a private nonce \(s\in\Nall\) and define the following frames given a protocol term \(t\):
  \begin{align*}
    \frameh t=~&\{\ax_1\mapsto\hfun(t,s),~\ax_2\mapsto\hfun(\1,s)\}\\
    \framen t=~&\{\ax_1\mapsto\hNode(t,s),~\ax_2\mapsto\hfun(\1,s)\}\\
    \frameb t=~&\{\ax_1\mapsto\hBool(t,s),~\ax_2\mapsto\hfun(\1,s)\}
  \end{align*}

  One shall observe that the whole deal with our reduction is about which instances of these three frames are reachable in which conditions. Hence first we prove a lemma investigating the static equivalence between some of them:

  \begin{lemma}
    \label{lem:nexp st-equiv}
    Let \(t\) be a message in normal form (\(t=t\norm\)). The following properties hold:
    \begin{itemize}
      \item[\((i)\)] \(\frameh t\StatEq\frameh\0\) \textit{iff} \(t\neq\1\)
      \item[\((ii)\)] \(\framen t\StatEq\frameh\0\) \textit{iff} \(\rootf(t)\neq\Node\)
      \item[\((iii)\)] \(\frameb t\StatEq\frameh\0\) \textit{iff} \(t\notin\B\)
    \end{itemize}
  \end{lemma}

  \begin{proof}
    We prove the three equivalences together by double implication.
    \begin{itemize}
      \item[(\(\Rightarrow\))] We prove the three properties by contraposition. We naturally proceed by exhibiting ground recipes \(\xi,\zeta\) witnessing the non-static-equivalence goal:

      \item[\sbt] \((i)\) We assume \(t=\1\) and we choose \(\xi=\ax_1\) and \(\zeta=\ax_2\): the conclusion follows from \(\xi\frameh t\norm=\hfun(\1,s)=\zeta\frameh t\norm\) and \(\xi\frameh\0\norm=\hfun(\0,s)\neq\hfun(\1,s)=\zeta\frameh\0\norm\).
      \item[\sbt] \((ii)\) We assume \(t=\Node(t_1,t_2)\) and we choose \(\xi = \TestNode (\ax_1)\): the conclusion follows from \(\msg{\xi\framen t}\) and \(\neg\msg{\xi\framen\0}\).
      \item[\sbt] \((iii)\) We assume \(t\in\{\0,\1\}\) and we choose \(\xi=\TestBool (\ax_1)\): the conclusion follows from \(\msg{\xi\frameb t}\) and \(\neg\msg{\xi\frameb\0}\).

      \item[(\(\Leftarrow\))] The key point is an observation about rewriting critical pairs (not specific to \(\R\)):
      \begin{itemize}
        \item[if:] \(u\) is a term;
        \item[if:] \(\sigma\) is a substitution such that for any \(m\in\im(\sigma)\) there exists no rule \(\ell\rightarrow r\) of \(\R\) such that \(m\) is unifiable with a subterm \(u\in\subterms\ell-\X\);
        \item[then:] for any rewriting sequence \(u\sigma\rightarrow^\star s\), it holds that \(s=u'\sigma\) for some \(u\rightarrow^\star u'\) (where \(u'\) is in normal form \textit{iff} \(s\) is in normal form. In particular \((u\sigma)\norm=(u\norm)\sigma\)).
      \end{itemize}

      One shall note that any frame \(\Phi\) investigated by the lemma (\(\frameh t\) when \(t\neq\1\),
      \(\framen t\) when \(\rootf t\neq\Node\) and \(\frameb t\) when \(t\notin\B\)) verifies the second hypothesis.
      As a consequence, if \(\xi\) is a ground recipe such that \(\axioms\xi\subseteq\dom(\Phi)\), then \(\msg{\xi\Phi}\) \textit{iff} \(\msg{\xi\frameh\0}\)
      (\textit{iff} \(\forall \zeta\in\subterms\xi,~\zeta\norm\in\termset(\sigc,\sig_0\cup\AX)\)). This settles the first item of the definition of static equivalence. As for the second item, let us fix two ground recipes \(\xi,\zeta\) such that \(\msg{\xi\Phi}\) and \(\msg{\xi\frameh\0}\).
      Let us then prove that \(\xi\Phi\norm= \zeta\Phi\norm\) \textit{iff} \(\xi\frameh\0\norm= \zeta\frameh\0\norm\) by induction on \((\xi\norm,\zeta\norm)\). Note that we will intensively (and implicitly) use the fact that \(\xi\Phi\norm=(\xi\norm)\Phi\) (same for \(\zeta\) and/or \(\frameh\0\)).

      \item[\sbt] \case 1 \(\xi\norm=f(\xi_1,\ldots,\xi_n)\) and \(\zeta\norm=f(\zeta_1,\ldots,\zeta_p)\) with \(f,g\in\sigc\).

      If \(f=g\) then the result follows from induction hypothesis and if \(f\neq g\) the conclusion is immediate (\(\xi\Phi\norm\neq\zeta\Phi\norm\) and \(\xi\frameh\0\norm\neq\zeta\frameh\0\norm\)).

      \item[\sbt] \case 2 \(\xi\norm\in\AX\) and \(\zeta\norm\in\AX\).

      If \(\xi=\zeta\) the conclusion is immediate and so is it when \(\xi\neq\zeta\) since \(\Phi(\ax_1)\neq\Phi(\ax_2)\) and \(\frameh\0(\ax_1)\neq\frameh\0(\ax_2)\).

      \item[\sbt] \case 3 \(\xi\norm\in\AX\) and \(\zeta\norm=f(\zeta_1,\ldots,\zeta_p)\) with \(f\in\sigc\).

      We argue that \(\xi\Phi\norm\neq\zeta\Phi\norm\) and \(\xi\frameh\0\norm\neq\zeta\frameh\0\norm\). Either of the two equalities being verified would indeed imply that \(s\in\{\zeta_2\Phi,\zeta_2\frameh\0\}\): this is impossible as \(\zeta_2\) is a ground recipe in normal form and \(s\in\Nall\) (one easily shows that \(\zeta_2\Phi\) and \(\zeta_2\frameh\0\) are either public names, constants, or termes of height 1 or more).

      As \(\msg{\xi\Phi}\) and \(\msg{\xi\frameh\0}\), the preliminary observation justifies that \(\xi\norm\) and \(\zeta\norm\) function symbols are all constructor. In particular no other cases than the three above need to be considered, which concludes the proof.\qedhere
    \end{itemize}
  \end{proof}

  With this lemma in mind, the proofs of Propositions \ref{prop:correctness checktree} and \ref{prop:correctness checksat} become quite straightforward:

  \propCorrectnessChecktree*

  \begin{proof}
    Let \(x\) be a message which is not a complete binary tree of height \(n\) whose leaves are booleans.

    \begin{itemize}[leftmargin=*]
      \item[] \case 1 there exists a position \(\vec\pi\in\B^\star\) such that \(|\vec \pi|=i\in\eint 0 {n-1}\) and \(\rootf{\recpos x{\vec\pi}}\neq\Node\).

      The result follows from Lemma \ref{lem:nexp st-equiv} after writing the following sequence of transitions:
      \begin{align*}
        \CheckTree x & \Cstep\epsilon \guessBinary{p_1,\ldots,p_i}.~\OutP c {\hNode(\recpos x{p_1\cdots p_i},s)}.~\OutP c {\hfun(\1,s)}\\
        & \Cstep\epsilon \OutP c {\hNode(\recpos x{\vec \pi},s)}.~\OutP c {\hfun(\1,s)}
      \end{align*}

      \item[] \case 2 there exists a position \(\vec\pi\in\B^n\) such that \(\recpos x{\vec\pi}\notin\B\).

      The result follows from Lemma \ref{lem:nexp st-equiv} after writing the following sequence of transitions:
      \begin{align*}
        \CheckTree x & \Cstep\epsilon \guessBinary{p_1,\ldots,p_n}.~\OutP c {\hNode(\recpos x{p_1\cdots p_n},s)}.~\OutP c {\hfun(\1,s)}\\
        & \Cstep\epsilon \OutP c {\hBool(\recpos x{\vec \pi},s)}.~\OutP c {\hfun(\1,s)} \qedhere
      \end{align*}
    \end{itemize}
  \end{proof}

  \propCorrectnessChacksat*

  \begin{proof}
    Let \(x\) be a complete binary tree of height \(n\) whose leaves are booleans, that is to say, a message such that \(\recpos x{\vec p}\in\B\) for all \(\vec p\in\B^n\). Naming \(x_0,\ldots,x_{2^n-1}\) the variables of \(\sem\Gamma_\varphi\) in this order, \(\val_x\) refers to the valuation mapping \(x_i\) to \(\recpos x {\vec p}\) where \(\vec p\) is the binary representation of \(i\) (of size \(n\) with padding head 0's).

    Let us now assume that \(\val_x\) does not satisfy \(\sem\Gamma_\varphi\). In particular there exists a clause of \(\sem\Gamma_\varphi\), say the \(i\)\textsuperscript{th} clause with \(i=\sum_{k=1}^m\pi_k2^{k-1}\), which is falsified by \(\val_x\). In particular, if the three variable of this clause are called \(x_1\), \(x_2\), \(x_3\) with respective negation bits \(b_1\), \(b_2\), \(b_3\), the following formula is evaluated to false (i.e. \(\0\)):
    \[\bigvee_{i=1}^3(~b_j=\val_x(x_j)~)\]
    Therefore, by choosing the sequence \(\pi_1,\ldots,\pi_m\) to instanciate the initial \(\guessBinary{p_1,\ldots,p_m}\) of \(\CheckSat x\), we obtain the following sequence of transitions, which concludes the proof:
    \begin{align*}
      \CheckSat x & \Cstep\epsilon \OutP c {\hfun(\0,s)}.~\OutP c {\hfun(\1,s)} \qedhere
    \end{align*}
  \end{proof}


  We finally gathered all the ingredients needed to prove the main lemma:

  \propCorrectnessCheckall*

  \begin{proof}
    We prove the result by double implication.
    \begin{itemize}[leftmargin=*]
      \item[\(\Rightarrow\)]
      Let us consider a valuation satisfying \(\varphi\) and let \(t\) be a message such that \(\val_t\) is equal to this valuation. 
      Then since the trace \(B \Cstep {\InP c t. \OutP c {\ax_1}. \OutP c {\ax_2}} \process \emptyset {\frameh \0}\) cannot be matched in \(A\)
      by Lemma \ref{lem:nexp st-equiv}, we obtain \(B \not \TraceIncl A\). Hence, since trace equivalence is the coarsest of the considered equivalence relations, we obtain the desired result.
      \item[\(\Leftarrow\)]
      By contraposition, let us suppose that \(\varphi\) is unsatisfiable, and let us prove that \(A \LabBis B\) (which implies all other equivalences).
      Let us consider \(\bisim\) the smallest reflexive symmetric relation on extended processes such that:
      \begin{enumerate}[label=\emph{\arabic*.}]
        \item \(A \bisim B\)
        \item \(A'(x) \bisim B'(x)\) for all message \(x\), where \(A = \InP c x. A'(x)\) and \(B = \InP c x. B'(x)\)
        \item \(P_i \bisim P_i'\), where:
        \begin{align*}
          P_i & = \process {\multi {\OutP c {t_{i+1}} \ldots \OutP c {t_p}}} {\{\ax_1 \mapsto t_1, \ldots, \ax_i \mapsto t_i\}}\\
          P_i' & = \process {\multi {\OutP c {t_{i+1}'} \ldots \OutP c {t_p'}}} {\{\ax_1 \mapsto t_1', \ldots, \ax_i \mapsto t_i'\}}
        \end{align*}
        and where $\{\ax_1 \mapsto t_1, \ldots, \ax_p \mapsto t_p\}
        \StatEq
        \{\ax_1 \mapsto t_1', \ldots, \ax_p \mapsto t_p'\}$
      \end{enumerate}
      It easily follows from Lemma \ref{lem:nexp st-equiv},\ref{prop:correctness checktree},\ref{prop:correctness checksat} that \(\bisim\) is a bisimulation up to \(\silentstep\), hence the conclusion. \qedhere
    \end{itemize}
  \end{proof}

\end{document}